\documentclass[11pt,a4paper]{article}
\pdfoutput=1

\usepackage{jheppub}
\makeatletter
\def\@fpheader{}
\makeatother

\usepackage[export]{adjustbox}
\usepackage{amsthm}
\usepackage{amsfonts}
\usepackage{bm}
\usepackage{dsfont}
\usepackage[titletoc]{appendix}
\usepackage{url}
\usepackage{braket}
\usepackage{qcircuit}
\usepackage{here}
\usepackage{mathrsfs}
\usepackage{comment}
\usepackage{tikz-cd}
\usepackage{accents}
\usepackage{enumitem}
\usepackage{array,longtable}
\usepackage{setspace}

\newtheorem*{assumption*}{Assumption}

\newtheorem{theorem}{Theorem}[section]
\newtheorem{corollary}[theorem]{Corollary}
\newtheorem{lemma}[theorem]{Lemma}
\newtheorem{proposition}[theorem]{Proposition}

\theoremstyle{definition}
\newtheorem{example}[theorem]{Example}
\newtheorem{definition}[theorem]{Definition}
\newtheorem*{definition*}{Definition}
\newtheorem{defn-prop}{Definition-Proposition}[section]
\newtheorem{axiom}[theorem]{Axiom}

\def\tr{{\mathrm{tr}}}
\def\str{{\mathrm{str}}}

 \def\CC{{\cal C}}

 \def\CO{{\cal O}}
 
 \def\CH{{\cal H}}

\def\be{\begin{equation}}
\def\ee{\end{equation}}
\def\ba{\begin{eqnarray}}
\def\ea{\end{eqnarray}}

\numberwithin{equation}{section}

\newcommand{\kket}[1]{\ket{#1 \rangle}}
\newcommand{\bbra}[1]{\bra{\langle #1}}
\newcommand{\bbrakket}[1]{\braket{\braket{#1}}}

\newcommand{\HBU}{{\mathcal{H}_\mathrm{BU}}}
\newcommand{\HBUpre}{{\mathcal{H}_\mathrm{BU}^\mathrm{pre}}}
\newcommand{\CBU}{{\mathcal{C}_\mathrm{BU}}}
\newcommand{\CBUpre}{{\mathcal{C}_\mathrm{BU}^\mathrm{pre}}}
\newcommand{\CBUHS}{{\mathcal{C}_\mathrm{BU}^\mathrm{HS}}}

\newcommand{\CBUvN}{{\mathcal{C}_\mathrm{BU}^\mathrm{vN}}}
\newcommand{\CBUvNL}{{\mathcal{C}_\mathrm{BU}^{\mathrm{vN},L}}}
\newcommand{\CBUvNR}{{\mathcal{C}_\mathrm{BU}^{\mathrm{vN},R}}}
\newcommand{\CBUfd}{{\mathcal{C}_\mathrm{BU}^\mathrm{f.d.}}}
\newcommand{\CBUold}{{\mathcal{C}_\mathrm{BU}^\mathrm{old}}}
\newcommand{\defined}{\stackrel{\mathrm{def}}{=}}
\newcommand{\Bord}{\mathrm{Bord}}

\newcommand{\BordX}{{\mathrm{Bord}_d^\mathcal{X}}}

\newcommand{\Hilb}{\mathrm{Hilb}}
\newcommand{\sHilb}{\mathrm{sHilb}}
\renewcommand{\Vec}{\mathrm{Vec}}
\newcommand{\sVec}{\mathrm{sVec}}

\newcommand{\cellmath}[1]{%
  \par\smallskip
  \noindent\makebox[\linewidth][c]{\(\displaystyle #1\)}
  \smallskip%
}

\title{Wormholes as red herrings: reflection positivity and the reconstruction of unitary quantum field theories}

\author[a]{Jacob McNamara}
\author[b]{and Zhencheng Wang}

\affiliation[a]{Simons Center for Geometry and Physics, Stony Brook University,\\
Stony Brook, NY 11794, USA}

\affiliation[b]{Department of Physics, University of Illinois Urbana-Champaign,\\
Urbana, IL 61801, USA}

\emailAdd{jmcnamara@scgp.stonybrook.edu}
\emailAdd{zcwang1@illinois.edu}

\abstract{As Coleman famously argued, the apparent breakdown of partition-function factorization in quantum gravity associated with Euclidean wormholes is a red herring, arising from a hidden average over an ensemble of theories. We present a direct analog of Coleman's argument for the apparent breakdown of Hilbert-space factorization associated with spatial wormholes, i.e., Einstein--Rosen bridges. Our main result is the following reconstruction theorem for quantum field theories: unitary QFTs are determined, up to unitary isomorphism, by their closed-manifold partition functions; every reflection-positive partition function arises from a unitary quantum field theory; and the states prepared by manifolds span the space of invariant states under the reconstructed theory's symmetry group. Interpreting the result gravitationally, we conclude that any apparent breakdown of Hilbert-space factorization is a red herring, arising from restricting to an incomplete spectrum of charged states.
}

\begin{document}

\maketitle

\section{Introduction}

Quantum gravity, though as yet undefined, is meant to be a quantum-mechanical theory of dynamical spacetime manifolds. In general, quantum mechanics instructs us to compute probability amplitudes by summing the amplitudes for each possible history. This perspective has proven extremely fruitful in quantum field theory (QFT), as codified in the Feynman path integral over quantum fields. By analogy, it seems very natural to attempt to define amplitudes in quantum gravity by summing over all possible spacetime manifolds in a gravitational path integral (GPI).

The GPI has come in and out of popularity over the years, with a recent surge of interest following the solution of Jackiw-Teitelboim (JT) gravity in terms of an explicit evaluation of the GPI \cite{Saad:2019lba}. We have learned that the GPI is smarter than previously thought: not only can it calculate the Bekenstein--Hawking entropy of black holes \cite{Gibbons:1976ue}; it can even reproduce the Page curve \cite{Page:1993wv} for evaporating black holes once replica wormholes are included \cite{Penington:2019npb,Almheiri:2019psf,Penington:2019kki,Almheiri:2019qdq,Almheiri:2019hni}. Thus, the GPI seems to give answers entirely consistent with unitary black hole evaporation, as expected in ultraviolet (UV) complete quantum gravity.

However, despite its successes, the GPI has a severe flaw compared to the familiar path integral of QFT. In the QFT path integral, the fluctuating quantum fields are local degrees of freedom living on top of spacetime. These local degrees of freedom may be cut and glued subject to a division of spacetime into coordinate patches, which makes the locality of the resulting QFT manifest. Mathematically, this cutting and gluing property is encoded by the statement that the dynamical fields in the QFT path integral form a \textit{sheaf}.

In contrast, the fluctuating field in the GPI is a bulk spacetime manifold pinned at its boundary to some region held fixed as a source. Given a division of the boundary into pieces, there is no canonical division of the bulk. This is especially true outside of the semiclassical regime, where we might hope to apply the Ryu--Takayanagi/Hubeny--Rangamani--Takayanagi \cite{Ryu:2006bv,Hubeny:2007xt} or quantum extremal surface \cite{Faulkner:2013ana,Engelhardt:2014gca} prescription, and associate each boundary region to its \textit{entanglement wedge} in the bulk. However, even in the semiclassical regime, the fields in the GPI do not form a sheaf. This is because the entanglement wedge of a disjoint union may be larger than the disjoint union of the entanglement wedges, containing what is known as an \textit{entanglement island} \cite{Almheiri:2019psf,Penington:2019npb,Almheiri:2019qdq,Penington:2019kki,Almheiri:2019hni}. The failure of the fields in the GPI to form a sheaf is a feature, not a bug, as the appearance of entanglement islands is essential in the recent successes of the GPI discussed above.

Consequently, these recent developments have raised a new tension, known as the \textit{Factorization Paradox}, arising from our strong belief in the holographic principle \cite{tHooft:1993dmi,Susskind:1994vu, Bousso:2002ju}. The \textit{holographic principle} asserts that UV-complete quantum gravity must admit a dual description in terms of local, non-gravitational degrees of freedom living on the boundary of spacetime. There is no tension between holography and the unitarity of black hole evaporation, and indeed, the two are intimately linked. The tension, instead, is between the locality of a holographic dual and the nonlocal nature of the GPI. This tension is most severe when the bulk contains a wormhole connecting two otherwise-disjoint components of the boundary. Holography tells us that the data associated to a disjoint boundary must factorize, but a wormhole leads to a connected, non-factorizing contribution to the GPI.

The Factorization Paradox is most commonly viewed as an issue arising from Euclidean (spacetime) wormholes, meaning wormholes whose Euclidean cross section is codimension-one in the bulk \cite{Hawking:1987mz,Giddings:1987cg,Hawking:1988ae,Coleman:1988cy,Giddings:1988cx,Giddings:1988wv,Witten:1999xp,Maldacena:2004rf,Arkani-Hamed:2007cpn,Marolf:2020xie}. A Euclidean wormhole describes the propagation of a baby (spatially closed) universe from one boundary component to another. This propagation of baby universes threatens to spoil the factorization of multi-boundary partition functions,
\begin{equation}\label{eq:factorized_Z}
    \zeta(M_1 \sqcup M_2) \stackrel{?}{=} \zeta(M_1)\zeta(M_2).
\end{equation}
Partition-function factorization is, by now, well understood: the necessary and sufficient condition for partition-function factorization is that the Hilbert space $\HBU$ of baby universe states is trivial \cite{Almheiri:2019qdq,Penington:2019kki,Marolf:2020xie,McNamara:2020uza,Banerjee:2022pmw,Usatyuk:2024mzs,Usatyuk:2024isz,Harlow:2025pvj,Abdalla:2025gzn,Harlow:2026hky,Zhao:2026mpl},
\begin{equation}\label{eq:hbu_equals_C}
    \HBU = \mathbb{C},
\end{equation}
so that there is no non-trivial propagation of baby universes.

Moreover, we understand what happens when partition-function factorization breaks down. As classically argued by Coleman \cite{Coleman:1988cy}, and subsequently expanded on by Giddings and Strominger \cite{Giddings:1988cx,Giddings:1988wv} as well as Marolf and Maxfield \cite{Marolf:2020xie}, this breakdown is always a red herring. When the GPI produces non-factorizing multi-boundary partition functions $\zeta$, they may always\footnote{\label{footnote:moment}Up to an unresolved moment problem (as discussed in, e.g., \cite{Marolf:2024jze,Dong:2024tjx,Chen:2025fwp}).} be reinterpreted as the ensemble average of individually-factorizing partition functions $\zeta_\alpha$ over a measure space of so-called \textit{$\alpha$-sectors},
\begin{equation}\label{eq:quenched_average_intro}
    \zeta(M_1 \sqcup \cdots \sqcup M_n) = \int d\alpha\ \zeta_\alpha(M_1) \cdots \zeta_\alpha(M_n).
\end{equation}
The original breakdown of partition-function factorization is reinterpreted as the presence of statistical correlations in the underlying ensemble. One may hope that by projecting to a fixed $\alpha$-sector, one could recover the partition function of a single UV-complete theory of quantum gravity.

However, partition-function factorization is not the end of the story. The factorization \eqref{eq:factorized_Z} of the partition function only ensures that the resulting quantum system is boundary-local in a very weak sense. In QFT, locality imposes a full tower of constraints, corresponding to the factorization of higher and higher categorical data associated to higher and higher codimension cuts of spacetime. Correspondingly, the Factorization Paradox consists of a full tower of possible breakdowns of factorization, which quantify the full tension between the GPI and the holographic principle.

In this paper, we move one step up this tower, and study the factorization of multi-boundary Hilbert spaces,
\begin{equation}\label{eq:factorizing_H}
    \mathcal{H}_{B_1 \sqcup B_2} \stackrel{?}{=} \mathcal{H}_{B_1} \otimes \mathcal{H}_{B_2}.
\end{equation}
Factorization of Hilbert spaces is certainly threatened by Euclidean wormholes, as the states on both sides are constrained to lie in the same $\alpha$-sector. However, projecting onto a fixed $\alpha$-sector does not fully resolve Hilbert-space factorization. This is because of the possibility of spatial wormholes, known as \textit{Einstein--Rosen (ER) bridges}, which connect disjoint spatial slices of the boundary through their bulk codimension-two horizon.

Hilbert-space factorization is less well understood than partition-function factorization. We expect that, in a UV-complete theory, Hilbert-space factorization is resolved by the ER = EPR\footnote{``EPR'' stands for Einstein--Podolsky--Rosen, who famously discussed quantum entanglement \cite{Einstein:1935rr}.} proposal \cite{Maldacena:2001kr,VanRaamsdonk:2010pw,Maldacena:2013xja}, so that every ER bridge is equivalent to an entangled state in a tensor-factorized Hilbert space.\footnote{Our notion of ER = EPR mainly concerns the direction ER $\Rightarrow$ EPR, and does not require the converse direction except in some very general, formal sense.} Until now, we have not had a clear characterization of when ER = EPR must hold, analogous to the necessary and sufficient condition \eqref{eq:hbu_equals_C} for partition-function factorization. Hilbert-space factorization has been studied in a number of specific contexts \cite{Harlow:2015lma,Harlow:2018tqv,Penington:2023dql,Kolchmeyer:2023gwa,Chua:2023ios,Boruch:2024kvv,Balasubramanian:2024yxk,Li:2024nft,Banerjee:2024fmh,Balasubramanian:2025zey,Balasubramanian:2025jeu}, with the predominant overarching theme being that Hilbert spaces seem to factorize once enough degrees of freedom are included.

While we have not known a sufficient condition, we have certainly known a necessary condition. As first observed by Harlow \cite{Harlow:2015lma}, ER = EPR can never hold if our bulk theory has a gauge symmetry $G$ with an incomplete spectrum of states \cite{Polchinski:2003bq,Banks:2010zn}. For instance, if the bulk theory contains a $U(1)$ gauge field without electrically charged matter, then a Reissner--Nordstr\"{o}m wormhole threaded by electric flux cannot possibly be identified with an entangled superposition of one-sided states. This necessary condition forms a motivational basis for related Swampland constraints including the absence of global 1-form symmetries \cite{Harlow:2018tng,Rudelius:2020orz,Heidenreich:2021xpr,McNamara:2021cuo,Cordova:2022rer} and the triviality of cobordism classes \cite{McNamara:2019rup}.

\subsection{Summary of results}\label{sec:summary_of_results}

In this paper, we provide a precise characterization of the breakdown of Hilbert-space factorization, analogous to \eqref{eq:hbu_equals_C}, working within the developing axiomatic framework for the GPI \cite{Marolf:2020xie,Gesteau:2020wrk,Colafranceschi:2023txs,Colafranceschi:2023urj,Marolf:2024adj,Chen:2025fwp,DiUbaldo:2026rly,Zhao:2026mpl}. In this axiomatic framework (reviewed in Section \ref{sec:axiomatic_approach_philosophy}), we fix a class $\mathcal{X}$ of manifolds with sources, viewed as the boundary conditions of the GPI, as well as a reflection-positive, factorizing partition function $\zeta$, viewed as the result of projecting said GPI onto a fixed $\alpha$-sector.

With $\mathcal{X}$ and $\zeta$ fixed, we construct a mathematically well-defined invariant, the \textit{baby universe category} $\CBU$, such that ER = EPR holds if and only if the baby universe category is trivial,
\begin{equation}\label{eq:triviality_of_CBU}
    \CBU = \Hilb.
\end{equation}
Here, $\Hilb$ denotes the category of complex Hilbert spaces and bounded linear maps.\footnote{One might have expected the constraint to be $\CBU = \sHilb$, the category of complex super-Hilbert spaces, or $\CBU = \Hilb_\mathbb{R}$, the category of real Hilbert spaces, motivated by \cite{Harlow:2023hjb}. For the definition of $\CBU$ used in most of this paper, \eqref{eq:triviality_of_CBU} is the correct constraint. See Section \ref{sec:fermionic_sources} for further discussion.} Heuristically, the objects of $\CBU$ are generated, under direct sum and projection onto subobjects, by the set of closed codimension-two slices of the bulk. The morphisms are quantum superpositions of bordisms modulo gravitational constraints. Our construction of $\CBU$ is directly motivated and enabled by the previous work \cite{Colafranceschi:2023urj, Marolf:2024adj} of Colafranceschi, Dong, Marolf, Zhang, and the second author of this paper.

When Hilbert-space factorization breaks down, our definition of $\CBU$ allows us to run a direct analog of Coleman's argument \cite{Coleman:1988cy} at one higher category level. The result is the following main theorem, which, while motivated by quantum gravity, is a mathematically rigorous result in formal QFT alone.

\begin{theorem}[Main Theorem]\label{thm:main}
    Let $\zeta$ be a Euclidean partition function defined on a class $\mathcal{X}$ of $d$-manifolds, possibly with sources, that satisfies Axioms \ref{axiom:finiteness}-\ref{axiom:reflection_positivity}. In particular, $\zeta$ need only be known on closed $\mathcal{X}$-manifolds, and must be finite, must factorize on disjoint unions, and must be reflection positive. Then the following are true:
    \begin{enumerate}[label=(\alph*), font=\normalfont\bfseries]
        \item There exists a unitary functorial QFT (Definition \ref{defn:functorial_QFT}), defined on the same class of $d$-manifolds, with partition function $\zeta$.

        \item Any two such QFTs are unitarily isomorphic.

        \item For such a QFT, let $G$ denote the compact Hausdorff group implementing the maximal unitary symmetry under which all sources in $\mathcal{X}$ are neutral. For any spatial manifold $B$, the closed span of the set of states preparable by Euclidean $\mathcal{X}$-manifolds with boundary $B$ is precisely the $G$-invariant subsector of the full Hilbert space of states on $B$.
    \end{enumerate}
\end{theorem}

\noindent Stated less formally, parts (a) and (b) of Theorem \ref{thm:main} assert that:
\begin{center}
\bfseries\itshape
Unitary quantum field theories are determined \\ by their closed-manifold partition function.
\end{center}
This is true both in the sense that a known unitary QFT can be reconstructed from its closed-manifold partition function, and in the stronger sense that an abstract closed-manifold partition function, known to be reflection positive, automatically produces a unitary QFT. Thus, the closed-manifold partition function of a unitary QFT is not a coarse invariant at all, as is sometimes stated; it is, instead, a complete invariant. Theorem \ref{thm:main} generalizes the perspective of Freed and Hopkins \cite{Freed:2016rqq} from unitary \textit{invertible} QFTs to arbitrary unitary QFTs. Let us caution that our notion of unitary functorial QFT, given in Definition \ref{defn:functorial_QFT}, includes only \textit{compact} QFTs, whose Euclidean closed-manifold partition function is finite.

The class $\mathcal{X}$ of manifolds can be taken to include sources if one wants, but can always be restricted to the bare minimum needed to define the QFT. For instance, we may take $\mathcal{X}$ to consist of unadorned Riemannian manifolds with whatever tangential structure (orientation, spin structure, etc.) is required. The entire collection of Hilbert spaces and Euclidean evolution operators are then uniquely reconstructed. One would only be forced to include sources in order to distinguish theories, such as symmetry-protected topological (SPT) or symmetry-enriched topological (SET) phases, which become trivial or indistinguishable when one forgets the coupling to background fields.

As a consequence of Theorem \ref{thm:main}, we are able to close a long-standing conjecture of Friedan and Shenker \cite{Friedan:1986ua} in the affirmative: unitary two-dimensional conformal field theories (CFTs) are determined by their all-genus partition function.\footnote{Theorem \ref{thm:main} also likely resolves a number of conceptual descendants of Friedan and Shenker's conjecture, such as \cite{Gaberdiel:2009rd,Gaberdiel:2010jf,Codogni:2019sub,Carpi:2026bfh}, though Theorem \ref{thm:main}, as stated, does not apply to gravitationally anomalous theories, such as two-dimensional chiral CFTs, which do not have a well-defined closed-manifold partition function.} This is because, for a CFT, all correlation functions of local operators may be regarded as matrix elements of Euclidean evolution operators via the state-operator correspondence. In the language of the CFT bootstrap, Theorem \ref{thm:main} states that imposing reflection positivity of the closed-manifold partition functions, for all manifolds at once, produces a complete set of bootstrap constraints, in any dimension. Theorem \ref{thm:main} also fits very well with the result \cite{Wen:2019ylt} of Wen and Wen, that certain $(2+1)$-dimensional topological orders (i.e., modular tensor categories), with identical modular data, are nevertheless distinguished by their partition functions on higher-genus mapping tori.

Applied to quantum gravity, parts (a) and (b) of Theorem \ref{thm:main} show that any breakdown of Hilbert-space factorization is, ultimately, a red herring, in the sense of Coleman \cite{Coleman:1988cy}. While the non-factorizing Hilbert spaces of states produced by cutting open the GPI are not the full Hilbert spaces of any holographic theory, they can always be embedded in the Hilbert spaces of an underlying boundary-local QFT, uniquely up to (non-unique) unitary isomorphism. In a precise sense (see Section \ref{sec:higher_averaging}), choosing such an embedding is directly analogous to projecting a non-factorizing partition function onto an $\alpha$-sector.

Part (c) of Theorem \ref{thm:main} provides us with a natural reinterpretation of the original breakdown of Hilbert-space factorization. Informally, part (c) states that a state can be prepared by manifolds if and only if it is neutral under the largest possible symmetry group. Equivalently, a state \textit{cannot} be prepared by manifolds if and only if it is \textit{charged} under a symmetry group. This is true both when we include no sources beyond the manifold itself, and also when we include any collection of sources we like. In particular, if we include enough sources to break every symmetry of our QFT, then every state can be prepared by manifolds with sources. Part (c) of Theorem \ref{thm:main} holds both for states on connected spatial slices and also, crucially, for states on disconnected spatial slices.

Interpreted in unitary theories of quantum gravity, we learn that completeness of the spectrum of charged states is not merely a \textit{necessary} condition for ER = EPR, as shown by Harlow \cite{Harlow:2015lma}, but is actually a \textit{sufficient} condition as well. Thus, we show that:
\begin{center}
\bfseries\itshape
Given unitarity, completeness of the \\ spectrum is equivalent to ER = EPR.
\end{center}
An important caveat, discussed in Section \ref{sec:bulk_interpretation_of_symmetry_group}, is that the required notion of completeness of the spectrum is completeness with respect to the total global symmetry group of the microscopic holographic dual. This group may or may not be visible in terms of gauge fields in a GPI or effective field theory (EFT) description we may have started with.
 
As a result, whenever partition-function factorization holds, but Hilbert-space factorization breaks down, we obtain the following reinterpretation: the non-factorizing multi-boundary Hilbert spaces of states preparable by the GPI are just the neutral subsector,
\begin{equation}\label{eq:quenched_gauging_intro}
    \mathcal{H}_{B_1 \sqcup \cdots \sqcup B_n} = \left( \widetilde{\CH}_{B_1} \otimes \cdots \otimes \widetilde{\CH}_{B_n} \right)^G,
\end{equation}
in the factorizing Hilbert spaces $\widetilde{\CH}_B$ in a unique underlying QFT. This neutral subsector would not be expected to factorize in the first place, because not every neutral state in a tensor product is the tensor product of neutral states.

The expression \eqref{eq:quenched_gauging_intro} for the multi-boundary Hilbert spaces produced by the GPI is directly analogous to the expression \eqref{eq:quenched_average_intro} for multi-boundary partition functions. In \eqref{eq:quenched_average_intro}, we write the multi-boundary partition functions as an average over $\alpha$-sectors. More specifically, it is a \textit{quenched} average, meaning that we first multiply the partition functions for each boundary and only take an average at the end. Similarly, \eqref{eq:quenched_gauging_intro} is a higher-categorical form of quenched averaging, in that we first tensor the Hilbert spaces for each boundary and only project onto $G$-invariant states at the end.\footnote{Let us contrast this operation with gauging the symmetry $G$, which can be viewed as a higher categorical \textit{annealed} average. First of all, when we gauge a symmetry, we project onto $G$ invariants on each spatial component independently. Secondly, we must also include twisted sectors, which has no analog in \eqref{eq:quenched_gauging_intro}.} This is a form of averaging both in the straightforward sense that the projector onto $G$-invariants can be defined by group-averaging, and in the more formal (but precise) sense discussed in Section \ref{sec:higher_averaging}. The structure \eqref{eq:quenched_gauging_intro} was previously observed by Maxfield \cite{Maxfield:2023mdj} in a beautiful toy example, reviewed in Section \ref{sec:Ex_TopoQM}, which served as a major point of inspiration for this paper.

For the majority of this paper, we discuss the issue of Hilbert-space factorization alone, assuming that partition-function factorization has already been resolved. Ultimately, however, it is more natural to consider both issues at once, as we discuss in Section \ref{sec:including_alpha_sectors}. In such a case, with a nontrivial ensemble of $\alpha$-sectors, we must generalize \eqref{eq:quenched_gauging_intro} to the direct integral,
\begin{equation}\label{eq:quenched_gauging_intro_with_alpha_sectors}
    \mathcal{H}_{B_1 \sqcup \cdots \sqcup B_n} = \int^\oplus d\alpha\ \left( \widetilde{\CH}_{B_1}^\alpha \otimes \cdots \otimes \widetilde{\CH}_{B_n}^\alpha \right)^{G_\alpha},
\end{equation}
of the $G_\alpha$-invariant sectors in the ensemble of underlying $\alpha$-sector QFTs.

\subsection{Overview of paper}\label{sec:overview}

At a high level, our paper is structured as follows. Sections \ref{sec:axiomatic_approach_philosophy}, \ref{sec:categories_of_sources}, \ref{sec:axioms_for_partition_function}, and \ref{sec:spaces_of_grav_states} form a review and formalization of the emerging axiomatic approach to the GPI, which provides the formal context in which we prove Theorem \ref{thm:main}. Its main output is the family of Hilbert spaces $\mathcal{H}_B$ of states prepared by manifolds, obtained using reflection positivity via a form of the Osterwalder--Schrader (OS)/Gelfand--Naimark--Segal (GNS) construction. Section \ref{sec:examples} illustrates Theorem \ref{thm:main} in a number of examples, including some low-dimensional theories where we can directly prove Theorem \ref{thm:main} by hand. Sections \ref{sec:CBU}, \ref{sec:fusion}, and \ref{sec:reconstruction} comprise the technical core of this paper, in which we construct the baby universe category $\CBU$ and use it to prove Theorem \ref{thm:main}. Finally, with Theorem \ref{thm:main} proven, Section \ref{sec:robustness_extensions} includes stress tests and extensions of our framework, and Section \ref{sec:conclusions_future_directions} concludes our discussion and points towards future directions.

Let us now describe our proof of Theorem \ref{thm:main}. As motivation, we briefly recall the details of Coleman's argument \cite{Coleman:1988cy} showing that the path integral over Euclidean wormholes may be interpreted as producing a classical ensemble over $\alpha$-sectors. From the modern perspective put forward by Marolf and Maxfield \cite{Marolf:2020xie}, this argument is based on understanding the structure of the baby universe Hilbert space $\HBU$, and the von Neumann algebra $\mathcal{A}_\mathrm{BU}$ \cite{Gesteau:2020wrk} acting on it. Both $\mathcal{A}_\mathrm{BU}$ and $\HBU$ are spanned by the collection of closed codimension-one slices $M$ of the bulk, with the algebra structure on $\mathcal{A}_\mathrm{BU}$ induced by disjoint union of manifolds. Linear combinations of geometric manifolds are subject to the Hamiltonian constraint, which is implemented at the level of $\HBU$ by a quotient by null states.

The key observation of Coleman's \cite{Coleman:1988cy} is that the disjoint union of manifolds is commutative, and so $\mathcal{A}_\mathrm{BU}$ is a commutative algebra. As a result, the quantum mechanics of baby universe classicalizes. More precisely, $\mathcal{A}_\mathrm{BU}$ and $\HBU$ can be identified with the spaces of bounded and square-integrable functions, respectively, on the spectrum $\mathrm{Spec}(\mathcal{A}_\mathrm{BU})$. This spectrum is canonically equipped with a measure induced by the GPI (or, equivalently, by the Hartle--Hawking no-boundary wavefunction \cite{Hartle:1983ai}). The points of $\mathrm{Spec}(\mathcal{A}_\mathrm{BU})$ are, precisely, the $\alpha$-sectors discussed above. Each $\alpha$-sector corresponds to a normal $*$-algebra homomorphism $\zeta_\alpha : \mathcal{A}_\mathrm{BU} \to \mathbb{C}$, which, when evaluated on a closed manifold $M$, produces a factorizing partition function $\zeta_\alpha(M)$.\footnote{Note that the measure space $\mathrm{Spec}(\mathcal{A}_\mathrm{BU})$ only admits a complete set of global points when it is discrete. When $\mathrm{Spec}(\mathcal{A}_\mathrm{BU})$ is diffuse, it will not actually be possible to project onto a single $\alpha$-sector, though it will still be possible to restrict to subsets of $\alpha$-sectors of non-zero measure.}

Our proof of Theorem \ref{thm:main} is a direct categorification of Coleman's argument. While $\mathcal{A}_\mathrm{BU}$ is the von Neumann algebra generated by closed codimension-one slices of the bulk, the baby universe category $\CBU$ is the monoidal category generated by closed codimension-two slices of the bulk. In more detail, the objects of $\CBU$ are given by formal direct sums of bulk codimension-two closed manifolds, viewed as higher-codimension analogs of baby universes. The morphisms in $\CBU$ are quantum superpositions of bulk codimension-one bordisms, and thus $\CBU$ is a quantum bordism category, generalizing the bordism groups of quantum gravity \cite{McNamara:2019rup}. These bulk codimension-one bordisms should be viewed as partial Cauchy slices of the bulk which evolve a codimension-two slice radially. Just as $\mathcal{A}_\mathrm{BU}$ is an algebra under disjoint union of manifolds, the category $\CBU$ is monoidal, again under disjoint union.

Both the morphisms and objects of $\CBU$ are subject to gravitational constraints. At the level of morphisms, this is simply the bulk Hamiltonian constraint, which is imposed by a quotient by null states as in the baby universe Hilbert space $\HBU$. At the level of objects, we have a higher-categorical version of the Hamiltonian constraint, enforcing that the bulk radial evolution of codimension-two slices must be pure gauge. Formally, this constraint is imposed by passing to an idempotent completion, as explained in Section \ref{sec:idempotent_completion}. We view the objects in $\CBU$ as possible horizons for ER bridges, and so the nontrivial objects in $\CBU$ represent the most general sort of unbreakable flux which may flow through an ER bridge and obstruct ER = EPR.

The specific categorical formalism we use in this paper is the framework of Cauchy complete W*-categories, which provides a natural setting for categorical quantum mechanics. We refer the reader to \cite{henriques2024completewcategories} for a review of the basic theory of W*-categories and a clear development of their role as categorical analogs of Hilbert spaces. A W*-category is an operator-algebraic category, all of whose endomorphism algebras are von Neumann algebras (otherwise known as W*-algebras). W*-categories are Cauchy complete, in the sense of \cite{henriques2024completewcategories}, when they admit arbitrary orthogonal direct sums of objects as well as images for every orthogonal projection operator. The baby universe category is a Cauchy complete W*-category, which is moreover a W*-tensor category\footnote{Our conventions for the term \textit{W*-tensor category} differ from \cite{henriques2024completewcategories}. For us, a W*-tensor category is required to be Cauchy complete and have a simple tensor unit. When the condition on the tensor unit is dropped, we will say we have a \textit{W*-multitensor category}.} under the monoidal structure induced by disjoint union. We reassure readers unfamiliar with the general theory of von Neumann algebras that $\CBU$ is, crucially, an \textit{atomic} W*-category, meaning each endomorphism algebra is a direct sum of type I factors and nothing else.

Just as the von Neumann algebra $\mathcal{A}_\mathrm{BU}$ is commutative, the W*-tensor category $\CBU$ is symmetric, again because the disjoint union of manifolds is commutative (in the strongest sense). We saw that the existence of $\alpha$-sectors arises from the basic fact that commutative von Neumann algebras always classicalize, reducing quantum mechanics to classical probability theory. Analogously, symmetric W*-tensor categories $\mathcal{C}$ always classicalize as,
\begin{equation}\label{eq:classicalization_of_category}
    \mathcal{C} = \mathrm{Rep}(G),
\end{equation}
for some compact group $G$.\footnote{Up to an important subtlety involving fermionic statistics, discussed in Section \ref{sec:reconstruction}.} This is true, at least, when they satisfy an appropriate finiteness condition, which we refer to as \textit{rigid generation} and define in Section \ref{sec:rg_W*_tensor}. We refer to \eqref{eq:classicalization_of_category} as ``classicalization,'' as it has reduced the quantum mechanical possibility of non-invertible symmetry to the representation theory of a classical (i.e., not quantum) group.

The classicalization \eqref{eq:classicalization_of_category} of rigidly generated symmetric W*-tensor categories is described in general by the framework of Tannakian reconstruction, and more specifically, by the Doplicher--Roberts (DR) reconstruction theorem \cite{DoplicherRoberts89,muger2007abstract}. The DR reconstruction theorem forms the technical core of our argument, and Theorem \ref{thm:main} follows immediately from DR reconstruction once we construct $\CBU$ and show that it is rigidly generated, by leveraging the trace inequality of \cite{Colafranceschi:2023urj}.

For readers familiar with the theory of Doplicher--Haag--Roberts (DHR) superselection sectors \cite{Doplicher:1969fvo,Doplicher:1969fvg,Doplicher:1971wk,Doplicher:1973at} in algebraic QFT (see \cite{Casini:2019kex,Casini:2020rgj,Casini:2021zgr,Benedetti:2022zbb,Benedetti:2024dku,Casini:2025lfn,Shao:2025mfj,Harlow:2025cqc} for recent discussion), we note that our main argument is simply an application of the philosophy of DHR superselection to the superselection sectors of the GPI defined by fixing the boundary of a spatial slice. These sectors break up into irreducible sectors, which are precisely the $\mu$-sectors of \cite{Colafranceschi:2023urj,Marolf:2024adj}. Thus, the superselection theory of boundaries of spatial slices in the GPI behaves like DHR superselection theory in the stable range of $d \geq 4$ spacetime dimensions.

In detail, this paper is organized as follows. In Section \ref{sec:axiomatic_approach_philosophy}, we describe the broad philosophy behind the developing axiomatic approach to quantum gravity. In Section \ref{sec:categories_of_sources}, we describe the formal structures needed on the class $\mathcal{X}$ of manifolds, possibly with sources, which serve as formal boundary conditions for the GPI. These structures are best phrased in terms of a non-unital geometric bordism category, $\BordX$, equipped with structures encoding unitarity. In Section \ref{sec:axioms_for_partition_function}, we review the axioms placed on the partition function $\zeta$, following \cite{Colafranceschi:2023urj}. We also provide a definition for unitary functorial QFTs in the style of Atiyah--Segal \cite{Atiyah1988TQFT,segal1988definition}, roughly following the approach of Kontsevich--Segal \cite{Kontsevich:2021dmb}.

In Section \ref{sec:spaces_of_grav_states}, we use a form of the GNS construction to define Hilbert spaces of states formally produced by manifolds, which we view as the natural Hilbert spaces obtained by cutting open the GPI. This construction is known in the TQFT literature, independently, as the \textit{universal construction} of \cite{Blanchet:1995TQFT}, and we refer to the constructed Hilbert spaces as the \textit{universal Hilbert spaces}. In this axiomatic framework, we state a precise formalization of the ER = EPR conjecture \cite{Maldacena:2001kr,VanRaamsdonk:2010pw,Maldacena:2013xja}, and set up the central questions of this paper: when do the Hilbert spaces produced from the GPI factorize, and if they do not, what is going on? In Section \ref{sec:examples}, we illustrate the answers provided by Theorem \ref{thm:main} in a number of explicit examples, including in particular the toy model of Maxfield \cite{Maxfield:2023mdj}.

Our proof of Theorem \ref{thm:main} begins in earnest in Section \ref{sec:CBU}, where we construct the baby universe category $\CBU$. Our construction proceeds in a series of technical steps. First, we define the pre-baby universe category $\CBUpre$, the free linearization of $\BordX$. The pre-baby universe category admits a canonical family of representations on the universal Hilbert spaces, and we define a W*-category $\CBUvN$ by taking a categorical double commutant. Applying standard results from Tomita--Takesaki theory, we reproduce the main results of \cite{Colafranceschi:2023urj,Marolf:2024adj} in our language: that $\CBUvN$ is equipped with a faithful normal semifinite trace, which is uniformly bounded below on projectors, proving that $\CBUvN$ is atomic. We then construct the baby universe category $\CBU$ itself by taking the W*-Cauchy completion of $\CBUvN$.

In Section \ref{sec:fusion}, we equip the baby universe category $\CBU$ with the symmetric monoidal fusion induced from the disjoint union of manifolds, and prove that $\CBU$ is rigidly generated. Finally, in Section \ref{sec:reconstruction}, we apply the DR reconstruction theorem to $\CBU$, concluding our proof of Theorem \ref{thm:main}.

Let us note that this paper is intended, mainly, to address the issue of Hilbert-space factorization, with applications in formal/mathematical QFT serving as an alternative source of motivation and interpretation of our results. A follow-up paper \cite{M-JF-R_wip} by Johnson-Freyd, Reutter, and the first author of this paper will present our results in a mathematically rigorous context, and explore the full categorical generality in which they apply. For a preliminary report on this work in progress, see \cite{johnsonfreyd2026rigidfirm}.

As our proof involves many moving pieces, we provide the following roadmap of the various constructions and theorems we apply to go from a bordism category $\BordX$ and a partition function $\zeta$ to a proof of Theorem \ref{thm:main}.
\begin{equation*}
    \begin{tikzcd}[column sep={5.2em}]
        \BordX \arrow[r, "\text{Linearize}"] & \CBUpre \arrow[r, "{\text{GNS via $\zeta$}}"] & \CBUvN \arrow[r, "\stackrel{\text{W*-Cauchy}}{\text{complete}}"] & \CBU \arrow[r, "\text{DR + \cite{Colafranceschi:2023urj}}"] & \text{Theorem \ref{thm:main}.}
    \end{tikzcd}
\end{equation*}
Thus, $\CBUpre$ is obtained from $\BordX$ by free linearization, $\CBUvN$ is obtained from $\CBUpre$ by a double commutant in the GNS representation induced by $\zeta$, $\CBU$ is obtained from $\CBUvN$ by taking a W*-Cauchy completion, and finally Theorem \ref{thm:main} follows by applying the DR reconstruction theorem to $\CBU$, with the crucial finiteness condition supplied, ultimately, by the results of \cite{Colafranceschi:2023urj}.

We also provide the following table (inspired by \cite{henriques2024completewcategories}) outlining the analogy between our results and the story of Coleman \cite{Coleman:1988cy}, Giddings and Strominger \cite{Giddings:1988cx,Giddings:1988wv}, and Marolf and Maxfield \cite{Marolf:2020xie}.

{
\setlength{\LTleft}{0pt}
\setlength{\LTright}{0pt}
\renewcommand{\arraystretch}{1.25}

\begin{longtable}{%
  |p{\dimexpr0.5\textwidth - 2\tabcolsep - 0.5\arrayrulewidth\relax}|
  p{\dimexpr0.5\textwidth - 2\tabcolsep - 0.5\arrayrulewidth\relax}|%
}
\hline
\makebox[\linewidth][c]{\textbf{Partition-function factorization}}
&
\makebox[\linewidth][c]{\textbf{Hilbert-space factorization}}
\\
\hline
\hline
\endfirsthead

\hline
\endhead
\hline
For closed $d$-manifolds $M_1,M_2$:
\cellmath{\zeta(M_1 \sqcup M_2) \stackrel{?}{=} \zeta(M_1)\zeta(M_2).} 
&
For closed $(d-1)$-manifolds $B_1,B_2$:
\cellmath{\mathcal{H}_{B_1 \sqcup B_2} \stackrel{?}{=} \mathcal{H}_{B_1} \otimes \mathcal{H}_{B_2}.}
\\
\hline

Linear span of closed $d$-manifolds:
\cellmath{\HBUpre = \mathrm{Span}_\mathbb{C}\{\ket{M}\}.}
&
Linearization of bordism category:
\cellmath{\CBUpre = \mathbb{C}[\BordX].}

\\
\hline

Inner product:
\cellmath{\braket{M_1 | M_2} = \zeta(M_2 \sqcup \overline{M}_1).}
&
Hilbert space-valued inner product:
\cellmath{\bbrakket{B_1 | B_2} = \mathcal{H}_{B_2 \sqcup \overline{B}_1 }.}

\\
\hline

Partition functions are overlaps with Hartle--Hawking state:
\cellmath{\zeta(M) = \braket{\varnothing | M}.}
&
Universal Hilbert spaces are overlaps with Hartle--Hawking object:
\cellmath{\mathcal{H}_B = \bbrakket{\varnothing | B}.}
\\
\hline

Cauchy completion:
\cellmath{\HBU = \mathrm{Compl}(\HBUpre).}
&
Double commutant and completion:
\cellmath{\CBUvN = (\CBUpre)'', \quad \CBU = \Hilb(\CBUvN).}

\\
\hline

Quotient by null states:
\cellmath{\ket{\psi_1} = \ket{\psi_2},}
in $\HBU$, whenever:
\cellmath{\braket{M | \psi_1} = \braket{M | \psi_2}, \quad \forall \ket{M} \in \HBUpre.}
&
Isomorphism in Cauchy completion:
\cellmath{\Psi_1 \cong \Psi_2,}
in $\CBU$, whenever:
\cellmath{\mathcal{H}_{\Psi_1 \sqcup \overline{B}} \cong \mathcal{H}_{\Psi_2 \sqcup \overline{B}}, \quad \forall B \in \CBUpre,}
by the Yoneda Lemma.
\\
\hline

Partition functions factorize if and only if:
\cellmath{\HBU=\mathbb{C}.}
&
Hilbert spaces factorize if and only if:
\cellmath{\CBU=\mathrm{Hilb}.}
\\
\hline

$\mathcal{A}_\mathrm{BU}$ is a commutative algebra:
\cellmath{M_1 \sqcup M_2 = M_2 \sqcup M_1.}
&
$\CBU$ is symmetric monoidal,
\cellmath{B_1 \sqcup B_2 = B_2 \sqcup B_1.}
\\
\hline
Regular $*$-algebra homomorphism:
\cellmath{\zeta_\alpha : \mathcal{A}_\mathrm{BU} \to \mathbb{C},}
giving an $\alpha$-sector partition function.
&
Fiber functor:
\cellmath{\CBU \to \sHilb,}
giving an $\alpha$-sector QFT.
\\
\hline

Measure space of $\alpha$-sectors:
\cellmath{\mathrm{Spec}(\mathcal{A}_\mathrm{BU}).}
&
Stack of fiber functors:
\cellmath{\mathrm{Spec}(\CBU) \simeq BG.}
\\
\hline

Baby universe algebra classicalizes:
\cellmath{\mathcal{A}_\mathrm{BU} = L^\infty\big(\mathrm{Spec}(\mathcal{A}_\mathrm{BU})\big)}
&
Baby universe category classicalizes:
\cellmath{\CBU = \sHilb\big(\mathrm{Spec}(\CBU)\big) \simeq \mathrm{sRep}(G)}
\\
\hline

Partition functions are ensemble averages:
\cellmath{\zeta(M) = \int d\alpha\ \zeta_\alpha(M)}
&
Universal Hilbert spaces are $G$-invariants:
\cellmath{\mathcal{H}_B = \widetilde{\mathcal{H}}_B^G = \int^\oplus_{BG} \widetilde{\mathcal{H}}_B}
\\
\hline

\end{longtable}
}
\vspace{1em}

\begin{figure}[h]
    \centering
    \includegraphics[width=0.9\linewidth]{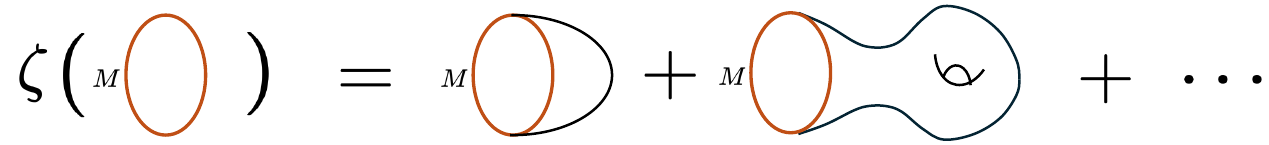}
    \caption{Schematically, the GPI computes the partition function $\zeta(M)$ by summing over all possible bulk manifolds with boundary $M$.}
    \label{fig:GPI}
\end{figure}

\section{An axiomatic approach to quantum gravity}\label{sec:axiomatic_approach_philosophy}

As discussed in the Introduction, the basic structure of the GPI is the following: given a bulk theory of gravity, we fix the boundary of spacetime and perform a path integral over all possible bulks. Thus, schematically, the GPI is given by
\begin{equation}\label{eq:GPI}
    \zeta(M) \defined \int_{\partial X = M} \mathcal{D}X\ e^{-S(X)},
\end{equation}
where $S(X)$ is the bulk action functional, and we always work in Euclidean signature. The GPI runs over the space of bulk manifolds $X$ with boundary $M$, as depicted in Figure \ref{fig:GPI}. Our convention for dimensions will be that $d$ denotes the dimension of the manifold $M$ held fixed at the boundary of the gravitational spacetime, while the bulk is $(d+1)$-dimensional, ignoring any possibly large internal dimensions.

Frequently, the bulk theory of gravity includes some dynamical fields $\phi_X$, including a dynamical metric tensor, which must also be integrated over in the GPI. For pedagogical simplicity, we usually assume that we have chosen Dirichlet boundary conditions along $M$, meaning that the bulk fields $\phi_X$ must restrict to a fixed choice $\phi_M$ of boundary fields serving as a classical source. The matching of bulk fields with boundary sources is known as the \textit{holographic dictionary}, originally described by \cite{Gubser:1998bc,Witten:1998qj} in the context of the anti-de Sitter/conformal field theory (AdS/CFT) correspondence \cite{Maldacena:1997re}. For more general boundary conditions, we still use language adapted to Dirichlet boundary conditions, but strictly speaking what we mean is that we turn on source fields for all possible deformations of the boundary condition. We will frequently suppress explicit reference to fields, considering them to be part of the structure of our manifolds, and always requiring isomorphisms of manifolds to be field-preserving.

The GPI \eqref{eq:GPI} is, in almost every case, only schematic. While it can sometimes be evaluated in a semiclassical saddle-point approximation, there is frequently little hope of fully evaluating the integral. The issues are more severe than the technical issues in evaluating the QFT path integral, for multiple reasons, including the conformal mode problem \cite{Gibbons:1978ac}, the choice of integration contour \cite{Gibbons:1978ac,Halliwell:1989dy,Kontsevich:2021dmb,Witten:2021nzp,Marolf:2022ybi}, and the measure for summing over topologies \cite{Hawking:1987mz,Hawking:1988ae,Coleman:1988cy,Giddings:1988cx,Giddings:1987cg,Giddings:1988wv,Maloney:2007ud,Saad:2019lba,Marolf:2020xie}. These issues arise, in large part, from the fact that there is no full justification for the Euclidean GPI arising from Wick rotation from Lorentzian signature.\footnote{See \cite{Louko:1995jw,Colin-Ellerin:2020mva,Marolf:2022ybi,Chen:2025leq,Held:2026huj,Held:2026bbo,Kolanowski:2026gii} for more on the GPI in Lorentzian signature.} It has been conjectured \cite{Harlow:2020bee} that the Euclidean GPI will only truly be justified in a UV-complete theory of quantum gravity.

Nevertheless, whatever quantum gravity is, it should at least come with a notion of allowed boundary conditions and a partition function $\zeta(M)$ for each allowed boundary condition $M$. This is especially clear in the case of asymptotically AdS boundaries, where we have a wealth of examples arising from the known landscape of AdS vacua of string theory. As a result, a new, axiomatic approach to quantum gravity has been the subject of increasing study in recent years \cite{Marolf:2020xie,Gesteau:2020wrk,Colafranceschi:2023urj,Marolf:2024adj,Chen:2025fwp,DiUbaldo:2026rly,Zhao:2026mpl}. In this approach, we simply assume that the partition function $\zeta(M)$, ostensibly computed by the GPI, is known from the outset. As a result, the GPI \eqref{eq:GPI} is axiomatized away, just as the path integral of QFT is axiomatized away in the Atiyah--Segal axiomatic approach to QFT.\footnote{The path integral is still quite useful in constructing and/or motivating examples of QFTs, and the GPI can likely play an analogous role.}

The currently developing axiomatic framework for quantum gravity is much more minimal than the Atiyah--Segal framework for QFT \cite{Atiyah1988TQFT,segal1988definition}. In the Atiyah--Segal framework, one axiomatizes not just the partition function, but also the Hilbert spaces and Euclidean evolution operators of the QFT. These assemble into a unitary symmetric monoidal functor out of a bordism category, as reviewed in Definition \ref{defn:functorial_QFT} below. Locality, motivated by the cutting and gluing property of the QFT path integral, is fundamentally baked into the Atiyah--Segal axioms. In contrast, for quantum gravity, we very emphatically do not want to assume locality, and so the current approach is to simply axiomatize the partition function $\zeta(M)$ by itself.

Technically, the emerging axiomatic framework consists of two main parts. The first part is a class $\mathcal{X}$ of source manifolds $M$, or even just abstract ``sources,'' that can serve as the input to the putative GPI. We discuss this class of source manifolds in Section \ref{sec:categories_of_sources}. The second part is the partition function $\zeta(M)$ itself, which is taken to satisfy a list of axioms laid out in Section \ref{sec:axioms_for_partition_function} (possibly dropping multiplicativity, i.e., Axiom \ref{axiom:multiplicativity}). The required properties of the class $\mathcal{X}$ of source manifolds are actually identical to those for a $d$-dimensional QFT, and thus it makes sense to ask whether a given gravitational partition function $\zeta$ arises from a $d$-dimensional QFT or not. When it does, we say that the quantum theory of gravity is \textit{holographic}, which is, crucially, not taken as an assumption in this paper.

Let us note that the axioms for the gravitational partition function are strictly weaker than those of QFT, in that the partition function of any QFT will automatically satisfy them. Thus, there is a formal sense, discussed in more detail below, in which any $d$-dimensional QFT can be viewed as a holographic quantum theory of gravity in $(d+1)$ dimensions. From this perspective, we view the state produced in a QFT by evolution over some manifold as a state in which the bulk simply \textit{is} that manifold, at least before imposing the bulk Hamiltonian constraint. Throughout, we will frequently switch back and forth between the QFT perspective, in which we are considering a candidate partition function for a $d$-dimensional QFT, and the gravitational perspective, in which we view the partition function and resulting quantum theory as arising from a $(d+1)$-dimensional GPI.

\section{Bordism categories of source manifolds}\label{sec:categories_of_sources}

In this section, we outline the formal properties we require of the class $\mathcal{X}$ of source manifolds. These formal properties are best expressed as structures on the geometric bordism category $\BordX$. Geometric bordism categories have been studied previously \cite{segal1988definition,stolz2004elliptic,ayala2009geometriccobordismcategories,stolz2011supersymmetric,stolz2012traces,Kontsevich:2021dmb,grady2026higher}, with the approaches of Stolz--Teichner \cite{stolz2012traces} and Kontsevich--Segal \cite{Kontsevich:2021dmb} providing the most direct motivation for our approach in this paper. The most important structure we require on $\BordX$ is what we call a \textit{unitary structure}, as defined in Section \ref{sec:OS_conjugation}, and as motivated by \cite{stehouwer2024unitary,stehouwer2024spin,ferrer2024dagger}.

Ultimately, nothing essential in our framework requires the source category to have anything to do with bordism. The necessary structures are purely categorical, and once we have the category $\BordX$ in hand, with the structures we discuss, we are free to forget how it was built. This will be shown explicitly in the follow-up work \cite{M-JF-R_wip} of Johnson-Freyd, Reutter, and the first author of this paper.

\subsection{Source data and locality}\label{sec:source_data}

To start, fix some class $\mathcal{X}$ of local source data on $d$-dimensional manifolds. This source data may include topological data such as an orientation or a spin structure, as well as continuous data, such as background fields of various kinds or marked loci representing operator insertions. One field of particular importance is a fixed Riemannian metric, which sources the bulk gravitational field. To fix terminology, we say that a choice of source data on a $d$-manifold is an \textit{$\mathcal{X}$-structure}, and a $d$-manifold $M$ equipped with an $\mathcal{X}$-structure is an \textit{$\mathcal{X}$-manifold}. Two $\mathcal{X}$-manifolds are considered to be the same when they are related by an \textit{$\mathcal{X}$-isomorphism}, i.e., an $\mathcal{X}$-preserving diffeomorphism. For instance, when $\mathcal{X}$ includes a Riemannian metric, $\mathcal{X}$-isomorphisms must be isometries.

Formally, the class $\mathcal{X}$ is a sheaf on the site of smooth $d$-manifolds and local diffeomorphisms \cite{Freed:2012bs}. The condition that $\mathcal{X}$ be a sheaf encodes the locality of the source data, as it asserts that we may cut and glue $\mathcal{X}$-structures coordinate patch by coordinate patch.\footnote{In other words, an $\mathcal{X}$-structure is something which transforms as an $\mathcal{X}$-structure.} We restrict to local source data because our goal is to probe locality; if our sources were not local, it would be impossible to ask the question. We assume that $\mathcal{X}$ is equipped with a topology, which we assume to be separable on compact $d$-manifolds.\footnote{Separability holds for any physically reasonable choice of source data, including spaces of finitely (or even countably) many smooth fields, or configuration spaces of marked loci.}

In this paper, we will need to study both $d$-dimensional and $(d-1)$-dimensional manifolds equipped with $\mathcal{X}$-structures. A standard issue is that, since $\mathcal{X}$-structures are only defined in dimension $d$, it does not make sense to ask for an $\mathcal{X}$-structure on a $(d-1)$-manifold $B$. The standard\footnote{See e.g. \cite{segal1988definition,stolz2004elliptic,stolz2011supersymmetric,freed2019lectures}.} way of handling this is to define a \textit{$(d-1)$-dimensional $\mathcal{X}$-manifold} to be a $(d-1)$-dimensional manifold $B$ equipped with a germ of an ambient $d$-dimensional $\mathcal{X}$-manifold in which $B$ is embedded as a co-oriented, codimension-one submanifold. We illustrate this definition in Figure \ref{fig:spacetime_germ}.

Physically, we view the manifold $B$ as a spatial slice of the ambient $d$-dimensional source spacetime. Saying that we only have a \textit{germ} means that we identify any two spacetimes if they are isomorphic in a neighborhood of $B$ via an isomorphism which acts trivially on $B$; thus, a germ only remembers an infinitesimal piece of spacetime around $B$. The co-orientation of $B$ in spacetime means that we have chosen an orientation for its normal bundle, which physically serves as the arrow of time.\footnote{The arrow of time has nothing to do with whether our theory admits a time-reversal symmetry, but instead indicates the direction of quantization. See e.g. \cite{freed2019lectures,Witten:2025ayw}.} $\mathcal{X}$-isomorphisms of $(d-1)$-dimensional $\mathcal{X}$-manifolds are required to preserve the arrow of time.

\begin{figure}
    \centering
\includegraphics[width=0.25\linewidth]{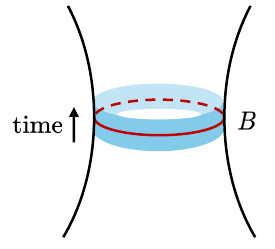}
    \caption{A spatial manifold $B$ (red) equipped with a germ of spacetime (shaded blue) and a co-orientation, i.e., an arrow of time. Whatever source fields we have included in the class $\mathcal{X}$ are defined on the germ of spacetime. Though suppressed from our notation, both the germ of spacetime, with its $\mathcal{X}$-structure, and the arrow of time are needed to specify an object of $\BordX$.}
\label{fig:spacetime_germ}
\end{figure}

We define an \textit{$\mathcal{X}$-bordism} $N : B_1 \to B_2$ between $(d-1)$-dimensional $\mathcal{X}$-manifolds to be a compact $d$-dimensional $\mathcal{X}$-manifold $N$ with boundary $\partial N = B_1 \sqcup B_2$, such that the arrows of time on $B_1$ and $B_2$ point inward and outward, respectively, and such that the germ of $N$ near its boundary agrees with the given germs on $B_1$ and $B_2$. We depict a few bordisms in Figure \ref{fig:germ_bordisms}. An $\mathcal{X}$-bordism is precisely what is needed to evolve quantum states from the incoming boundary $B_1$ to the outgoing boundary $B_2$, both in the context of QFT and in the GPI. Given $\mathcal{X}$-bordisms
\begin{equation}
    N_1 : B_1 \to B_2, \quad N_2 : B_2 \to B_3,
\end{equation}
we define the composite bordism $N_2 \circ N_1: B_1 \to B_3$ by gluing along the shared boundary $B_2$,
\begin{equation}
    N_2 \circ N_1 = N_2 \cup_{B_2} N_1,
\end{equation}
as depicted in Figure \ref{fig:germ_bordisms}. The glued bordism has a canonical $\mathcal{X}$-structure, obtained by gluing the $\mathcal{X}$-structures of $N_1$ and $N_2$ along the shared germ of an $\mathcal{X}$-structure around $B_2$. We consider two $\mathcal{X}$-bordisms to be isomorphic if they are related by an $\mathcal{X}$-isomorphism which restricts to the identity on the boundary germs.

\begin{figure}
    \centering
    \includegraphics[width=0.5\linewidth]{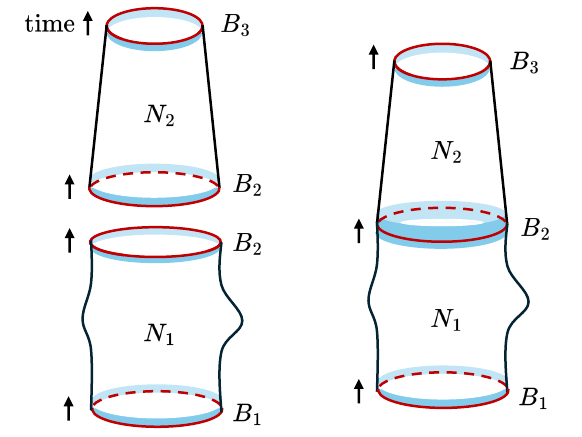}
    \caption{We depict two bordisms, $N_1:B_1\to B_2$, $N_2:B_2\to B_3$, as well as their composite $N_2 \circ N_1 : B_1 \to B_3$. The source and target of each bordism are determined by the arrows of time on its boundary. In defining the composite, we glue the $\mathcal{X}$-structures along the given germ of spacetime along $B_2$.}
    \label{fig:germ_bordisms}
\end{figure}

The collections of $(d-1)$-dimensional $\mathcal{X}$-manifolds and $d$-dimensional $\mathcal{X}$-bordisms form the objects and morphisms of the $\mathcal{X}$\textit{-bordism category}, $\BordX$. For most physically relevant choices of source data $\mathcal{X}$, the bordism category $\BordX$ will be non-unital, meaning that it will lack identity morphisms. This is because all non-empty bordisms are required to have some finite Euclidean duration, as measured by some invariant like the distance between the boundaries or the total volume, which can only increase under composition.

The main exception is when the class $\mathcal{X}$ is \textit{topological}. Formally, this means that the sheaf $\mathcal{X}$ is locally constant, in that specifying an $\mathcal{X}$-structure in a small neighborhood of any point is enough to uniquely specify an $\mathcal{X}$-structure in any contractible region containing that point. A class of topological source data $\mathcal{X}$ is the same thing \cite{lurie2008classification,ayala2015factorization,ayala2017local,freed2019lectures} as a \textit{tangential structure} in the sense of \cite{lashof1963poincare,Freed:2016rqq}. When $\mathcal{X}$ is topological, the bordism category $\BordX$ is unital. This is because, with no way to measure the duration of a bordism, cylinders can serve as identity morphisms.

Working with non-unital categories is a bit difficult because they do not have a notion of categorical isomorphism. One approach, taken for instance in \cite{stolz2012traces}, is to adjoin all $\mathcal{X}$-isomorphisms of $(d-1)$-dimensional manifolds as additional morphisms in $\BordX$. The composition of an $\mathcal{X}$-isomorphism with a bordism is defined by twisting the appropriate boundary by the given isomorphism. This approach produces a unital category, but not a category whose morphisms can always be viewed as finite-length bordisms.

Our approach will be different, following a suggestion made by Kontsevich and Segal in \cite{Kontsevich:2021dmb}.\footnote{And made to us by Theo Johnson-Freyd.} For us, the category $\BordX$ will be genuinely non-unital, but we will use the notion of $\mathcal{X}$-isomorphism in place of categorical isomorphism whenever we consider whether two objects or morphisms of $\BordX$ are the same.\footnote{This is, in fact, how we typically think of objects of bordism categories anyway, even in the topological case. This is because the categorical isomorphisms in topological bordism categories include not just isomorphisms of manifolds, but also $h$-cobordisms. See \cite{ferrer2024dagger} for further discussion.} As a result, in the various categorical structures we discuss and construct on $\BordX$, all coherence isomorphisms are taken to be $\mathcal{X}$-isomorphisms. As in \cite{Kontsevich:2021dmb}, our reason for doing this is analytic, as we discuss below. By abuse of notation, we will still use $u \circ N$ or $N \circ u$ to denote the composition of an $\mathcal{X}$-isomorphism $u$ with a bordism $N$, which makes sense even though $\mathcal{X}$-automorphisms are not morphisms in our category $\BordX$.

The bordism category $\BordX$ allows us to cut $\mathcal{X}$-manifolds along spatial slices, which is sufficient to encode locality in time. Let us now discuss locality in space. If we wanted to discuss full spatial locality, we should continue cutting our spatial slices in codimension-two, codimension-three, and so on, until we could discuss arbitrarily small neighborhoods of each spacetime point. In the topological case, this extended cutting is formalized by the definition of bordism $d$-categories \cite{lurie2008classification,calaque2019note}, and there is no obstruction in defining geometric bordism $d$-categories. However, in geometric QFTs, there are severe analytic issues related to cutting in higher codimension, including the expected Type III von Neumann algebras which appear once one cuts a spatial slice along a codimension-two interface.

For our goal of understanding Hilbert-space factorization, we so not need to address these analytic issues. Instead, we restrict ourselves to a trivial-seeming special case: cutting a disjoint spatial slice $B_1 \sqcup B_2$ along the \textit{empty} codimension-two interface separating $B_1$ from $B_2$. While this cut may not be so interesting in QFT, it is highly nontrivial in quantum gravity. This is because we can have a nontrivial codimension-two cut of a bulk spatial slice which does not extend out to the boundary. For instance, the horizon, or cross-section, of an ER bridge connecting $B_1$ and $B_2$ through the bulk is precisely such a cut.

Categorically, the possibility of cutting a disjoint spatial manifold $B_1 \sqcup B_2$ into its components is encoded in the symmetric monoidal structure on the bordism category $\BordX$. We recall that a \textit{symmetric monoidal structure} on a category $\mathcal{C}$ is a unital fusion operation,
\begin{equation}
    \otimes : \mathcal{C} \times \mathcal{C} \to \mathcal{C}, \quad \mathds{1} \in \mathcal{C},
\end{equation}
which is associative and commutative up to coherent isomorphism (in the strongest sense). For $\BordX$, the symmetric monoidal structure is given by disjoint union of manifolds, with unit the empty $(d-1)$-manifold,
\begin{equation}
    \sqcup : \BordX \times \BordX \to \BordX, \quad \varnothing^{d-1} \in \BordX.
\end{equation}
For categorically inclined readers, we note that the monoidal structure on the $\mathcal{X}$-bordism 1-category can be viewed as the composition of endomorphisms of the empty $(d-2)$-manifold in the $\mathcal{X}$-bordism 2-category. This is the precise sense in which splitting a disjoint union $B_1 \sqcup B_2$ into its components is a special case of cutting in codimension-two.

Some examples of topological bordism categories are the categories $\Bord_d^O, \Bord_d^{SO}$, and $\Bord_d^{Spin}$, of unoriented, oriented, and spin manifolds respectively. Each of these has a Riemannian analog, where, in addition to the tangential structure, we also require the choice of a Riemannian metric. Other examples of source data include background scalar fields, background gauge bundles with connection, configurations of points, lines, or general loci marked by operator labels, and foliations of fixed type.

\subsection{Unitary structures}\label{sec:OS_conjugation}

In this paper, our goal is to study unitary quantum systems, whose states live in Hilbert spaces and admit probabilistic interpretations via the Born rule. In order to make sense of unitarity, and its Euclidean avatar, reflection positivity, the class $\mathcal{X}$ of source manifolds must be equipped with additional structure. Parts of this structure were discussed recently in \cite{Witten:2025ayw} in the context of the GPI. Our approach is directly motivated by \cite{Freed:2016rqq,Penneys2018UnitaryDF,stehouwer2024unitary,stehouwer2024spin,ferrer2024dagger}.

\begin{figure}
    \centering
    \includegraphics[width=0.6\linewidth]{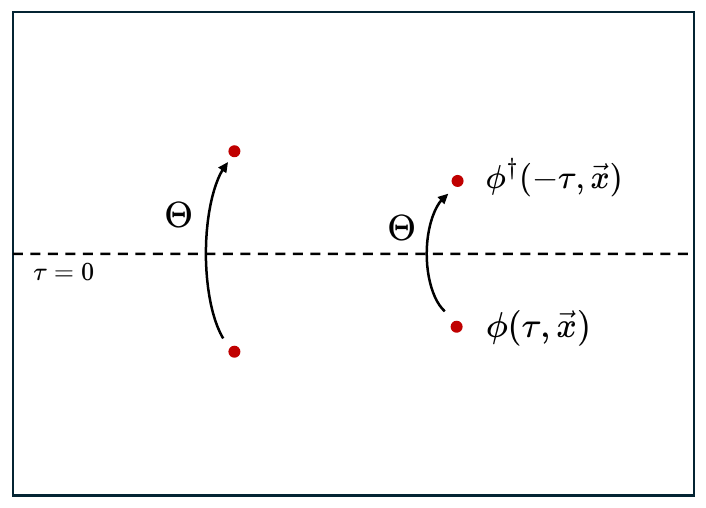}
    \caption{The Osterwalder--Schrader conjugation map $\Theta:\phi(\tau,\vec{x}) \mapsto \phi^\dagger (-\tau,\vec{x})$. This conjugation map is necessary in order to turn Euclidean correlation functions into an inner product, as required to formulate the condition of reflection positivity.}
    \label{fig:OS_conjugation}
\end{figure}

As motivation, let us recall the Osterwalder--Schrader (OS)/Wightman approach \cite{OsterwalderSchrader1973,OsterwalderSchrader1975,Wightman1956,StreaterWightman2000} to axiomatic QFT. In this approach, one starts with Euclidean correlation functions $\braket{-}_\Omega$ of fields $\phi(\tau, \vec x)$ on $\mathbb{R}^d$, ostensibly computed in some vacuum state $\ket{\Omega}$. In order to turn these Euclidean correlation functions into an inner product, one needs the \textit{OS conjugation map} $\Theta$, which acts on fields by
\begin{equation}
    \Theta : \phi(\tau, \vec x) \mapsto \phi^\dagger(- \tau, \vec x),
\end{equation}
as depicted in Figure \ref{fig:OS_conjugation}. Given a collection $\phi$ of fields supported on the lower half-space, one defines formal states $\ket{\phi \Omega}$, with inner product defined by,
\begin{equation}
    \braket{\phi_1 \Omega | \phi_2 \Omega} \defined \braket{\Theta(\phi_1) \phi_2}_\Omega.
\end{equation}
It is only at this stage, once OS conjugation has been applied, that one asks for reflection positivity.

Thus, we need an OS conjugation map on the bordism category $\BordX$. In particular, we need to know how to glue two bordisms,
\begin{equation}\label{eq:two_bordisms_rep_two_states}
    M_1, M_2 : \varnothing \to B,
\end{equation}
representing two states on $B$ into a closed manifold which represents the corresponding wavefunction overlap. At a technical level, what we need is a \textit{$\dagger$-structure}, or \textit{dagger-structure}, on $\BordX$. A $\dagger$-structure is a generalization of the operation on the category $\Hilb$ which takes a linear map to its adjoint. Concretely, a $\dagger$-structure on $\BordX$ is a contravariant involution,
\begin{equation}
    (-)^\dagger : \BordX \to (\BordX)^\mathrm{op},
\end{equation}
which is the identity on objects and takes a bordism $N : B_1 \to B_2$ to an adjoint bordism $N^\dagger : B_2 \to B_1$. Given bordisms $M_1, M_2$ as in \eqref{eq:two_bordisms_rep_two_states} representing two states, the closed manifold $M_1^\dagger \circ M_2$ represents their wavefunction overlap, as depicted in Figure \ref{fig:inner_product}. Moreover, we demand that every $\mathcal{X}$-isomorphism is unitary, in that its image under $\dagger$ is its inverse.

\begin{figure}
    \centering
    \includegraphics[width=0.8\linewidth]{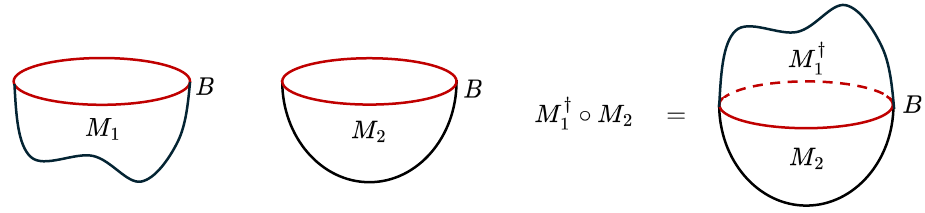}
    \caption{For two bordisms $M_1,M_2 : \varnothing \to B$ (left), representing two states on $B$, the corresponding wavefunction overlap is represented by the closed manifold $M_1^\dagger \circ M_2$ (right) obtained by reflecting $M_1$ across $B$ and gluing.}
    \label{fig:inner_product}
\end{figure}

To find a $\dagger$-structure, let us first note that the bordism category $\BordX$ is naturally equipped with a contravariant duality involution,
\begin{equation}\label{eq:duality_functor}
    (-)^\vee : \BordX \to (\BordX)^\mathrm{op},
\end{equation}
defined on objects by reversing the arrow of time and on bordisms by swapping the source and target. However, the duality involution $(-)^\vee$ does not typically define a $\dagger$-structure, because it is not typically the identity on objects. For example, if an oriented manifold $B \in \Bord_d^{SO}$ is not diffeomorphic to its orientation reversal, then there is no $SO$-isomorphism between $B$ and $B^\vee$ as objects of $\Bord_d^{SO}$. This is because any isomorphism $B \to B^\vee$, as we have defined it, must act by time reversal on the germ of spacetime around $B$. Thus, an isomorphism $B \to B^\vee$ must also reverse the orientation of space in order to preserve the orientation of spacetime. As a result, we cannot use the dual functor $(-)^\vee$, on its own, as an analog of OS conjugation.

While not a $\dagger$-structure, the functor $(-)^\vee$ is still quite useful. In fact, the object $B^\vee$ is a non-unital version of a dual object for $B$. This means that we have a family of canonical isomorphisms,
\begin{equation}
    \BordX(B_1 \sqcup B \to B_2) \cong \BordX(B_1 \to B_2 \sqcup B^\vee), \quad N \mapsto N^{\vee_B},
\end{equation}
defined by reversing the arrow of time on the boundary component $B$ of the bordism $N$. Thus, the symmetric monoidal category $\BordX$ is \textit{rigid} monoidal, in a non-unital sense, meaning that each object admits a dual.

It may come as a surprise to readers familiar with TQFT that $\BordX$ is rigid, as dualizability of objects in topological bordism categories is directly connected to the finite dimensionality of Hilbert spaces in TQFT. This is because one may use duals to turn the identity map $X \to X$ of any object into a coevaluation map $\mathds{1} \to X \otimes X^\vee$, which is only possible in finite dimensional vector spaces.

In contrast, in the non-unital rigid monoidal category $\BordX$, we cannot define coevaluation maps because we do not have identity morphisms in the first place. Nevertheless, we may still turn any bordism $A \to B$ of finite length into a bordism $\varnothing \to B \sqcup A^\vee$. For comparison, consider the category $\Hilb^{\mathrm{HS}}$, whose objects are Hilbert spaces and whose morphisms are Hilbert--Schmidt operators. Any Hilbert--Schmidt operator $\widehat{\mathcal{O}} : \mathcal{H}_1 \to \mathcal{H}_2$ can be turned into a state $\ket{\CO} \in \mathcal{H}_2 \otimes \mathcal{H}_1^\vee$; infinite dimensionality causes no issues here, because the identity operator on an infinite-dimensional Hilbert space is never Hilbert--Schmidt. For this reason, finite length bordisms behave in many ways like Hilbert--Schmidt operators.\footnote{In fact, bordisms even behave like trace-class operators, as discussed in \cite{stolz2012traces}. This is because any finite length bordism can be decomposed as a composition of two shorter bordisms, and the product of any two Hilbert--Schmidt operators is always trace class.} This is the essential analytic reason why we do not include identity morphisms in the category $\BordX$.

Returning to the problem of equipping $\BordX$ with a $\dagger$-structure, we take a cue from the natural $\dagger$-structure on $\Hilb$ defined by taking adjoints of operators. In $\Hilb$, the adjoint $(-)^\dagger$ is defined not just in terms of the dual $(-)^\vee$, but in terms of the conjugate dual $\overline{(-)}^\vee$, where
\begin{equation}
    \overline{(-)} : \Hilb \to \Hilb,
\end{equation}
is the functor which takes a Hilbert space to its complex conjugate. Now, $\overline{(-)}^\vee$ is not quite the identity on objects. To remedy this, we must combine $\overline{(-)}^\vee$ with the canonical unitary isomorphism $\rho_\mathcal{H} : \mathcal{H} \to \overline{\mathcal{H}}^\vee$ induced by the inner product.  Concretely, $\rho_\CH$ maps the ket vector $\ket{\psi} \in \CH$ to the corresponding bra vector $\bra{\psi} \in \overline{\CH}^\vee$. The matrix elements of $\rho_\CH$, when computed in any linear basis, are the components $G_{\bar \imath j}$ of the Gram matrix of the inner product on $\CH$.

Following \cite{stehouwer2024unitary,stehouwer2024spin}, we refer to a choice of isomorphism $\rho_\mathcal{H} : \CH \to \overline{\CH}^\vee$ as a \textit{Hermitian structure} with respect to the conjugate dual functor $\overline{(-)}^\vee$, and note that the objects of $\Hilb$ are precisely vector spaces equipped with a Hermitian structure which is, moreover, positive definite. Using the conjugate, dual, and Hermitian structures, the $\dagger$-structure on $\Hilb$ is defined by,
\begin{equation}\label{eq:dagger_on_hilb}
    \big( T : \mathcal{H}_1 \to \mathcal{H}_2\big) \mapsto \big(T^\dagger \defined \rho_{\mathcal{H}_1}^{-1} \circ \overline{T}^\vee \circ \rho_{\mathcal{H}_2} : \mathcal{H}_2 \to \mathcal{H}_1\big).
\end{equation}
The Hermitian structures $\rho_{\mathcal{H}_1}^{-1}, \rho_{\mathcal{H}_2}$ in \eqref{eq:dagger_on_hilb} are invisible when $T^\dagger$ is computed in an orthonormal basis, but are necessary when it is computed in a general linear basis, as we have,
\begin{equation}
    (T^\dagger)_i^j = G_{i \bar \imath} T_{\bar \imath}^{\bar \jmath} G^{j \bar \jmath},
\end{equation}
using the Hermitian metric and its inverse to raise and lower indices.

\begin{figure}
    \centering
    \includegraphics[width=0.9\linewidth]{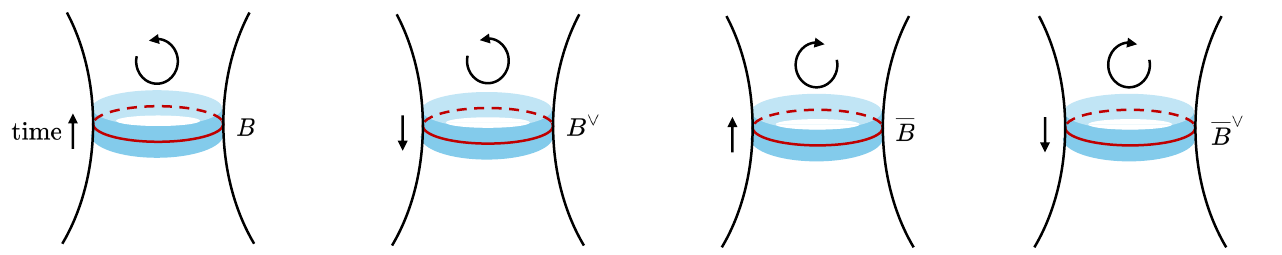}
    \caption{Starting from a spatial manifold $B \in \mathrm{Bord}^{\mathcal X}_d$ (far left), reversing the arrow of time gives the dual $B^\vee$ (left), while reversing the orientation (and conjugating sources) in spacetime gives the conjugate $\overline{B}$ (right). Finally, the conjugate dual $\overline{B}^{\vee}$ (far right) is obtained by applying both operations. The vertical arrows denote the arrows of time, while the circular arrows denote the orientation of the ambient spacetime. Neither the dual $(-)^\vee$ nor the conjugate $\overline{(-)}$ move points of spacetime, acting only on the auxiliary structures given by the arrow of time and ambient $\mathcal{X}$-structure, respectively.}
    \label{fig:conjugation_object}
\end{figure}

By analogy, we require that $\mathcal{X}$ be equipped with a conjugation involution $\overline{(-)}$, taking any $\mathcal{X}$-manifold $M$ to its conjugate $\overline{M}$. At the level of the bordism category, $\overline{(-)}$ defines a covariant conjugation involution,
\begin{equation}
    \overline{(-)} : \BordX \to \BordX.
\end{equation}
Heuristically, the conjugation map reverses orientations and takes complex source fields to their complex conjugates. More precisely, $\overline{(-)}$ defines what we mean by the complex conjugate of a source field in the first place. We depict the conjugation operation $\overline{(-)}$, as well as the dual $(-)^\vee$, in Figure \ref{fig:conjugation_object}.

As in $\Hilb$, the conjugate dual $\overline{(-)}^\vee$ is not the identity on objects on the nose. In the case of $\BordX$, the conjugate dual $\overline{(-)}^\vee$ acts on a spatial manifold $B$ by reversing the arrow of time and conjugating sources on the germ of spacetime. Thus, we need to equip every manifold $B \in \BordX$ with a Hermitian structure,
\begin{equation}
    \rho_B : B \xrightarrow{\sim} \overline{B}^\vee,
\end{equation}
in order to equip $\BordX$ with a $\dagger$-structure. To understand the conjugate dual of an object $B \in \BordX$, consider the following example.

\begin{example}\label{ex:Hermitian_str_SO}
    Let $\Bord_d^{SO}$ denote the topological bordism category of oriented $d$-manifolds. The conjugate dual $\overline{B}^\vee$ of an object $B \in \Bord_d^{SO}$ is the exact same manifold, with the ambient orientation of spacetime flipped, and with the arrow of time reversed. To identify $B$ with $\overline{B}^\vee$, we flip the germ of spacetime,
    \begin{equation}
        \rho_B : B \to \overline{B}^\vee, \quad  (\tau, \vec{b}) \mapsto (- \tau, \vec{b}),
    \end{equation}
    fixing the spatial coordinate $\vec{b} \in B$ and flipping the Euclidean time coordinate $\tau$. One might object that $\rho_B$ seems to be orientation-reversing, but this is not the case: $\overline{(-)}$ is defined to reverse the ambient orientation, while $(-)^\vee$ has nothing to do with the orientation. Note that $\rho_B$ also preserves the orientation of space induced from the orientation of spacetime via the arrow of time.
\end{example}

As in Example \ref{ex:Hermitian_str_SO}, the Hermitian structures $\rho_B$ in $\BordX$ should act as time reversal, fixing space point-wise, and flipping the germ of spacetime about space. We emphasize that $\rho_B$, like the arrow of time itself, has nothing to do with whether our theory admits a time-reversal symmetry. Instead, $\rho_B$ serves as the structure we use to turn kets into bras. One may worry, with justification, that when $\mathcal{X}$ includes the choice of a metric (or any continuous source field), a general germ of spacetime will not admit any time-reversing isometry. We return to this point in Section \ref{sec:germs}, with the upshot being that we have to restrict to germs of spacetime which do admit a time-reversal isometry.

By combining conjugation, duality, and the Hermitian structures, we can finally equip $\BordX$ with a $\dagger$-structure, defined by,
\begin{equation}\label{eq:dagger_on_BordX}
    \big( N : A \to B \big) \mapsto \big( N^\dagger \defined \rho_A^{-1} \circ \overline{N}^\vee \circ \rho_B : B \to A \big).
\end{equation}
Physically, $N^\dagger$ is defined by first using $\rho_B$ to turn an incoming ket state $\ket{\psi}$ on $B$ into the corresponding bra $\bra{\psi}$, propagating $\bra{\psi}$ backwards in time along $\overline{N}^\vee$, and finally using $\rho_A^{-1}$ to turn the resulting bra state back into a ket. We use the $\dagger$-structure to define the closed manifold obtained by pairing two bordisms,
\begin{equation}
    \braket{M_1 | M_2}_\Bord \defined M_1^\dagger \circ M_2, \quad M_1, M_2 : \varnothing \to B,
\end{equation}
from nothing to a fixed spatial slice $B$.

\begin{figure}
    \centering
    \includegraphics[width=0.8\linewidth]{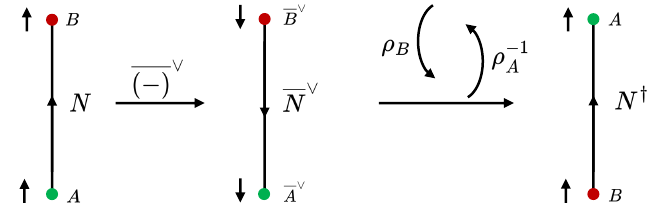}
    \caption{To define the $\dagger$-structure on $\BordX$ we start with a bordism $N : A \to B$ (left), applying the conjugate dual functor to obtain $\overline{N}^\vee : \overline{B}^\vee \to \overline{A}^\vee$ (middle). This operation reverses the arrows of time on both boundaries and reverses the orientation of spacetime. Finally, we obtain the adjoint $N^\dagger : B \to A$ by conjugating $\overline{N}^\vee$ with the Hermitian structures $\rho_A^{-1}$ and $\rho_B$, which we view geometrically as a $\pi$ rotation.}
    \label{fig:conjugation_bordism}
\end{figure}

The conjugate functor $\overline{(-)}$ and the Hermitian structure $\rho_B$ are required to satisfy a collection of coherence axioms, which we discuss in detail in Appendix \ref{app:daggers_duals_conjugates}. Specifically, up to slight differences in conventions, we ask that the induced $\dagger$-structure \eqref{eq:dagger_on_BordX} makes $\BordX$ \textit{strong fermionically $\dagger$-compact} in the sense of \cite{stehouwer2024unitary,stehouwer2024spin}. Equivalently (for our purposes), we ask for the non-unital, rigid, symmetric monoidal category $\BordX$ to be a $\dagger$-category \textit{with unitary duality}, in the sense of \cite{ferrer2024dagger}. In this paper, we refer to either of these equivalent structures as a \textit{unitary structure}, and a functor which preserves them as a \textit{unitary functor}.

Let us spell out the additional constraints which enter into the definition of a unitary structure. First, we require an additional family of coherence isomorphisms $\overline{B}^\vee \cong \overline{B^\vee}$. For us, these isomorphisms will simply be identity maps, since reversing the arrow of time and conjugating sources commute on the nose. Via these isomorphisms, we obtain what is known as the \textit{twist automorphism},
\begin{equation}
    \theta_B \defined \rho_{\overline{B}^\vee} \circ \rho_B : B \to B,
\end{equation}
using that $\overline{(-)}$ and $(-)^\vee$ commute and are involutive on the nose. We then require that $\theta_B$ satisfies,
\begin{equation}
\theta_B^2 = \mathrm{id}_B,
\end{equation}
as $\mathcal{X}$-automorphisms of $B$. A unitary structure on $\BordX$ is a non-linear analog of a balanced unitary dual functor (UDF) on a unitary tensor category \cite{Penneys2018UnitaryDF}. We say that a bordism category $\BordX$ equipped with a unitary structure is a \textit{unitary bordism category}.

To illustrate the physical meaning of the twist $\theta_B$, consider the following example.

\begin{example}\label{ex:twist_on_spin_bord}
    Let $\Bord_d^{Spin}$ denote the topological bordism category of $d$-manifolds equipped with spin structures. As discussed in \cite{Freed:2016rqq,stehouwer2024unitary,stehouwer2024spin}, the spin bordism category $\Bord_d^{Spin}$ admits a canonical unitary structure with the following property: for every spin $(d-1)$-manifold, the twist automorphism $\theta_B : B \to B$ is the $2 \pi$ rotation spin diffeomorphism, which acts as the identity on the points of $B$ and acts on spin frames as a $2 \pi$ rotation.
\end{example}

As illustrated by Example \ref{ex:twist_on_spin_bord}, the twist automorphism $\theta_B$ corresponds, physically, to a $2 \pi$ rotation of $B$.\footnote{Analogously, $\rho_B$ corresponds to a $\pi$ rotation. However, as demonstrated algebraicly in Appendix \ref{app:daggers_duals_conjugates}, $\rho_B$ and $\rho_{B^\vee}$ (as well as $\rho_{\overline{B}}$) correspond to $\pi$ rotations in opposite directions, because of the issue discussed in \cite{stehouwer2024unitary,stehouwer2024unitary} related to the positivity of inner products on the duals of super-Hilbert spaces.} The condition $\theta_B^2 = \mathrm{id}_B$ is thus the statement that a $4 \pi$ rotation must act trivially, which is physically motivated by the fact that our spatial manifolds are not embedded anywhere, and thus cannot have anyonic behavior under rotation. We say that a unitary bordism category $\BordX$ is \textit{bosonic} if $\theta_B = \mathrm{id}_B$ for all objects $B$, and is \textit{fermionic} otherwise. This interpretation of $\theta_B$ is justified by the spin-statistics theorem for unitary functorial QFTs \cite{Freed:2016rqq,stehouwer2024unitary,stehouwer2024spin}, which asserts that $\theta_B$ (measuring \textit{spin}) must map to fermion parity $(-1)^F$ (measuring \textit{statistics}) in any unitary functorial QFT (Definition \ref{defn:functorial_QFT}).

Example \ref{ex:twist_on_spin_bord} illustrates an additional important point: equipping the bordism category $\BordX$ with a unitary structure is a nontrivial step, which may involve physically meaningful choices. To see this, consider the example of \cite{stehouwer2024unitary}, given by the spin bordism category $\Bord_1^{Spin}$ in $d = 1$ dimensions. Now, note that a spin structure on a 1-manifold is precisely the same data as an orientation and a separate $\mathbb{Z}_2$ gauge bundle. This is because an orientation on a 1-manifold already trivializes the tangent bundle, and removing the usual twisting that distinguishes a spin structure from a $\mathbb{Z}_2$ bundle. As a result, $\Bord_1^{Spin}$ is equivalent, as a symmetric monoidal category, to the bordism category $\Bord_1^{{SO} \times \mathbb{Z}_2}$ of oriented 1-manifolds equipped with a $\mathbb{Z}_2$ bundle. However, the two categories are distinguished by their unitary structures. While $\theta_B$ is the $2 \pi$ spin rotation in $\Bord_1^{Spin}$, it is simply the identity in $\Bord_1^{{SO} \times \mathbb{Z}_2}$.

While equipping $\BordX$ with a unitary structure is nontrivial in general, at least in the topological case $\BordX$ admits a canonical unitary structure whenever $\mathcal{X}$ is a \textit{stable tangential structure}, meaning, some structure on the stabilized tangent bundle \cite{Freed:2016rqq,ferrer2024dagger}. Stable tangential structures include, for instance: orientations, spin structures, pin structures, and background gauge fields for any finite symmetry. While the analog of this canonical unitary structure is not clear in the geometric case, we can use the same canonical unitary structure if $\mathcal{X}$ is the class of Riemannian manifolds equipped, additionally, with some stable tangential structure. Moreover, many choices of geometric source data seem to admit natural unitary structures. For instance, a complex metric is taken by $\overline{(-)}$ to its complex conjugate \cite{Kontsevich:2021dmb}.

\subsubsection{Restrictions on germs and sources}\label{sec:germs}

In order for a geometric bordism category $\BordX$ to admit a unitary structure, we are forced to place a few restrictions on the germs of spacetime and the types of sources we allow. The first, discussed above, is that we must restrict to germs which admit a time-reversal symmetry that fixes the spatial slice, so that every object $B \in \BordX$ is equipped with a Hermitian structure $\rho_B : B \to \overline{B}^\vee$. Thankfully, this restriction is also physically natural, as this is the only case in which we expect quantum states to form a Hilbert space in any case \cite{Kontsevich:2021dmb}. We must also demand that the source fields turned on around the spatial slices are invariant under $\overline{(-)}$. For instance, a complex scalar field may only take real values on the fixed spatial slice, and must be taken to its complex conjugate under time reversal.

\begin{figure}
    \centering
    \includegraphics[width=0.7\linewidth]{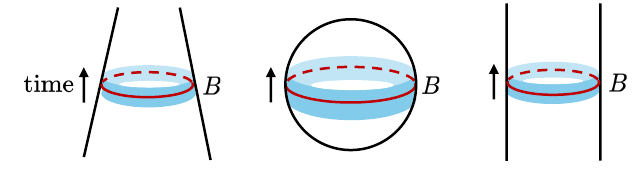}
    \caption{For a spatial manifold $B$, there may be multiple non-isomorphic germs of spacetime in which it might be embedded. This germ might admit no additional symmetry (left), time-reversal symmetry (middle), or time-reversal and time-translation symmetry (right).}
    \label{fig:different_germs}
\end{figure}

For the sake of simplicity, we make the additional assumption that the germs of spacetime admit an infinitesimal time translation symmetry. In particular, we assume the germ of a metric around an object $B \in \BordX$ takes the form,
\begin{equation}
    ds^2 = N_B(\vec x)^2 d\tau^2 + g_{i j}^B(\vec x) dx^i dx^j,
\end{equation}
up to diffeomorphism, for some metric $g_{ij}^B$ and lapse function $N_B$ on $B$.

This restriction on germs ensures that each object $B \in \BordX$ admits a one-parameter semigroup of self-adjoint self-bordisms, given by \textit{cylinders},
\begin{equation}\label{eq:self_adjoint_semigroup}
    C_B(\beta) = [0, \beta] \times B, \quad C_B(\beta_1) \circ C_B(\beta_2) = C_B(\beta_1 + \beta_2), \quad C_B(\beta)^\dagger = C_B(\beta),
\end{equation}
as used in \cite{Colafranceschi:2023urj,Marolf:2024adj}. Here, $\beta > 0$ denotes the parameter controlling the semigroup. Note that $C_B(\beta)$ is only a Riemannian product when the lapse $N_B$ is constant, in which case we will always rescale $\tau$ to set $N_B = 1$ and take $\beta$ to denote the actual geometric length of $C_B(\beta)$. In a unitary QFT, the evolution operator along $C_B(\beta)$ will be $e^{- \beta \widehat{H}_B}$, where $\widehat{H}_B$ is the Hamiltonian operator on $B$ associated to the lapse function $N_B$. In the topological case, we can drop the parameter $\beta$, as the cylinders of different length are all isomorphic and serve as identity bordisms.

While the restriction to germs with a time-reversal symmetry is structurally important, the restriction to germs with a time-translation symmetry is essentially cosmetic (though necessary for defining a Hamiltonian operator). Without this restriction, we would replace the semigroup \eqref{eq:self_adjoint_semigroup} with the filtered category of bordisms into $B$, as was used in \cite{Kontsevich:2021dmb} to define what they called the \textit{continuity property}. This generalization will play a central role in the follow-up work \cite{M-JF-R_wip, johnsonfreyd2026rigidfirm}, in terms of the firmness property (in the sense of \cite{quillen1997module}) of Riemannian bordism categories.

The second structural restriction concerns fermionic sources. The axioms defining a unitary structure on $\BordX$ demand that the twist automorphisms $\theta_B$ form a natural automorphism,
\begin{equation}
    \theta : \mathrm{id}_\BordX \Rightarrow \mathrm{id}_\BordX,
\end{equation}
of the identity functor on $\BordX$. Concretely, this means that the twists commute with all bordisms,
\begin{equation}
    \theta_{B} \circ N = N \circ \theta_{A}, \quad \forall N : A \to B.
\end{equation}
As a result, in order for $\BordX$ to admit a unitary structure, the class $\mathcal{X}$ of source fields cannot include any fermionic fields or operator insertions, meaning sources which transform nontrivially under $\theta_B$.

To avoid confusion, we emphasize that $\BordX$ can still include source fields, such as spin structures or spin$^c$ gauge fields, which allow the presence of fermionic degrees of freedom in a QFT defined on $\mathcal{X}$-manifolds. As an illustrative example, note that while a background spin$^c$ gauge field $A_\mu$ couples to dynamical spinor fields $\psi$, this interaction takes the form $\overline{\psi} A_\mu \gamma^\mu \psi$, so $A_\mu$ only couples to the fermion bilinear $\overline{\psi} \gamma^\mu \psi$, a bosonic operator. What we do not allow are fields, such as background spinor fields, which would couple directly to fermionic operators.

Though structurally important to our construction throughout most of this paper, there is a clear fix, which we outline in Section \ref{sec:fermionic_sources}. The issue with fermionic sources is not really unitarity, but symmetric monoidality: fermionic sources must anticommute at spacelike separation, rather than commute. For comparison, consider the symmetric monoidal category $\sVec$ of super-vector spaces, whose symmetry is defined by the Koszul sign rule. Importantly, the morphisms in $\sVec$ are only the bosonic linear maps. If we tried to include fermionic linear maps as morphisms in $\sVec$, while still using the Koszul sign rule, we would obtain a category which is not symmetric monoidal in the usual sense.\footnote{The resulting category is a symmetric monoidal category \textit{enriched} over $\sVec$. Our approach for including fermionic sources, described in Section \ref{sec:fermionic_sources}, requires working with $\sVec$-enriched categories from the start.}

\subsubsection{Traces, super-traces, kets, and bras}\label{sec:traces_and_fermions}

One of the most important benefits, for this paper, of choosing a unitary structure on $\BordX$ is that it allows us to define the trace of a bordism. Heuristically, the trace of a bordism $N : B \to B$ is obtained by gluing the outgoing boundary of $N$ to the incoming boundary, in order to obtain a closed manifold.

However, there is a crucial subtlety in the definition of the trace involving fermions. To illustrate this, consider again the spin bordism category $\Bord_d^{Spin}$. The spin manifold obtained by directly gluing the two ends of a spin cylinder $C_B$ to each other is $S^1_\mathrm{p} \times B$, where $S^1_\mathrm{p}$ denotes a circle with periodic (non-bounding) spin structure. This is because, by directly gluing the two ends, we require fermion fields to come back to themselves on the nose when carried around the temporal circle.

In any unitary fermionic QFT, the spin manifold $S^1_\mathrm{p} \times B$ computes the super-trace,
\begin{equation}
    S^1_\mathrm{p} \times B \rightsquigarrow \str\big(e^{- \beta \widehat{H}_B}\big) \defined \tr\big((-1)^{\widehat{F}_B} e^{-\beta \widehat{H}_B}\big),
\end{equation}
of the (un-normalized) thermal density matrix on $B$. If we wanted the ordinary trace, we should instead work with the spin manifold $S^1_\mathrm{ap} \times B$, where $S^1_\mathrm{ap}$ denotes a circle with the anti-periodic (bounding) spin structure. More precisely, in a not-necessarily unitary fermionic QFT, the spin manifold $S^1_\mathrm{ap} \times B$ computes,
\begin{equation}\label{eq:str_minus_one_to_the_F}
    S^1_\mathrm{ap} \times B \rightsquigarrow \str\big(\widehat{\theta}_B e^{- \beta \widehat{H}_B}\big) = \tr\big((-1)^{\widehat{F}_B} \widehat{\theta}_B e^{- \beta \widehat{H}_B}\big),
\end{equation}
where $\widehat{\theta}_B$ denotes the operator on Hilbert space implementing the twist automorphism $\theta_B$. In a unitary fermionic QFT, which must satisfy spin-statistics, we have $\widehat{\theta}_B = (-1)^{\widehat{F}_B}$, so \eqref{eq:str_minus_one_to_the_F} reduces to $\tr(e^{- \beta \widehat{H}_B})$, as expected.

As a result, the operation which glues the two ends of a bordism $N : B \to B$ directly corresponds not to the trace, but to the super-trace. Thus, we define the super-trace of a bordism to be,
\begin{equation}
    \str_\Bord(N) \defined N/(B_\mathrm{in} \sim B_\mathrm{out}), \quad N : B \to B,
\end{equation}
the closed manifold obtained by directly gluing the incoming boundary $B_\mathrm{in}$ of $N$ to the outgoing boundary $B_\mathrm{out}$. We then define the trace of a bordism by,
\begin{equation}
    \tr_\Bord(N) \defined \str(\theta_B \circ N) = N/(B_\mathrm{in} \sim_{\theta_B} B_\mathrm{out}), \quad N : B \to B,
\end{equation}
the closed manifold obtained by gluing the two boundaries $B_\mathrm{in}, B_\mathrm{out}$ via the twist automorphism $\theta_B$. Our assumption $\theta_B^2 = \mathrm{id}_B$ ensures that $\str_\Bord$ and $\tr_\Bord$ are the only two traces we may define on $\BordX$, as any other power of $\theta_B$ reduces to either $\mathrm{id}_B$ or $\theta_B$.\footnote{In the language of category theory, we say that the trace is \textit{balanced}, or \textit{spherical}.} By our restriction to bosonic sources only, both $\str_\Bord$ and $\tr_\Bord$ are cyclic.

\begin{figure}
    \centering
    \includegraphics[width=\linewidth]{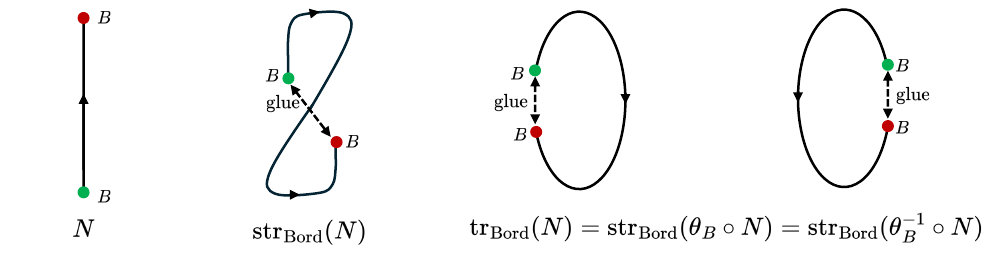}
    \caption{For a bordism $N:B\to B$ (far left), the super-trace $\str_\Bord(N)$ is defined by gluing the incoming and outgoing boundaries directly (middle left). Note that, if $N$ is embedded in some higher-dimensional space, this requires us to cross $N$ over itself if we do not want to accidentally rotate either boundary. The trace $\tr_\Bord(N)$ is instead defined by rotating one (or either) boundary by $2 \pi$ before gluing (middle right), and does not require crossing $N$ over itself at all. Finally, if we had not assumed $\theta_B^2 = \mathrm{id}_B$, we would obtain yet a third trace, defined by rotating by $2 \pi$ in the opposite direction (far right). This last distinction disappears for us, since $\theta_B^{-1} = \theta_B$ by assumption.}
    \label{fig:traces}
\end{figure}

As a final note, we observe that the trace on $\BordX$ induces a trace pairing for any two parallel bordisms,
\begin{equation}\label{eq:trace_pairing_on_bord}
    (N_1, N_2) \mapsto \tr_\Bord(N_1^\dagger \circ N_2), \quad N_1, N_2 : A \to B,
\end{equation}
which is analogous to the trace pairing on operators between any two Hilbert spaces. To understand the trace pairing, note that the trace pairing on operators $S, T : \mathcal{H}_1 \to \mathcal{H}_2$ satisfies,
\begin{equation}\label{eq:trace_of_operators_equals_HS_inner_prod}
    \tr(S^\dagger T) = S_i^j T_j^i = (S^{j \bar \imath} G_{\bar \imath i})(T^{i \bar \jmath} G_{\bar \jmath j}) = \braket{S | T},
\end{equation}
where $\ket{S}, \ket{T} \in \mathcal{H}_2 \otimes \overline{\mathcal{H}}_1$ are the states obtained by raising one index using the Hermitian metric.

The trace pairing \eqref{eq:trace_pairing_on_bord} on bordisms satisfies an analogous equation. First, for a bordism $N : A \to B$, we define the \textit{ket bordism} associated to $N$ by,
\begin{equation}\label{eq:ket_of_bord}
    \ket{N}_\Bord : \varnothing \to B \sqcup \overline{A}, \quad \ket{N}_\Bord \defined (N \circ \rho_A^{-1})^{\vee_{\overline{A}^\vee}}.
\end{equation}
Thus, $\ket{N}_\Bord$ is defined by pre-composing with $\rho_A^{-1}$ to obtain a bordism $\overline{A}^\vee \to B$, then partially dualizing to obtain a bordism $\varnothing \to B \sqcup \overline{A}$. For two parallel bordisms $N_1, N_2 : A \to B$, the analog of \eqref{eq:trace_of_operators_equals_HS_inner_prod} is,
\begin{equation}\label{eq:tr_pairing_bord_is_bord_overlap}
    \tr_\Bord(N_1^\dagger \circ N_2) = \braket{N_1 | N_2}_\Bord,
\end{equation}
where the \textit{bra bordism} $_\Bord\bra{N_1}$ is defined by,
\begin{equation}
    _\Bord\bra{N} \defined \ket{N}_\Bord^\dagger,
\end{equation}
and the inner product in \eqref{eq:tr_pairing_bord_is_bord_overlap} simply denotes composition. We prove \eqref{eq:tr_pairing_bord_is_bord_overlap} in Appendix \ref{app:daggers_duals_conjugates} as Proposition \ref{prop:trace_pairing_is_inner_product}, as well as establish a number of other useful properties of the traces and kets associated to bordisms.

\subsection{The bordism group}\label{sec:bordism_group}

We close our discussion of the source category $\BordX$ by describing an important invariant: the \textit{bordism group},
\begin{equation}
    \Omega_{d-1}^\mathcal{X} \defined \pi_0 (\BordX),
\end{equation}
whose elements are equivalence classes $[B]$ of $(d-1)$-dimensional $\mathcal{X}$-manifolds. In $\Omega_{d-1}^\mathcal{X}$, we have $[B_1] = [B_2]$ whenever $B_1$ and $B_2$ are \textit{$\mathcal{X}$-bordant}, meaning connected by any $\mathcal{X}$-bordism whatsoever. At least with our restrictions on germs, $\mathcal{X}$-bordism is an equivalence relation: it is reflexive due to the existence of cylinders, it is symmetric via the $\dagger$-structure, and it is transitive due to composition of bordisms. The symmetric monoidal structure on $\BordX$ endows $\Omega_{d-1}^\mathcal{X}$ with the structure of an abelian group,
\begin{equation}
    [B_1] + [B_2] = [B_1 \sqcup B_2], \quad 0 = [\varnothing], \quad -[B] = [\overline{B}].
\end{equation}

We will see that the bordism group provides a classical obstruction to ER = EPR for any (non-zero) theory defined on smooth $\mathcal{X}$-manifolds, as anticipated in \cite{McNamara:2022xkg}. This classical obstruction is a shadow of the full obstruction to ER = EPR described by the baby universe category $\CBU$, and ultimately we obtain a grading of $\CBU$ by bordism classes. In order to avoid the obstruction arising from bordism, one would need to generalize $\BordX$ to a category of bordisms with singularities \cite{McNamara:2019rup}, as discussed for instance in \cite{lurie2008classification}. All of our constructions and main results extend to bordisms with singularities without significant modification.

\section{Axioms for the partition function}\label{sec:axioms_for_partition_function}

Having described our requirements on the sources of the GPI, we now describe our requirements on the partition function $\zeta$, ostensibly computed by the GPI (though the axioms can apply equally well regardless of how $\zeta$ was computed). We then compare these minimal axioms for $\zeta$ to an axiomatization of unitary QFT in the style of Atiyah--Segal.

\subsection{Abstract partition functions}

Fix a unitary bordism category $\BordX$. The first axiom specifies that $\zeta(M)$ is a known, finite number for all closed $d$-dimensional $\mathcal{X}$-manifolds $M$.

\begin{axiom}[Finiteness]\label{axiom:finiteness}
    For each closed $d$-dimensional source manifold $M$, we are given a finite number $\zeta(M) \in \mathbb{C}$. We assume that $\zeta(M)$ is invariant under $\mathcal{X}$-isomorphism.
\end{axiom}

\noindent Gravitationally, our assumption that $\zeta(M)$ is finite for closed manifolds $M$ may be motivated by the observations of \cite{Hamada:2021yxy} related to the finiteness of gravitational amplitudes. In QFT, the finiteness of $\zeta(M)$ implies that a QFT with partition function $\zeta$ must be \textit{compact}, as discussed below.

We assume that the partition function $\zeta(M)$ satisfies the following two kinematic conditions.

\begin{axiom}[Reality]\label{axiom:reality}
    We assume that $\zeta$ takes conjugation of source manifolds to complex conjugation, meaning that $\zeta(\overline{M}) = \overline{\zeta(M)}$ for all closed $M$.
\end{axiom}

\begin{axiom}[Multiplicativity]\label{axiom:multiplicativity}
    We assume that $\zeta$ is multiplicative on disjoint unions, so that we have $\zeta(M_1 \sqcup M_2) = \zeta(M_1) \zeta(M_2)$ for any two source manifolds $M_1, M_2$. Moreover, we assume that $\zeta(\varnothing) = 1$.
\end{axiom}

\noindent Taken together, Axioms \ref{axiom:finiteness}, \ref{axiom:reality}, and \ref{axiom:multiplicativity} assert that $\zeta$ defines a $*$-monoid homomorphism from the monoid of closed $d$-dimensional $\mathcal{X}$-manifolds under disjoint union to the complex numbers under multiplication.

Axiom \ref{axiom:reality} may be motivated in quantum gravity from the fact that CPT must be a gauge symmetry in the bulk \cite{Harlow:2020bee}, in the sense that conjugating the boundary of the GPI must produce a complex-conjugate partition function.

Axiom \ref{axiom:multiplicativity}, though unnatural from the perspective of the GPI, is motivated by the assumption that we have already resolved any breakdown of partition-function factorization, and is made for technical simplicity (see Section \ref{sec:including_alpha_sectors} for a sketch of how Axiom \ref{axiom:multiplicativity} might be dropped). The resolution of partition-function factorization might be achieved by projecting onto an $\alpha$-sector, or, equivalently, by finely tuning our given GPI to ensure partition-function factorization \cite{Blommaert:2021fob,Saad:2021rcu,Saad:2021uzi,Gesteau:2024gzf}. Alternatively, we could simply assume we are working with a UV-complete theory of quantum gravity, such as string theory, in which the issue of partition-function factorization never arises in the first place \cite{McNamara:2020uza, Eberhardt:2021jvj}.\footnote{This last approach, while justified given a UV-complete theory, also likely renders the whole question of Hilbert-space factorization moot.}

In addition to satisfying Axioms \ref{axiom:finiteness}-\ref{axiom:multiplicativity}, we expect the partition function to depend continuously on sources, including the choice of Riemannian metric. In \cite{Colafranceschi:2023urj}, a very mild continuity assumption was made, requiring only continuity in the length of an embedded cylinder. We make a stronger continuity assumption, as it seems to be required for a few auxiliary results (though, interestingly, not for the proof of our main theorem).\footnote{We thank Tom Banks for encouraging us to make this stronger continuity assumption.}

\begin{axiom}[Continuity]\label{axiom:continuity}
    We assume that $\zeta(M)$ depends continuously on sources in our chosen topology on the space of $\mathcal{X}$-structures on $M$.
\end{axiom}

\noindent Combining Axiom \ref{axiom:continuity} with our assumption that $\zeta$ is invariant under isomorphisms of $\mathcal{X}$-manifolds, we could alternatively say that $\zeta$ is a continuous function on the moduli space of closed $d$-dimensional $\mathcal{X}$-manifolds. We make no assumptions about the behavior of $\zeta$ at the boundaries of moduli space where source fields such as the metric degenerate (in contrast to \cite{Friedan:1986ua}). Note that Axiom \ref{axiom:continuity} is vacuous when $\mathcal{X}$ is topological.

The four axioms already described are kinematic, and the space of partition functions that satisfy them is not meaningfully smaller than the space of all possible functions $\zeta(M)$. This is not the case for our last axiom: reflection positivity, the Euclidean avatar of unitarity. As discussed above, reflection positivity is required in order for our partition function to have a consistent interpretation in terms of wavefunction overlaps. From the perspective of the GPI, reflection positivity is not at all obvious \cite{Marolf:2020xie,Colafranceschi:2023urj,Colafranceschi:2023txs,DiUbaldo:2026rly,Maloney:2007ud,Keller:2014xba,Benjamin:2019stq,Benjamin:2020mfz,Maxfield:2020ale,DiUbaldo:2023hkc}. Nevertheless, circumstantially, the GPI seems to behave consistently with unitarity in more ways than may have been naively expected, including giving an accurate count of black hole microstates and a derivation of the Page curve, as discussed in the Introduction.

To formulate reflection positivity, consider any $(d-1)$-manifold $B \in \BordX$, and two bordisms from the empty set into $B$.
\begin{equation}
    M_1, M_2 \in \BordX(\varnothing \to B).
\end{equation}
We view the bordisms $M_1, M_2$ as ways to prepare two states $\ket{M_1}, \ket{M_2}$ on the spatial manifold $B$. As discussed in Section \ref{sec:OS_conjugation}, the closed manifold corresponding to the wavefunction overlap of those two states is given by,
\begin{equation}
    \braket{M_1 | M_2}_\Bord = M_1^\dagger \circ M_2,
\end{equation}
as depicted above in Figure \ref{fig:inner_product}. Thus, knowing nothing about any underlying QFT, we already know that the wavefunction overlap must be given by,
\begin{equation}\label{eq:formal_inner_product}
    \braket{M_1 | M_2} = \zeta\big(\braket{M_1|M_2}_\Bord\big) = \zeta(M_1^\dagger \circ M_2).
\end{equation}
This expression is analogous to the inner product derived from Euclidean correlation functions in the OS/Wightman approach to axiomatic QFT. Knowing \eqref{eq:formal_inner_product} is enough to test whether the inferred inner product is positive definite.

\begin{axiom}[Reflection Positivity]\label{axiom:reflection_positivity}
    For all $B \in \BordX$, all finite collections of manifolds $M_i \in \BordX(\varnothing \to B)$, and all choices $c_i \in \mathbb{C}$ of coefficients, we assume that
    \begin{equation}\label{eq:formal_inner_product_linear_comb}
        \sum_{i, j} \overline{c}_ic_j \zeta(M_i^\dagger \circ M_j) \geq 0.
    \end{equation}
    In other words, we assume that the Hermitian matrix with elements $\zeta(M_i^\dagger \circ M_j)$ defines a positive semidefinite Gram matrix for all $B \in \BordX$.
\end{axiom}

A simple consequence of reflection positivity is that the partition function on a \textit{double}, $\zeta(M^\dagger \circ M)$, must be positive. However, in general, we need to consider arbitrary sesquilinear combinations, for the simple fact that a Hermitian matrix with positive entries on the diagonal is not necessarily positive semidefinite. The only general exception is when the matrix is known to be rank-one (as will be the case for invertible theories), in which case positivity of $\zeta$ on doubles is sufficient \cite{Freed:2016rqq}.

Compared to Axioms \ref{axiom:finiteness}-\ref{axiom:continuity}, which are purely kinematic and easy to satisfy, reflection positivity is an extremely strong assumption about $\zeta$. For any fixed $B$ and any finite collection of states, reflection positivity is a finite system of inequalities. However, reflection positivity for all manifolds and spatial slicings at once is an infinite system of inequalities. As a result, reflection positivity can cut the space of partition functions down dramatically, as is familiar in the context of the bootstrap.

We note that reflection positivity is particularly constraining when applied to \textit{disconnected} spatial slices. For instance, applying reflection positivity to disconnected spatial slices directly implies integrality of the degeneracies of every energy level, which is sometimes taken as an additional assumption in the context of the bootstrap. We will see this in practice in the example discussed in Section \ref{sec:Ex_TopoQM}, based on \cite{Marolf:2020xie,Maxfield:2023mdj,Colafranceschi:2023urj,Marolf:2024adj}. We will see it again at the core of the proof of the DR reconstruction theorem \cite{DoplicherRoberts89,muger2007abstract}, and thus also at the core of our proof of Theorem \ref{thm:main}. If one is looking for a driving force behind Theorem \ref{thm:main}, it is reflection positivity applied to disconnected spatial slices, which is precisely also where the question of Hilbert-space factorization arises in the first place.

\subsection{Functorial quantum field theories}\label{sec:functorial_QFT}

Let us now give our working definition of a unitary functorial QFT, in the spirit of Atiyah and Segal. Our definition takes its most direct inspiration from the approach of Kontsevich and Segal \cite{Kontsevich:2021dmb}, though with the restriction (discussed in Section \ref{sec:germs}) to spatial slices with infinitesimal time-reversal and time-translation symmetries.

\begin{definition}\label{defn:functorial_QFT}
    A \textit{unitary quantum field theory} of dimension $d$ is a non-degenerate, unitary, symmetric monoidal functor
    \begin{equation}
        \BordX \to \sHilb, \quad B \mapsto \widetilde{\CH}_B, \quad N \mapsto \widehat{N},
    \end{equation}
    from a unitary bordism category to the category of super-Hilbert spaces.\footnote{As is standard in the physics literature, we take a \textit{super-Hilbert space} to be a Hilbert space equipped with an orthogonal decomposition into bosonic and fermionic sub-spaces. In particular, for us, the squared norm of a fermionic state is a positive real number. This choice is closely related to our conventions for unitary structures, as used in Appendix \ref{app:daggers_duals_conjugates}.}
\end{definition}

\noindent We have taken the target of our QFTs to be $\sHilb$, in order to accommodate fermionic statistics. By the spin-statistics theorem, every unitary QFT on a bosonic bordism category factors through the subcategory $\Hilb \subset \sHilb$ of purely bosonic Hilbert spaces, and so for bosonic theories we could take the target to be $\Hilb$ with no loss of generality. We will suppress the prefix ``super'' in super-Hilbert space for the remainder of this section.

Let us unpack the various terms in Definition \ref{defn:functorial_QFT}. Firstly, a QFT defined on $\BordX$ consists of the following basic data:

\begin{itemize}

    \item {\bf Spaces of quantum states:} For each spatial manifold $B$, a Hilbert space $\widetilde{\mathcal{H}}_B$.\footnote{We use the notation $\widetilde{\mathcal{H}}_B$ for the Hilbert spaces of an actual QFT to distinguish them from the universal Hilbert spaces $\mathcal{H}_B$ constructed in Section \ref{sec:universal_construction} and used throughout.}

    \item {\bf Euclidean evolution operators:} For each bordism $N : B_1 \to B_2$, a bounded linear operator,
    \begin{equation}
        \widehat{N} : \widetilde{\mathcal{H}}_{B_1} \to \widetilde{\mathcal{H}}_{B_2},
    \end{equation}
    which evolves quantum states along the bordism $N$.
    
\end{itemize}

\noindent  The Hilbert spaces and evolution operators must satisfy a list of constraints, given as follows:

\begin{itemize}
    
\item {\bf Functoriality:} The condition that the operators assigned to composable bordisms multiply to the operator assigned to their composite,
\begin{equation}
    \big(N_1 \circ N_2 = N_3\big) \implies \big( \widehat{N}_1 \widehat{N}_2 = \widehat{N}_3 \big).
\end{equation}
Functoriality encodes locality in time, as motivated by the QFT path integral.

\item {\bf Non-degeneracy:} The replacement, given the non-unital nature of $\BordX$, for the standard condition that a functor must preserve identity morphisms. Specifically, we demand that, for any spatial manifold $B$, the evolution operators $\widehat{C}_B(\beta)$ along cylinder bordisms converge,
\begin{equation}
    \widehat{C}_B(\beta) \to \mathrm{id}_{\widetilde{\CH}_B}, \quad \beta \to 0,
\end{equation}
to the identity operator in the strong operator topology. Non-degeneracy is closely related to the continuity property of \cite{Kontsevich:2021dmb}, as well as the notion of firmness \cite{quillen1997module} which will be used in \cite{M-JF-R_wip}. In the topological case, non-degeneracy reduces to unitality.

Non-degeneracy is needed to exclude uninteresting counterexamples to our main theorem, which involve additional decoupled sectors on which every bordism acts by the zero operator. Such decoupled sectors would, formally, be spaces of infinite-energy states, and should not be included in the physically-relevant notion of QFT.

\item {\bf Symmetric monoidality:} The data of unitary natural isomorphisms,
\begin{equation}
    \widetilde{\mathcal{H}}_{B_1 \sqcup B_2} = \widetilde{\mathcal{H}}_{B_1} \otimes \widetilde{\mathcal{H}}_{B_2}, \quad \widetilde{\mathcal{H}}_\varnothing = \mathbb{C},
\end{equation}
compatible with the symmetric monoidal structures. Symmetric monoidality encodes a minimal amount of spatial locality.

\item {\bf Unitarity:} The condition that the QFT be a unitary functor. Concretely, first, the operator assigned to the adjoint of a bordism must be the adjoint operator,
\begin{equation}
    \widehat{N^\dagger} = \widehat{N}^\dagger.
\end{equation}
Moreover, we require the data of an anti-unitary CPT conjugation map,
\begin{equation}
    \widehat{\mathcal{T}} : \widetilde{\mathcal{H}}_B \to \widetilde{\mathcal{H}}_{\overline{B}},
\end{equation}
satisfying $\widehat{\mathcal{T}}^2 = 1$ and a naturality condition with respect to bordisms. In fact, $\widehat{\mathcal{T}}$ is not additional data, as one form of the CPT theorem asserts that we may extract a canonical choice of $\widehat{\mathcal{T}}$ assuming only that our QFT is a non-degenerate symmetric monoidal $\dagger$-functor.\footnote{We thank Luuk Stehouwer for this observation.} The anti-unitary map $\widehat{\mathcal{T}}$ is the only universal data related to the CPT theorem when a QFT is quantized on an arbitrary spatial manifold $B$, as recently discussed in detail in \cite{Witten:2025ayw}.\footnote{If $B$ happens to admit an $\mathcal{X}$-isomorphism $\mathcal{R} : \overline{B} \to B$, then the combination $\widehat{\mathcal{R}}\widehat{\mathcal{T}} : \widetilde{\CH}_B \to \widetilde{\CH}_B$ is more closely related to the CRT operator traditionally studied on Minkowski space, particularly when $\mathcal{R}$ reflects a single spatial coordinate. It is worth noting an important special case, which is when the $\mathcal{X}$-isomorphism $\mathcal{R}$ is the identity on points. Such an isomorphism exists for all $B$ in the unoriented bosonic case, and exists when $B$ is oriented in the pin$^\pm$ fermionic case (as required to lift the identity map to a pin$^\pm$-diffeomorphism). In such cases, we typically denote the resulting anti-unitary operator $\widehat{\mathcal{R}}\widehat{\mathcal{T}}$ by $\widehat{\mathsf{T}}$ instead, as discussed in \cite{Witten:2025ayw}, and refer to it as a \textit{time-reversal symmetry}. Note that, even in the fermionic case, we take $\widehat{\mathcal{T}}^2 = 1$. The use of the pin$^\pm$ isomorphism $\mathcal{R} : \overline{B} \to B$ in defining $\widehat{\mathsf{T}}$ ensures that we have $\widehat{\mathsf{T}}^2 = (-1)^{\widehat{F}}$ in the pin$^+$ case and $\widehat{\mathsf{T}}^2 = 1$ in the pin$^-$ case, as expected (see \cite[Equation 3.34]{Freed:2016rqq}).}

\end{itemize}

Given a unitary QFT, we may evaluate it on closed $d$-dimensional $\mathcal{X}$-manifolds $M$, viewed as $\mathcal{X}$-bordisms $M : \varnothing \to \varnothing$. Using monoidality, the corresponding evolution operators are linear maps $\widehat{M} : \mathbb{C} \to \mathbb{C}$, and hence correspond to multiplication by a complex number. We denote this number by $\zeta(M)$, and refer to it as the \textit{partition function} of the QFT on $M$. Functoriality and unitarity imply that we have
\begin{equation}
    \zeta(M_1 \sqcup M_2) = \zeta(M_1) \zeta(M_2), \quad \zeta(\varnothing) = 1, \quad \zeta(\overline{M}) = \overline{\zeta(M)}.
\end{equation}
Finally, positivity of the inner products on the Hilbert spaces $\widetilde{\mathcal{H}}_B$ implies that $\zeta$ is reflection positive.

We have thus shown that four of our axioms for the partition function (namely, Axioms \ref{axiom:finiteness}, \ref{axiom:reality}, \ref{axiom:multiplicativity}, and \ref{axiom:reflection_positivity}) follow from Definition \ref{defn:functorial_QFT}, and so our axioms for the partition function are, almost, strictly weaker than our axioms for a unitary QFT. We have intentionally neglected to assume that the evolution operators $\widehat{N}$ depend continuously on source fields, besides the strong convergence $\widehat{C}_B(\beta) \to 1$ as $\beta \to 0$. Remarkably, continuity of $\zeta$ in sources actually implies continuity of $\widehat{N}$ in sources (in the weak-$*$ topology). This will follow as a consequence of Theorem \ref{thm:main}, as the unique reconstructed QFT has this property by Proposition \ref{prop:continuity_of_morphisms}, and our proof of Theorem \ref{thm:main} does not require us to assume any continuity of our QFT besides the continuity of $\zeta$.

Definition \ref{defn:functorial_QFT} is an axiomatization of what is typically called a \textit{compact QFT} in the physics literature. The exact definition of a compact QFT varies, but the standard requirements include a finite thermal partition function when quantized on a compact spatial manifold (this implies that the Hamiltonian on a compact spatial manifold has a discrete energy spectrum). Definition \ref{defn:functorial_QFT} requires the partition function on every compact Euclidean manifold to be finite, which includes the thermal partition functions as a special case. Thus, we exclude theories with continuous spectra, such as the Liouville CFT in $d = 2$, as well as theories with unbounded Hagedorn growth of the spectrum. We also exclude non-compact TQFTs, such as the BF theories arising from the so-called ``flat gauging'' of continuous gauge groups (see, e.g., \cite{Brennan:2024fgj,Bonetti:2024cjk}).

One may worry that the closed-manifold partition function of a physical QFT is typically UV divergent, and thus requires regularization and renormalization. Our notion of QFT axiomatizes the result of having performed this renormalization in some fixed choice of regularization scheme and counterterms. One may further worry that this makes our framework non-universal, meaning, counterterm-dependent. We assuage this worry in Section \ref{sec:counterterms}, showing that our construction depends on counterterms only up to the ambiguity in defining fermion parity for states on certain manifolds (for example, the ambiguity in defining fermion parity in the Ramond sector of the superstring worldsheet).

\section{Spaces of gravitational states}\label{sec:spaces_of_grav_states}

We now turn to the problem of building Hilbert spaces $\mathcal{H}_B$ of states prepared by manifolds with boundary $B$, given only a closed-manifold partition function $\zeta$ on a unitary bordism category $\BordX$. Viewed gravitationally, the Hilbert spaces $\mathcal{H}_B$ consist of all states preparable by cutting open the GPI. If we have included sources for all degrees of freedom in the bulk EFT, then this is also precisely the set of states obtained by allowing effective field theory (EFT) degrees of freedom to evolve and interact gravitationally.

Pieces of this construction have appeared previously in the recent quantum gravity literature \cite{Marolf:2020xie,Gesteau:2020wrk,Colafranceschi:2023urj,Marolf:2024adj,Chen:2025fwp,Zhao:2026mpl,DiUbaldo:2026rly}, as well as in the formal TQFT literature under the name of the \textit{universal construction} \cite{Blanchet:1995TQFT} (see also \cite{freedman2005universal,Calegari:2008cw, khovanov2020universal}), which we adopt here. The universal construction is a form of the standard Gelfand--Naimark--Segal (GNS) construction \cite{Gelfand:1943normed,Segal:1947operator}, as used in the OS/Wightman approach to axiomatic QFT. In the GNS construction, one starts with a state on a $*$-algebra, and uses the state to build a Hilbert space representation of formally defined states.

In a precise sense, the abstract partition function $\zeta$ defines a state (in fact, a tracial state) on the unitary bordism category $\BordX$, and the universal construction is the corresponding form of the GNS construction.\footnote{We thank Theo Johnson-Freyd for this perspective.} Thus, we will use $\zeta$ to build a formal, non-degenerate Hilbert space representation of $\BordX$, or in other words, a non-degenerate unitary functor,
\begin{equation}
    \BordX \to \Hilb,
\end{equation}
which is, almost, a unitary QFT. However, as we discuss in Section \ref{sec:ER=EPR}, the universal construction fails to define a QFT by failing to produce tensor-factorized Hilbert spaces on disconnected spatial manifolds. This failure of the universal construction to define a unitary QFT leads to the central question of study in this paper.

\subsection{The universal construction}\label{sec:universal_construction}

To start, fix $B \in \BordX$. The \textit{universal pre-Hilbert space} $\mathcal{H}_B^\mathrm{pre}$ is defined to be the free vector space
\begin{equation}
    \mathcal{H}_B^\mathrm{pre} = \mathbb{C}\left[\BordX(\varnothing \to B)\right],
\end{equation}
of finite $\mathbb{C}$-linear combinations of formal states $\ket{M}$ associated to $\mathcal{X}$-bordisms $M: \varnothing \to B$. We equip $\mathcal{H}_B^\mathrm{pre}$ with the sesquilinear form,
\begin{equation}\label{eq:inner_prod_on_H_pre}
        \Big\langle\sum_i c_i M_i \Big| \sum_j c_j M_j\Big\rangle \defined \sum_{i, j} \overline{c}_ic_j \zeta(M_i^\dagger \circ M_j),
\end{equation}
which is positive semidefinite by reflection positivity. As we have restricted to purely bosonic source fields (see Section \ref{sec:germs}), all states in the universal pre-Hilbert spaces will be bosonic, even when $\BordX$ is fermionic in the sense of including background spin structures or similar.

Next, the \textit{universal Hilbert space} $\mathcal{H}_B$ is defined to be the Cauchy completion of the pre-Hilbert space $\mathcal{H}_B^\mathrm{pre}$ with respect to the seminorm induced by the inner product \eqref{eq:inner_prod_on_H_pre}. Taking this Cauchy completion automatically includes the quotient by the subspace of null states, which forms the closure of the origin in the non-Hausdorff topology induced by the seminorm. This Cauchy completion is automatically a Hilbert space, by construction. Note that, modulo null states, any components of $M$ which are not connected to $B$ may be replaced by their partition function using the multiplicativity of $\zeta$.

Let us comment on the interpretation of the state $\ket{M}$. In QFT, the state $\ket{M}$ is the state prepared by a Euclidean path integral over $M$. In contrast, gravitationally, $\ket{M}$ represents the state prepared by a GPI with $M$ held fixed in the Euclidean past, as illustrated in Figure \ref{fig:GPI_Ham_constraint}. But, since bulk time evolution in quantum gravity is pure gauge, one could justifiably say that the state $\ket{M}$ is a state where $M$ \textit{is} the bulk geometry.\footnote{At least in the sense of a field-basis eigenstate, if not in the sense of a coherent state.} Moreover, the source fields correspond holographically to dynamical fields in the bulk, so the configuration of source fields on $M$ specifies the precise configuration of the corresponding bulk fields. Note that this holographic interpretation is always available, regardless of how $\zeta$ was computed. Thus, any QFT can formally be regarded as a theory of quantum gravity, in the somewhat tautological sense that it has a subspace of quantum states parametrized by manifolds.

\begin{figure}
    \centering
    \includegraphics[width=0.4\linewidth]{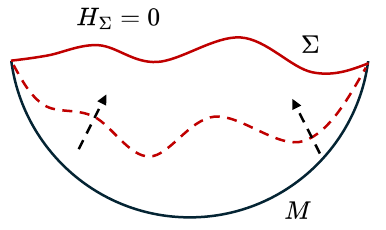}
    \caption{The Euclidean GPI prepares a state $\ket{M}$ by integrating over bulk spacetimes with $M$ held fixed in the Euclidean past. The final slice of this path integral, denoted $\Sigma$, is where the prepared wavefunction is read off, which must satisfy the non-perturbative Hamiltonian constraint $H_\Sigma=0$. However, we can alternatively view this state as a field-basis eigenstate in which the bulk simply \textit{is} the manifold $M$. These two descriptions, being related by bulk time evolution, are gauge equivalent.}
    \label{fig:GPI_Ham_constraint}
\end{figure}

The holographic interpretation of the state $\ket{M}$ is unambiguous in the pre-Hilbert space $\mathcal{H}_B^\mathrm{pre}$, where each state admits a unique expression as a superposition of manifold states $\ket{M}$. As $M$ is only defined up to $\mathcal{X}$-isomorphisms fixing $B$, the states $\ket{M} \in \mathcal{H}_B^\mathrm{pre}$ automatically solve the bulk momentum constraint.\footnote{If we wanted, we could have let $\mathcal{H}_B^\mathrm{pre}$ be the free vector space on the set of $\mathcal{X}$-bordisms $M : \varnothing \to B$ considered up to \textit{equality}, not just up to $\mathcal{X}$-isomorphism. Then the states in $\mathcal{H}_B^\mathrm{pre}$ would \textit{not} solve the bulk momentum constraint. Nevertheless, the completed Hilbert spaces $\mathcal{H}_B$ would stay the same, since the difference of isomorphic manifolds would necessarily be null by our assumption that $\zeta$ is invariant under $\mathcal{X}$-isomorphism.} However, we have not taken into account the bulk Hamiltonian constraint, so one should view $\mathcal{H}_B^\mathrm{pre}$ as a kinematic vector space where we can unambiguously describe the geometry and topology of the bulk spatial slice without having to worry about bulk time evolution.

In contrast, $\mathcal{H}_B$ is a Hilbert space of physical states arising from projecting states in the kinematic vector space $\mathcal{H}_B^\mathrm{pre}$ onto the space of solutions to the non-perturbative Hamiltonian constraint. This projection takes the form of the quotient by null states.\footnote{See \cite{Held:2025mai} for recent discussion of the relationship between this approach and more traditional approaches to solving the Hamiltonian constraint.} As a result, the states in $\mathcal{H}_B$ generically do not have unique expressions as linear combinations of manifolds. Gravitationally, this corresponds to the fact that the geometry \cite{dewitt1967quantum,Kuchar:1993ne,Isham:1992ms} and topology \cite{Jafferis:2017tiu} of a bulk Cauchy slice are not gauge invariant after the Hamiltonian constraint is imposed. One can ask how well the inner product on $\mathcal{H}_B$ distinguishes the states $\ket{M}$ for different $M$, in order to get a diagnostic of how close the bulk may be to semiclassical \cite{Miyaji:2015woj,Lashkari:2015hha,Bak:2017rpp,Marolf:2017kvq,Belin:2018fxe,Czech:2023zmq}.

A natural worry is that the universal Hilbert spaces we have constructed are too big to possibly be the Hilbert spaces of a QFT, as we have states in $\mathcal{H}_B$ for every source manifold $M$ with boundary $B$. However, as shown in Appendix \ref{app:source_continuity}, the universal Hilbert spaces $\mathcal{H}_B$ are always separable, and the states $\ket{M} \in \mathcal{H}_B$ vary continuously as a function of sources. These results are direct consequences of continuity of $\zeta$ and our assumption that the spaces of sources on compact manifolds are seperable. 

So far, we have built Hilbert spaces $\mathcal{H}_B$ for each spatial manifold $B \in \BordX$. Now, we supply the data of Euclidean evolution operators $\widehat{N} : \mathcal{H}_{A} \to \mathcal{H}_{B}$ for each bordism $N : A \to B$. By definition, $\widehat{N}$ acts by composition,
\begin{equation}\label{eq:definition_of_N_hat}
    \widehat{N} \sum_i c_i \ket{M_i} \defined \sum_i c_i \ket{N \circ M_i}.
\end{equation}
A priori, \eqref{eq:definition_of_N_hat} defines an operator $\widehat{N} : \mathcal{H}_{A}^\mathrm{pre} \to \mathcal{H}_{B}^\mathrm{pre}$, which will extend uniquely to the universal Hilbert spaces provided it is bounded in operator norm.

To show that $\widehat{N}$ is bounded, we use the fundamental trace inequality of \cite{Colafranceschi:2023urj}, reproduced in Appendix \ref{app:trace_inequality} as Proposition \ref{prop:trace_inequality}. Corollary \ref{cor:operator_norm_bound} supplies the bound,
\begin{equation}
    \lVert \widehat{N} \rVert \leq \sqrt{\braket{N | N}},
\end{equation}
on the operator norm of $\widehat{N}$, where the state $\ket{N} \in \mathcal{H}_{B \sqcup \overline{A}}^\mathrm{pre}$ is defined as in \eqref{eq:ket_of_bord} by using the unitary structure of $\BordX$ to bend the incoming boundary $A$ of $N$ into an outgoing boundary component $\overline{A}$. Thus, $\widehat{N}$ is bounded, and extends uniquely to a map $\CH_A \to \CH_B$ between the Cauchy completions.

We now show that the assignment $B \mapsto \CH_B, N \mapsto \widehat{N}$ assembles into a non-degenerate unitary functor $\BordX \to \Hilb$. First, we check functoriality. Given composable bordisms, $N_1 \circ N_2 = N_3$, we have
\begin{equation}
    \widehat{N}_1 \widehat{N}_2 \ket{M} = \widehat{N}_1 \ket{N_2 \circ M} = \ket{N_1 \circ N_2 \circ M} = \widehat{N}_3 \ket{M},
\end{equation}
on manifold states, and thus for all states by linearity and boundedness. For unitarity, first, we compute,
\begin{equation}
    \braket{M_1 | \widehat{N^\dagger} | M_2} = \braket{M_1 | N^\dagger \circ M_2} = \zeta ( M_1^\dagger \circ N^\dagger \circ M_2 ) = \braket{M_1 \circ N | M_2} = \braket{M_1 | \widehat{N}^\dagger | M_2},
\end{equation}
and thus $\widehat{N^\dagger} = \widehat{N}^\dagger$ for manifold states, and again for all states by linearity and boundedness. Then, the anti-linear CPT conjugation map $\widehat{\mathcal{T}}$ is defined at the level of the universal pre-Hilbert spaces by,
\begin{equation}
    \widehat{\mathcal{T}} \sum_i c_i \ket{M_i} \defined \sum_i \overline{c}_i \ket{\overline{M}_i}.
\end{equation}
The operator $\widehat{\mathcal{T}}$ is anti-unitary due to the reality of $\zeta$, and thus extends uniquely to an anti-unitary map $\CH_B \to \CH_{\overline{B}}$. Finally, non-degeneracy is given by Proposition \ref{prop:approx_identity}. Altogether, the universal construction specifies a non-degenerate unitary functor $\BordX \to \Hilb$, as claimed.

One immediate consequence of the definition of the universal Hilbert spaces $\mathcal{H}_B$ is that, whenever the bordism class $[B] \in \Omega_{d-1}^\mathcal{X}$ is nonzero, the universal Hilbert space $\mathcal{H}_B$ must vanish. This is simply because, when $[B] \neq 0$, there simply are no bordisms $M : \varnothing \to B$, and so we cannot produce any states at all. More formally, $\mathcal{H}_B^\mathrm{pre}$ is the free vector space on the empty set, which is just the zero vector space.

\subsection{Hilbert-space factorization and ER = EPR}\label{sec:ER=EPR}

Above, we used the universal construction to define a non-degenerate unitary functor $\BordX \to \Hilb$, which nearly defines a functorial QFT. However, there is a key axiom of functorial QFT which we have not yet verified. This is the condition that the functor $\BordX \to \Hilb$ must be symmetric monoidal, in that we should have natural unitary isomorphisms
\begin{equation}\label{eq:monoidality?}
    \mathcal{H}_{B_1} \otimes \mathcal{H}_{B_2} \stackrel{?}{=} \mathcal{H}_{B_1 \sqcup B_2}, \quad \CH_\varnothing \stackrel{?}{=} \mathbb{C},
\end{equation}
compatible with the symmetry. Thus, we ask: is the universal construction symmetric monoidal?

Let us first discuss the second condition $\CH_\varnothing = \mathbb{C}$. Gravitationally, states in $\HBU$ correspond to states with closed bulk spatial slices, and so $\CH_\varnothing$ is the \textit{baby universe Hilbert space} $\HBU$. Multiplicativity ensures that the inner product,
\begin{equation}
    \braket{M_1 | M_2} = \zeta(M_1^\dagger \sqcup M_2) = \overline{\zeta(M_1)} \zeta(M_2),
\end{equation}
on $\mathcal{H}_\varnothing^\mathrm{pre}$ is rank-one, so $\HBU$ is one-dimensional. Moreover, the \textit{Hartle--Hawking no-boundary wavefunction} $\ket{\varnothing} \in \HBU$ \cite{Hartle:1983ai}, prepared by the empty $d$-manifold, has norm,
\begin{equation}
    \braket{\varnothing | \varnothing} = \zeta(\varnothing) = 1,
\end{equation}
again by multiplicativity. Thus, $\HBU$ is one dimensional, with a preferred unit vector, and is thus canonically unitarily isomorphic to $\mathbb{C}$.

Under this isomorphism, a state $\ket{M} \in \HBU$ prepared by a closed manifold is mapped to the complex number $\braket{\varnothing | M} = \zeta(M)$. Moreover, the operator $\widehat{M} : \HBU \to \HBU$ maps $\ket{\varnothing}$ to $\ket{M} = \zeta(M) \ket{\varnothing}$, and so $\widehat{M}$ is simply multiplication by the complex number $\zeta(M)$. So, the functor $\BordX \to \Hilb$ built by the universal construction preserves the tensor unit, and even takes a closed manifold to the given partition function $\zeta$.

However, the first part of monoidality, namely, Hilbert-space factorization, fails. The reason is that, unlike the states of a local QFT, the states in the universal construction are not prepared in a local way. The main issue is that we can have connected $d$-manifolds $M$ whose boundary is disjoint, $\partial M = B_1 \sqcup B_2$. Viewed gravitationally, such states correspond to spatial wormholes, or \textit{Einstein--Rosen (ER) bridges}. In general, we only have a comparison map
\begin{equation}\label{eq:comparison_map}
    \sqcup: \mathcal{H}_{B_1} \otimes \mathcal{H}_{B_2} \to \mathcal{H}_{B_1 \sqcup B_2}, \quad \ket{M_1} \otimes \ket{M_2} \mapsto \ket{M_1 \sqcup M_2},
\end{equation}
between the tensor products of universal Hilbert spaces and universal Hilbert spaces on the disjoint union. Multiplicativity implies that \eqref{eq:comparison_map} is an isometric embedding, but it has no reason to be surjective in general. Instead, the maps \eqref{eq:comparison_map}, together with the Hartle--Hawking state $\ket{\varnothing} \in \HBU$, endow the universal construction with a \textit{lax} symmetric monoidal structure \cite{Blanchet:1995TQFT}, which is better than nothing but not enough to qualify as a local QFT.

As discussed in the Introduction, there is a proposal, in the context of quantum gravity, for how ER bridges could fit inside a tensor product Hilbert space. This proposal goes by the name of \textit{ER = EPR} \cite{Maldacena:2001kr,Maldacena:2013xja,VanRaamsdonk:2010pw}, and asserts that, modulo the non-perturbative Hamiltonian constraint \cite{Jafferis:2017tiu}, ER bridges should all be equivalent to entangled superpositions of disconnected states. The most well-studied example is the state of the AdS-Schwarzschild wormhole, which has been identified \cite{Maldacena:2001kr} with the thermofield double state,
\begin{equation}\label{eq:TFD}
    \ket{\mathrm{TFD}_\beta} = \sum_n e^{- \beta E_n/2} \ket{n} \otimes \ket{\overline{n}},
\end{equation}
in the holographically dual CFT quantized on $S^{d-1} \sqcup \overline{S^{d-1}}$. In our language, the argument of \cite{Maldacena:2001kr} is that the AdS-Schwarzschild wormhole is the state prepared by the cylinder $C_{S^{d-1}}(\beta/2)$, viewed as a bordism $\varnothing \to S^{d-1} \sqcup \overline{S^{d-1}}$. Computed in an underlying holographic CFT, the state $\ket{C_{S^{d-1}}(\beta/2)}$ is precisely the thermofield double state $\ket{\mathrm{TFD}_\beta}$, as it is obtained by partially dualizing the Euclidean evolution operator $e^{- \beta \widehat{H}_{S^{d-1}}/2}$ into a state.

In our framework, we can give an entirely rigorous definition of the ER = EPR proposal as follows.

\begin{definition}
    Let $\zeta$ be a partition function on a unitary bordism category $\BordX$, satisfying Axioms \ref{axiom:finiteness}-\ref{axiom:reflection_positivity}. We say \textit{ER = EPR holds} for $\mathcal{X}$ and $\zeta$ if the canonical isometric map,
    \begin{equation}
        \sqcup : \mathcal{H}_{B_1} \otimes \mathcal{H}_{B_2} \to \mathcal{H}_{B_1 \sqcup B_2}
    \end{equation}
    is an isomorphism for all $B_1, B_2$.\footnote{In other words, ER = EPR holds if the canonical lax monoidal structure on the universal construction is actually strong monoidal.}
\end{definition}

When ER = EPR holds, every ER bridge lies in the subspace spanned by superpositions of disconnected states, and may therefore be written as an entangled superposition of disconnected states modulo null states. Let us emphasize that ER = EPR is a property of both the partition function $\zeta$ and the class $\mathcal{X}$. This is because ER = EPR, requiring a complete set of one-sided states, depends sensitively on which one-sided states we think we have. Even if a partition function $\zeta$ is, in principle, compatible with ER = EPR (for instance, whenever $\zeta$ comes from an actual QFT), we may have artificially or unknowingly restricted our class $\mathcal{X}$ of sources beyond what is necessary to prepare a complete set of one-sided states.

To illustrate this, note that one obvious way in which ER = EPR might fail is if the bordism group $\Omega_{d-1}^\mathcal{X}$ is non-zero \cite{McNamara:2019rup,McNamara:2022xkg}, because manifolds $B$ with nonzero bordism class necessarily have $\mathcal{H}_B = 0$. Nevertheless, we will almost always have $\mathcal{H}_{B \sqcup \overline{B}} \neq 0$, as this two-sided universal Hilbert space contains at least the cylinder states $\ket{C_B(\beta)}$. In order to restore ER = EPR, we would need to enlarge our source category $\BordX$ to include background \textit{cobordism defects} \cite{McNamara:2019rup} in order to trivialize the bordism group $\Omega_{d-1}^\mathcal{X}$. The only other way out would be if the cylinder state $\ket{C_B(\beta)}$ was null, which happens if and only if we have $\zeta(S^1_{\mathrm{ap}} \times B) = 0$. In such a case, the nontrivial bordism class $[B] \neq 0$ is not an obstruction to ER = EPR, because an underlying QFT would not need any states on $B$. In such a case, we should say that the corresponding 1-form global symmetry in the bulk is gauged \cite{McNamara:2019rup}.

We illustrate the obstruction to ER = EPR arising from bordism with the following example.

\begin{example}
    Suppose $\zeta$ is the partition function of a spin CFT in $d = 2$ dimensions, but we neglect to include any spin fields in our class $\mathcal{X}$ of sources. Then, even though the underlying spin CFT has perfectly factorizing Hilbert spaces, we cannot produce any states in the Ramond-sector Hilbert space $\mathcal{H}_{S^1_\mathrm{p}}$, because the bordism class $[S^1_\mathrm{p}] \in \Omega_1^{Spin}$ is nonzero. Nevertheless, we could still produce states on $S^1_\mathrm{p} \sqcup S^1_\mathrm{p}$, including the Ramond-sector thermofield double states. Note that including spin fields in our class $\mathcal{X}$ of sources would trivialize the bordism class $[S^1_\mathrm{p}] \in \Omega_1^\mathcal{X}$, removing the obstruction to ER = EPR.
\end{example}

The bordism group $\Omega_{d-1}^\mathcal{X}$ of source manifolds, viewed as semiclassical configurations of the bulk EFT as in \cite{McNamara:2019rup, Kaidi:2024cbx}, is insufficient to fully capture the breakdown of ER = EPR. In particular, the semiclassical bordism group is only well-defined in a particular duality frame, and only directly captures charges associated with the topology of configuration space. Physically, the issue is that $\Omega_{d-1}^\mathcal{X}$ collapses too many configurations into a single equivalence class. To fully capture the breakdown of ER = EPR, we should only identify two configurations if they are actually equivalent under the gravitational gauge constraints, not merely if they are connected by some trajectory in the GPI.

The baby universe category $\CBU$, as constructed in Section \ref{sec:CBU}, remedies this issue. While $\CBU$ is graded by $\Omega_{d-1}^\mathcal{X}$, the baby universe category remembers the actual quantum state of the spacetime manifold, and not merely its semiclassical bordism class. We will prove the following result, showing that $\CBU$ is a sharp diagnostic of the breakdown of ER = EPR.

\begin{theorem}\label{thm:condition_for_ER=EPR}
    Let $\zeta$ be a partition function on a unitary bordism category $\BordX$ satisfying Axioms \ref{axiom:finiteness}-\ref{axiom:reflection_positivity}. Then ER = EPR holds if and only if the associated baby universe category is trivial,
    \begin{equation}\label{eq:CBU_is_trivial_iff_ER=EPR}
        \CBU = \Hilb.
    \end{equation}
\end{theorem}

As a result, \eqref{eq:CBU_is_trivial_iff_ER=EPR} provides a formalization, in the axiomatic framework for quantum gravity studied in this paper, of the Swampland cobordism conjecture \cite{McNamara:2019rup} in bulk codimension two. Theorem \ref{thm:condition_for_ER=EPR} is a direct analog of the condition $\HBU = \mathbb{C}$ for partition-function factorization, which is itself a formalization, in the same framework, of the Swampland cobordism conjecture in bulk codimension one \cite{McNamara:2020uza}.

When ER = EPR fails, other familiar features of a local QFT will break down as well. One of these is the relationship between the partition function on a Euclidean circle $S^1_\beta$ of length $\beta$ and the trace of the (un-normalized) thermal density matrix $e^{- \beta \widehat{H}}$. This is because, as traces require the use of the monoidal structure to define, they need not be preserved by a lax monoidal functor. In Appendix \ref{app:geometric_symmetries}, we define the Hamiltonian operators $\widehat{H}_B$ acting on the universal Hilbert spaces, and prove the bound (Proposition \ref{prop:thermal_bound}),
\begin{equation}\label{eq:thermal_bound}
    \mathrm{tr}_{\mathcal{H}_B} (e^{- \beta \widehat{H}_B}) \leq \zeta(S^1_{\beta} \times B),
\end{equation}
with equality if and only if $\mathcal{H}_B \otimes \mathcal{H}_{\overline{B}} = \mathcal{H}_{B \sqcup \overline{B}}$. In \eqref{eq:thermal_bound}, the circle must be given the anti-periodic spin structure when necessary.

Gravitationally, the inequality \eqref{eq:thermal_bound} has a natural interpretation in the regime where we expect black holes to dominate the thermal ensemble. The thermal partition function $\zeta(S^1_{\beta} \times B)$ computed by the GPI includes Euclidean black hole configurations in which the thermal circle shrinks, and gives a result compatible with the Bekenstein--Hawking entropy. However, $\mathrm{tr}_{\mathcal{H}_B} (e^{- \beta \widehat{H}_B})$ only counts those black hole microstates which actually exist as one-sided states preparable by our chosen class of sources. If we have included sources for all fields in the bulk EFT, then these are precisely the black hole states which can be formed by the gravitational collapse of EFT degrees of freedom. Morally, one can view $\mathrm{tr}_{\mathcal{H}_B} (e^{- \beta \widehat{H}_B})$ as being computed by a GPI with an additional restriction that the thermal circle is not allowed to shrink in the bulk (see \cite{Harlow:2020bee} for further comments).

\subsection{Symmetries and missing states}

While the universal construction does not define a unitary QFT, let us suppose, temporarily, that there is an underlying QFT with partition function $\zeta$, whose factorizing Hilbert spaces we denote $\widetilde{\CH}_B$. There is a canonical map $\mathcal{H}_B \to \widetilde{\mathcal{H}}_B$ from the universal Hilbert spaces to the Hilbert spaces $\widetilde{\CH}_B$, defined by mapping the formal state $\ket{M}$ in the universal Hilbert space to the state $\ket{\psi_M} \in \widetilde{\mathcal{H}}_B$ prepared, in the actual QFT, by the path integral over $M$ in the underlying QFT. In fact, this map defines an isometric embedding $\mathcal{H}_B \subset \widetilde{\mathcal{H}}_B$, since, by construction, the squared norm of the state $\ket{M}$ is $\zeta(M^\dagger \circ M)$, which agrees with the norm of $\ket{\psi_M}$ by assumption.

Thus, if an underlying QFT exists, the universal construction fails to satisfy ER = EPR only because it is missing some states. This conclusion fits naturally with the bound \eqref{eq:thermal_bound}, which would follow immediately from the presence of an isometric embedding $\mathcal{H}_B \subset \widetilde{\mathcal{H}}_B$ compatible with the Hamiltonian (as it would be).

Supposing an underlying QFT does exist, we ask: which states \textit{should} be missed by the universal construction? Well, if the underlying QFT has a global symmetry $G$, and our class $\mathcal{X}$ of sources does not include any fields charged under $G$, then every state in the universal construction must be neutral under $G$. Nevertheless, in a two-boundary Hilbert space, we may prepare states with non-vanishing one-sided charges, so long as the charges on the two sides are precisely correlated so that the net charge cancels. For example, the thermofield double state \eqref{eq:TFD} includes fluctuations over all states in the one-sided Hilbert space, including any one-sided charged states, despite the fact that the net charge of the thermofield double state is precisely zero.

Thus, in any underlying QFT with a nontrivial global symmetry, ER = EPR cannot possibly hold unless we include a complete set of charges in our collection $\mathcal{X}$ of sources, as was first observed by Harlow in \cite{Harlow:2015lma}. Our main result, Theorem \ref{thm:main}, asserts that: first of all, an underlying QFT does exist, and is specified uniquely up to unitary isomorphism. Moreover, the states charged under a global symmetry are the \textit{only} states which the universal construction misses, so the clear obstruction described by Harlow \cite{Harlow:2015lma} is, in fact, the only obstruction.

The same logic applies to something which is not usually considered a symmetry: the ``symmetry'' implemented by the $(d-1)$-dimensional topological operators defined by stacking invertible $(d-1)$-dimensional QFTs on top of the identity operator. In QFT, this is not usually considered a global symmetry, as it acts on a fixed spatial manifold $B$ by a theory-independent c-number, entirely determined by the bordism class $[B]$, which can be removed by a counterterm by definition. Nevertheless, this symmetry obstructs ER = EPR just the same: it corresponds to the $(\Omega_{d-1}^\mathcal{X})^\vee$ global symmetry described in the Swampland cobordism conjecture \cite{McNamara:2019rup},\footnote{This symmetry was independently rediscovered in \cite{Friedan:2023vxx}.} obstructing ER = EPR as discussed above.

As a final comment, we note that the obstruction arising from symmetry also means that any theory which requires fermions cannot possibly satisfy ER = EPR, at least within the framework we have developed so far. This is because any fermionic QFT has a global $\mathbb{Z}_2^F$ symmetry corresponding to fermion parity,\footnote{The bulk dual of the boundary global symmetry $\mathbb{Z}_2^F$ is a dynamical spin structure in the bulk, and the incompleteness of its spectrum is related to the 1-form symmetry discussed in \cite{Cheung:2024ypq}.} and as discussed in Section \ref{sec:germs}, we cannot include fermionic sources in a unitary bordism category. We describe how our framework must be modified to accommodate fermionic sources in Section \ref{sec:fermionic_sources}.

\section{Examples}\label{sec:examples}

In this section, we describe the universally constructed Hilbert spaces in a few simple, low-dimensional examples. In each example, we study a unitary bordism category equipped with a partition function $\zeta$ satisfying Axioms \ref{axiom:finiteness}-\ref{axiom:reflection_positivity}, and directly verify the conclusions of Theorem \ref{thm:main} to the extent that we are able. Our examples are ordered in terms of increasing complexity, and, correspondingly, decreasing tractability of reproducing Theorem \ref{thm:main} by hand.

\subsection{Topological quantum mechanics}
\label{sec:Ex_TopoQM}

Our first example is bosonic topological quantum mechanics (TQM) without time reversal symmetry. Thus, we take $d = 1$, and fix a partition function $\zeta$ on the unitary bordism category $\Bord_1^{SO}$ of oriented 1-manifolds. This example ends up being simple enough that we can completely prove Theorem \ref{thm:main} by hand. This example is motivated by the GPI of Marolf and Maxfield \cite{Marolf:2020xie}, and more specifically, by Maxfield's detailed analysis \cite{Maxfield:2023mdj} of the geometrically preparable states in TQM.

First, by finiteness, reality, and multiplicativity, $\zeta$ is specified by the single real number $\zeta(S^1) = k \in \mathbb{R}$. Reflection positivity, applied simply, implies that $k \geq 0$, as the circle is obtained as the double of a semicircle. The partition function $\zeta$ wants to be the partition function of a topological quantum mechanical theory with $k$ degenerate ground states, but for this to make sense we would need $k$ to be a non-negative integer.

Incredibly, reflection positivity does actually imply that $k \in \mathbb{Z}_{\geq 0}$. This follows from an argument at the heart of Tannakian reconstruction (see \cite{DoplicherRoberts89,deligne1990categories,deligne2002categories,muger2007abstract,EGNO}), which we now review in this special case. Note that every object in $\Bord_1^{SO}$ consists of some number $n_+$ of positively oriented points $\mathrm{pt}^+$, and some number $n_-$ of negatively oriented points, $\mathrm{pt}^-$. As shorthand, we will denote such an oriented $0$-manifold by the pair $(n_+, n_-)$.

Consider the diagonal universal Hilbert spaces $\mathcal{H}_{(n, n)}$. We define a homomorphism,
\begin{equation}\label{eq:hom_from_group_ring}
    \mathbb{C}[S_n] \to \CH_{(n, n)},
\end{equation}
from the group ring on the symmetric group as follows. For each permutation $\sigma \in S_n$, we obtain a bordism from $n$ points to $n$ points, all positively oriented, in which the $i$-th incoming point connects to the $\sigma(i)$-th outgoing point, for $1 \leq i \leq n$. This defines a homomorphism,
\begin{equation}
    S_n \to \Bord_1^{SO}\big((n,0) \to (n,0)\big).
\end{equation}
We define the state $\ket{\sigma} \in \mathcal{H}_{n, n}$ to be the state obtained by taking the corresponding bordism $(n, 0) \to (n, 0)$ and using the unitary structure to turn the $n$ incoming positively oriented points into $n$ outgoing negatively oriented points. Extending linearly, we obtain the homomorphism \eqref{eq:hom_from_group_ring}. For illustrative purposes, a few examples of the states $\ket{\sigma}$ for $n = 3$ are given by:
\begin{align}
    \sigma_1 = \mathrm{id},\qquad |\sigma_1\rangle &= \Ket{\includegraphics[height=1.35cm,valign=c]{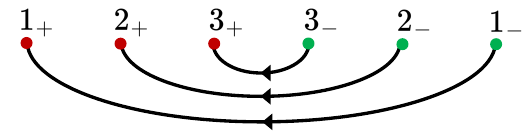}}\\[0.5em]
    \sigma_2 = (12)(3),\qquad |\sigma_2\rangle &= \Ket{\includegraphics[height=1.2cm,valign=c]{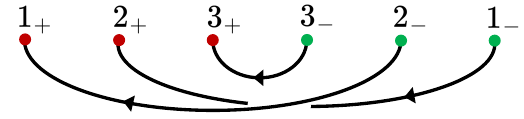}}\\[0.5em]
    \sigma_3 = (123),\qquad |\sigma_3\rangle &= \Ket{\includegraphics[height=1.15cm,valign=c]{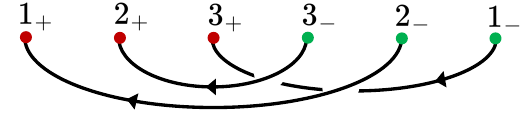}}
\end{align}

Now, following \cite{Marolf:2020xie,Colafranceschi:2023urj}, consider the state,
\begin{equation}
    \ket{P_{\Lambda^n}} = \frac{1}{n!} \sum_{\sigma \in S_n} (-1)^{|\sigma|} \ket{\sigma},
\end{equation}
corresponding to the projector in the group ring $\mathbb{C}[S_n]$ onto the $n$-th antisymmetric tensor power. A simple calculation shows that we have,
\begin{equation}\label{eq:antisymmetric_tensor}
    \braket{P_{\Lambda^n} | P_{\Lambda^n}} = \binom{k}{n} \defined \frac{k (k - 1) \cdots (k - n + 1)}{n!}.
\end{equation}
Thus, if $k$ is not a non-negative integer, the squared norm \eqref{eq:antisymmetric_tensor} will become negative for sufficiently large $n$, violating reflection positivity.

Given that $k \in \mathbb{Z}_{\geq 0}$, we may interpret $\zeta$ as the partition function of any unitary TQFT on $\Bord_1^{SO}$ whose Hilbert space $\widetilde{\CH}_{(1, 0)}$ on a single positively-oriented point is $k$-dimensional (the rest of the TQFT is determined by the cobordism hypothesis). Up to unitary isomorphism, there is only one such Hilbert space, given by $\mathbb{C}^k$, and so all unitary TQFTs with $\zeta(S^1) = k$ are unitarily equivalent. We note that any such TQFT is subject to a $U(k)$ global symmetry rotating the orthonormal basis of the $k$-dimensional space of ground states.\footnote{With time reversal symmetry, the global symmetry would be $O(k)$, we would replace oriented manifolds with unoriented ones, and we would replace the group ring $\mathbb{C}[S_n]$ with the Brauer algebra.}

Now, we ask, which states of this TQFT are produced by the universal construction? An obvious constraint, arising from the bordism group $\Omega_0^{SO} = \mathbb{Z}$, is that the universal construction only produces nonzero states in the diagonal Hilbert spaces $\widetilde{\mathcal{H}}_{(n, n)}$ on an equal number $n$ of positively and negatively oriented points. The diagonal Hilbert spaces $\widetilde{\mathcal{H}}_{(n, n)}$ are precisely the spaces of neutral states under the diagonal subgroup $U(1) \subset U(k)$, which acts by the phase $e^{i (n_+ - n_-) \theta}$ in the Hilbert space $\widetilde{\mathcal{H}}_{(n_+, n_-)}$. The diagonal subgroup $U(1) \subset U(k)$ is the Pontryagin dual of the bordism group $\Omega_0^{SO} = \mathbb{Z}$, acting as a global symmetry as described in \cite{McNamara:2019rup}.

However, even in the diagonal Hilbert spaces $\widetilde{\mathcal{H}}_{(n, n)}$, not every state is neutral under the full symmetry group $U(k)$, as discussed in \cite{Maxfield:2023mdj}. As a representation of the symmetry group $U(k)$, we have
\begin{equation}
    \widetilde{\mathcal{H}}_{(n, n)} = \Box^{\otimes n_+} \otimes \overline{\Box}^{\otimes n_-},
\end{equation}
where $\Box$ and $\overline{\Box}$ denote the fundamental and anti-fundamental representations of $U(k)$, respectively. By Schur--Weyl duality, we have a decomposition
\begin{equation}\label{eq:Schur_Weyl}
    \Box^{\otimes n} = \bigoplus_{D} \rho_{S_n}^D \otimes \rho_{U(k)}^D, \quad \mathrm{rows}(D) \leq k
\end{equation}
as representations of $S_n \times U(k) $, where $S_n$ acts to permute tensor factors and $U(k)$ acts diagonally. The sum \eqref{eq:Schur_Weyl} runs over all Young diagrams $D$ with $n$ boxes and at most $k$ rows, and $\rho_{S_n}^D$ and $\rho_{U(k)}^D$ are the associated irreducible representations of $S_n$ and $U(k)$, respectively. Thus, we have
\begin{equation}\label{eq:invts_in_Schur_Weyl}
    (\widetilde{\mathcal{H}}_{(n, n)})^{U(k)} = \bigoplus_{D} \rho_{S_n}^D \otimes \overline{\rho_{S_n}^D}, \quad \mathrm{rows}(D) \leq k.
\end{equation}
Again, the sum runs over all Young diagrams $D$ with $n$ boxes and at most $k$ rows.

Let us now match the invariant subspace $(\widetilde{\mathcal{H}}_{(n, n)})^{U(k)}$ with the states prepared by the universal construction. Up to unimportant disconnected components, any 1-manifold with boundary $(n, n)$ produces one of the states $\ket{\sigma}$ discussed above. Thus, the map $\mathbb{C}[S_n] \to \CH_{(n, n)}$ is surjective. By the Peter--Weyl theorem, we have
\begin{equation}\label{eq:Peter--Weyl}
    \mathbb{C}[S_n] = \bigoplus_{D} \rho_{S_n}^D \otimes \overline{\rho_{S_n}^D},
\end{equation}
where now the sum runs over all Young diagrams with $n$ boxes, with no restriction on the number of rows. Comparing with \eqref{eq:invts_in_Schur_Weyl}, we see that the kernel of $\mathbb{C}[S_n] \to \CH_{(n, n)}$ consists of the summands labeled by Young diagrams with greater than $k$ rows. Moreover, the universal construction produces the entire neutral sector under the $U(k)$ global symmetry, for any oriented 0-manifold in $\Bord_1^{SO}$, as shown in \cite{Maxfield:2023mdj}.

In fact, \cite{Maxfield:2023mdj} gives an explicit description of the space of null states in the universal pre-Hilbert spaces $\CH_{(n,n)}^\mathrm{pre}$, corresponding to the summands in \eqref{eq:Peter--Weyl} labeled by Young diagrams with more than $k$ rows. The argument is that the squared norm of the state obtained from a central projector in $\mathbb{C}[S_n]$, corresponding to some irreducible representation of $S_n$, is the dimension of the corresponding irreducible representation of $U(k)$. But this dimension vanishes if and only if the corresponding Young diagram has more than $k$ rows, and so states labeled by a Young diagram with more than $k$ rows are null. This generalizes \eqref{eq:antisymmetric_tensor} from the antisymmetric tensor representation to an arbitrary representation.

\subsection{Adding end-of-the-world branes}\label{sec:ex_EotW_branes}

Let us make a simple modification of the previous example, motivated by \cite{Marolf:2020xie} and \cite{McNamara:2019rup}. We adjoin some number $k_\mathrm{EotW}$ of end-of-the-world (EotW) branes, namely, $k_\mathrm{EotW}$ additional morphisms
\begin{equation}
    \psi_1, \dots, \psi_{k_\mathrm{EotW}} : \varnothing \to \mathrm{pt}^+,
\end{equation}
as well as their formal conjugates $\overline{\psi}_i$ and adjoints $\psi_i^\dagger$. We view the morphisms $\psi_i$ geometrically as oriented 1-dimensional bordisms-with-singularities, obtained by allowing free endpoints marked by an index $i = 1, \dots, k_\mathrm{EotW}$  (see Figure \ref{fig:EotW}). We denote the resulting unitary category by $\Bord_1^{{SO} + \mathrm{EotW}}$, and note that $\Omega_0^{{SO} + \mathrm{EotW}} = 0$ as long as $k_\mathrm{EotW} > 0$, since each EotW brane $\psi_i$ provides a bordism-with-singularities trivializing the class $[\mathrm{pt}^+]$.

\begin{figure}
    \centering
    \includegraphics[width=0.5\linewidth]{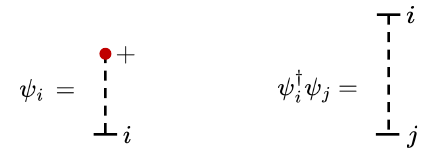}
    \caption{End-of-the-world (EotW) branes provide $k_\mathrm{EotW}$ additional morphisms $\psi_i:\varnothing \to \mathrm{pt}^+$, which we view as bordisms with singularities from the empty set to $\mathrm{pt}^+$ (left). The allowed type of singularity is a free boundary labeled by the index $i$. The compositions $\psi_i^\dagger \circ \psi_j$ are closed manifolds with singularities (right), which evaluate to the inner product matrix $\braket{\psi_i | \psi_j}$.}
    \label{fig:EotW}
\end{figure}

Suppose we have some partition function $\zeta$ on $\Bord_1^{{SO} + \mathrm{EotW}}$, satisfying Axioms \ref{axiom:finiteness}-\ref{axiom:reflection_positivity}. By restricting $\zeta$ to the sub-category $\Bord_1^{SO}$, we may apply the reasoning of the previous section to conclude that $\zeta(S^1)$ is given by some non-negative integer $k \in \mathbb{Z}_{\geq 0}$. The remaining data in $\zeta$ is the Gram matrix,
\begin{equation}\label{eq:inner_prod_of_eotws}
    \braket{\psi_i | \psi_j} = \zeta(\psi_i^\dagger \psi_j).
\end{equation}
Our goal, now, is to constrain the rank $r$ of the matrix \eqref{eq:inner_prod_of_eotws}.

Following \cite{Marolf:2020xie}, note that we may construct precisely $r$ orthonormal states
\begin{equation}
    \ket{\phi_a} = \sum_{i = 1}^{k_\mathrm{EotW}} c_a^i \ket{\psi_i}, \quad a = 1, \dots, r, \quad \braket{\phi_a | \phi_b} = \delta_{ab}
\end{equation}
from the states $\ket{\psi_i} \in \mathcal{H}_{\mathrm{pt}^+}^\mathrm{pre}$. Consider the state,
\begin{equation}
    \ket{P_{\mathrm{EotW}^\perp}} \defined \ket{\mathrm{id}_{\mathrm{pt}^+}} - \sum_{a = 1}^r \ket{\phi_a \sqcup \overline{\phi}_a} \in \mathcal{H}_{\mathrm{pt}^+ \sqcup \mathrm{pt}^-},
\end{equation}
where we have used the unitary structure to bend $\mathrm{id}_{\mathrm{pt}^+}$ into a bordism $\varnothing \to \mathrm{pt}^+ \sqcup \mathrm{pt}^-$. By construction, $\ket{P_{\mathrm{EotW}^\perp}}$ is a state obtained, by bending, from the formal projector onto the orthogonal complement of the subspace of $\mathbb{C}^k$ spanned by the EotW brane states. We compute,
\begin{equation}
    \braket{P_{\mathrm{EotW}^\perp} | P_{\mathrm{EotW}^\perp}} = k - 2 \sum_{a = 1}^r \delta_{aa} + \sum_{a = 1}^r \delta_{a a} \delta_{a a} = k - r. 
\end{equation}
Thus, reflection positivity implies that the rank $r$ of the inner product matrix \eqref{eq:inner_prod_of_eotws} is bounded above by the dimension $k$ of the underlying Hilbert space, $r \leq k$.

The bound $r \leq k$ is necessarily the only constraint, as given it we may always find a unitary TQFT with EotW branes and with partition function $\zeta$, by choosing $k_\mathrm{EotW}$ vectors in $\mathbb{C}^k$ with inner product matrix $\braket{\psi_i| \psi_j}$. Moreover, any two such choices are unitarily equivalent, as, up to a unitary change of basis, the Gram matrix is the only invariant of an ordered list of vectors in a Hilbert space.

Now, let us ask: which new states do the EotW branes add into the universal Hilbert spaces? Without EotW branes, we saw previously that we could already produce all states invariant under the $U(k)$ symmetry acting on the $k$-dimensional space of ground states. Let us fix an underlying TQFT, corresponding to a fixed choice of $k_\mathrm{EotW}$ vectors $\ket{\psi_i} \in \mathbb{C}^k$ with inner product matrix \eqref{eq:inner_prod_of_eotws}. Let us also choose some orthonormal basis $\ket{\phi_a}$ of the linear span of the states $\ket{\psi_i}$, as above. This orthonormal basis defines an isometric embedding $\mathbb{C}^r \hookrightarrow \mathbb{C}^k$, which precisely describes the additional states produced by the EotW branes.

Consider the unitary group $G = U(k - r)$ of the orthogonal complement $\mathbb{C}^{k - r} \hookrightarrow \mathbb{C}^k$ of the space $\mathbb{C}^r \hookrightarrow \mathbb{C}^k$ of EotW brane states. We claim that the universal Hilbert spaces of the theory with EotW branes are precisely the $U(k - r)$-invariant subspaces of the Hilbert spaces of the underlying TQFT, on any number of positively and negatively oriented points. One inclusion is obvious, as we cannot produce any charged states under $U(k - r)$ using smooth manifolds or EotW branes, by construction.

For the other direction, consider any tensor product of copies of $\mathbb{C}^k$ and its conjugate. We can prepare any state in this tensor product which is given by a tensor product of any number of vectors in $\mathbb{C}^r \subset \mathbb{C}^k$ and any number of invariant tensors under $U(k)$. Decompose the tensor product as a direct sum of tensor products of $\mathbb{C}^r$, $\mathbb{C}^{k - r}$, and their conjugates. We can clearly produce any state in the tensor factors $\mathbb{C}^r$, and we can also prepare any $U(k-r)$-invariant state in tensor products of $\mathbb{C}^{k-r}$ and its conjugate by applying Schur--Weyl duality to the states $\ket{P_{\mathrm{EotW}^\perp}}$. Thus, we can produce all $U(k-r)$-invariant states, as claimed.

This example demonstrates the general manner in which the universal Hilbert spaces change upon adding additional sources. If the original universal Hilbert spaces were given by the $G$-invariants, and we add new sources which break the symmetry to a subgroup $H \subset G$, then we will necessarily produce \textit{all} the states that are invariant under $H$ (though not necessarily under $G$). Again, we see that the only obstruction to states being geometrically preparable is the presence of a symmetry group.

\subsection{Fermionic topological quantum mechanics}\label{sec:ex_fTQM}

Our next example is closely related to the example in Section \ref{sec:Ex_TopoQM}, but generalized from bosonic TQM to fermionic TQM. Thus, we replace the oriented bordism category $\Bord_1^{SO}$ with the spin bordism category $\Bord_1^{Spin}$. Just as the example of bosonic TQM is related to the GPI of Marolf and Maxfield \cite{Marolf:2020xie}, the example of fermionic TQM is related to the fermionic generalization of Marolf and Maxfield's model considered in \cite{Balasubramanian:2020jhl}.

Let us first discuss the obvious consequences of reflection positivity for fermionic TQM. By finiteness, reality, and multiplicativity, the partition function on $\Bord_1^{Spin}$ is determined by two real numbers,
\begin{equation}
    \zeta(S^1_\mathrm{ap}) \in \mathbb{R}, \quad \zeta(S^1_\mathrm{p}) \in \mathbb{R},
\end{equation}
obtained by evaluating $\zeta$ on circles equipped with the anti-periodic (bounding) and periodic (non-bounding) spin structures, as illustrated in Figure \ref{fig:NS_R_circles}.

\begin{figure}
    \centering
    \includegraphics[width=0.4\linewidth]{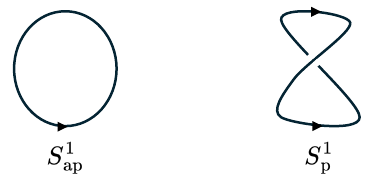}
    \caption{The anti-periodic circle $S^1_\mathrm{ap}$ and the periodic circle $S^1_\mathrm{p}$. Note that, in traversing $S^1_\mathrm{ap}$, one undergoes a $2 \pi$ rotation, while in traversing $S^1_\mathrm{p}$ there is no net rotation. Note also that $S^1_\mathrm{ap}$ is manifestly a double, while $S^1_\mathrm{p}$ is not reflection-symmetric. In fermionic topological quantum mechanics, a simple application of reflection positivity implies $\zeta(S^1_\mathrm{ap}) \geq 0$, and places no constraint on $\zeta(S^1_\mathrm{p})$.}
    \label{fig:NS_R_circles}
\end{figure}

What are the obvious consequences of reflection positivity, analogous to the condition $\zeta(S^1) \geq 0$ in the bosonic case? To answer this question, we need to know which circle is the double, $M^\dagger \circ M$, of the spin-bordism
\begin{equation}
    M : \varnothing \to \mathrm{pt}^+ \sqcup \mathrm{pt}^-, \quad M = [0, 1],
\end{equation}
given by a semicircle. Using the canonical unitary structure on $\Bord_1^{Spin}$, we may regard $M$ as the identity spin-bordism $\mathrm{id}_{\mathrm{pt}^+}$ of the positively-oriented point, and compute
\begin{equation}
    M^\dagger \circ M = \mathrm{tr}_\mathrm{Bord}(\mathrm{id}_{\mathrm{pt}^+}) = \mathrm{str}_\mathrm{Bord}(\theta_{\mathrm{pt}^+} \circ \mathrm{id}_{\mathrm{pt}^+}) = S^1_\mathrm{ap},
\end{equation}
as discussed in \cite{Freed:2016rqq,stehouwer2024unitary}. Even if we first twisted $M$ by $\theta_{\mathrm{pt}^+}$ on either boundary, we would still obtain $M^\dagger \circ M = S^1_\mathrm{ap}$, as the two additional copies of $\theta_{\mathrm{pt}^+}$ would cancel.

Thus, a simple application of reflection positivity requires the partition function on $S^1_\mathrm{ap}$ to be positive, $\zeta(S^1_\mathrm{ap}) \geq 0$, and places no restriction on the partition function $\zeta(S^1_\mathrm{p})$. As discussed in Section \ref{sec:traces_and_fermions}, this is the physically correct conclusion. This is because in an actual fermionic TQFT $\zeta(S^1_\mathrm{ap})$ computes the dimension of the super-Hilbert space of ground states on $\mathrm{pt}^+$, while $\zeta(S^1_\mathrm{p})$ computes the super-dimension. Let us note that this is a meaningfully different positivity constraint than we would have obtained if we considered bosonic TQM with a unitary $\mathbb{Z}_2$ internal symmetry, as the $\dagger$-structure associated to $\Bord_1^{{SO} \times \mathbb{Z}_2}$ makes the periodic $\mathbb{Z}_2$ bundle a double, not the anti-periodic $\mathbb{Z}_2$ bundle.

By a similar argument as the one presented in Section \ref{sec:Ex_TopoQM} for bosonic TQM, we could show that there exist non-negative integers $p, q \in \mathbb{Z}_{\geq 0}$ such that we have
\begin{equation}
    \zeta(S^1_\mathrm{ap}) = p + q, \quad \zeta(S^1_\mathrm{p}) = p - q.
\end{equation}
and thus identify $\zeta$ with the partition function of the fermionic TQM with super-Hilbert space $\mathbb{C}^{p | q}$ on a single positively oriented point. For a full argument, see our review of the proof of the DR reconstruction theorem in Section \ref{sec:DR_reconstruction}, and in particular the twisting of the symmetry used to prove integrality of dimensions. Using similar logic as in Section \ref{sec:Ex_TopoQM}, we would also conclude that the states prepared by the universal construction are precisely the states invariant under the $U(p) \times U(q)$ symmetry acting on the $(p | q)$-dimensional space of ground states.

Let us use this example to illustrate a key point, regarding the necessary exchange statistics of missing states. Consider the case of a single missing fermionic state, namely,
\begin{equation}\label{eq:part_fun_a_fermion}
    \zeta(S^1_\mathrm{ap}) = +1, \quad \zeta(S^1_\mathrm{p}) = -1.
\end{equation}
We have two bordisms,
\begin{equation}
    \mathrm{id}, s : \mathrm{pt}^+ \sqcup \mathrm{pt}^+ \to \mathrm{pt}^+ \sqcup \mathrm{pt}^+,
\end{equation}
given by the identity and the swap bordisms, respectively.

Via the unitary structure, we may view both $\mathrm{id}$ and $s$ as bordisms from $\varnothing$ to the disjoint union of two positively and two negatively oriented points. We thus have corresponding states,
\begin{equation}
    \ket{\mathrm{id}} = \Ket{\includegraphics[height=1.25cm,valign=c]{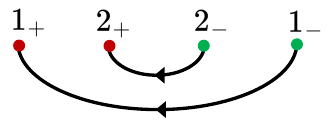}}, \quad \ket{s} = \Ket{\includegraphics[height=1.1cm,valign=c]{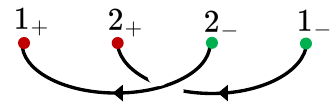}},
\end{equation}
in the four-boundary universal Hilbert space $\mathcal{H}_{\mathrm{pt}^+ \sqcup \mathrm{pt}^+ \sqcup \mathrm{pt}^- \sqcup \mathrm{pt}^-}$. The inner product matrix can be computed from \eqref{eq:part_fun_a_fermion}, and is given by,
\begin{equation}
    \begin{array}{c | c c}
        & \ket{\mathrm{id}} & \ket{s} \\
        \hline
        \bra{\mathrm{id}} & +1 & -1 \\
        \bra{s} & -1 & +1
    \end{array}
\end{equation}
As a result, modulo null states, we have $\ket{s} \sim - \ket{\mathrm{id}}$, i.e.,
\begin{equation}
    \Ket{\includegraphics[height=1.25cm,valign=c]{figures/FTQM_no_swap.pdf}} \sim - \Ket{\includegraphics[height=1.1cm,valign=c]{figures/FTQM_swap.pdf}}.
\end{equation}

We see that the unitary structure on $\Bord_1^{Spin}$, together with reflection positivity, allows us to directly deduce the exchange statistics for missing states. This is the basic mechanism by which the partition function $\zeta$ tells us which missing states should be bosons, and which should be fermions. We might hope to adjoin these missing states by adding EotW branes as in the previous example, but this will fail due to the inability of our formalism, as described so far, to accommodate fermionic sources. We demonstrate this failure in Section \ref{sec:fermionic_sources}, as well as provide the resolution.

\subsection{Non-topological quantum mechanics}\label{sec:ex_QM}

As a final example in $d = 1$, let us consider non-topological quantum mechanics, where we take $\mathcal{X}$ to consist of the choice of an orientation and a Riemannian metric. Let us assume that we have an actual compact quantum mechanical system, meaning a Hilbert space $\mathcal{H}$ with Hamiltonian $\widehat{H}$ such that the thermal partition function,
\begin{equation}
    \zeta(S^1_\beta) = \tr(e^{- \beta \widehat{H}}) < \infty,
\end{equation}
is finite for all $\beta$. Compactness implies that $\widehat{H}$ has a discrete spectrum, bounded from below, consisting of energy levels $E_0 < E_1 < \cdots$ with integer degeneracies $d_n \in \mathbb{Z}_{\geq 0}$. Thus, we have,
\begin{equation}\label{eq:spectral_rep_thermal_part_fun}
    \zeta(S^1_\beta) = \sum_{n \geq 0} d_n e^{- \beta E_n}.
\end{equation}

In fact, it is possible to fully bootstrap this example. One could consider carefully chosen states in the universal Hilbert space $\CH_{\mathrm{pt}^+ \sqcup \mathrm{pt}^-}$, corresponding to the spectral projectors of the presumptive Hamiltonian operator. By applying an anti-symmetrization argument as in Section \ref{sec:Ex_TopoQM}, one could conclude that $\zeta(S^1_\beta)$ may be written, as in \eqref{eq:spectral_rep_thermal_part_fun}, as a sum over a discrete spectrum with integer degeneracies at each energy level. However, the analysis needed to run this bootstrap is subsumed by the analysis used in our proof of Theorem \ref{thm:main}.

Instead, we use this example to illustrate an interesting possibility regarding the nature of the symmetry group $G$. In a quantum mechanical system, it is quite easy to find symmetries, as a symmetry need only commute with the Hamiltonian operator. Thus, the full unitary symmetry of the Hamiltonian $\widehat{H}$ discussed above is simply,
\begin{equation}
    G = \prod_n U(d_n),
\end{equation}
the product of the unitary groups rotating degenerate states within each energy level. As topologized in the strong operator topology, $G$ is a compact Hausdorff group, being an infinite product of Lie groups. However, it is not a Lie group.

Let us consider the states in the universal Hilbert space $\CH_{\mathrm{pt}^+ \sqcup \mathrm{pt}^-}$. This is given by the closed linear span, in $\mathcal{H} \otimes \overline{\CH}$, of the thermofield double states $\ket{\mathrm{TFD}_\beta}$ produced by intervals of length $\beta/2$ as $\beta$ varies. Certainly, $\CH_{\mathrm{pt}^+ \sqcup \mathrm{pt}^-}$ is contained in the $G$-invariant subspace, as every thermofield double state is $G$-invariant. Conversely, in each microcanonical window around an energy level $E_n$, there is a unique $G$-invariant state, which is maximally entangled within the microcanonical window and pairs each one-sided state with its CPT conjugate. These states may all be obtained as linear combinations of thermofield double states $\ket{\mathrm{TFD}_\beta}$ by taking an inverse Laplace transform. Thus, the states in the two-boundary universal Hilbert space $\CH_{\mathrm{pt}^+ \sqcup \mathrm{pt}^-}$ are precisely the $G$-invariant states.

As a result, we see that the group $G$ described in Theorem \ref{thm:main} can really be a compact Hausdorff group, and is not always a compact Lie group. That being said, this example is rather specific to quantum mechanics ($d = 1$), and the condition for being a symmetry in higher dimensions is much more restrictive, as a symmetry is required to commute with the local stress-energy tensor. We do not know any example in higher dimensions where the group $G$ is not a compact Lie group.

\subsection{Two-dimensional Ising CFT}\label{sec:ex_Ising}

Our final example considered in this section is given by the two-dimensional Ising CFT, considered as a theory on the geometric bordism category $\Bord_2^{\mathrm{Riem}, SO}$ of oriented two-dimensional Riemannian manifolds.\footnote{We need to choose Riemannian metrics, and not just conformal structures, in order to trivialize the Weyl anomaly and get a well-defined partition function.} The Ising CFT has three Virasoro primary operators: $\mathds{1}, \sigma$, and $\epsilon$. It also has a $\mathbb{Z}_2$ spin-flip symmetry, under which $\sigma$ is charged. Under the state-operator correspondence, every state in the Hilbert space on $S^1$ is a linear combination of the states corresponding to $\mathds{1}, \sigma$, and $\epsilon$, as well as their Virasoro descendants. We ask: which of these states are produced by the path integral over a smooth Riemann surface with one boundary circle? By the state-operator correspondence, this is the same as asking which operators acquire a one-point function on some closed Riemann surface.

On the sphere, the identity operator $\mathds{1}$ certainly acquires a one-point function, as it corresponds to the vacuum under the state-operator correspondence. However, there are more operators which acquire a sphere one-point function if we allow an arbitrary metric on the sphere, or, equivalently, keep the metric fixed but vary the local coordinate sewing our operator insertion into the sphere. This change of local coordinate is, precisely, the action of the Virasoro group on the space of local operators. Thus, any state in the vacuum Virasoro module, corresponding to the stress tensor and its descendants, can be produced by some possibly-infinite linear combination of appropriately chosen metrics at genus zero. Moreover, if we parametrize our choice of metric by the choice of Virasoro transformation used to sew in the local coordinate, the null states in the sense of the universal construction are, precisely, the null states in the vacuum Virasoro module. By an analogous argument, it is sufficient, at any genus, to determine which Virasoro primaries acquire a one-point function.

Let us now move to genus one, where we get additional states. In particular, the operator $\epsilon$ acquires a torus one-point function \cite{Itzykson:1986pj,DiFrancesco:1987ez,Bagger:1988yc}, and so the corresponding state and all of its descendants are produced by the path integral over a genus one surface with one circular boundary. More precisely, this path integral will produce some linear combinations of descendants of $\mathds{1}$ and $\epsilon$, and we may subtract off any descendants of $\mathds{1}$ by taking linear combinations with the states already produced at genus zero. Ultimately, at genus zero and one, we will produce all Virasoro descendants of $\mathds{1}$ and $\epsilon$. At this point, nothing new happens at higher genus, as we have already produced all neutral states under the $\mathbb{Z}_2$ spin flip symmetry.

Now, let us consider the universal construction for the Hilbert space on two circles, which corresponds to asking about which pairs of operators acquire a non-zero two-point function on some possibly-disconnected Riemann surface. Disconnected manifolds will produce the pairs of operators,
\begin{equation}\label{eq:Ising_two_circles1}
    (\mathds{1}, \mathds{1}), \quad (\epsilon, \mathds{1}), \quad (\mathds{1}, \epsilon), \quad (\epsilon, \epsilon),
\end{equation}
and descendants thereof under the two independent copies of Virasoro.

However, we will also get new states from two-point functions on connected Riemann surfaces. This will include linear combinations of the pairs \eqref{eq:Ising_two_circles1} and their descendants, together with the pair,
\begin{equation}\label{eq:Ising_two_circles2}
    (\sigma, \sigma),
\end{equation}
and its descendants, from the non-zero $\sigma$ two-point function on the sphere. Note that the pairs \eqref{eq:Ising_two_circles1} and \eqref{eq:Ising_two_circles2}, with their descendants, correspond to all the $\mathbb{Z}_2$-neutral states in the two-circle Hilbert space, consisting of the tensor products of neutral states with neutral states or charged states with charged states, but not the tensor products of neutral states with charged states or vice-versa. By similar arguments, we learn that, for any number of circles, the states in the Ising CFT preparable by manifolds are, precisely, the $\mathbb{Z}_2$-neutral states.

It is worth emphasizing that the Ising CFT has, in addition to the $\mathbb{Z}_2$ spin-flip symmetry, a non-invertible Kramers--Wannier symmetry under which both $\sigma$ and $\epsilon$ are charged. However, while this non-invertible symmetry is sufficient to prevent $\epsilon$ from acquiring a sphere one-point function, it is insufficient to prevent $\epsilon$ from acquiring a torus one-point function. The non-invertible symmetry is broken by the nontrivial topology of the torus, as explained in \cite{McNamara:2021cuo,Heckman:2024obe}. It is possible that Theorem \ref{thm:main} could be used to prove the conjecture of \cite{McNamara:2021cuo,Heckman:2024obe} that \textit{every} non-invertible symmetry is broken on some topology.\footnote{This conjecture does not immediately follow from Theorem \ref{thm:main} due to the possibility that some element of $G$ secretly acts non-invertibly on extended operators, while acting invertibly on the Hilbert space on any spatial manifold. We find this possibility unlikely, but ruling it out would require us to generalize our framework to the fully-extended case.}

\section{The baby universe category}\label{sec:CBU}

In this section, we turn to the central object of study in this paper: the baby universe category $\CBU$ associated to an abstract partition function $\zeta$ on a unitary bordism category $\BordX$ that satisfies axioms \ref{axiom:finiteness}-\ref{axiom:reflection_positivity}. Our construction of $\CBU$ in this section is, essentially, a formalization and retelling of the previous work \cite{Colafranceschi:2023urj,Marolf:2024adj} of Colafranceschi, Dong, Marolf, Zhang, and the second author of this paper.

\subsection{Towards a quantum bordism category}

The core conceptual idea of our construction of $\CBU$ is that, while the universal Hilbert spaces $\mathcal{H}_B$ miss some of the states of the putative QFT underlying $\zeta$, these states are still secretly present as the would-be one sided microstates of ER bridges in $\mathcal{H}_{B \sqcup \overline{A}}$, as $A$ ranges over all possible objects of $\BordX$. In general, while ER = EPR may fail, the states of ER bridges will nevertheless try their best to behave as if they were entangled states. For instance, the diagonal sector $\mathcal{H}_{B \sqcup \overline{B}}$ contains the ER bridge prepared by the cylinder $C_B(\beta)$, which wants to be the thermofield double state at inverse temperature $\beta$.

As we have seen in the examples of Section \ref{sec:examples}, the ER bridge states formally built from manifolds know a remarkably large amount about the missing states. In fact, example by example, we found that the ER bridges always seem to know \textit{everything} about the missing states, up to the action of a unitary symmetry group acting on the missing states. Our goal, in this section and the next two, is to systematize the process of distilling information about the missing states from the collection of ER bridges.

We do this by following a general categorical principle known as Grothendieck's relative point of view. The idea is that, to understand some mathematical object $Y$ in some category $\mathcal{C}$, one should consider maps $X \to Y$ from every other object $X \in \mathcal{C}$, rather than only consider maps $\mathds{1} \to Y$ from some unit object $\mathds{1}$. For instance, if we are doing any sort of geometry, and $\mathds{1}$ denotes a single point, then a map $\mathds{1} \to Y$ is a point of $Y$, while a map $X \to Y$ is an $X$-parametrized family of points of $Y$. While the bare set of points of $Y$ does not fully capture its geometry, the knowledge of how its points fit into parametrized families does. This approach to algebraic geometry is known as the \textit{functor of points}, and is justified by the Yoneda Lemma.

Applying Grothendieck's relative point of view, we learn that we have been making a mistake by only building states on $B$ from bordisms $\varnothing \to B$. Instead, we should try to build states on $B$ from bordisms $N : A \to B$, viewed as recipes for producing states on $B$ given a reservoir of arbitrary states on $A$. Given that we do not actually \textit{have} a reservoir of arbitrary states on $A$, the best we can do is prepare the ER bridge state $\ket{N}$ in the Hilbert space $\mathcal{H}_{B \sqcup \overline{A}}$. We depict this chain of logic in Figure \ref{fig:relative_POV}. Our goal now, following Grothendieck's relative point of view, is to assemble the Hilbert spaces $\mathcal{H}_{B \sqcup \overline{A}}$ into some category in which $A$ and $B$ live as objects.

\begin{figure}
    \centering
    \includegraphics[width=0.8\linewidth]{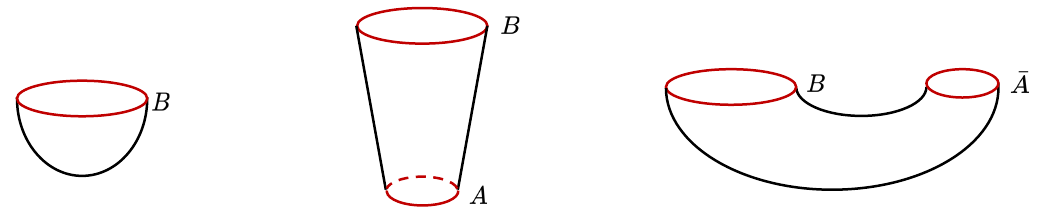}
    \caption{In the universal Hilbert space $\CH_B$, we only produce states using bordisms $\varnothing \to B$ (left). Applying Grothendieck's relative point of view, we should attempt to prepare states on $B$ from bordisms $A \to B$ (center), viewed as recipes for evolving arbitrary states on $A$ into states on $B$. Given that we do not \textit{have} arbitrary states on $A$, the best we can do is prepare the associated ER bridge state in $\CH_{B \sqcup \overline{A}}$ (right).}
    \label{fig:relative_POV}
\end{figure}

The most direct approach is to simply take $\mathcal{H}_{B \sqcup \overline{A}}$ to be the space of morphisms from $A$ to $B$. By the fundamental trace inequality of \cite{Colafranceschi:2023urj} (Proposition \ref{prop:trace_inequality}), composition of bordisms induces a jointly continuous bilinear composition map,
\begin{equation}
    \mathcal{H}_{C \sqcup \overline{B}} \times \mathcal{H}_{B \sqcup \overline{A}} \to \mathcal{H}_{C \sqcup \overline{A}},
\end{equation}
defined by using the unitary structure of $\BordX$ to compose along the intermediate boundary $B$, as described in Proposition \ref{cor:composition_is_jointly_continuous}. The resulting category is typically non-unital, just like $\BordX$, as the cylinder states $\ket{C_B(\beta)} \in \mathcal{H}_{B \sqcup \overline{B}}$ typically do not converge as $\beta \to 0$. To understand the problem, consider what would happen if we had ER = EPR, so that the two-boundary Hilbert space did factorize as $\mathcal{H}_{A \sqcup \overline{B}} = \mathcal{H}_A \otimes \overline{\mathcal{H}}_B$. In this case, we would recognize $\mathcal{H}_{A \sqcup \overline{B}}$ as the space of Hilbert--Schmidt operators $\mathcal{H}_A \to \mathcal{H}_B$, and, as mentioned in Section \ref{sec:OS_conjugation}, the identity operator on an infinite-dimensional Hilbert space is never Hilbert--Schmidt. Thus, the universal Hilbert space $\mathcal{H}_{B \sqcup \overline{A}}$ is a space of what we might call ``Hilbert--Schmidt morphisms'' from $A$ to $B$. We will call the resulting category the \textit{Hilbert--Schmidt baby universe category}, and denote it by $\CBUHS$.

There is nothing wrong with this direct approach, and we could prove all of our main results by working with $\CBUHS$ alone. However, our main approach in this paper will be slightly more involved, taking a detour through the world of von Neumann algebras. We do this for a few reasons: firstly, in order to make more of a direct connection to the motivating previous work \cite{Colafranceschi:2023urj, Marolf:2024adj}, and secondly, in order to provide a more easily generalizable framework. As we will see, the fact that composition of bordisms is jointly continuous in the Hilbert space topology is a lucky accident, arising physically from the fact that we are dealing with closed spatial manifolds and, ultimately, atomic (hence in particular Type I) von Neumann algebras. That being said, our motivating goal is to understand the complete breakdown of boundary locality in quantum gravity. Any higher-categorical extension of our framework would require us to cut spatial manifolds open along higher-codimension entangling surfaces, where we expect to find Type III von Neumann algebras.

To motivate our construction, note that if $\mathcal{H}_{B \sqcup \overline{A}}$ is the analog of the space of Hilbert--Schmidt operators between some Hilbert spaces associated to $A$ and $B$, what we want is the analog of the space of all bounded operators. These spaces will form the morphisms in a category $\CBUvN$, which is a \textit{W*-category}, the multi-object analog of a von Neumann algebra. To construct $\CBUvN$, we will follow \cite{Colafranceschi:2023urj}, and take a double commutant in a natural Hilbert space representation induced by $\zeta$.

\subsection{The algebraic core}\label{sec:analytic_core}

The first step in our construction is to linearize the bordism category $\BordX$, so that we can begin working with bordisms as operators before taking any completions.

\begin{definition}\label{defn:CBUpre}
    The \textit{pre-baby universe category} $\CBUpre$ is the free linearization $\mathbb{C}[\BordX]$ of the unitary bordism category of source manifolds.
\end{definition}

\noindent Unpacking Definition \ref{defn:CBUpre}, the objects of $\CBUpre$ are the same as those of $\BordX$, and the morphism spaces of $\CBUpre$ are the free vector spaces,
\begin{equation}
    \CBUpre(A \to B) = \mathbb{C}[\BordX(A \to B)].
\end{equation}
More concretely, a morphism $\mathcal{O} : A \to B$ in $\CBUpre$ is a finite linear combination,
\begin{equation}
    \mathcal{O} = \sum_i c_i N_i, \quad N_i : A \to B
\end{equation}
of $\mathcal{X}$-bordisms, and composition is induced bilinearly from the composition of bordisms. We extend the unitary and symmetric monoidal structures of $\BordX$ (anti-)linearly to $\CBUpre$. The pre-baby universe category $\CBUpre$ is the multi-object version of the surface algebras of \cite{Colafranceschi:2023urj}, in that the endomorphism algebra of an object $B \in \CBUpre$ is precisely the (left) surface algebra of $B$ considered in \cite[Section 3.1]{Colafranceschi:2023urj}.

The unitary monoidal category $\CBUpre$ admits a canonical trace, defined on bordisms by
\begin{equation}\label{eq:trace_on_CBUpre}
    \mathrm{tr}(N) = \zeta\big(\tr_\Bord(N)\big),
\end{equation}
and extended linearly, where we recall from Section \ref{sec:traces_and_fermions} that $\mathrm{tr}_\Bord(N) = \str_\Bord(\theta_B N)$ denotes the closed $\mathcal{X}$-manifold obtained from a bordism $N : B \to B$ by gluing its outgoing boundary to its incoming boundary after first acting with the twist $\theta_B$. The trace automatically satisfies many desired properties, such as
\begin{equation}
        \tr(\mathcal{O}_1 \circ \mathcal{O}_2) = \tr(\mathcal{O}_2 \circ \mathcal{O}_1), \quad \tr(\mathcal{O}^\dagger) = \overline{\tr(\mathcal{O})}, \quad \tr(\mathcal{O}_1 \sqcup \mathcal{O}_2) = \tr(\mathcal{O}_1) \tr(\mathcal{O}_2),
\end{equation}
which follow from naturality of $\theta_B$, reality, and multiplicativity, respectively. Moreover, the trace is normalized, in that $\tr(\varnothing) = 1$.

It will frequently be useful to convert a morphism $\mathcal{O} : A \to B$ in $\CBUpre$ into a morphism $\varnothing \to B \sqcup \overline{A}$. We do this by using the unitary structure, exactly as in $\BordX$, extending \eqref{eq:ket_of_bord} from Section \ref{sec:traces_and_fermions} linearly. Further, as discussed in Section \ref{sec:traces_and_fermions} and proven in Proposition \ref{prop:trace_pairing_is_inner_product}, the trace pairings and inner products agree in $\BordX$. By applying $\zeta$, we learn that the same holds in $\CBUpre$:
\begin{equation}\label{eq:trace_equals_inner_prod}
    \tr(\mathcal{O}_1^\dagger \circ \mathcal{O}_2) = \braket{\mathcal{O}_1 | \mathcal{O}_2}.
\end{equation}
By reflection positivity, we learn that the trace on $\CBUpre$ is positive. In fact, our entire construction depends on $\zeta$ only through the induced trace \eqref{eq:trace_on_CBUpre}, viewed as a positive tracial state on the $\dagger$-category $\CBUpre$.

By extending the universal construction linearly, we obtain a Hilbert space representation
\begin{equation}\label{eq:univ_const_extended_linearly}
    \CBUpre \to \Hilb,
\end{equation}
where we recall that a Hilbert space representation of a linear $\dagger$-category $\mathcal{C}$ is a linear $\dagger$-functor $\mathcal{C} \to \Hilb$. However, this representation is too coarse: it only remembers how to compose a bordism $N : B_1 \to B_2$ with bordisms $M : \varnothing \to B_1$, as these are the only bordisms which prepare states in the universal Hilbert space $\mathcal{H}_{B_1}$. For instance, if there are no bordisms $M : \varnothing \to B_1$ at all, so that $[B_1]$ is nontrivial in the bordism group $\Omega_{d-1}^\mathcal{X}$, then the operator $\widehat{N} : \mathcal{H}_{B_1} \to \mathcal{H}_{B_2}$ is simply zero.

Applying Grothendieck's relative point of view, we should instead consider how $N$ acts on states $\ket{M} \in \mathcal{H}_{B_1 \sqcup \overline{A}}$ prepared by bordisms $M : A \to B_1$, for arbitrary $A$. To describe this action, let us introduce the following notation for the \textit{partial composition} of morphisms $\CO : B_1 \to B_2$ with kets $\ket{\CO'}$ associated to morphisms $\CO':A \to B_1$:
\begin{equation}\label{eq:partial_composition_left}
    \CO \circ_{B_1} \ket{\CO'} \defined \ket{\CO \circ \CO'}.
\end{equation}
Partial composition allows us to construct a family of \textit{$\overline{A}$-twisted representations},
\begin{equation}\label{eq:A_twisted_rep}
    \CBUpre \to \Hilb, \quad B \mapsto \mathcal{H}_{B \sqcup \overline{A}}, \quad \CO \mapsto \widehat{\CO}_{\overline{A}},
\end{equation}
where $\widehat{\CO}_{\overline{A}}$ acts by partial composition as in \eqref{eq:partial_composition_left}, leaving the boundary $\overline{A}$ fixed. As for the universal construction, the operator $\widehat{\CO}_{\overline{A}}$ is well-defined and bounded by the trace inequality (Proposition \ref{prop:trace_inequality}), which provides the bound (Corollary \ref{cor:operator_norm_bound}),
\begin{equation}\label{eq:norm_bound_by_Hilb_norm}
    \lvert\lvert\widehat{\mathcal{O}}_{\overline{A}}\rvert\rvert \leq \sqrt{\braket{\mathcal{O} | \mathcal{O}}},
\end{equation}
uniformly in $A$.

We define the \textit{operator norm} of a morphism $\mathcal{O}$ in $\CBUpre$ to be,
\begin{equation}\label{eq:operator_norm}
    \lvert \lvert \mathcal{O} \rvert \rvert \defined \sup_A\ \lvert\lvert\widehat{\mathcal{O}}_{\overline{A}}\rvert\rvert,
\end{equation}
the supremum of the operator norms of $\mathcal{O}_{\overline{A}}$ over all $A$. This operator norm, arising from a family of Hilbert space representations, automatically satisfies the C*-axioms, and so $\CBUpre$ is a pre-C*-category in the sense of \cite{bunke2020additive}.

As $\CBUpre$ is equipped with a unitary structure, and in particular a $\dagger$-structure, we should also consider the adjoint version of Grothendieck's relative point of view. In this adjoint version, we view a state in $\CH_{B \sqcup \overline{A}}$, not as a generalization of a ket vector $\ket{M}$ prepared by a bordism $M : \varnothing \to B$, but as a generalization of a bra vector $\bra{M'}$ corresponding to a bordism $M' : A \to \varnothing$. Thus, for composable bordisms $\overline{\CO} : \overline{A}_1 \to \overline{A}_2$ and $\CO' : A_1 \to B$, we introduce the conjugate partial composition,
\begin{equation}\label{eq:partial_composition_right}
    \overline{\CO} \circ_{\overline{A_1}} \ket{\CO'} \defined \ket{\CO' \circ \CO^\dagger},
\end{equation}
as justified by Proposition \ref{prop:dinatuality_of_kets}. We then define the $B$-twisted representation of the complex conjugate category $\overline{\CBUpre}$ by,
\begin{equation}
    \overline{\CBUpre} \to \Hilb, \quad \overline{A} \mapsto \mathcal{H}_{B \sqcup \overline{A}}, \quad \overline{\CO} \mapsto \widehat{\overline{\CO}}_B
\end{equation}
where $\widehat{\overline{\CO}}_B$ acts by conjugate partial composition as in \eqref{eq:partial_composition_right}, with $B$ held fixed. As with the $\overline{A}$-twisted representations of $\CBUpre$, Corollary \ref{cor:operator_norm_bound} gives a uniform bound on the analogous operator norms in the $B$-twisted representations of $\overline{\CBUpre}$.

Because of the $\dagger$-structure, the $\overline{A}$-twisted representations the $B$-twisted representations are not independent: the operation of taking a bordism $M : A \to B$ to its adjoint $M^\dagger : B \to A$ defines a \textit{modular conjugation} map,
\begin{equation}\label{eq:modular_conjugation}
    \widehat{J} : \mathcal{H}_{B \sqcup \overline{A}} \to \mathcal{H}_{A \sqcup \overline{B}}, \quad \ket{M} \mapsto \ket{M^\dagger},
\end{equation}
which we prove is antiunitary in Corollary \ref{corr:Hermitian_and_dagger_is_antiunitary}. Equivalently, modular conjugation is simply the $\dagger$-structure on $\CBUHS$. As we prove in Corollary \ref{cor:modular_conjugation}, modular conjugation exchanges the action of a morphism $\mathcal{O} : B_1 \to B_2$ in the $\overline{A}$-twisted representation of $\CBUpre$ with the action of the conjugate morphism $\overline{\mathcal{O}} : \overline{B}_1 \to \overline{B}_2$ in the $A$-twisted representation of $\overline{\CBUpre}$. As a consequence, we may simply speak of ``the operator norm,'' as the operator norms defined by composition on either side agree.

The actions of morphisms $\mathcal{O} : B_1 \to B_2$ and $\mathcal{O}' : \overline{A}_1 \to \overline{A}_2$ on the two-boundary Hilbert space $\CH_{B \sqcup \overline{A}}$ commute. Thus, the $\overline{A}$-twisted representations of $\CBUpre$ and the $B$-twisted representations of $\overline{\CBUpre}$ combine to a representation,
\begin{equation}\label{eq:regular_rep}
    \CBUpre \otimes \overline{\CBUpre} \to \Hilb, \quad (B, A) \mapsto \mathcal{H}_{B \sqcup \overline{A}},
\end{equation}
of the tensor product category $\CBUpre \otimes \overline{\CBUpre}$, as depicted in Figure \ref{fig:2bdy_operators}. If we only consider diagonal sectors $\mathcal{H}_{B \sqcup \overline{B}}$, then the actions of $\CBUpre$ and $\overline{\CBUpre}$ on \eqref{eq:regular_rep} are simply the actions of the surface algebras on the left and right boundaries considered in \cite{Colafranceschi:2023urj, Marolf:2024adj}. The representation \eqref{eq:regular_rep} is obtained by completing the regular action of $\CBUpre$ on itself in the Hilbert space topology derived from $\zeta$, so we will call it the \textit{regular representation} (\textit{induced by $\zeta$}, when we need to make the dependence on $\zeta$ explicit).

\begin{figure}
    \centering
    \includegraphics[width=0.4\linewidth]{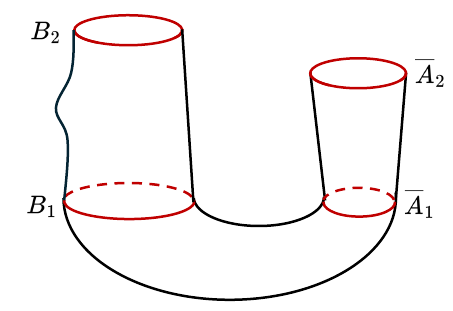}
    \caption{The two-boundary Hilbert space $\CH_{B_1\sqcup\overline A_1}$ is the representation space for the regular representation. Morphisms $B_1 \to B_2$ in $\CBUpre$ act on the left boundary $B_1$, and morphisms $\overline{A}_1 \to \overline{A}_2$ act on the right boundary $\overline{A}_1$, in such a way that the two actions commute. Modular conjugation $\widehat{J}$ acts by reflecting the whole figure horizontally.}
    \label{fig:2bdy_operators}
\end{figure}

Because of null states, the regular representation is not faithful, and so the operator norm \eqref{eq:operator_norm} is, strictly speaking, merely a seminorm. Let us say that a morphism $\mathcal{O} : B_1 \to B_2$ in $\CBUpre$ is \textit{null} if the operators $\widehat{\mathcal{O}}_{\overline{A}}$ vanish for all $A$, or equivalently, if $\lvert\lvert \mathcal{O} \rvert \rvert = 0$. The subspaces of null morphisms in $\CBUpre$ form an ideal of morphisms $\mathcal{N} \subset \CBUpre$, and we may form the quotient category $\CBUpre/\mathcal{N}$, defined by taking the quotient of each morphism space by the linear subspace of null morphisms. By definition, the regular representation \eqref{eq:regular_rep} descends to a faithful representation of $\CBUpre/\mathcal{N}$.

In fact, we claim that a morphism $\mathcal{O} : B_1 \to B_2$ is null if and only if the corresponding state $\ket{\mathcal{O}}$ is a null state. Certainly, if $\ket{\mathcal{O}}$ is null as a state, then we have $\lvert\lvert\mathcal{O}\rvert\rvert = 0$ by \eqref{eq:norm_bound_by_Hilb_norm}, so $\mathcal{O}$ is null as a morphism. Conversely, suppose $\mathcal{O}$ is null as a morphism. Consider the states
\begin{equation}
    \widehat{\mathcal{O}}_{\overline{B}_1} \ket{C_{B_1}(\beta)} = \ket{\mathcal{O} \circ C_{B_1}(\beta)} = C_{\overline{B}_1}(\beta) \circ_{\overline{B}_1} \ket{\mathcal{O}},
\end{equation}
in $\mathcal{H}_{B_2 \sqcup \overline{B}_1}$. Since $\mathcal{O}$ is null as a morphism, the state $\widehat{\mathcal{O}}_{B_1} \ket{C_{B_1}(\beta)}$ must vanish. But by Proposition \ref{prop:approx_identities_in_two_sided_reps}, the cylinder operators converge to the identity in the strong operator topology as $\beta \to 0$, so the states $C_{\overline{B}_1}(\beta) \circ_{\overline{B}_1} \ket{\mathcal{O}}$ converge to $\ket{\mathcal{O}}$. Thus, $\ket{\mathcal{O}}$ is null, as claimed.

\subsubsection{Simplifications in the topological case}

Momentarily, we will complete $\CBUpre/\mathcal{N}$ by taking a double commutant. However, we first note a tremendous simplification in the topological case. This simplification arises from the general bound,
\begin{equation}\label{eq:thermal_bound_for_top_simplification}
    \mathrm{tr}_{\mathcal{H}_B} (e^{- \beta \widehat{H}_B}) \leq \zeta(S^1_\beta \times B),
\end{equation}
proven in Proposition \ref{prop:thermal_bound} and discussed previously in Section \ref{sec:ER=EPR}. In the topological case, \eqref{eq:thermal_bound_for_top_simplification} reduces to,
\begin{equation}\label{eq:TQFT_finite_dim}
    \dim(\mathcal{H}_B) \leq \zeta(S^1 \times B) < \infty.
\end{equation}
Thus, the universal Hilbert spaces of a topological partition function are finite-dimensional. This is not obvious from the start, because there are frequently infinitely many diffeomorphism classes of smooth manifolds with any fixed boundary.

As we have seen, a morphism $\mathcal{O}$ in $\CBUpre$ is null if and only if the corresponding state $\ket{\CO}$ is null. Thus, the morphism spaces in $\CBUpre/\mathcal{N}$ are also finite dimensional, and so $\CBUpre/\mathcal{N}$ needs no operator algebraic completion (or, rather, any completion would just give us $\CBUpre/\mathcal{N}$ back). So we are free to work with $\CBUpre/\mathcal{N}$ on its own, a C*-category with finite dimensional morphism spaces. Readers who are only interested in the topological case may freely skip to Section \ref{sec:idempotent_completion}, and replace the W*-Cauchy completion discussed there with the algebraic semisimple completion.

\subsection{Taking a double commutant}

Following \cite{Colafranceschi:2023urj,Marolf:2024adj}, we now form an operator-algebraic completion of the pre-C*-category $\CBUpre$ by taking a double commutant.

\begin{definition}\label{defn:CBU_vN}
    The \textit{left von Neumann completion} $\CBUvNL$ of $\CBUpre$ is the double commutant W*-category associated to the action of $\CBUpre$ on the regular representation \eqref{eq:regular_rep}. Analogously, the \textit{right von Neumann completion} $\CBUvNR$ is the double commutant W*-category associated to the action of $\overline{\CBUpre}$ on the regular representation.\footnote{To stave off a possible confusion, note that the action of $\overline{\CBUpre}$ on the regular representation, whose double commutant defines the right von Neumann completion $\CBUvNR$, is nevertheless still a \textit{left} action, though it can equivalently be viewed as a right action of $\CBUpre$ using the $\dagger$-structure.}
\end{definition}

\noindent While we have defined both a left and right completion of $\CBUpre$, modular conjugation \eqref{eq:modular_conjugation} induces a canonical anti-unitary spatial equivalence of categories,
\begin{equation}
    \CBUvNR = \overline{\CBUvNL}.
\end{equation}
By convention, we will use the un-adorned notation $\CBUvN$ below to denote the left von Neumann completion $\CBUvNL$, and simply refer to it as ``the von Neumann completion.''

Let us spell out the definition of the double commutant W*-category, as defined in \cite{ghez1985w, bunke2020additive}.\footnote{Technically, our definition of the double commutant is a slight generalization of the definition given in \cite{ghez1985w, bunke2020additive}, as they define the double commutant in a one-sided representation. Our definition of $\CBUvN$ coincides with the double commutant of \cite{ghez1985w,bunke2020additive} in the (enormous) direct-sum representation $B \mapsto \bigoplus_{A} \mathcal{H}_{B \sqcup \overline{A}}$, while the single commutant of \cite{ghez1985w,bunke2020additive} corresponds to the linking algebra \cite{ghez1985w} of our single commutant.} First of all, the \textit{commutant} $(\CBUpre)'$ of the left action of $\CBUpre$ on the regular representation is the category whose objects are the same as $\overline{\CBUpre}$, and whose morphisms $\overline{T} : \overline{A}_1 \to \overline{A}_2$ are uniformly bounded natural transformations from the $\overline{A}_1$-twisted representation of $\CBUpre$ to the $\overline{A}_2$-twisted representation. More explicitly, a morphism $\overline{T} : \overline{A}_1 \to \overline{A}_2$ in $(\CBUpre)'$ is a family of operators,
\begin{equation}
    \widehat{\overline{T}}_B : \mathcal{H}_{B \sqcup \overline{A}_1} \to \mathcal{H}_{B \sqcup \overline{A}_2},
\end{equation}
parametrized by $B \in \CBUpre$, which commute with the left action of $\CBUpre$ in that
\begin{equation}\label{eq:commutation_relation}
    \widehat{\overline{T}}_{B_2} \widehat{\mathcal{O}}_{\overline{A}_1} = \widehat{\mathcal{O}}_{\overline{A}_2} \widehat{T}_{B_1},
\end{equation}
for all morphisms $\mathcal{O} : B_1 \to B_2$ in $\CBUpre$. We have a canonical inclusion,
\begin{equation}
    \overline{\CBUpre} \to (\CBUpre)',
\end{equation}
due to the uniform bound \eqref{eq:norm_bound_by_Hilb_norm}, whose kernel is precisely $\overline{\mathcal{N}}$.

Iterating, the \textit{double commutant} $(\CBUpre)''$ is the category whose objects are the same as $\CBUpre$, and whose morphisms $\mathcal{O} : B_1 \to B_2$ are uniformly bounded natural transformations from the $B_1$-twisted representation of $(\CBUpre)'$ to the $B_2$-twisted representation. Concretely, a morphism $\mathcal{O} : B_1 \to B_2$ in $(\CBUpre)''$ is a family of operators
\begin{equation}
    \widehat{O}_{\overline{A}} : \mathcal{H}_{B_1 \sqcup \overline{A}} \to \mathcal{H}_{B_2 \sqcup \overline{A}},
\end{equation}
which commute with all morphisms $\overline{T}$ in $(\CBUpre)'$ as in \eqref{eq:commutation_relation}. Analogously, there is a canonical inclusion,
\begin{equation}\label{eq:faithful_inclusion}
    \CBUpre \to (\CBUpre)'',
\end{equation}
whose kernel is $\mathcal{N}$. The double commutant $(\CBUpre)''$ (even the single commutant $(\CBUpre)'$) is automatically a W*-category \cite{ghez1985w}. Note that while $\CBUpre$ is a non-unital category, the (single or double) commutant is automatically unital, as identity operators commute with everything.

Our definition of $\CBUvN$ as a double commutant, while useful for establishing formal properties, is far from illuminating. To build a better understanding of $\CBUvN$, note that the image of $\CBUpre$ in $\CBUvN$ under the canonical map \eqref{eq:faithful_inclusion} is dense in the weak operator topology. This follows from the W*-categorical double commutant theorem \cite[Theorem 4.2]{ghez1985w} combined with the non-degeneracy of the regular representation (as proven in Proposition \ref{prop:approx_identities_in_two_sided_reps}). Thus, $\CBUvN$ has the following, more physical, description: its objects are closed $(d-1)$-dimensional $\mathcal{X}$-manifolds, just as in $\BordX$, and its morphisms are weak operator limits of superpositions of $\mathcal{X}$-bordisms. Taking weak operator limits automatically factors through the quotient by the ideal $\mathcal{N}$ of null morphisms, as $\mathcal{N}$ is the closure of the origins in the morphism spaces in $\CBUpre$ with respect to the weak operator topology induced by the regular representation.

As in any W*-category, the endomorphism algebras
\begin{equation}
    \mathcal{A}_B \defined \CBUvN(B \to B),
\end{equation}
in $\CBUvN$ are von Neumann algebras. In fact, these von Neumann algebras are precisely the left von Neumann algebras $\mathcal{A}_B^L$ studied in \cite{Colafranceschi:2023urj,Marolf:2024adj}.\footnote{The endomorphism algebra of the conjugate object $\overline{B} \in \overline{\CBUvN} = \CBUvNR$ is the corresponding right von Neumann algebra $\mathcal{A}_B^R$.} To see this, note that \cite{Marolf:2024adj} proved that the representations of the surface algebras (the endomorphism algebras in $\CBUpre$) on the off-diagonal sectors $\mathcal{H}_{B \sqcup \overline{A}}$ extend to representations of the von Neumann algebras $\mathcal{A}_B^L$ defined via the representation on the diagonal sector $\mathcal{H}_{B \sqcup \overline{B}}$. This means that the weak operator topology induced on the surface algebras by the diagonal representation agrees with that induced by the entire regular representation, and so their completions agree as well.

By construction, the regular representation admits commuting actions of the left and right von Neumann completions $\CBUvNL$ and $\CBUvNR$, and thus defines a representation,
\begin{equation}\label{eq:standard_form_rep}
    \CBUvNL \otimes \CBUvNR \to \Hilb, \quad (B, A) \mapsto \mathcal{H}_{B \sqcup \overline{A}}.
\end{equation}
By Tomita's commutation theorem \cite{TakesakiI},\footnote{Technically, by applying Tomita's commutation theorem to the associated linking algebras \cite{ghez1985w}.} we have
\begin{equation}
    \widehat {J} \CBUvNL \widehat{J} = (\CBUvNL)',
\end{equation}
identifying the commutant of $\CBUvNL$ as its image $\widehat {J} \CBUvNL \widehat{J} = \CBUvNR = \overline{\CBUvNL}$ under modular conjugation. As a result, we learn that the regular representation is precisely the standard-form representation of the W*-category $\CBUvN$, in which we have faithful actions of $\CBUvN$ and $\overline{\CBUvN}$ which are each the commutant of the other.

In any W*-category, the standard-form representation plays a key role. As discussed in \cite{henriques2024completewcategories}, the functor,
\begin{equation}
    \CBUvN \otimes \overline{\CBUvN} \to \Hilb, \quad (B, A) \mapsto \mathcal{H}_{B \sqcup \overline{A}},
\end{equation}
mapping a pair of objects in the W*-category $\CBUvN$ to the standard-form representation is also known as the \textit{Hilbert space-valued inner product}, as it behaves in many ways as a categorical analog of an inner product. We will use the notation
\begin{equation}\label{eq:categorical_wavefunction_overlep}
    \bbrakket{A | B} \defined \mathcal{H}_{B \sqcup \overline{A}},
\end{equation}
to denote this Hilbert space-valued inner product in order to recall the boundary-state formalism. We will also refer to the Hilbert space-valued inner product as the \textit{categorical wavefunction overlap}.

Let us note that the unitary structure on $\BordX$ endows the $\CBUvN$ with an additional piece of structure beyond its $\dagger$-structure, which is the anti-linear involution defined by
\begin{equation}\label{eq:bi_involutive_structure}
    \overline{(-)} : \CBUvN \to \overline{\CBUvN}, \quad B \mapsto \overline{B}.
\end{equation}
This involution, together with the $\dagger$-structure, equip $\CBUvN$ with a \textit{bi-involutive structure} in the sense of \cite{henriques2024completewcategories}.

\subsubsection{A bit of Tomita--Takesaki theory}\label{sec:Tomita_Takesaki}

Our construction of the von Neumann completion $\CBUvN$ is a natural generalization of a standard construction in Tomita--Takesaki theory \cite{TakesakiI,TakesakiII}. In this section, we spell out this analogy, and apply standard techniques from Tomita--Takesaki theory to quickly reproduce the key analytic results of \cite{Colafranceschi:2023urj} concerning the structure of the von Neumann endomorphism algebras $\mathcal{A}_B$ in $\CBUvN$.

In Tomita--Takesaki theory, one starts with what is known as a \textit{Hilbert algebra}: a pre-Hilbert space $\mathcal{H}^\mathrm{pre}$ equipped with a (typically non-unital) $*$-algebra structure satisfying a number of axioms. One then forms the Hilbert space completion $\mathcal{H}$ of the pre-Hilbert space $\mathcal{H}^\mathrm{pre}$, and considers the action of $\mathcal{H}^\mathrm{pre}$ on $\mathcal{H}$ induced by left multiplication. The \textit{left von Neumann algebra} $\mathcal{A}$ of the Hilbert algebra $\mathcal{H}^\mathrm{pre}$ is the double commutant,
\begin{equation}
    \mathcal{A} = (\mathcal{H}^\mathrm{pre})'',
\end{equation}
of this collection of left-multiplication operators. By writing the Tomita operator $\widehat{S} : a \mapsto a^\dagger$ in terms of its polar decomposition, $\widehat{S} = \widehat{J} \widehat{\Delta}^{1/2}$, one defines the modular conjugation operator $\widehat{J}$ and the operator $\widehat{\Delta}$ which implements modular flow. Finally, using modular conjugation, one proves that $\mathcal{H}$ is a standard form representation of the von Neumann algebra $\mathcal{A}$.

Though not so named, this is precisely the formalism used by \cite{Colafranceschi:2023urj} to construct von Neumann algebras $\mathcal{A}_B$ from the surface algebras $\CBUpre(B \to B)$.\footnote{ZW thanks Marius Junge for explaining this point.} Above, we have generalized the construction of \cite{Colafranceschi:2023urj} to the natural multi-object version. The pre-baby universe category $\CBUpre$ is an example of what might be called a \textit{Hilbert algebroid}, namely, a multi-object version of a Hilbert algebra. By analogy, then, the von Neumann completion $\CBUvN$ might be called the \textit{left von Neumann algebroid} of the Hilbert algebroid $\CBUpre$. To demystify the naming conventions, recall that a one-object linear category is entirely determined by a single algebra, namely the endomorphism algebra of its single object. An \textit{algebroid} is the multi-object version of an algebra: a linear category. Thus, ``von Neumann algebroid'' is simply a synonym for what we have been calling a W*-category.

The Hilbert algebroid $\CBUpre$, and its endomorphism Hilbert algebras \cite{Colafranceschi:2023urj}, are quite special among all possible Hilbert algebroids and algebras. In a general Hilbert algebra, one does not require the map $a \mapsto a^\dagger$ to be anti-unitary. If it is, as in the morphism spaces of $\CBUpre$, one obtains what is known as a \textit{unimodular Hilbert algebra}, which necessarily has trivial modular flow, $\widehat{\Delta} = 1$. Moreover, in any Hilbert algebra, the formula,
\begin{equation}
    \varphi(a^\dagger a) \defined \braket{a | a},
\end{equation}
on positive elements $a^\dagger a$ in the unimodular Hilbert algebra $\mathcal{H}^\mathrm{pre}$ extends to define a faithful normal semifinite (n.s.f.) weight on the left von Neumann algebra $\mathcal{A}$. When the Hilbert algebra is unimodular, the weight $\varphi$ is \textit{tracial}, $\varphi(a^\dagger a) = \varphi(a a^\dagger)$, and defines an n.s.f. trace on $\mathcal{A}$. This trace agrees precisely with the trace \eqref{eq:trace_on_CBUpre} on the unimodular Hilbert algebroid $\CBUpre$, which extends to the n.s.f. trace on $\CBUvN$ constructed and studied in \cite{Colafranceschi:2023urj}. The existence of an n.s.f. trace on $\CBUvN$ immediately implies that the von Neumann endomorphism algebras $\mathcal{A}_B$ are direct integrals of type I and II factors only.

However, we are not yet done, as the Hilbert algebroid $\CBUpre$ is even better than unimodular: the composition on $\CBUpre$ is not merely separately bounded, but jointly bounded in the Hilbert space norm, according to the fundamental trace inequality of \cite{Colafranceschi:2023urj} (Proposition \ref{prop:trace_inequality}), 
\begin{equation}
    \tr(\mathcal{O}_1^\dagger \circ \mathcal{O}_2^\dagger \circ \mathcal{O}_2 \circ \mathcal{O}_1) \leq \tr(\mathcal{O}_2^\dagger \circ \mathcal{O}_2) \tr(\mathcal{O}_1^\dagger \circ \mathcal{O}_1).
\end{equation}
Applied to orthogonal projectors $P : B \to B$ in $\CBUvN$ satisfying $P^2 = P$, $P^\dagger = P$, the trace inequality reduces to
\begin{equation}
    \tr(P) \leq \tr(P)^2,
\end{equation}
or equivalently,
\begin{equation}\label{eq:lower_bound_on_trace}
    \tr(P) \geq 1,
\end{equation}
for non-zero projectors by the faithfulness of the trace.

As a result, the W*-category $\CBUvN$ is \textit{atomic}, meaning every projector (even ones with infinite trace) is an orthogonal sum of minimal projectors. As a result, the von Neumann endomorphism algebras $\mathcal{A}_B$ are a direct sum of type I factors, as shown in \cite{Colafranceschi:2023urj}. The atomicity of $\CBUvN$ will play an absolutely crucial role in all of our main results.\footnote{It also means that the full formalism of Tomita--Takesaki theory is overkill for $\CBUvN$. We include it here to prepare for future extensions of our framework to higher codimension.}

Atomic von Neumann algebras have a historically important alternative description. By taking the subspace of Hilbert--Schmidt operators inside an atomic von Neumann algebra $\mathcal{A}$, equipped with the Hilbert--Schmidt norm,
\begin{equation}
    \lvert \lvert a \rvert \rvert_2 \defined \sqrt{\tr(a^\dagger a)},
\end{equation}
one obtains what is known as an \textit{H*-algebra}, and conversely, every non-degenerate H*-algebra arises uniquely in this way \cite{Ambrose1945}. The Hilbert--Schmidt category $\mathcal{C}_\mathrm{BU}^\mathrm{HS}$, embedded in $\CBUvN$ via the bound \eqref{eq:norm_bound_by_Hilb_norm}, is the associated \textit{H*-category}, by which we mean a multi-object analog of an H*-algebra.

\subsection{Categorical Cauchy completion}\label{sec:idempotent_completion}

Above, we constructed a W*-category $\CBUvN$ by taking the double commutant of $\CBUpre$ in the regular representation, or equivalently, completing $\CBUpre$ in the corresponding weak operator topology. We are now prepared to define the baby universe category $\CBU$ itself. Having already taken a completion at the level of morphisms, we now need to take a completion at the level of objects.

The first step in completing the set of objects in $\CBUvN$ is very simple: we form the direct-sum completion $(\CBUvN)^\oplus$, whose objects are possibly-infinite formal orthogonal direct sums of objects in $\CBUvN$. More invariantly, $(\CBUvN)^\oplus$ has objects associated to any categorical wavefunction,
\begin{equation}\label{eq:categorical_superposition_of_projectors}
    \Psi = \bigoplus_i \mathcal{H}_i \otimes B_i, \quad \mathcal{H}_i \in \Hilb, \quad B_i \in \CBUvN.
\end{equation}
The tensor product $\mathcal{H} \otimes B$ should be understood as the result of stacking a decoupled quantum system, with Hilbert space $\mathcal{H}$, on top of whatever quantum system $B$ represents. The choice of an orthonormal basis for $\mathcal{H}$ induces an isomorphism,
\begin{equation}
    \mathcal{H} \otimes X \cong X^{\oplus \mathrm{dim}(\mathcal{H})},
\end{equation}
so $\mathcal{H} \otimes B$ can always be written as formal direct sums of objects in $\CBUvN$. We extend the bi-involutive structure $B \mapsto \overline{B}$ on $\CBUvN$ to $(\CBUvN)^\oplus$ by anti-linearity over $\Hilb$; for instance, we set $\overline{\mathcal{H} \otimes B} = \overline{\mathcal{H}} \otimes \overline{B}$.

The passage from $\CBUvN$ to its direct-sum completion $(\CBUvN)^\oplus$ should be viewed as analogous to forming the free $\mathbb{C}$-vector spaces $\mathcal{H}_B^\mathrm{pre}$ from the spaces of geometric bordisms $\BordX(\varnothing \to B)$ in the universal construction. In both cases, we allow ourselves to form arbitrary linear superpositions of geometric manifolds, with the only differences being the dimensions of the manifolds and the category level of the coefficients. The final step in our construction of $\CBU$ will be analogous, then, to the Cauchy completion and quotient by null states which produces the universal Hilbert spaces $\mathcal{H}_B$ from the universal pre-Hilbert spaces $\mathcal{H}_B^\mathrm{pre}$.

As motivation, let us recall the gravitational interpretation of the objects in $\CBUvN$ given by closed $(d-1)$-dimensional $\mathcal{X}$-manifolds. Given two manifolds $A, B \in \CBUvN$, the categorical wavefunction overlap $\bbrakket{A|B} = \mathcal{H}_{B \sqcup \overline{A}}$ is the Hilbert space of ER bridges from $A$ to $B$, possibly including the trivial sort of ER bridge where $A$ and $B$ are not connected. For any fixed ER bridge, there may be some sort of unsplittable gauge flux flowing along the bridge from $A$ to $B$, which may or may not be visible in the EFT description.

Merely fixing the boundary conditions $A, B$ at the boundary of an ER bridge does not necessarily fully fix the flux flowing through the bridge. Moreover, any fixed flux may appear with some large multiplicity. We would like to determine the irreducible unsplittable fluxes which may flow through an ER bridge and obstruct ER = EPR. Thus, our goal now is to break the objects $A, B \in \CBUvN$ into their constituent irreducible flux sectors.

We begin by recalling the categorical notion of the image of a projector. If $X \in \mathcal{C}$ is an object of some category, and $P : X \to X$ is a projector, $P^2 = P$, we say the \textit{image of $P$}, if it exists, is an object $\mathrm{im}(P) \in \mathcal{C}$ equipped with inclusion and projection morphisms,
\begin{equation}
    i:\mathrm{im}(P) \to X, \quad \pi : X \to \mathrm{im}(P),
\end{equation}
such that $i \circ \pi = P$ and $\pi \circ i = \mathrm{id}_{\mathrm{im}(P)}$. When an image $\mathrm{im}(P)$ exists, it is unique up to unique isomorphism. In a W*-category, we typically assume our projectors are orthogonal, $P^\dagger = P$, and demand that the inclusion and projection are related by $i^\dagger = \pi$. Given an orthogonal projector $P : X \to X$ in a W*-category, we obtain an orthogonal direct sum decomposition,
\begin{equation}
    X \cong \mathrm{im}(P) \oplus \mathrm{im}(1-P),
\end{equation}
whenever both $P$ and $1-P$ have images.

However, the existence of images for projectors is far from guaranteed. Let us illustrate this with an example.

\begin{example}
    Consider topological quantum mechanics as discussed in Section \ref{sec:Ex_TopoQM}, specified by $\zeta(S^1) = k \in \mathbb{Z}_{\geq 0}$. Consider the morphism $P : \mathrm{pt}^+ \sqcup \mathrm{pt}^- \to \mathrm{pt}^+ \sqcup \mathrm{pt}^-$ in $\CBUvN$ given by,
    \begin{equation}
        P=\frac{1}{k}\includegraphics[height=2cm,valign=c]{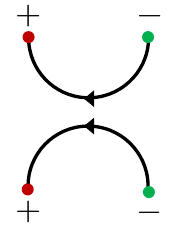}
    \end{equation}
    By construction, we have $P^\dagger = P$, and moreover we have $P^2 = P$ in $\CBUvN$, as a closed circle may be replaced by the number $k$.
    
    Now, $P$ happens to have an image in $\CBUvN$, given by the empty 0-manifold $\varnothing$, as can be seen visually by the fact that $P$ may be split into two pieces by a horizontal cut. The inclusion and projection morphisms are given by,
    \begin{equation}
        i= \frac{1}{\sqrt{k}}\,\includegraphics[height=1cm,valign=c]{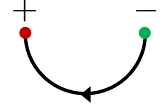}, \qquad \pi = i^\dagger = \frac{1}{\sqrt{k}}\, \includegraphics[height=1cm,valign=c]{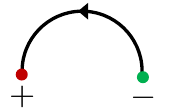},
    \end{equation}
    which satisfy $i\circ \pi = P$ and $\pi \circ i=\mathrm{id}_\varnothing$ as required.
    Thus, we would like to conclude that
    \begin{equation}\label{eq:decomp_adjoint}
        \mathrm{pt}^+ \sqcup \mathrm{pt}^- \stackrel{?}{=} \varnothing \oplus \mathrm{im}(1-P),
    \end{equation}
    as objects in $\CBUvN$.
    
    However, while $P$ has an image in $\CBUvN$, the complementary projector $1-P$ does not. We can see this by passing to the known underlying TQM, whose Hilbert space on $\mathrm{pt}^+$ is $\mathbb{C}^k$, or as a representation of $U(k)$, the fundamental $\Box$. The decomposition \eqref{eq:decomp_adjoint} wants to be the decomposition,
    \begin{equation}
        \Box \otimes \overline{\Box} = \mathds{1} \oplus \mathrm{Ad}_{\mathfrak{su}(k)},
    \end{equation}
    of $U(k)$ representations. But, while we do have a manifold $\varnothing$ corresponding to the trivial rep $\mathds{1}$, we do not have any manifold corresponding to $\mathrm{Ad}_{\mathfrak{su}(k)}$. Thus, we are not able to break $\mathrm{pt}^+ \sqcup \mathrm{pt}^-$ into a direct sum of manifolds in $\CBUvN$.
\end{example}

To remedy the failure of $\CBUvN$ to admit images for all projectors, we pass to what is known as the \textit{idempotent completion} of $\CBUvN$, otherwise known as the \textit{Karoubi envelope}.\footnote{See \cite{Binder:2019zqc} for a physicist-friendly overview of idempotent completion.} The idempotent completion of $\CBUvN$ is the category $\hat{\mathcal{C}}_\mathrm{BU}^\mathrm{vN}$ obtained by adjoining formal images $\mathrm{im}(P)$ of all orthogonal projectors $P$. For a pair of orthogonal projectors $P_1, P_2$ on objects $B_1, B_2$ respectively, we define the spaces of morphisms in the idempotent completion by,
\begin{equation}
    \hat{\mathcal{C}}_\mathrm{BU}^\mathrm{vN}\big(\mathrm{im}(P_1) \to \mathrm{im}(P_2)\big) = P_2 \circ \CBUvN(B_1 \to B_2) \circ P_1.
\end{equation}
Thus, a morphism $\mathrm{im}(P_1) \to \mathrm{im}(P_2)$ is a morphism $\mathcal{O} : B_1 \to B_2$ such that
\begin{equation}
    P_2 \circ \mathcal{O} = \mathcal{O} = \mathcal{O} \circ P_1.
\end{equation}
For example, the identity morphism of $\mathrm{im}(P)$ is the projector $P$ itself.

We have a fully faithful embedding $\CBUvN \to \hat{\mathcal{C}}_\mathrm{BU}^\mathrm{vN}$ mapping $B \mapsto \mathrm{im}(\mathrm{id}_B)$, which we will simply denote by $B$. One can easily check that $\mathrm{im}(P)$ actually is an image of $P$ in $\hat{\mathcal{C}}_\mathrm{BU}^\mathrm{vN}$, with inclusion and projection maps both given by $P$ itself. By the uniqueness of images, when they exist, any images of projectors that already existed in $\CBUvN$ will be canonically isomorphic to the new formal images added in the idempotent completion. Finally, we note that $\hat{\mathcal{C}}_\mathrm{BU}^\mathrm{vN}$ inherits a bi-involutive structure from $\CBUvN$, defined by $\overline{\mathrm{im}(P)} = \mathrm{im}(\overline{P})$.

What is the interpretation of the object $\mathrm{im}(P) \in \hat{\mathcal{C}}_\mathrm{BU}^\mathrm{vN}$ obtained as the formal image of a projector $P$ on $A$? To answer this, note that the categorical wavefunction overlap,
\begin{equation}
    \bbrakket{\mathrm{im}(P) | B} \subset \bbrakket{A | B} = \mathcal{H}_{B \sqcup \overline{A}},
\end{equation}
is simply the image of the conjugate projector $\overline{P}$ acting on the component $\overline{A}$ via the regular representation. Thus, an ER bridge from $\mathrm{im}(P) \subset A$ to $B$ is a superposition of ER bridges from $A$ to $B$, which happens to lie in the image of the projector $P$ acting (by pre-composition) on the incoming boundary $A$. We will occasionally abuse notation by treating objects $X, Y$ in the idempotent completion as if they were manifolds, for instance writing $\mathcal{H}_{Y \sqcup \overline{X}}$ for $\bbrakket{X | Y}$, or writing $C_X(\beta)$ for the projection of a cylinder $C_B(\beta)$ to the image of a projector $X \subset B$.

Taken together, the process of taking the direct-sum completion and the idempotent completions, in either (equivalent) order, is known as \textit{W*-Cauchy completion} \cite{henriques2024completewcategories}. For any W*-category $\mathcal{C}$, we use,
\begin{equation}
    \Hilb(\CC) \defined (\hat{\CC})^{\oplus} = \widehat{\CC^{\oplus}},
\end{equation}
to denote the W*-Cauchy completion. We can now give a definition of the baby universe category itself.

\begin{definition}\label{defn:CBU}
    The baby universe category is the W*-Cauchy completion,
    \begin{equation}
        \CBU = \Hilb(\CBUvN),
    \end{equation}
    of the von Neumann completion $\CBUvN$ of the pre-baby universe category $\CBUpre$.
\end{definition}

As mentioned above, the process of W*-Cauchy completion, and in particular the idempotent completion, is a categorical analog of the Cauchy completion of a pre-Hilbert space. The categorical analog of the quotient by null states is the following. While the objects of $\CBUpre$ have a simple, geometric description as closed $(d-1)$-dimensional manifolds, the objects in $\CBU$ have no such simple description. In particular, while every object of $\CBU$ may be written as a categorical superposition of images of projectors on manifolds, this expression need not be unique. This is because there may be nontrivial isomorphisms between images of different projectors acting on different manifolds. Before taking the idempotent completion, the objects $B \in \CBUpre$ are linearly independent, but once we pass to $\CBU$ and expand them in terms of their component pieces, we may learn that they satisfy categorical linear relations.

Gravitationally, let us recall from Section \ref{sec:universal_construction} that $\mathcal{H}_B^\mathrm{pre}$ is a vector space of kinematic states, while $\mathcal{H}_B$ is a Hilbert space of physical states arising from projecting onto the solutions to the Hamiltonian constraint equation associated to bulk time evolution. Analogously, $\CBUpre$ is a category of kinematic objects, while $\CBU$ is a category of physical objects.\footnote{The objects of $\CBU$, given by categorical quantum states of bulk-codimension two slices, may be viewed as a sort of motive of spacetime. This is because the passage from, say, the category of smooth projective varieties to the category of pure Chow motives is also given by linearization and idempotent completion (see, e.g., \cite{milne2013motives}).} We view the nontrivial isomorphisms in $\CBU$ as a categorical version of the Hamiltonian constraint, arising from the fact that the categorical evolution of a codimension-two slice of a bulk spacetime along codimension-one ER bridges is pure gauge.

Let us describe another way of seeing that passing to the W*-Cauchy completion includes a categorical analog of taking the quotient by null states. First, note that two states $\ket{\psi_1}, \ket{\psi_2}$ are equivalent modulo null states if and only if they have the same wavefunction overlap with all states,
\begin{equation}\label{eq:same_mod_null_iff_same_overlap}
    \ket{\psi_1} \sim \ket{\psi_2} \iff \braket{\phi | \psi_1} = \braket{\phi | \psi_2},\ \forall \ket{\phi}.
\end{equation}

While the null state $\ket{\psi_2} - \ket{\psi_1}$ has no categorical analog, as we cannot subtract objects in a category, we do have an analog of \eqref{eq:same_mod_null_iff_same_overlap}. Let $\Psi_1, \Psi_2 \in \CBU$ be any objects. Then we have,
\begin{equation}\label{eq:W*_Yoneda_as_mod_null}
    \Psi_1 \cong \Psi_2 \iff \bbrakket{\Phi | \Psi_1} \cong \bbrakket{\Phi | \Psi_2},\ \forall \Phi \in \CBU.
\end{equation}
In fact, \eqref{eq:W*_Yoneda_as_mod_null} does not merely describe a logical implication, but an isomorphism between the space of unitary isomorphisms $\Psi_1 \cong \Psi_2$ in $\CBU$ and the space of unitary natural isomorphisms $\bbrakket{-|\Psi_1} \cong \bbrakket{-|\Psi_2}$ between the functors $\bbrakket{-|\Psi_i} : \overline{\CBU} \to \Hilb$ defined by taking a categorical wavefunction overlap.

More formally, \eqref{eq:W*_Yoneda_as_mod_null} is the \textit{W*-categorical Yoneda Lemma} \cite[Prop. 4.12]{henriques2024completewcategories}, which asserts that the W*-categorical Yoneda embedding,
\begin{equation}\label{eq:W*_Yoneda}
    \CC \to \mathrm{Fun}_\text{W*}(\overline{\CC}, \Hilb), \quad \Psi \mapsto \bbrakket{- | \Psi},
\end{equation}
is fully faithful, for any W*-category $\CC$. Moreover, the second part of \cite[Prop. 4.12]{henriques2024completewcategories} asserts further that \eqref{eq:W*_Yoneda} presents $\mathrm{Fun}_\text{W*}(\overline{\CC}, \Hilb)$ as the W*-Cauchy completion $\Hilb(\CC)$. Thus, in the sense of \eqref{eq:W*_Yoneda_as_mod_null}, passing from $\CC$ to $\Hilb(\CC)$ involves a categorical quotient by null states.

\subsubsection{Decomposition into $\mu$-sectors}

Our goal, now, is to describe the structure of the Cauchy complete W*-category $\CBU$. This structure is entirely determined by the collection of projectors in the von Neumann endomorphism algebras $\mathcal{A}_B$, as studied in \cite{Colafranceschi:2023urj}. As we saw in Section \ref{sec:Tomita_Takesaki}, the trace inequality implies that the von Neumann algebras $\mathcal{A}_B$ are atomic, so that every projector is an orthogonal direct sum of minimal projectors. This translates into a sort of infinitary semisimplicity for $\CBU$: every object $X \in \CBU$, including those associated to manifolds $B$, decomposes as a possibly-infinite orthogonal direct sum of simple objects.\footnote{Equivalently, we have a canonical equivalence $\CBU = \mathrm{Rep}(\mathcal{Z}(\CBU))$, where the commutative atomic von Neumann algebra $\mathcal{Z}(\CBU)$ is the W*-algebraic center of the W*-category $\CBU$, i.e., the center of its linking algebra. The algebra $\mathcal{Z}(\CBU)$ was called the universal central algebra in \cite{Colafranceschi:2023urj,Marolf:2024adj}.}

These simple objects represent the irreducible flux sectors which may flow through an ER bridge, and correspond precisely to the $\mu$-sectors of \cite{Colafranceschi:2023urj,Marolf:2024adj}. Thus, we make the following definition.

\begin{definition}
    A \textit{$\mu$-sector} is a simple object in the baby universe category $\CBU$.
\end{definition}

\noindent Each $\mu$-sector may be realized as the formal image of some minimal projector $P \in \mathcal{A}_B$ acting on some manifold $B$, but this representation need not be unique, as discussed above, due to the existence of nontrivial isomorphisms in $\CBU$. For each $\mu$-sector, we define the conjugate sector $\overline{\mu}$ via the bi-involutive structure inherited from $\CBUvN$ (and ultimately from $\BordX$).

The atomic nature of $\CBU$ allows us to do a straightforward sort of categorical quantum mechanics in the complete W*-category $\CBU$, as described in general in \cite{henriques2024completewcategories}. The set of $\mu$-sectors forms a categorical orthonormal basis for $\CBU$, in that we have
\begin{equation}
    \bbrakket{\mu | \nu } = \begin{cases} \mathbb{C}, & \mu \cong \nu, \\ 0, & \mu \ncong \nu. \end{cases}
\end{equation}
Moreover, we have the following resolution of the identity functor,
\begin{equation}\label{eq:resolution_of_the_identity}
    \mathrm{id}_{\CBU} = \bigoplus_\mu \kket{\mu} \bbra{\mu},
\end{equation}
where the summands $\kket{\mu} \bbra{\mu}$ denote the compositions of the functors $\bbra{\mu} : \CBU \to \Hilb$, given by taking the Hilbert-space valued inner product with $\mu$, with the functors $\kket{\mu} : \Hilb \to \CBU$, mapping $\mathcal{H} \mapsto \mathcal{H} \otimes \mu$.

Acting with the resolution of the identity \eqref{eq:resolution_of_the_identity} on the objects $B \in \CBU$ prepared by manifolds, we obtain a direct-sum decomposition,
\begin{equation}\label{eq:categorical_wavefunction_of_B}
    B = \bigoplus_\mu \kket{\mu} \bbrakket{\mu | B} = \bigoplus_\mu \mathcal{H}_B^\mu \otimes \mu,
\end{equation}
where the \textit{$\mu$-sector Hilbert spaces} $\mathcal{H}_B^\mu$, as studied in \cite{Colafranceschi:2023urj,Marolf:2024adj}, are defined by,
\begin{equation}
    \mathcal{H}_B^\mu \defined \bbrakket{\mu | B}.
\end{equation}
The decomposition \eqref{eq:categorical_wavefunction_of_B} of $B$ into $\mu$-sectors carries over to the von Neumann endomorphism algebras $\mathcal{A}_B$. If we let $\mathcal{A}_B^\mu = \mathcal{B}(\mathcal{H}_B^\mu)$ denote the Type I factor of all bounded operators on $\mathcal{H}_B^\mu$, we obtain the canonical direct-sum decomposition
\begin{equation}
    \mathcal{A}_B = \bigoplus_\mu \mathcal{A}_B^\mu,
\end{equation}
studied in \cite{Colafranceschi:2023urj,Marolf:2024adj}.

Let us comment on the gravitational meaning of the $\mu$-sector Hilbert space $\mathcal{H}_B^\mu$. While the universal Hilbert space $\mathcal{H}_B$ is the Hilbert space of gravitational states whose spatial boundary is $B$, the $\mu$-sector Hilbert space $\mathcal{H}_B^\mu$ is the Hilbert space of ``half-ER bridges,'' namely, ER bridges from $B$ to the irreducible flux sector $\mu$.\footnote{Our ``half-ER bridges'' should be compared to the ``half-wormholes'' considered in \cite{Saad:2021rcu,Saad:2021uzi}.} To make this intuition sharp, we compute
\begin{equation}\label{eq:two_sided_decomp}
    \mathcal{H}_{B \sqcup \overline{A}} = \bbrakket{A | B} = \bigoplus_\mu \bbrakket{A | \mu} \bbrakket{\mu | B} = \bigoplus_\mu \mathcal{H}_{B}^\mu \otimes \mathcal{H}_{\overline{A}}^{\overline{\mu}},
\end{equation}
as shown for the diagonal sector $A = B$ in \cite{Colafranceschi:2023urj} and for the general case in \cite{Marolf:2024adj}. Equation \eqref{eq:two_sided_decomp} shows that every ER bridge from $A$ to $B$ can be decomposed uniquely as a superposition of entangled half-ER bridges, subject to the constraint that the flux $\mu$ on both sides must match, as depicted in Figure \ref{fig:half_ER_bridges}. This constraint is all that obstructs ER = EPR, and so we have proven Theorem \ref{thm:condition_for_ER=EPR}, that ER = EPR holds if and only if $\CBU = \Hilb$.

\begin{figure}
    \centering
    \includegraphics[width=0.9\linewidth]{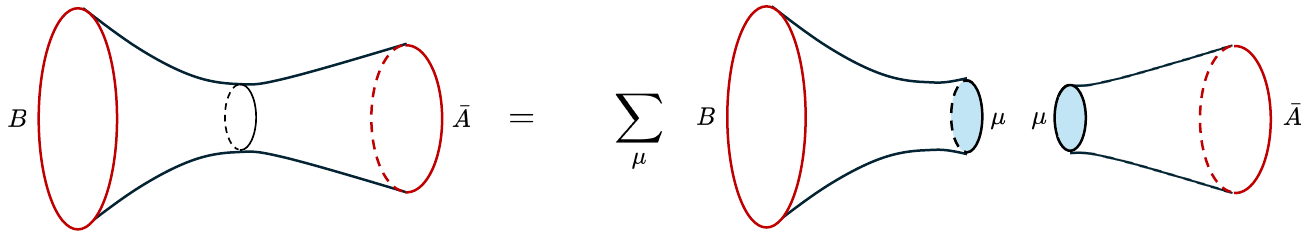}
    \caption{The structure of the baby universe category $\CBU$ implies that any Einstein--Rosen bridge can be decomposed as a sum over half-ER bridges with different irreducible fluxes $\mu$ (as previously derived in \cite{Colafranceschi:2023urj,Marolf:2024adj}). This is the only obstruction to ER = EPR, and so we learn that ER = EPR if and only if $\CBU$ has only one simple object, or more canonically, if $\CBU = \Hilb$.}
    \label{fig:half_ER_bridges}
\end{figure}

Among the $\mu$-sectors, there is one which plays a privileged role, given by the object $\varnothing \in \CBU$ associated to the empty $(d-1)$-dimensional $\mathcal{X}$-manifold, which we call the \textit{Hartle--Hawking object}. Multiplicativity of $\zeta$ implies that $\mathcal{A}_\varnothing = \mathbb{C}$, and so the object $\varnothing \in \CBU$ is automatically simple, corresponding by itself to a $\mu$-sector. By definition, the $\varnothing$-sector Hilbert spaces are simply the universal Hilbert spaces,
\begin{equation}
    \mathcal{H}_B^\varnothing = \bbrakket{\varnothing | B} = \mathcal{H}_B,
\end{equation}
and moreover, the isometric image of $\mathcal{H}_B \otimes \mathcal{H}_{\overline{A}}$ inside $\mathcal{H}_{B \sqcup \overline{A}}$ under the disjoint union map is precisely the $\mu = \varnothing$ summand in \eqref{eq:two_sided_decomp}. For this reason, the privileged $\mu$-sector $\varnothing$ was called $\mu = \otimes$ in \cite{Colafranceschi:2023urj,Marolf:2024adj}. Let us note further that the algebra $\mathcal{A}_\varnothing = \mathbb{C}$ is precisely the baby universe von Neumann algebra $\mathcal{A}_\mathrm{BU}$ discussed in the Introduction, and so the irreducibility of $\varnothing \in \CBU$ is precisely tied to the factorization of the partition function.

We note that each $\mu$-sector can be unambiguously assigned a bordism class,
\begin{equation}
    [\mu] \in \Omega_{d-1}^\mathcal{X},
\end{equation}
defined to be the bordism class of any manifold $B \in \CBU$ which contains $\mu$ as a sub-object. This is unambiguous, as any two manifolds $B_1, B_2$ which both contain $\mu$ as a subobject must admit a non-zero morphism $B_1 \to B_2$ in $\CBU$ (defined, for instance, by projecting $B_1$ onto $\mu$ and then embedding $\mu$ in $B_2$). But this is only possible if $B_1$ and $B_2$ are connected by a bordism in $\BordX$, and so we have $[B_1] = [B_2]$.

As a result, the baby universe category $\CBU$ is canonically graded by the bordism group $\Omega_{d-1}^\mathcal{X}$. This means that the only way ER = EPR could hold, if $\Omega_{d-1}^\mathcal{X} \neq 0$, would be if every manifold $B$ with $[B] \neq 0$ vanished in $\CBU$. This only happens when $\zeta(S^1 \times B) = 0$, as discussed in Section \ref{sec:ER=EPR}, so that an underlying QFT does not admit any states on $B$. An isomorphism $B = 0$ in $\CBU$ is a clear indication that the associated 1-form symmetry of \cite{McNamara:2019rup} is gauged, as $B = 0$ is precisely the consequence of the Gauss law associated to gauging this 1-form symmetry. It would mean that the manifold $B$ is not a consistent background for the bulk theory of quantum gravity.

\subsubsection{Finite-dimensional objects}\label{sec:finite_dimensional}

We have seen that $\CBU$, being a Cauchy complete atomic W*-category, is an infinitary version of a semisimple category. For some purposes, and particularly in order to eventually apply the DR reconstruction theorem, it will be helpful to work with an actual semisimple category. Thus, let us define $\CBUfd$ to be the full sub-category of $\CBU$ containing all finite direct sums of $\mu$-sectors. $\CBUfd$ is semisimple by construction. Importantly, the full baby universe category may be easily recovered from $\CBUfd$ by ``extension of scalars'' from $\Vec$ to $\Hilb$. More formally, $\CBU$ is recovered from $\CBUfd$ by W*-Cauchy completion,
\begin{equation}
    \CBU = \Hilb(\CBUfd),
\end{equation}
which corresponds to allowing infinite direct sums when it previously only had finite ones. Thus, the structure of $\CBU$ is entirely determined by $\CBUfd$, even though the objects $B \in \CBU$ prepared by manifolds will typically not be finite-dimensional (outside of topological theories).

We may obtain an important alternate characterization of $\CBUfd$ through the use of the canonical trace on $\CBU$. For any object $X \in \CBU$, we define its \textit{categorical dimension} by,
\begin{equation}\label{eq:categorical_dimension}
    d_X = \mathrm{dim}(X) \defined \tr(\mathrm{id}_X),
\end{equation}
and we say that $X$ is \textit{finite-dimensional} if it has finite categorical dimension. By positivity of the trace, we have $d_X \geq 0$ for all $X \in \CBU$, and in fact, we even have $d_X \geq 1$ for any nonzero object $X$, by the lower bound \eqref{eq:lower_bound_on_trace} on the trace of any projector. As the trace on $\CBU$ is semifinite, each $\mu$-sector is necessarily finite-dimensional, as otherwise we could find a smaller, finite-dimensional sub-object of $\mu$. As a result, an object is finite dimensional if and only if it is a finite direct sum of $\mu$-sectors, and so $\CBUfd$ is precisely the full subcategory of $\CBU$ on the finite-dimensional objects.

\section{Fusion of manifolds}\label{sec:fusion}

Having constructed the baby universe category $\CBU$ associated to a partition function $\zeta$ (satisfying Axioms \ref{axiom:finiteness}-\ref{axiom:reflection_positivity}) on a unitary bordism category $\BordX$, we now discuss the canonical symmetric monoidal structure on $\CBU$ induced by the disjoint union of manifolds. At this point, our construction genuinely moves beyond \cite{Colafranceschi:2023urj,Marolf:2024adj}, outside of reframing established results. From a concrete perspective, what we do in this section is take the suggested next step of \cite{Colafranceschi:2023urj,Marolf:2024adj} and consider universal Hilbert spaces with three or more boundaries.

\subsection{The symmetric monoidal structure}\label{sec:monoidal_structure}

We want to equip $\CBU$ with a symmetric monoidal structure induced from the disjoint union of manifolds. At the level of the pre-baby universe category $\CBUpre$, the free linearization of $\BordX$, there is no issue: we simply take the free linearization of the symmetric monoidal structure on $\BordX$. In more detail, the symmetric monoidal structure takes two objects $B_1, B_2 \in \CBUpre$ to their disjoint union,
\begin{equation}
    B_1 \sqcup B_2,
\end{equation}
with unit the empty $(d-1)$-manifold $\varnothing$, as in $\BordX$. The disjoint union of morphisms,
\begin{equation}
    \mathcal{O} = \sum_i c_i N_i, \quad \mathcal{O}' = \sum_j c_j' N_j',
\end{equation}
is defined bilinearly,
\begin{equation}\label{eq:fusion_of_bordisms}
    \mathcal{O} \sqcup \mathcal{O}' \defined \sum_{i, j} c_i c_j' (N_i \sqcup N_j').
\end{equation}
We note that the disjoint union of a null morphism in $\CBUpre$ with anything is automatically null, by multiplicativity of $\zeta$.

To extend this symmetric monoidal structure to $\CBU$, it suffices to extend it to the von Neumann completion $\CBUvN$. This is because the operation of W*-Cauchy completion automatically transfers symmetric monoidal structures. We simply define the disjoint union of a formal direct sum of objects by extending the disjoint union bilinearly over $\Hilb$, and we define the disjoint union of formal images of projectors to be the formal image of the disjoint union of the projectors,
\begin{equation}
    \mathrm{im}(P) \sqcup \mathrm{im}(P') \defined \mathrm{im}(P \sqcup P').
\end{equation}
Moreover, all coherence data (associators, symmetry, etc.) automatically transfer as well.

The only analytic issue, then, is defining the disjoint union of the new morphisms in $\CBUvN$ obtained in the process of taking the double commutant in the regular representation,
\begin{equation}
    \CBUpre \times \overline{\CBUpre} \to \Hilb, \quad (B, A) \mapsto \mathcal{H}_{B \sqcup \overline{A}}.
\end{equation}
Let us consider how a disjoint union $N \sqcup N'$ of bordisms $N : B_1 \to B_2$, $N' : B_1' \to B_2'$ acts in the regular representation. By definition, the disjoint union $N \sqcup N'$ acts on states $\ket{M} \in \mathcal{H}_{B_1 \sqcup B_1' \sqcup \overline{A}}$ by gluing the incoming boundary of $N$ to the component $B_1$ of the boundary of $M$ and the incoming boundary of $N'$ to the component $B_1'$ (see Figure \ref{fig:3bdy_wormholes}). Moreover, the two possible orders of gluing commute.

\begin{figure}
    \centering
    \includegraphics[width=0.6\linewidth]{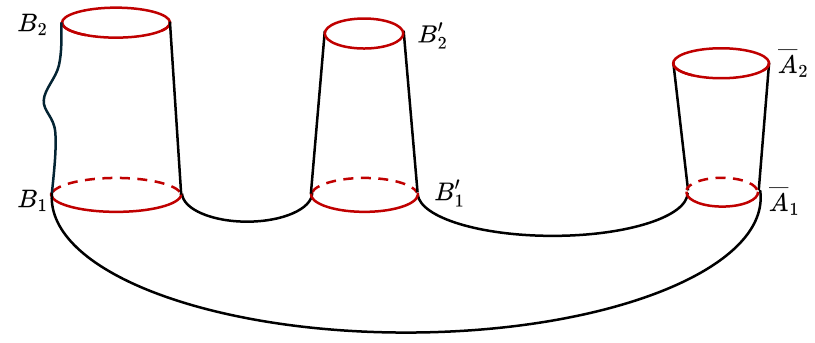}
    \caption{The action of a disjoint union $N \sqcup N'$ of bordisms $N: B_1 \to B_2, N':B_1' \to B_2'$ on a state in the three-boundary Hilbert space $\CH_{B_1\sqcup B_1' \sqcup \overline{A}}$. Each boundary can be acted on by bordisms, and these three actions all commute. In particular, the disjoint union $N \sqcup N'$ lands in the commutant of the action of bordisms on $\overline{A}_1$.}
    \label{fig:3bdy_states}
\end{figure}

As a result, we see that the functor,
\begin{equation}
    \CBUpre \times \CBUpre \times \overline{\CBUpre} \to \Hilb, \quad (B, B', A) \mapsto \mathcal{H}_{B \sqcup B' \sqcup \overline{A}},
\end{equation}
can be viewed as a restriction of the regular representation in multiple compatible ways, depending on how we divide $B \sqcup B' \sqcup \overline{A}$ as a disjoint union of two components. However, by construction, the action of $\CBUpre$ on any Hilbert space in the regular representation extends uniquely to a normal representation of the von Neumann completion $\CBUvN$, so we have a functor
\begin{equation}\label{eq:triple_rep}
    \CBUvN \times \CBUvN \times \overline{\CBUvN} \to \Hilb,
\end{equation}
defined by the same formula. But under \eqref{eq:triple_rep}, the image of $\CBUvN \times \CBUvN$ lands in the commutant $(\overline{\CBUvN})'$ of the action of $\overline{\CBUvN}$ on $A$. This commutant is precisely given by $\CBUvN$ itself by Tomita's commutation theorem, as discussed above. Thus, the functor \eqref{eq:triple_rep} induces a symmetric monoidal fusion $\sqcup : \CBUvN \times \CBUvN \to \CBUvN$ extending the disjoint union on $\CBUpre$.

In fact, there is a natural generalization of the functor \eqref{eq:triple_rep} for any number of boundaries, given by
\begin{equation}
    \CBUvN \times \cdots \times \CBUvN \times \overline{\CBUvN} \to \Hilb, \quad (B_1, \cdots, B_n, A) \mapsto \mathcal{H}_{B_1 \sqcup \cdots \sqcup B_n \sqcup \overline{A}}.
\end{equation}
These functors encode, collectively, all the data of the symmetric monoidal structure on $\CBUvN$. For instance, the four-boundary Hilbert spaces ($n = 3$) encode the associators for $\CBUvN$, and the five-boundary Hilbert spaces ($n = 4$) enforce the pentagon equation.

We note, for later reference, that the symmetric monoidal structure of $\CBU$ respects its grading by the bordism group $\Omega_{d-1}^\mathcal{X}$, as both the fusion of manifolds in $\CBU$ and the addition of cobordism classes are defined by disjoint union of manifolds.

\subsection{Duals and rigidity}\label{sec:duals_and_rigidity}

While the non-unital categories $\BordX$, $\CBUpre$, and even $\mathcal{C}_\mathrm{BU}^\mathrm{HS}$ are non-unital rigid monoidal categories (which we recall means that every object has a dual), the von Neumann completion $\CBUvN$ and the baby universe category $\CBU$ itself are not. This is because, in forming the von Neumann completion, we have introduced weak operator limits of morphisms $\mathcal{O}_i : B_1 \to B_2$ which may not converge when considered as states $\ket{\mathcal{O}_i}$. These morphisms include, in particular, the identity morphisms for any infinite-dimensional object.

However, some objects $X \in \CBU$ do admit duals, which are given by the conjugates $\overline{X}$ when they exist. In fact, an object in $X \in \CBU$ admits a dual if and only if it is finite-dimensional. This is because an object $X \in \CBU$ has a dual if and only if $\mathrm{id}_X$ is Hilbert--Schmidt, if and only if the limit as $\beta \to 0$ of the cylinder states $\ket{C_X(\beta)}$ exists, if and only if we may use that limit to define evaluation and coevaluation maps,
\begin{equation}
    \mathrm{ev}_X : \overline{X} \sqcup X \to \varnothing, \quad \mathrm{coev}_X : \varnothing \to X \sqcup \overline{X},
\end{equation}
forming a standard solution to the conjugate equations.\footnote{We get a standard solution, and thus a balanced unitary dual functor \cite{Penneys2018UnitaryDF}, by our assumption that the unitary structure on $\BordX$ is balanced in the same sense.}

Thus, we see that the full subcategory $\CBUfd \subset \CBU$ on the finite-dimensional objects is also exactly the full subcategory on the dualizable objects. This immediately implies that $\CBUfd$ is closed under disjoint union, as the fusion of dualizable objects in any monoidal category is always dualizable (with dual the fusion of the duals). Thus, $\CBUfd$ becomes a symmetric monoidal category in its own right. In fact, we may see this more directly: the trace on $\CBU$ is monoidal, in that we have,
\begin{equation}
    \tr(\mathcal{O}_1 \sqcup \mathcal{O}_2) = \tr(\mathcal{O}_1) \tr(\mathcal{O}_2).
\end{equation}
Thus, the categorical dimension is multiplicative under disjoint union, and so the disjoint union of finite-dimensional objects must be finite-dimensional.

As a result, $\CBUfd$ is a rigid, symmetric monoidal, idempotent complete C*-category with simple tensor unit, otherwise known as a \textit{unitary tensor category}. Moreover, many of the structures we have discussed reduce to canonical structures on any unitary tensor category, including the canonical balanced unitary dual functor (UDF) $X \mapsto \overline{X}$ \cite{Penneys2018UnitaryDF}, the induced positive trace, and the twist automorphism $\theta_X$.

The description of the monoidal structure on $\CBU$ given above, while geometrically concrete, is very algebraically abstract. We now give a description of the fusion ``in components,'' namely, in terms of the $\mu$-sectors which generate $\CBU$. Consider a triple $(\mu, \nu, \rho)$ of simple objects in $\CBU$. We define the \textit{fusion Hilbert space} to be the inner product,
\begin{equation}
    \mathcal{V}_{\mu \nu}^\rho \defined \bbrakket{\rho | \mu \sqcup \nu},
\end{equation}
whose dimension $N_{\mu \nu}^\rho = \dim(\mathcal{V}_{\mu \nu}^\rho)$ is the \textit{fusion coefficient}. By expanding $\mu \sqcup \nu$ into a complete set of simple objects, we learn that
\begin{equation}
    \mu \sqcup \nu = \bigoplus_\rho \mathcal{V}_{\mu \nu}^\rho \otimes \rho,
\end{equation}
as objects in $\CBU$. By the arguments given above (Section \ref{sec:finite_dimensional}), $\mu \sqcup \nu$ must be finite-dimensional, which implies that $\sum_\rho N_{\mu \nu}^\rho < \infty$. Thus, the fusion rule for $\mu$-sectors is finite: each fusion Hilbert space is finite dimensional, and only finitely many $\rho$ appear in the disjoint union $\mu \sqcup \nu$.

\subsection{Multi-boundary Hilbert spaces}\label{sec:multi_boundary_Hilb_spaces}

We may use the fusion Hilbert spaces to obtain a description of the multi-boundary universal Hilbert spaces, for any number of boundaries. First, we obtain an expression for the $\rho$-sector Hilbert spaces for a disjoint union, $\mathcal{H}_{B \sqcup B'}^\rho$. We compute,
\begin{equation}
    B \sqcup B' = \bigoplus_{\mu, \nu} \mathcal{H}_B^\mu \otimes \mathcal{H}_{B'}^\nu \otimes (\mu \sqcup \nu) = \bigoplus_{\mu, \nu, \rho} \mathcal{H}_B^\mu \otimes \mathcal{H}_{B'}^\nu \otimes \mathcal{V}_{\mu \nu}^\rho \otimes \rho.
\end{equation}
Comparing with the expansion of $B \sqcup B'$ itself into $\rho$-sectors, we obtain a canonical isomorphism,
\begin{equation}
    \mathcal{H}_{B \sqcup B'}^\rho = \bigoplus_{\mu, \nu} \mathcal{H}_B^\mu \otimes \mathcal{H}_{B'}^\nu \otimes \mathcal{V}_{\mu \nu}^\rho.
\end{equation}
We also directly obtain a description of the three-boundary universal Hilbert spaces, by applying \eqref{eq:two_sided_decomp} to the disjoint union $(B \sqcup B') \sqcup \overline{A}$:
\begin{equation}\label{eq:three_boundary_decomposition}
    \mathcal{H}_{B \sqcup B' \sqcup \overline{A}} = \bigoplus_{\mu, \nu, \rho} \mathcal{H}_B^\mu \otimes \mathcal{H}_{B'}^\nu \otimes \mathcal{V}_{\mu \nu}^\rho \otimes \mathcal{H}_{\overline{A}}^{\overline{\rho}},
\end{equation}
as depicted in Figure \ref{fig:3bdy_wormholes}.

\begin{figure}
    \centering
    \includegraphics[width=0.7\linewidth]{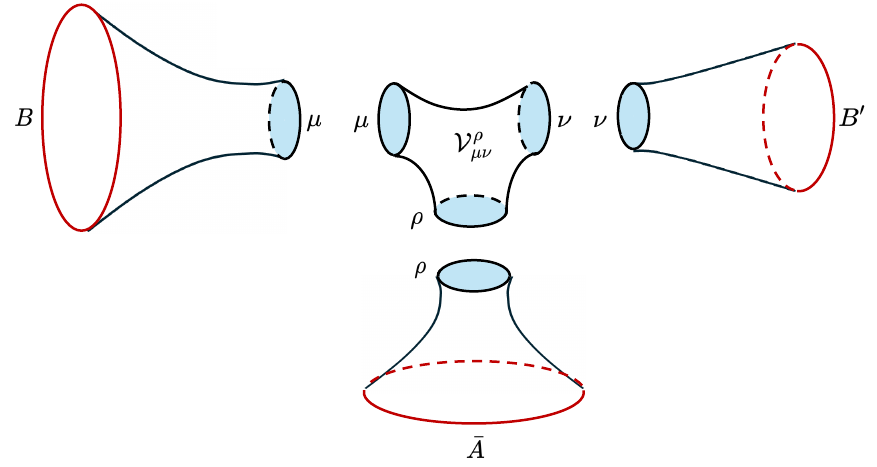}
    \caption{The fusion on $\CBU$ allows us to obtain a description of the three-boundary Hilbert space $\CH_{B \sqcup B' \sqcup \overline{A}}$: its states are obtained by sewing together three half-ER bridges with a state in the fusion Hilbert space $\mathcal{V}_{\mu \nu}^\rho$.}
    \label{fig:3bdy_wormholes}
\end{figure}

The expression \eqref{eq:three_boundary_decomposition} suggests a more concrete way to discover the fusion on $\CBU$. The von Neumann algebra $\mathcal{A}_B \otimes \mathcal{A}_{B'} \otimes \mathcal{A}_{\overline{A}}$ naturally acts on the three-boundary Hilbert space $\mathcal{H}_{B \sqcup B' \sqcup \overline{A}}$. We could then diagonalize this action with respect to the commuting $\mu, \nu$, and $\overline{\rho}$-sector central projectors in $\mathcal{A}_B$, $\mathcal{A}_{B'}$, and $\mathcal{A}_{\overline{A}}$ respectively. We would obtain a block diagonal representation of the type I factors comprising the three von Neumann algebras, and the multiplicity of the representation associated to the triple $(\mu, \nu, \overline{\rho})$ would be precisely the fusion coefficient $N_{\mu \nu}^\rho$. It would then take some work to show that this multiplicity is, in fact, independent of $B, B'$, and $\overline{A}$; one benefit of our formal approach is that it makes this independence manifest. Similarly, our approach makes the associativity of the fusion manifest, while a more concrete approach would involve hard work to show associativity.

Similar logic as discussed above applies to any number of boundaries. In general, we recover the universal Hilbert space for a disjoint union as,
\begin{equation}\label{eq:overlap_with_HH_object}
    \mathcal{H}_{B_1 \sqcup \cdots \sqcup B_n} = \bbrakket{\varnothing | B_1 \sqcup \cdots \sqcup B_n},
\end{equation}
the categorical wavefunction overlap of their fusion with the tensor unit: the Hartle--Hawking object $\varnothing \in \CBU$. Interpreted gravitationally, \eqref{eq:overlap_with_HH_object} describes the space of multi-boundary ER bridges. By expanding \eqref{eq:overlap_with_HH_object} into intermediate channels, we could write multiple different explicit expressions for the multi-boundary Hilbert spaces in terms of the $\mu$-sector Hilbert spaces and the fusion Hilbert spaces. These expressions would all be canonically isomorphic, due to the associators and pentagon equation of $\CBU$.

We note that equation \eqref{eq:overlap_with_HH_object} is analogous to the equation (see, e.g., \cite{Marolf:2020xie}),
\begin{equation}\label{eq:overlap_with_HH_wavefunction}
    \zeta(M_1 \sqcup \cdots \sqcup M_n) = \braket{\varnothing | M_1 \sqcup \cdots \sqcup M_n},
\end{equation}
relating a non-factorizing gravitational partition function with multiple boundary components to the wavefunction overlap with the Hartle--Hawking state $\ket{\varnothing}$ in the baby universe Hilbert space $\HBU$.

\section{Reconstruction}\label{sec:reconstruction}

In this section, we conclude our proof of Theorem \ref{thm:main}, by applying the Doplicher--Roberts (DR) reconstruction theorem \cite{DoplicherRoberts89,muger2007abstract} to the baby universe category.

\subsection{Rigidly generated W*-tensor categories}\label{sec:rg_W*_tensor}

Before we begin introducing the DR reconstruction theorem, we pause to introduce the following notion. We do this because the DR reconstruction theorem, as usually stated, does not directly apply to $\CBU$, as discussed further below.

\begin{definition}\label{defn:rigidly_generated_W*_tensor_cat}
    Let $\mathcal{C}$ be a W*-tensor category, which we recall means a W*-tensor category in the sense of \cite{henriques2024completewcategories} which is additionally Cauchy complete and whose tensor unit is simple. We say that $\mathcal{C}$ is \textit{rigidly generated} if it is generated by some collection of dualizable objects. We say the same for a W*-multitensor category, whose unit may not be simple.
\end{definition}

By \cite{Longo:1997yc}, we do not need to separately assume that a rigidly generated W*-tensor category is atomic, as the irreducibility of the unit implies that every dualizable object is a finite direct sum of simple objects. Conversely, every finite direct sum of simple objects is dualizable, as every simple object is a subobject of a dualizable object by rigid generation. As we have seen, $\CBU$ is rigidly generated, as it is generated by the subcategory $\CBUfd$ of dualizable objects.

In order to understand the general structure of rigidly generated W*-tensor categories, let $\mathrm{UTC}$ denote the $(2,1)$-category of unitary tensor categories, unitary monoidal functors, and unitary monoidal natural transformations, and let $\mathrm{W^*TC}^\mathrm{r.g.}$ denote the $(2,1)$-category of rigidly generated W*-tensor categories, W*-tensor functors, and unitary monoidal natural transformations. Consider the 2-functors,
\begin{equation}\label{eq:adjoint_equiv_tensor_categories}
    \begin{tikzcd}
        \mathrm{UTC} \arrow[rr, bend left, "\Hilb(-)"] && \mathrm{W^*TC}^\mathrm{r.g.} \ar[ll, bend left, "(-)^\mathrm{f.d.}"]
    \end{tikzcd}
\end{equation}
given respectively, by taking the W*-Cauchy completion and restricting to the full sub-category of dualizable objects. The second 2-functor $(-)^\mathrm{f.d.}$ is well defined, because any monoidal functor takes dualizable objects to dualizable objects, and the tensor product of two dualizable objects is always dualizable.

These 2-functors \eqref{eq:adjoint_equiv_tensor_categories} form an adjoint equivalence of $(2,1)$-categories, and so from an external perspective, there is really no difference at all between unitary tensor categories and rigidly generated W*-tensor categories. As a result, we may transfer any categorical properties whatsoever between $\mathrm{UTC}$ and $\mathrm{W^*TC}^\mathrm{r.g.}$, as well as their analogous symmetric monoidal versions. For instance, by equipping every symmetric unitary tensor category with its canonical balanced UDF, we learn that any rigidly generated symmetric W*-tensor category is equipped with a canonical bi-involutive structure, which restricts to this balanced UDF on the subcategory of dualizable objects.

More significantly, we can directly transfer the DR reconstruction theorem, which is usually stated as a theorem about unitary tensor categories, to a result about rigidly generated symmetric W*-tensor categories. We could have skipped this section, and neglected to present Definition \ref{defn:rigidly_generated_W*_tensor_cat}, by applying the DR reconstruction theorem to the unitary tensor category $\CBUpre$, and then repeatedly translating these results by hand to $\CBU$ itself. Indeed, from the perspective of the fusion of $\mu$-sectors, $\CBU$ itself is auxiliary. However, $\CBU$ plays a very important role for us, as it admits a canonical functor $\BordX \to \CBU$, embedding the geometric manifolds of $\BordX$ into the linear category $\CBU$. As discussed above, this map does not typically land in the unitary tensor category $\CBUfd$, since geometric manifolds require infinite-dimensional Hilbert spaces in most geometric QFTs.

A second reason to work with rigidly generated W*-tensor categories, as opposed to unitary tensor categories, is that there is no analogous adjoint equivalence once we drop the simplicity of the tensor unit, hence passing to rigidly generated W*-multitensor categories. This is because the results of \cite{Longo:1997yc} no longer apply in the multitensor setting, and a rigidly generated W*-multitensor category need not be atomic. For instance, consider the W*-multitensor category of measurable fields of Hilbert spaces over the interval $[0, 1]$, equipped with the fiberwise tensor product. The tensor unit is the constant field $\mathbb{C}$, which is a dualizable generating object all on its own. We discuss this generalization in Section \ref{sec:including_alpha_sectors}, in the context of a diffuse space of $\alpha$-sectors.

\subsection{Statistics and fiber functors}\label{sec:statistics_fiber_functors}

We now recall some preliminary notions which enter in the statement of the DR reconstruction theorem. First, given a simple object $\mu$ in a rigidly generated symmetric W*-tensor category $\mathcal{C}$, consider the twist automorphism $\theta_\mu : \mu \to \mu$ defined by the canonical bi-involutive structure. As $(-1)^F$ squares to one, $\theta_\mu$ may only take the values $\pm 1$, since $\mu$ is simple. We say that $\mu$ is a \textit{boson} (respectively a \textit{fermion}) if $\theta_\mu = +1$ (respectively $\theta_\mu = -1$), More generally, an object $X$ is bosonic (respectively fermionic) if it is a direct sum of bosons (respectively fermions) only, and define the \textit{fermion parity} $F_X$ of $X$ to be $(+1)$ (respectively $(-1)$), so that we have $\theta_B = (-1)^{F_X}$ for homogeneous objects $X$. We say that the category $\mathcal{C}$ is \textit{bosonic} if every simple object in $\mathcal{C}$ is a boson (so that $\theta_X = \mathrm{id}_X$ for all $X$), and otherwise that $\mathcal{C}$ is \textit{fermionic}.

The central objects of interest in DR reconstruction, in the bosonic case, are \textit{fiber functors}, namely, symmetric W*-tensor functors
\begin{equation}\label{eq:fiber_functor}
    \mathcal{Z}: \mathcal{C} \to \Hilb,
\end{equation}
from a bosonic symmetric W*-tensor category $\mathcal{C}$ to $\Hilb$. In the fermionic case, we must instead consider \textit{super-fiber functors},
\begin{equation}\label{eq:super_fiber_functor}
    \mathcal{Z} : \mathcal{C} \to \sHilb,
\end{equation}
to the category of super-Hilbert spaces. By the categorical spin-statistics theorem, a super-fiber functor must take bosons to bosons and fermions to fermions, and thus a fiber functor on $\mathcal{C}$ will land in the subcategory $\Hilb \subset \sHilb$ if and only if $\mathcal{C}$ is bosonic. Thus, nothing is lost if we always assume that the target of our fiber functors is $\sHilb$, and we will drop the prefix ``super,'' calling either \eqref{eq:fiber_functor} or \eqref{eq:super_fiber_functor} a fiber functor.

\subsubsection{Fiber functors as QFTs}\label{sec:fiber_functors_on_CBU}

What would a fiber functor on the baby universe category $\CBU$ give us? By pre-composing with the canonical unitary symmetric monoidal functor $\BordX \to \CBU$, we would obtain a unitary symmetric monoidal functor,
\begin{equation}
    \mathcal{Z} : \BordX \to \sHilb,
\end{equation}
or in other words, a unitary functorial QFT! The condition of non-degeneracy for this QFT follows from the non-degeneracy of the functor $\BordX \to \CBU$ itself (Proposition \ref{prop:approx_identities_in_two_sided_reps}). Moreover, the partition function of any QFT produced this way will automatically be $\zeta$, as a closed $\mathcal{X}$-manifold $M$ maps to multiplication by the complex number $\zeta(M)$ already in $\CBU$.

In fact, the converse is true as well: any unitary QFT on $\BordX$ with partition function $\zeta$ automatically extends uniquely to a fiber functor $\mathcal{Z} : \CBU \to \sHilb$. We prove this formally in Appendix \ref{app:univ_property}, but at a heuristic level, note that the action of $\mathcal{Z}$ on manifolds is entirely determined. If $\widetilde{\CH}_B$ is the (super-)Hilbert space of the unitary QFT quantized on $B$, we must define,
\begin{equation}
    \mathcal{Z}(B) = \widetilde{\CH}_B,
\end{equation}
and similarly for morphisms given by bordisms. This means that the unitary QFT automatically lifts to $\CBUpre$, and in fact lifts to $\CBU$ as shown in Appendix \ref{app:univ_property}.

What we prove in Appendix \ref{app:univ_property} is even stronger: the same conclusion holds even if we replace $\sHilb$ with any rigidly generated symmetric W*-tensor category $\mathcal{D}$. More specifically, Proposition \ref{prop:univ_property_CBU} shows that $\CBU$ has the following universal property: any $\mathcal{D}$-valued QFT on $\BordX$ (Definition \ref{defn:D_valued_QFT}), with partition function $\zeta$, admits a unique lift to a symmetric W*-tensor functor $\CBU \to \mathcal{D}$ which makes the following diagram commute,
\begin{equation}
    \begin{tikzcd}
        & \CBU \arrow[dr, dashed, "\exists !"] & \\
        \BordX \arrow[ur] \arrow[rr] & & \mathcal{D}
    \end{tikzcd}
\end{equation}
In other words, $\CBU$ is the universal rigidly generated symmetric W*-tensor category receiving a QFT with partition function $\zeta$ from $\BordX$. This universal property characterizes $\CBU$ uniquely up to unique equivalence.

\subsection{Doplicher--Roberts reconstruction}\label{sec:DR_reconstruction}

We now complete our proof of Theorem \ref{thm:main} by recalling the DR reconstruction theorem, in three parts: Theorems \ref{thm:existence_of_fiber_functors}, \ref{thm:uniqueness_of_fiber_functors}, and \ref{thm:concrete_Tannaka}. We will refer to these as parts (a), (b), and (c), as they mirror parts (a), (b), and (c) of Theorem \ref{thm:main}.

As discussed above, the version of the DR reconstruction theorem proven in \cite{DoplicherRoberts89, muger2007abstract} applies to a symmetric unitary tensor category, such as $\mathcal{C}^\mathrm{f.d.}$. However, we will instead describe DR reconstruction as a theorem about rigidly generated symmetric W*-tensor categories, by transfer along the adjoint equivalence \eqref{eq:adjoint_equiv_tensor_categories}.

\subsubsection{Existence of fiber functors}

The first part of the DR reconstruction theorem states the following:

\begin{theorem}[Doplicher--Roberts, part (a)]\label{thm:existence_of_fiber_functors}
    Let $\mathcal{C}$ be a rigidly generated symmetric W*-tensor category. Then there exists a fiber functor $\mathcal{Z} : \mathcal{C} \to \sHilb$.
\end{theorem}

Applying Theorem \ref{thm:existence_of_fiber_functors} to the rigidly generated baby universe category $\CBU$, we obtain a fiber functor $\CBU \to \sHilb$. Using the equivalence between fiber functors on $\CBU$ and unitary QFTs with partition function $\zeta$ (Proposition \ref{prop:univ_property_CBU}), we prove part (a) of Theorem \ref{thm:main}:

\begin{theorem}[Theorem \ref{thm:main}, part (a)]
    Let $\zeta$ be a partition function on a unitary bordism category $\BordX$ satisfying Axioms \ref{axiom:finiteness}-\ref{axiom:reflection_positivity}. Then there exists a unitary QFT with partition function $\zeta$.
\end{theorem}

\noindent Let us outline the proof of Theorem \ref{thm:existence_of_fiber_functors}, as a good deal of it is physically illuminating. In fact, much of the proof has, remarkably, recently reappeared in the context of the GPI \cite{Marolf:2020xie,Colafranceschi:2023urj,Marolf:2024adj}, at least in the bosonic case. The most important step is to prove that the categorical dimension $d_X$ of any finite-dimensional object $X \in \mathcal{C}$ is a non-negative integer, $d_X \in \mathbb{Z}_{\geq 0}$. This is necessary for any fiber functor to exist: fiber functors preserve categorical dimensions, and the dimensions of finite-dimensional objects in $\sHilb$ are non-negative integers.\footnote{Note that the categorical dimensions in $\sHilb$ are given by $\dim(\mathbb{C}^{p | q}) = p + q$, as we have chosen the canonical positive unitary structure on $\sHilb$.}

The standard proof that $d_X \in \mathbb{Z}_{\geq 0}$ is a simple generalization of the anti-symmetrization argument we saw in Section \ref{sec:Ex_TopoQM}. First, we twist the symmetry $s_{X, Y}$ on the possibly-fermionic category $\mathcal{C}$ by,
 \begin{equation}\label{eq:bosonization_symmetry}
     s_{X_\mathrm{b}, Y_\mathrm{b}} \defined (-1)^{F_X F_Y} s_{X, Y},
 \end{equation}
in order to obtain a bosonic rigidly-generated symmetric W*-tensor category $\mathcal{C}_\mathrm{b}$, with objects $X_\mathrm{b}$ for every object $X \in \mathcal{C}$. Physically, the object $X_\mathrm{b} \in \mathcal{C}_\mathrm{b}$ corresponding to a fermionic object $X \in \mathcal{C}$ is obtained by dressing $X$ with a transparent fermion, while the bosonic objects are left untouched. This leaves the categorical dimensions of all objects unchanged, as the categorical dimensions do not depend on the choice of symmetry.

We then consider the dimension of the antisymmetric tensor powers $\Lambda^n X_\mathrm{b}$ of the finite-dimensional object $X_\mathrm{b} \in \mathcal{C}_\mathrm{b}$.\footnote{Because of the twist, this corresponds to taking antisymmetric tensor powers of the bosonic part of $X$ and symmetric tensor powers of the fermionic part of $X$ in the original category $\mathcal{C}$.} A simple calculation shows
\begin{equation}
    \dim(\Lambda^n X_\mathrm{b}) = \binom{d_X}{n} \defined \frac{d_X (d_X - 1) \cdots (d_X - n + 1)}{n!},
\end{equation}
which is eventually negative, violating positivity, unless $d_X \in \mathbb{Z}_{\geq 0}$. For the $\mu$-sectors in $\CBU$, this argument, and the non-negative integers $d_\mu$, previously appeared in \cite{Colafranceschi:2023urj,Marolf:2024adj} (where our $d_\mu$ was called $n_\mu$).

Given that the categorical dimensions $d_X$ of all finite-dimensional objects $X \in \mathcal{C}$ are non-negative integers, the action of a fiber functor $\mathcal{Z} : \mathcal{C} \to \sHilb$ is fixed, at least up to unitary isomorphism. First, the action of $\mathcal{Z}$ on simple objects $\mu \in \mathcal{C}$ must be
\begin{equation}\label{eq:candidate_fiber_functor}
    \mathcal{Z}(\mu) = \mathbb{C}^{d_\mu},
\end{equation}
with even grading when $\mu$ is a boson and odd grading when $\mu$ is a fermion. That immediately implies that any fiber functor $\mathcal{Z} : \CBU \to \sHilb$ is given on manifold objects $B \in \CBU$ by,
\begin{equation}\label{eq:extended_Hilbert_space}
    \mathcal{Z}(B) = \bigoplus_\mu \mathcal{H}_B^\mu \otimes \mathbb{C}^{d_\mu},
\end{equation}
graded according to the fermion-parity of the $\mu$-sectors. The Hilbert spaces \eqref{eq:extended_Hilbert_space} were discussed in detail in \cite{Colafranceschi:2023urj,Marolf:2024adj}, where they were called the \textit{extended Hilbert spaces}. We will use $\widetilde{\CH}_B$ to denote the Hilbert space \eqref{eq:extended_Hilbert_space}; up to unitary isomorphism, it is necessarily the Hilbert space of any unitary QFT with partition function $\zeta$. Moreover, as a fiber functor takes traces to traces, the trace on $\CBU$ defined by $\zeta$ is identified with the actual Hilbert-space trace on $\widetilde{\CH}_B$; this was one of the central points of \cite{Colafranceschi:2023urj,Marolf:2024adj}.

The action of $\mathcal{Z}$ on morphisms is also uniquely fixed, as functors between atomic Cauchy complete W*-categories are determined up to unitary isomorphism by their action on simple objects. For instance, for $\CBU$, morphisms must act block-diagonally in the expansion \eqref{eq:extended_Hilbert_space} into $\mu$-sectors, with the action on the $\mu$-sector Hilbert space $\mathcal{H}_B^\mu$ induced from the action on the regular representation of $\CBU$.

We comment briefly on the gravitational interpretation of \eqref{eq:extended_Hilbert_space}. If the object $\mu \in \CBU$ represents the endpoint of a half-ER bridge, then $d_\mu$ is the number of additional black hole microstates needed in order to turn this half-ER bridge into an actual one-sided black hole state. Under the assignment $\mu \mapsto [\mu] \in \Omega_{d-1}^\mathcal{X}$ of a bordism class to each $\mu$-sector, we see that each of these additional $d_\mu$ black hole microstates is a cobordism defect, in the sense of \cite{McNamara:2019rup}: an additional degree of freedom in the UV completion of an incomplete theory of quantum gravity needed in order to trivialize the bordism class $[\mu]$.

The remainder of the proof of Theorem \ref{thm:existence_of_fiber_functors} boils down to showing that the functor $F$ defined above is symmetric monoidal. In doing so, one is forced to make arbitrary choices, which correspond to the choice of a coherent family of unitary isomorphisms,
\begin{equation}\label{eq:fiber_functor_data}
\phi_{\mu \nu}^\mathcal{Z} : \mathbb{C}^{d_\mu} \otimes \mathbb{C}^{d_\nu} \to \bigoplus_{\rho} \mathcal{V}_{\mu \nu}^\rho \otimes \mathbb{C}^{d_\rho},
\end{equation}
called a \textit{tensorator}, $\phi_{\mu \nu}^\mathcal{Z} : \mathcal{Z}(\mu) \otimes \mathcal{Z}(\nu) \to \mathcal{Z}(\mu \otimes \nu)$, which must satisfy a list of coherence axioms. Unitary isomorphisms \eqref{eq:fiber_functor_data}, with no constraints, clearly always exist, because the categorical dimensions are multiplicative, $d_\mu d_\nu = \sum_\rho N_{\mu \nu}^\rho d_\rho$. Showing that solutions to the coherence axioms exist involves heavy-duty categorical algebra; for the details, we refer the reader to \cite{muger2007abstract}.

\subsubsection{Uniqueness of fiber functors}

We have seen that, by Theorem \ref{thm:existence_of_fiber_functors}, any rigidly generated symmetric W*-tensor category admits a fiber functor. The second part of the DR reconstruction theorem states the following:

\begin{theorem}[Doplicher--Roberts, part (b)]\label{thm:uniqueness_of_fiber_functors}
    Let $\mathcal{C}$ be a rigidly generated symmetric W*-tensor category. Then any two unitary fiber functors $\mathcal{Z}_1, \mathcal{Z}_2 : \mathcal{C} \to \sHilb$ are unitarily isomorphic.
\end{theorem}

Applying Theorem \ref{thm:uniqueness_of_fiber_functors} to $\CBU$, we prove part (b) of Theorem \ref{thm:main}, stating that unitary QFTs are uniquely determined, up to unitary isomorphism, by their closed manifold partition functions:

\begin{theorem}[Theorem \ref{thm:main}, part (b)]
    Any two unitary QFTs $\BordX \to \sHilb$ which share the same partition function $\zeta$ on all closed $\mathcal{X}$-manifolds are unitarily isomorphic.
\end{theorem}

\noindent More concretely, Theorem \ref{thm:uniqueness_of_fiber_functors} asserts that any two choices $\phi_{\mu, \nu}^{\mathcal{Z}_1}, \phi_{\mu, \nu}^{\mathcal{Z}_2}$ of tensorators \eqref{eq:fiber_functor_data}, defining two fiber functors, are related by some family of unitaries,
\begin{equation}\label{eq:unitary_natural_monoidal}
    g_\mu : \mathbb{C}^{d_\mu} \to \mathbb{C}^{d_\mu}, \quad \left(\bigoplus\nolimits_\rho \mathrm{id}_{\mathcal{V}_{\mu \nu}^\rho} \otimes g_\rho \right) \circ \phi_{\mu, \nu}^{\mathcal{Z}_1} = \phi_{\mu \nu}^{\mathcal{Z}_2} \circ (g_\mu \otimes g_\nu),
\end{equation}
comprising a unitary monoidal natural isomorphism $g : \mathcal{Z}_1 \Rightarrow \mathcal{Z}_2$. As with finding tensorators $\phi_{\mu \nu}^\mathcal{Z}$, finding a collection of unitaries $g_\mu$ satisfying the coherence axioms involves heavy-duty categorical algebra, for which we again refer the reader to \cite{muger2007abstract} for details.

\subsubsection{The symmetry group}
\label{sec:symmetry_group}

The final part of the DR reconstruction theorem concerns, not the uniqueness of the fiber functor $\mathcal{Z}$ itself, but instead the uniqueness of the unitary isomorphism $g : \mathcal{Z}_1 \Rightarrow \mathcal{Z}_2$ between any two fiber functors. To discuss this uniqueness, let us consider the groupoid $\mathrm{Spec}(\mathcal{C})$, whose objects are fiber functors $\mathcal{C} \to \sHilb$, and whose morphisms are unitary monoidal natural isomorphisms. We use the notation $\mathrm{Spec}(\mathcal{C})$ for this groupoid, as it is the categorical analog of the spectrum of a commutative C*-algebra (see Section \ref{sec:higher_averaging} and \cite{Johnson-Freyd:2015fua} for more on this perspective).

Theorem \ref{thm:existence_of_fiber_functors} states that $\mathrm{Spec}(\mathcal{C})$ is nonempty, and Theorem \ref{thm:uniqueness_of_fiber_functors} states that $\mathrm{Spec}(\mathcal{C})$ is connected. Thus, $\mathrm{Spec}(\mathcal{C})$ is characterized (up to equivalence) by its fundamental group at some basepoint $\mathcal{Z} \in \mathrm{Spec}(\mathcal{C})$, defined to be the group
\begin{equation}
    G \defined \mathrm{Aut}^{\otimes, \dagger}(\mathcal{Z}),
\end{equation}
of unitary monoidal natural automorphisms $g : \mathcal{Z} \Rightarrow \mathcal{Z}$ of the chosen fiber functor. Expanding $g$ in components as in \eqref{eq:unitary_natural_monoidal} defines an embedding $G \subset \prod\nolimits_\mu U(d_\mu)$ as a closed subgroup of a product of unitary groups, and thus $G$ is a compact Hausdorff group.\footnote{$G$ is a Lie group if and only if $\mathcal{C}$ admits a finite-dimensional tensor generator.} If $G$ is nontrivial, then while any two fiber functors are unitarily isomorphic, they will not be \textit{uniquely} unitarily isomorphic, which is why we still speak of ``choosing'' a fiber functor.

A useful point of comparison is the example of topological quantum mechanics, as considered in Section \ref{sec:Ex_TopoQM}, with $\zeta(S^1) = k$, whose Hilbert space on a point is some $k$-dimensional Hilbert space. Of course, every $k$-dimensional Hilbert space is unitarily equivalent to $\mathbb{C}^k$. However, there is no canonical choice for this unitary equivalence, as an abstract $k$-dimensional Hilbert space does not come with a canonical orthonormal basis. This fact is absolutely essential to many quantum-mechanical phenomena, including but not limited to: degenerate perturbation theory, Berry connections, anomalies, and the very possibility of Yang--Mills fields. As we will see, the $U(k)$ ambiguity of choosing a $k$-dimensional Hilbert space is precisely the ambiguity in choosing a fiber functor for the baby universe category $\CBU$ associated to this example.

In general, the automorphism group $G$ of a fiber functor $\mathcal{Z} : \CBU \to \sHilb$ has a natural interpretation in terms of the associated unitary QFT: it is merely a faithfully-acting unitary symmetry group, in the usual sense. The component of a unitary monoidal natural transformation $g : \mathcal{Z} \Rightarrow \mathcal{Z}$ at an object $B \in \CBU$ is a unitary operator,
\begin{equation}\label{eq:symmetry_operator}
    \widehat{U}_B(g) : \widetilde{\mathcal{H}}_B \to \widetilde{\mathcal{H}}_B,
\end{equation}
and we have $\widehat{U}_B(g_1) \widehat{U}_B(g_2) = \widehat{U}_B(g_1 g_2)$. Moreover, the operators \eqref{eq:symmetry_operator} tensor factorize on disjoint unions, $\widehat{U}_{B_1 \sqcup B_2}(g) = \widehat{U}_{B_1}(g) \otimes \widehat{U}_{B_2}(g)$, and commute with the Euclidean evolution operator along any bordism, as they are components of a natural transformation. In particular, they commute with the evolution operators $e^{- \beta \widehat{H}_B}$ associated to Euclidean cylinders, which implies that they commute with the Hamiltonian operators $\widehat{H}_B$.

By the same argument, $G$ is a symmetry under which all source fields in the class $\mathcal{X}$ are neutral. In fact, it is precisely the maximal such symmetry, as any such symmetry will define a unitary monoidal natural automorphism of $\mathcal{Z}$.\footnote{Non-invertible symmetries will \textit{not} define monoidal natural endomorphisms, as they cannot be pushed past an arbitrary bordism. See \cite{McNamara:2021cuo,Heckman:2024obe} for more discussion, and see \cite{muger2007abstract} for a proof that a unitary fiber functor does not even admit any non-invertible natural endomorphisms.} Finally, we note that under the decomposition \eqref{eq:extended_Hilbert_space} of $\widetilde{\mathcal{H}}_B$ into $\mu$-sectors, the symmetry operator $\widehat{U}_B(g)$ acts only on the tensor factors $\mathbb{C}^{d_\mu}$, and not at all on the $\mu$-sector Hilbert spaces $\mathcal{H}_B^\mu$, again by naturality.

As we are considering fiber functors valued in $\sHilb$, the group
\begin{equation}
    \mathbb{Z}_2^F = \mathrm{Aut}^{\otimes, \dagger}(\mathrm{id}_\sHilb),
\end{equation}
generated by the fermion parity automorphism $(-1)^F$ always acts on $\mathcal{Z}$ by unitary monoidal natural transformations. Thus, we obtain a canonical map $\mathbb{Z}_2^F \to G$. The element $(-1)^F \in G$ in the image of this map is always central, and so $G$ is what the literature on Tannakian reconstruction calls a \textit{super-group}: a group $G$ with a choice of central involution $(-1)^F$. By the spin-statistics theorem, the element $(-1)^F \in G$ is nontrivial when $\mathcal{C}$ is fermionic, and is trivial when $\mathcal{C}$ is bosonic.

Physically, for fiber functors on $\CBU$, the symmetry of the associated QFT is what is known as a \textit{Spin-$G$ symmetry}, meaning that the faithfully acting symmetry is given by the diagonal quotient,
\begin{equation}\label{eq:spin_G}
    \frac{{Spin}(d) \times G}{\mathbb{Z}_2^F},
\end{equation}
of the product of the Euclidean spin rotation group and $G$, by the spin-statistics theorem. For example, a Spin$^c$ symmetry is a Spin-$G$ symmetry for $G = U(1)_{{spin}^c}$, namely, the group $U(1)$ with $(-1)^F$ identified with $-1 \in U(1)$. For bosonic QFTs, for which $\CBU$ is bosonic, the quotient \eqref{eq:spin_G} reduces to,
\begin{equation}
    {SO}(d) \times G,
\end{equation}
as $(-1)^F \in G$ is trivial. The fact that $\mathbb{Z}_2^F$ always appears as a nontrivial symmetry of any fermionic QFT is due to our restriction forbidding the class $\mathcal{X}$ from including any fermionic sources, as discussed in Section \ref{sec:germs}.

By the observations discussed above, a fiber functor $\mathcal{Z} : \mathcal{C} \to \sHilb$ canonically lifts to a symmetric W*-tensor functor,
\begin{equation}\label{eq:lift_to_sRep}
    \mathcal{Z}_G : \mathcal{C} \to \mathrm{sRep}(G),
\end{equation}
where $\mathrm{sRep}(G)$ denotes the rigidly generated symmetric W*-tensor category of unitary \textit{super-representations} of $G$, namely, unitary representations of $G$ on super-Hilbert spaces such that $(-1)^F \in G$ acts by the fermion-parity automorphism. For instance, $\mathrm{sRep}(\mathbb{Z}_2^F)$ is just $\sHilb$ itself, and $\mathrm{sRep}(U(1)_{\text{spin}^c})$ is the category of representations of $U(1)$ on super-Hilbert spaces that satisfy the spin-charge relation. If $\mathcal{C}$ is bosonic, then $(-1)^F \in G$ is trivial, and we have a canonical equivalence $\mathrm{sRep}(G) = \mathrm{Rep}(G)$ of symmetric W*-tensor categories.

We may now state the final part of the DR reconstruction theorem, which, with a fiber functor given, is the concrete reconstruction theorem of Tannaka \cite{tannaka1939dualitatssatz} and Krein \cite{Krein1949Duality}.

\begin{theorem}[Doplicher--Roberts, part (c)]\label{thm:concrete_Tannaka}
    Let $\mathcal{C}$ be a rigidly generated symmetric W*-tensor category, and let $\mathcal{Z} : \mathcal{C} \to \sHilb$ be a fiber functor (whose existence is guaranteed by Theorem \ref{thm:existence_of_fiber_functors}). Then the canonical lift $\mathcal{Z}_G : \mathcal{C} \to \mathrm{sRep}(G)$ described in \eqref{eq:lift_to_sRep} is an equivalence of symmetric W*-tensor categories.
\end{theorem}

Applied to a fiber functor on $\CBU$, we obtain a complete characterization of the $\mu$-sectors: they are precisely the irreducible unitary super-representations of a symmetry $G$ of the associated unitary QFT. As the $\mu$-sectors characterize the breakdown of ER = EPR, we learn that the only obstruction to ER = EPR, in any theory of quantum gravity, is the incompleteness of the spectrum of charged states under a symmetry group $G$. Moreover, the identification of $\CBU$ with $\mathrm{sRep}(G)$ allows us to more explicitly identify the multi-boundary universal Hilbert spaces. Finally, by computing the right-hand side of \eqref{eq:overlap_with_HH_object} in $\mathrm{sRep}(G)$, we prove part (c) of Theorem \ref{thm:main}.

\begin{theorem}[Theorem \ref{thm:main}, part (c)]
Fix a unitary QFT on $\BordX$, with partition function $\zeta$, and with Hilbert spaces $\widetilde{\CH}_B$. Then the universally constructed Hilbert spaces associated to $(\mathcal{X}, \zeta)$ are given by the invariant sector,
\begin{equation}\label{eq:quenched_gauging_proven}
    \mathcal{H}_{B_1 \sqcup \cdots \sqcup B_n} = \left( \widetilde{\CH}_{B_1} \otimes \cdots \otimes \widetilde{\CH}_{B_n} \right)^G,
\end{equation}
under the maximal symmetry group $G$ under which all sources in $\mathcal{X}$ are neutral.
\end{theorem}

As a result, we may reinterpret any breakdown of ER = EPR as the result of projecting onto the neutral sector of an actual unitary QFT.

\subsection{Higher ensemble averaging}\label{sec:higher_averaging}

Let us now deliver on our promise to explain the precise sense in which projecting onto $G$-invariants is, literally, an ensemble average. As motivation, let us first return to a non-factorizing partition function (thus dropping Axiom \ref{axiom:multiplicativity}). In such a case, we will have an ensemble of $\alpha$ sectors, and in each $\alpha$-sector, an underlying partition function $\zeta_\alpha$. The original partition function $\zeta$ is recovered as the average,
\begin{equation}
    \zeta(M) = \int d\alpha\ \zeta_\alpha(M),
\end{equation}
as described in the Introduction.

However, as discussed by Marolf and Maxfield \cite{Marolf:2020xie}, the universal Hilbert spaces $\CH_B$ will also, themselves, be given by an ensemble average over $\alpha$. More precisely, the universal Hilbert spaces will be given by,
\begin{equation}
    \CH_B = \int^\oplus d\alpha\ \CH_B^\alpha,
\end{equation}
where $\CH_B^\alpha$ is the universal Hilbert space constructed from $\zeta_\alpha$, and we recall that the direct integral denotes the Hilbert space of square integrable global sections of some field of Hilbert spaces. In particular, we have,
\begin{equation}\label{eq:HBU_direct_integral}
    \HBU = \int^\oplus d\alpha\ \CH_\mathrm{BU}^\alpha = \int^\oplus d\alpha\ \mathbb{C},
\end{equation}
using that $\CH_\mathrm{BU}^\alpha = \mathbb{C}$ for each $\alpha$-sector.
As a result, we see that the correct notion of ``ensemble average'' for a field of Hilbert spaces is just their direct integral.

Now, we return to our main case of a factorizing partition function $\zeta$ (thus restoring Axiom \ref{axiom:multiplicativity}). We have seen that,
\begin{equation}
    \CBU \cong \mathrm{Rep}(G), \quad \mathcal{H}_B = \widetilde{\CH}_B^G,
\end{equation}
for some compact Hausdorff group $G$ (and we restrict to the bosonic case for simplicity). We recall a standard fact in mathematics: for any sort of mathematical object $X$, an action of a group $G$ on $X$ is the same thing as a family,
\begin{equation}
    \begin{tikzcd}
        X \arrow[r] & X_G \arrow[d] \\
        & BG
    \end{tikzcd}
\end{equation}
of objects over the stack (or groupoid) $BG$ with fiber $X$. Here, $X_G$ denotes the homotopy quotient of $X$ by the $G$ action.

To see why this is the case, recall that the stack $BG$ consists of a single point with automorphism group $G$. To specify a family over $BG$, one simply specifies the fiber $X$ over the single point, and then specifies transition data encoding how $X$ comes back to itself after traveling along an automorphism. This transition data, consisting of a compatible family of automorphisms of $X$ for each $g \in G$, is the same thing as an action of $G$ on $X$. Further, in complete generality, the space of global sections over $BG$ is the space $X^G$ of homotopy fixed points. This is because a global section is a choice of element $x \in X$ (the value of the section at the single point) which glues to itself under the transition functions, meaning $g \cdot x = x$ for all $g \in G$.

As a result, a Hilbert space $\mathcal{H}$ with a unitary $G$-representation is simply a field of Hilbert spaces over $BG$. Moreover, the direct integral,
\begin{equation}\label{eq:invts_is_int_over_BG}
    \int_{BG}^\oplus \mathcal{H} = \mathcal{H}^G,
\end{equation}
defined to be the space of global sections, is simply the $G$-invariant subspace. Thus, we may rewrite \eqref{eq:quenched_gauging_proven} as a single direct integral,
\begin{equation}\label{eq:quenched_gauging_as_direct_integral}
    \mathcal{H}_{B_1 \sqcup \cdots \sqcup B_n} = \int^\oplus_{BG} \widetilde{\CH}_{B_1} \otimes \cdots \otimes \widetilde{\CH}_{B_n}.
\end{equation}
Moreover, the baby universe category itself can be written as,
\begin{equation}\label{eq:CBU_ensemble_avg}
    \CBU \cong \mathrm{Rep}(G) = \Hilb(BG) = \int_{BG}^\oplus \Hilb,
\end{equation}
analogously to \eqref{eq:HBU_direct_integral}, where here $\Hilb(BG)$ denotes the W*-category of measurable fields of Hilbert spaces over $BG$, or in other words, global sections of the trivial $\Hilb$ bundle of categories.

The expression \eqref{eq:invts_is_int_over_BG} has yet another interpretation in terms of higher algebraic geometry (or, really, higher spectral theory). First, we note that for the case of a non-factorizing partition function, we may evaluate the direct integral \eqref{eq:HBU_direct_integral} as,
\begin{equation}
    \HBU = \int^\oplus d\alpha\ \mathbb{C} = L^2\big(\mathrm{Spec}(\mathcal{A}_\mathrm{BU})\big),
\end{equation}
identifying $\mathrm{Spec}(\mathcal{A}_\mathrm{BU})$ with the space of $\alpha$-sectors. Recall from \ref{sec:DR_reconstruction} that $BG$ is, not quite canonically, simply the stack $\mathrm{Spec}(\CBU)$ of fiber functors on $\CBU$. Thus, we obtain the entirely canonical expression,
\begin{equation}
    \CBU = \int^\oplus_{\mathrm{Spec}(\CBU)} \Hilb,
\end{equation}
expressing the rigidly generated symmetric W*-tensor category $\CBU$ as the categorical ring of functions on its spectrum. We may also rephrase \eqref{eq:quenched_gauging_as_direct_integral} as a direct integral over $\mathrm{Spec}(\CBU)$.

Thus, our main result states: the universal Hilbert spaces $\mathcal{H}_B$ produced by the gravitational path integral are the ensemble average (meaning, direct integral) of a measurable field of Hilbert spaces $\widetilde{\CH}_B$ over the stack of fiber functors on $\CBU$. This phrasing is completely canonical, as it does not require us to non-canonically choose a fiber functor at any point. As we discuss in Section \ref{sec:including_alpha_sectors}, this phrasing also applies, verbatim, when we drop multiplicativity of $\zeta$ and consider both partition function and Hilbert-space factorization at once.

\subsection{Reconstruction in examples}

We now illustrate our entire construction in three examples, where we can rely on known constructions in the literature to directly compute the entire baby universe category at once.

\subsubsection{Reconstructing topological quantum mechanics}

For our first example, we return to the theory of topological quantum mechanics (TQM) on $\Bord_1^{SO}$ considered in Section \ref{sec:Ex_TopoQM}, based on \cite{Maxfield:2023mdj}, defined by the non-negative integer $\zeta(S^1) = k$. Recall that the objects of $\Bord_1^\mathrm{SO}$ are given by pairs $(n_+, n_-)$ of non-negative integers, where $n_\pm$ denotes the number of positively or negatively oriented points, respectively. We saw that there was a constraint, arising from the bordism group $\Omega_1^{SO}$, that a bordism $(n_+, n_-) \to (m_+, m_-)$ only existed if $n_+ + m_- = m_+ + n_-$. However, when this constraint is satisfied, any bordism is equivalent, modulo disconnected components, to some linear combination of bordisms specified by a permutation in $S_{n_+ + m_-}$ specifying how the points are connected.

By directly computing the null states, \cite{Maxfield:2023mdj} showed that the null states have a description where the null states correspond precisely to the linear combinations of permutations labeled by Young diagrams with more than $k$ rows. Thus, we learn that, when $n_+ + m_- = m_+ + n_-$, we have,
\begin{equation}\label{eq:hom_in_CBU_TQM}
    \CBU\big( (n_+, n_-) \to (m_+, m_-)\big) = \mathbb{C}[S_{n_+ + m_-}]/\{\text{rows}(D) > k\},
\end{equation}
which, by Schur--Weyl duality, is precisely the same as the space of $U(k)$-intertwiners,
\begin{equation}
    \Box^{\otimes n_+} \otimes \overline{\Box}^{\otimes n_-} \to \Box^{\otimes m_+} \otimes \overline{\Box}^{\otimes m_-}.
\end{equation}
As a result, we see that the symmetric W*-tensor functor $\CBU \to \mathrm{Rep}(U(k))$ mapping $\mathrm{pt}^+ \mapsto \Box$ is an equivalence of categories, as expected from the fact that the universal construction produced all the $U(k)$-invariant states in the underlying TQM.

Up to translation from algebraic representation theory to unitary representation theory, the description \eqref{eq:hom_in_CBU_TQM} of the morphism spaces in $\CBU$ precisely matches the morphisms in the formal category $\mathrm{Rep}(GL_t)$ defined by Deligne and Milne in \cite{deligne1982tannakian}, with $k$ playing the role of the formal variable $t$. The null states discussed in \cite{Maxfield:2023mdj} are in a precise correspondence with the Schur functors \cite{deligne2002categories} which annihilate $\Box \in \mathrm{Rep}(U(k))$. This connection, between Maxfield's example \cite{Maxfield:2023mdj} and the Tannakian formalism of Deligne and Milne \cite{deligne1982tannakian,deligne1990categories,deligne2002categories}, was a large part of the formal inspiration for this paper.

\subsubsection{Reconstructing a fermion}\label{sec:reconstructing_a_fermion}

Our second example will continue from Section \ref{sec:ex_fTQM}, where we discussed the fermionic TQM on $\Bord_1^{Spin}$ with a single fermionic state, defined by,
\begin{equation}
    \zeta(S^1_\mathrm{ap}) = +1, \quad \zeta(S^1_\mathrm{p}) = -1.
\end{equation}
Let us first note that, in $\CBU$, the semicircle bordism,
\begin{equation}
   \includegraphics[height=1cm,valign=c]{figures/TQM_pi.pdf} : \mathrm{pt}^+ \sqcup \mathrm{pt}^- \to \varnothing,
\end{equation}
is actually a unitary isomorphism, due to the one-dimensionality of the four-boundary Hilbert space. Thus, $\mathrm{pt}^+$ is invertible, with inverse $\mathrm{pt}^-$.

As a consequence, we do not need to label objects in $\CBU$ produced by manifolds by a pair $(n_+, n_-)$ of integers, counting the positively and negatively oriented points, and can instead label them by the net number $n = n_+ - n_-$ of points. Moreover, each object $n \in \CBU$ is simple, since they are all invertible. As a result, we have an equivalence,
\begin{equation}\label{eq:rep_U(1)}
    \CBU \simeq \mathrm{Rep}(U(1)),
\end{equation}
at least as a monoidal category, where $n$ is mapped to the representation of charge $n$. Without considering statistics, this is simply the $k = 1$ case of the previous example.

However, \eqref{eq:rep_U(1)} is not correct at the level of \textit{symmetric} monoidal categories. This is because, as seen in Section \ref{sec:ex_fTQM}, the swap bordism,
\begin{equation}
    s_{\mathrm{pt}^+, \mathrm{pt}^+} : \mathrm{pt}^+ \sqcup \mathrm{pt}^+ \to \mathrm{pt}^+ \sqcup \mathrm{pt}^+,
\end{equation}
is equivalent, modulo null states, to $- \mathrm{id}_{\mathrm{pt}^+ \sqcup \mathrm{pt}^+}$. Thus, in $\CBU$, the object $\mathrm{pt}^+$ has fermionic self-statistics. We may also see this more directly: consider the $2 \pi$-twist automorphism,
\begin{equation}
    \theta_{\mathrm{pt}^+} : \mathrm{pt}^+ \to \mathrm{pt}^+,
\end{equation}
which is realized in this topological bordism category by the nontrivial spin interval. We compute,
\begin{equation}
    \braket{\theta_{\mathrm{pt}^+} | \mathrm{id}_{\mathrm{pt}^+}} = \zeta(S^1_\mathrm{p}) = -1,
\end{equation}
and so, modulo null states, we have $\theta_{\mathrm{pt}^+} = - \mathrm{id}_{\mathrm{pt}^+}$, and so $\mathrm{pt}^+$ is a fermion.

Taking this into account, we see that, as a symmetric W*-tensor category, we have,
\begin{equation}
    \CBU \simeq \mathrm{sRep}(U(1)_{\text{spin}^c}),
\end{equation}
the category of representations of $U(1)$ on super-Hilbert spaces which satisfy the spin-charge relation.

\subsubsection{Reconstructing a two-dimensional TQFT}

Our final example will be in $d = 2$, for a partition function $\zeta$ defined on the bordism category $\Bord_2^{SO}$ of oriented surfaces. Two-dimensional unitary oriented TQFTs are extremely simple, and are all of the following form: the TQFT decomposes as a direct sum over a finite set of $k \in \mathbb{Z}_{\geq 0}$ vacua. The theory in each vacuum is invertible, and uniquely characterized by a positive number $\lambda_i > 0$ setting the Euler counterterm. The partition function on a connected surface $\Sigma_g$ of genus $g$ is given by,
\begin{equation}
    \zeta(\Sigma_g) = \sum_{i = 1}^k \lambda_i^{2 - 2 g}.
\end{equation}
For simplicity, we consider the case when we have $\lambda_i = 1$ for all $i$, in which we have $\zeta(\Sigma_g) = k$ for all connected Riemann surfaces.

In constructing $\CBU$ for $\zeta$, we are free to ignore any handles, as modulo null states the genus of a surface with boundary makes no difference (we are free to ignore disconnected components, as always, for a similar reason). Thus, modulo null states, each bordism in $\CBU$ is determined uniquely by the induced partition of its set of boundaries into bulk connected components. If we use $n \cdot S^1 \in \Bord_2^{SO}$ to denote the disjoint union of $n$ circles, and label partitions by Young diagrams, we obtain a spanning set of morphisms,
\begin{equation}
    D \in \CBU( n \cdot S^1 \to m \cdot S^1 ), \quad \text{boxes}(D) = n + m.
\end{equation}
Composition is defined formally by convolving the induced partitions, with any handles which might appear removed, and any closed surfaces evaluating to the number $k$ of vacua.

This description of $\CBU$ for the two-dimensional TQFT with $\zeta(\Sigma_g) = k$ is, precisely, the formal category $\mathrm{Rep}(S_t)$ constructed by Deligne in \cite{deligne2007categorie}, again up to a translation from unitary to algebraic representation theory. The number $k$ of vacua in the TQFT plays the role of the formal variable $t$, and moreover, when $t = k \in \mathbb{Z}_{\geq 0}$ we simply obtain, precisely, $\CBU = \mathrm{Rep}(S_k)$, after quotienting by null states. Under this identification, $S^1 \in \CBU$ maps to the defining permutation representation $\mathbb{C}^k$ of $S_k$, and the hemisphere $\varnothing \to S^1$ maps to the inclusion of the trivial representation of $S_k$ into $\mathbb{C}^k$ as the vector $(1, \dots, 1)$ (of squared norm $k$, compatibly with $\zeta(S^2) = k$).

Our formalism shows, then, that the universal construction for this two-dimensional TQFT produces, precisely, the set of states which are neutral under the $S_k$ symmetry permuting the vacua. Note that $\CBU$ is smart enough to know the difference between TQM and two-dimensional TQFT: while the total symmetry of a space of $k$ degenerate ground states is $U(k)$, the only valid local symmetry of the two-dimensional TQFT is $S_k$, due to the superselection into vacua.

\section{Robustness and extensions}
\label{sec:robustness_extensions}

We have completed our construction of $\CBU$ and proof of Theorem \ref{thm:main}. Now, in this section, we discuss the robustness of our framework as well as some simple extensions. First, we discuss the extent to which $\CBU$ depends on the choice of counterterms, and how it behaves under stacking of theories. We then describe how fermionic sources require us to generalize our framework, as is necessary to restore ER = EPR in any fermionic theory. Finally, we sketch how our framework might be generalized once Axiom \ref{axiom:multiplicativity} is dropped, so that we have a non-factorizing partition function and are forced to deal with $\alpha$-sectors.

\subsection{Counterterm independence and stacking}\label{sec:counterterms}

From the perspective of formal high-energy theory, a natural concern is that our construction requires us to know the closed-manifold partition function $\zeta$ as a fixed finite number. As discussed in Section \ref{sec:functorial_QFT}, continuum QFTs are frequently considered only up to equivalence under stacking of counterterms, meaning invertible unitary QFTs $\BordX \to \sHilb^\times$. We now discuss how our formalism works for invertible QFTs and behaves under stacking, in order to establish the sense in which $\CBU$ is counterterm-independent.

First, we need a way to discuss invertible theories knowing only their closed-manifold partition function. As invertible theories are, by definition, theories which are invertible under stacking, we make the following definition.

\begin{definition}\label{defn:invertible_theory}
    Let $\zeta$ be a partition function on $\BordX$, satisfying Axioms \ref{axiom:finiteness}-\ref{axiom:reflection_positivity}. We say that $\zeta$ is \textit{invertible} if $\zeta(M) \neq 0$ for all closed manifolds $M$, and the inverse partition function $\zeta^{-1}$ is also reflection positive (i.e., satisfies Axiom \ref{axiom:reflection_positivity}).
\end{definition}

\noindent We only need to assume that $\zeta^{-1}$ is reflection positive, as it automatically satisfies Axioms \ref{axiom:finiteness}-\ref{axiom:continuity} as soon as $\zeta$ does, provided $\zeta(M) \neq 0$ for all $M$.\footnote{It would be interesting to attempt to prove the conjecture of \cite{Freed:2016rqq} in our framework, that deformation classes of not-necessarily-topological invertible unitary QFTs are classified by a stable $d$-dimensional bordism group. This would likely require developing a notion of stable geometric structure, which might also supply a canonical unitary structure on $\BordX$, as in the stable topological case.}

Definition \ref{defn:invertible_theory} is a good definition of an invertible partition function due to the following standard fact: a Gram matrix $M_{ij}$, whose Hadamard (entrywise) inverse $M_{ij}^{-1}$ is also a Gram matrix, necessarily has rank one. As a result, the universal Hilbert spaces $\mathcal{H}_B$ associated to an invertible partition function are one-dimensional whenever the bordism class $[B] \in \Omega_{d-1}^\mathcal{X}$ vanishes, and vanish when $[B] \neq 0$ (as is always the case). This gives us a very simple description of $\CBU$ for any invertible theory,
\begin{equation}\label{eq:invertible_CBU}
    \CBU \simeq \Hilb[\Omega_{d-1}^\mathcal{X}]^{\sigma}.
\end{equation}
The superscript $\sigma$ denotes twisting the symmetry by a group homomorphism $\sigma : \Omega_{d-1}^\mathcal{X} \to \{\mathrm{b}, \mathrm{f}\}$, a theory-dependent assignment of fermion parity to bordism classes. For example, note that while the trivial theory on $\Bord_2^{Spin}$ is bosonic, the invertible Arf-invariant theory (the Kitaev chain \cite{Kitaev:2000nmw}) makes the Ramond-sector circle $[S^1_\mathrm{p}]$ a fermion.

While \eqref{eq:invertible_CBU} seems to forget most of the data of a given invertible theory, it is not lost. Instead, the invertible theory is encoded in the functor
\begin{equation}
    \BordX \to \CBU \simeq \Hilb[\Omega_{d-1}^\mathcal{X}]^\sigma,
\end{equation}
which includes knowledge of $\zeta$ by construction. We note that the choice of equivalence $\CBU \simeq \Hilb[\Omega_{d-1}^\mathcal{X}]^\sigma$ is the same as the choice of a fiber functor $\CBU \to \sHilb$, hence also the same as a choice of actual invertible QFT with partition function $\zeta$.

Let us now consider stacking two arbitrary partition functions $\zeta_1, \zeta_2$, each of which satisfies Axioms \ref{axiom:finiteness}-\ref{axiom:reflection_positivity}. The stacked theory, by definition, has partition function,
\begin{equation}\label{eq:stacking}
    \zeta(M) \defined \zeta_1(M) \zeta_2(M).
\end{equation}
By the Schur product theorem, the stacked partition function is automatically reflection positive. In general, $\CBU$ does not behave simply under stacking, as the tensor product of two QFTs may have additional symmetries beyond those which were already present in either theory; for example, when we stack two copies of the same QFT, we get a new symmetry which exchanges the copies.

However, for each manifold $B \in \BordX$, we do have a canonical isometric map between the universal Hilbert spaces defined with $\zeta$, $\zeta_1$, and $\zeta_2$ respectively, defined on manifolds by,
\begin{equation}\label{eq:stacking_isometry}
    \mathcal{H}_B^\zeta \to \mathcal{H}_B^{\zeta_1} \otimes \mathcal{H}_B^{\zeta_2}, \quad \ket{M}_\zeta \mapsto \ket{M}_{\zeta_1} \otimes \ket{M}_{\zeta_2}, \quad M \in \BordX(\varnothing \to B)
\end{equation}
and extended linearly. Note that this map only makes sense because our sources are classical, so there is no issue with cloning; indeed, the very operation \eqref{eq:stacking} of stacking theories only works because our manifolds are classical sources.

The maps \eqref{eq:stacking_isometry} extend uniquely to a symmetric monoidal W*-tensor functor,
\begin{equation}\label{eq:diagonal_under_stacking}
    \mathcal{C}_\mathrm{BU}^\zeta \to \mathcal{C}_\mathrm{BU}^{\zeta_1} \otimes \mathcal{C}_\mathrm{BU}^{\zeta_2},
\end{equation}
such that the following diagram commutes,
\begin{equation}\label{eq:diagonal_stacking_commutative_diagram}
    \begin{tikzcd}
        \BordX \arrow[r] \arrow[d, "\Delta"'] & \mathcal{C}_\mathrm{BU}^\zeta \arrow[d] \\
        \BordX \times \BordX \arrow[r] & \mathcal{C}_\mathrm{BU}^{\zeta_1} \otimes \mathcal{C}_\mathrm{BU}^{\zeta_2}
    \end{tikzcd}
\end{equation}
This follows from applying the universal property of $\mathcal{C}_\mathrm{BU}^\zeta$ (Proposition \ref{prop:univ_property_CBU}) to the bottom path through \eqref{eq:diagonal_stacking_commutative_diagram}. Again, we see that stacking depends crucially on the classical nature of source manifolds, which allows us to define the diagonal map $\Delta$ on the left side of \eqref{eq:diagonal_stacking_commutative_diagram}.

Interpreted in terms of symmetry groups via Theorem \ref{thm:main}, the functor \eqref{eq:diagonal_under_stacking} induces a group homomorphism,
\begin{equation}\label{eq:stacking_group_hom}
    G_{\zeta_1} \times G_{\zeta_2} \to G_{\zeta},
\end{equation}
between the respective symmetry groups, encoding the fact that the product of the symmetry groups of the QFTs associated to $\zeta_1$ and $\zeta_2$ does, at least, act as a symmetry of the QFT associated to $\zeta$, even if it is not the full symmetry.

However, we can show a bit more. By construction, the functor \eqref{eq:diagonal_under_stacking} maps the $[B]$-graded component of $\mathcal{C}_\mathrm{BU}^\zeta$ to the $([B], [B])$-graded component of $\mathcal{C}_\mathrm{BU}^{\zeta_1} \otimes \mathcal{C}_\mathrm{BU}^{\zeta_2}$, for any bordism class $[B] \in \Omega_{d-1}^\mathcal{X}$. What this means is that it restricts to a functor,
\begin{equation}\label{eq:relative_tensor_over_bordism}
    \mathcal{C}_\mathrm{BU}^\zeta \to \mathcal{C}_\mathrm{BU}^{\zeta_1} \otimes_{\Omega_{d-1}^\mathcal{X}} \mathcal{C}_\mathrm{BU}^{\zeta_2},
\end{equation}
where the notation $\otimes_{\Omega_{d-1}^\mathcal{X}}$ denotes the full sub-category of the ordinary tensor product on the diagonally-graded objects. Physically, \eqref{eq:relative_tensor_over_bordism} encodes the fact that the dual cobordism symmetry of \cite{McNamara:2019rup, Friedan:2023vxx} acts by the same c-number in all three QFTs, and so the map \eqref{eq:stacking_group_hom} factors through the diagonal quotient $(G_{\zeta_1} \times G_{\zeta_2})/(\Omega_{d-1}^\mathcal{X})^\vee$.

The functor \eqref{eq:relative_tensor_over_bordism} directly leads us to a form of counterterm-independence for the baby universe category. Suppose we stack an invertible partition function $\zeta_\mathrm{c.t.}$, representing a counterterm, on top of an arbitrary partition function $\zeta$,
\begin{equation}
    \zeta'(M) = \zeta_\mathrm{c.t.}(M) \zeta(M).
\end{equation}
By \eqref{eq:relative_tensor_over_bordism}, we obtain a functor,
\begin{equation}\label{eq:stacking_CBU_invertible}
    \mathcal{C}_\mathrm{BU}^{\zeta'} \to \mathcal{C}_\mathrm{BU}^{\zeta_\mathrm{c.t.}} \otimes_{\Omega_{d-1}^\mathcal{X}} \mathcal{C}_\mathrm{BU}^{\zeta},
\end{equation}
which is actually a unitary equivalence of categories. This can be seen in a number of ways: for one, the map \eqref{eq:stacking_isometry} defined from the Hadamard product of a rank-one Gram matrix with any Gram matrix is an isometry, and for another, the functor induced by stacking $\zeta'$ with $\zeta_\mathrm{c.t.}^{-1}$ is an inverse to \eqref{eq:stacking_CBU_invertible}.

The equivalence \eqref{eq:stacking_CBU_invertible} is canonical, but if we have an actual invertible QFT with partition function $\zeta$ (and not just its isomorphism class), we even get an equivalence,
\begin{equation}
    \mathcal{C}_\mathrm{BU}^{\zeta'} \simeq \Hilb[\Omega_{d-1}^\mathcal{X}]^\sigma \otimes_{\Omega_{d-1}^\mathcal{X}} \mathcal{C}_\mathrm{BU}^{\zeta} = (\mathcal{C}_\mathrm{BU}^{\zeta})^\sigma,
\end{equation}
where $(\mathcal{C}_\mathrm{BU}^{\zeta})^\sigma$ denotes the twisting of the symmetry on $\mathcal{C}_\mathrm{BU}^{\zeta}$ by $\sigma$.

As a result, $\CBU$ is nearly counterterm independent, with the only ambiguity being the ambiguity in assigning fermion parity to states on nontrivial bordism classes (such as in the Ramond sector of a two-dimensional spin CFT). This matches the fact that the symmetry group $G$ of a unitary QFT is a well-defined, universal invariant, which does not depend on counterterms.

\subsection{Including fermionic sources}\label{sec:fermionic_sources}

We now remedy a major flaw of the framework developed in the bulk of this paper, which is that it cannot accommodate fermionic sources. As discussed in Section \ref{sec:traces_and_fermions}, the issue is not unitarity, but symmetric monoidality. We illustrate this issue with the following example, continuing from Sections \ref{sec:ex_fTQM} and \ref{sec:reconstructing_a_fermion}.

\begin{example}\label{ex:cannot_add_fermion_in_bordism}
    Consider the partition function on $\Bord_1^{Spin}$ given by,
    \begin{equation}
        \zeta(S^1_\mathrm{ap}) = +1, \quad \zeta(S^1_\mathrm{p}) = -1.
    \end{equation}
    As we have seen, this example has an underlying fermionic TQM with a single fermionic state on $\mathrm{pt}^+$, and has baby universe category $\CBU = \mathrm{sRep}(U(1)_\mathrm{\text{spin}^c})$. Suppose we wanted to adjoin an end-of-the-world-brane to break the $U(1)_\mathrm{\text{spin}^c}$ global symmetry, as in Section \ref{sec:ex_EotW_branes}, or as in our review of the proof of the DR reconstruction theorem. This would correspond to adjoining some new bordism $\psi : \varnothing \to \mathrm{pt}^+$, which could be viewed as an interval with a free endpoint marked by $\psi$. No matter how hard we work, any geometric bordism $\psi$ will satisfy,
    \begin{equation}\label{eq:bad_fermion_eotw}
        \theta_{\mathrm{pt}^+} \circ \psi = \psi, \quad s_{\mathrm{pt}^+ \sqcup \mathrm{pt}^+} \circ (\psi \sqcup \psi) = \psi \sqcup \psi,
    \end{equation}
    since the twist and symmetry are monoidal natural transformations, which necessarily trivialize on $\varnothing$.

    Consider the three states $\ket{\mathrm{id}_{\mathrm{pt}^+}}, \ket{\theta_{\mathrm{pt}^+}}, \ket{\psi \sqcup \psi^\dagger} \in \CH_{\mathrm{pt}^+ \sqcup \mathrm{pt}^-}$. By reflection positivity, we must have $\braket{\psi | \psi} \geq 0$, and unless it vanishes we can rescale to assume $\ket{\psi}$ is normalized. Using \eqref{eq:bad_fermion_eotw}, we compute the inner product matrix,
    \begin{equation}
    \begin{array}{c | c c c}
        & \ket{\mathrm{id}_{\mathrm{pt}^+}} & \ket{\theta_{\mathrm{pt}^+}} & \ket{\psi \sqcup \psi^\dagger} \\
        \hline
        \bra{\mathrm{id}_{\mathrm{pt}^+}} & +1 & -1 & +1 \\
        \bra{\theta_{\mathrm{pt}^+}} & -1 & +1 & +1 \\
        \bra{\psi \sqcup \psi^\dagger} & +1 & +1 & +1
    \end{array}
    \end{equation}
    This is not positive semidefinite, and the only way out would be for $\psi$ to be a null morphism. The issue is that we tried to adjoin a necessarily-bosonic state to a theory which needed a fermionic one.
\end{example}

As demonstrated by this example, the issue arises from our requirement that the twist and symmetry be symmetric monoidal. The resolution is physically obvious: we should make fermionic sources anti-commute at spacelike separation, and transform by a sign when acted on by the twist. However, in the non-linear bordism category $\BordX$, there is no way to parse ``anti-commute'' or ``transform by a sign,'' as we cannot multiply a bordism by $(-1)$.

The resolution, to both issues, is that we need to start in a linear category. This is harmless for the previously considered case without fermionic sources, as our construction proceeds through the free linearization $\CBUpre = \mathbb{C}[\BordX]$ anyway. Starting with a linear category is also necessary if we already have some linear structure on sources, as it allows us to parse the constraint that $\zeta$ should depend linearly on sources. For example, if we start with some vector space of local operators, and we want $\mathcal{X}$ to include points marked by operator insertions, we would need to pick a basis of operators if we wanted to define $\zeta$ on a non-linear bordism category.

Thus, let us assume we are given a linear category $\CBUpre$, which we formally think of as some linearization of a unitary bordism category. We need to modify a few of the definitions of a unitary structure, such as requiring the conjugation and $\dagger$-structure,
\begin{equation}
    (-)^\dagger: \CBUpre \to (\overline{\CBUpre})^\mathrm{op}, \quad \overline{(-)} : \CBUpre \to \overline{\CBUpre},
\end{equation}
to be anti-linear, rather than linear, as before.

Now, in order to include fermionic sources, we should not ask for $\CBUpre$ to be a symmetric monoidal category, but instead ask for it to be a symmetric monoidal \textit{sVec-enriched category}. This means that the spaces of morphisms in $\CBUpre$ are super-vector spaces, and that the naturality condition for symmetric monoidality is,
\begin{equation}
    s_{Y_1, Y_2} \circ (\mathcal{O}_1 \otimes \mathcal{O}_2) = (-1)^{|\mathcal{O}_1| |\mathcal{O}_2|} (\mathcal{O}_2 \otimes \mathcal{O}_1) \circ s_{X_2, X_1},
\end{equation}
for morphisms $\mathcal{O}_i : X_i \to Y_i$ of parity $|\mathcal{O}_i|$. We pick up a factor of $(-1)^{|\mathcal{O}_1||\mathcal{O}_2|}$ because we need to use the symmetry on $\sVec$, the base of enrichment, to swap $\mathcal{O}_1 \otimes \mathcal{O}_2$ to $\mathcal{O}_2 \otimes \mathcal{O}_1$.

Asking for $\CBUpre$ to be symmetric monoidal in the $\sVec$-enriched sense encodes the fact that fermionic sources should anti-commute at spacelike separation. We also need to encode the fact that the twist should act on fermionic morphisms by a sign. We do this by asking for the twist to be a natural transformation,
\begin{equation}\label{eq:fermion_parity_twisted}
    \theta : \mathrm{id}_\CBUpre \Rightarrow \mathrm{id}_\CBUpre^{F}.
\end{equation}
The functor $\mathrm{id}_\CBUpre^F$ is a canonical twisted identity functor on any $\sVec$-enriched category, defined on objects $B$ and morphisms $\mathcal{O}$ by,
\begin{equation}\label{eq:F_twisted_identity_functor}
    \mathrm{id}_\CBUpre^F : B \mapsto B, \quad \mathcal{O} \mapsto (-1)^{|\mathcal{O}|} \mathcal{O}.
\end{equation}
Concretely, \eqref{eq:fermion_parity_twisted} means that we have,
\begin{equation}
    \theta_B \circ \mathcal{O} = (-1)^{|\mathcal{O}|} \mathcal{O} \circ \theta_A, \quad \mathcal{O} : A \to B,
\end{equation}
as desired. To show that this modification works, let us reconsider Example \ref{ex:cannot_add_fermion_in_bordism}.

\begin{example}
    Let $\Bord_1^{Spin}$ and $\zeta$ be as in Example \ref{ex:cannot_add_fermion_in_bordism}. Let us, again, adjoin a new morphism $\psi : \varnothing \to \mathrm{pt}^+$, which we can now take to be fermionic, $|\psi| = 1$. We now have,
    \begin{equation}\label{eq:good_fermion_eotw}
        \theta_{\mathrm{pt}^+} \circ \psi = -\psi, \quad s_{\mathrm{pt}^+ \sqcup \mathrm{pt}^+} \circ (\psi \sqcup \psi) = - \psi \sqcup \psi,
    \end{equation}
    because of the twisted naturality of the twist and symmetry. We compute, with $\ket{\theta_{\mathrm{pt}^+}}$, $\ket{\psi \sqcup \psi^\dagger} \in \CH_{\mathrm{pt}^+ \sqcup \mathrm{pt}^-}$ as in Example \ref{ex:cannot_add_fermion_in_bordism},
    \begin{equation}
    \begin{array}{c | c c c}
        & \ket{\mathrm{id}_{\mathrm{pt}^+}} & \ket{\theta_{\mathrm{pt}^+}} & \ket{\psi \sqcup \psi^\dagger} \\
        \hline
        \bra{\mathrm{id}_{\mathrm{pt}^+}} & +1 & -1 & +1 \\
        \bra{\theta_{\mathrm{pt}^+}} & -1 & +1 & -1 \\
        \bra{\psi \sqcup \psi^\dagger} & +1 & -1 & +1
    \end{array}
    \end{equation}
    This is positive semidefinite, as desired.
\end{example}

We now proceed in parallel with the rest of this paper, building universal super-Hilbert spaces, forming the regular representation induced by $\zeta$, and taking a double commutant to define a $\sVec$-enriched atomic W*-category $\CBUvN$. When it comes time to take the W*-Cauchy completion, the idempotent completion is unchanged, but we modify the direct-sum completion to include formal categorical wavefunctions,
\begin{equation}\label{eq:sHilb_wavefunction}
    \bigoplus_i \mathcal{H}_i \otimes X_i, \quad \mathcal{H}_i \in \sHilb, \quad X_i \in \CBUvN,
\end{equation}
with coefficient super-Hilbert spaces instead of merely Hilbert spaces. We define the baby universe category to be the result of this process,
\begin{equation}\label{eq:super_CBU}
    \CBU \defined \sHilb(\CBUvN),
\end{equation}
the $\sVec$-enriched W*-Cauchy completion of $\CBUvN$. We note that \eqref{eq:super_CBU} produces a category which is not merely $\sVec$-enriched, but also \textit{sHilb-tensored}, meaning we can tensor a super-Hilbert space with an object of $\CBU$ and get another well-defined object. We will call an $\sVec$-enriched, $\sHilb$-tensored W*-category a \textit{super W*-category}.

To avoid confusion, let us note that even without fermionic sources, the definition \eqref{eq:super_CBU} of the baby universe category will not agree with the category constructed and studied for most of this paper. For disambiguation, in the remainder of this section only, we will always take $\CBU$ to denote the $\sVec$-enriched category constructed by \eqref{eq:super_CBU}, and use $\CBUold$ to denote the category from Definition \ref{defn:CBU} we have studied for most of this paper (when it is well-defined, i.e., when we do not have fermionic sources). We have the relationship,
\begin{equation}
    \CBU = \sHilb \otimes_\Hilb \CBUold.
\end{equation}
Working with $\CBU$, as opposed to $\CBUold$, means that we are doing categorical quantum mechanics over $\sHilb$, not over $\Hilb$; for example the categorical wavefunction overlap will land in $\sHilb$. Moreover, the necessary and sufficient condition for ER = EPR is,
\begin{equation}
    \CBU = \sHilb,
\end{equation}
corresponding to $\CBUold = \Hilb$.

The last step of our main argument, given by DR reconstruction, needs a small modification, because as-stated it does not apply to a rigidly generated super W*-tensor category. Nevertheless, we can use structure already present in $\CBU$ (and in fact, already secretly encountered) to reduce to a case where DR reconstruction applies.

As motivation, recall that a \textit{real structure} on a $\mathbb{C}$-vector space $V$ is an anti-linear automorphism $\mathcal{T}$ with $\mathcal{T}^2 = 1$. Given a real structure, the invariant subspace $V_\mathbb{R}$ is a real vector space, and we have a canonical $\mathbb{C}$-linear isomorphism,
\begin{equation}\label{eq:tensor_over_R}
    V = \mathbb{C} \otimes_\mathbb{R} V_\mathbb{R},
\end{equation}
which takes $\mathcal{T}$ to the canonical real structure (complex conjugation) on the tensor factor $\mathbb{C}$. Thus, a real structure is \textit{Galois descent data} from $\mathbb{C}$ to $\mathbb{R}$: a twisted action of the Galois group,
\begin{equation}
    \mathrm{Gal}(\mathbb{C}/\mathbb{R}) = \mathbb{Z}_2^\mathcal{T},
\end{equation}
realizing $V$ as having arisen from the underlying $\mathbb{R}$-vector space $V_\mathbb{R}$.

By analogy, given a Cauchy complete super W*-category $\CC$, we define a \textit{bosonic structure} to be a twisted $\mathbb{Z}_2$ natural automorphism of the identity functor $\mathrm{id}_\CC$,
\begin{equation}
    \theta : \mathrm{id}_\mathcal{C} \Rightarrow \mathrm{id}_\mathcal{C}^{F}, \quad \theta^2 = 1.
\end{equation}
The twisted identity functor $\mathrm{id}_\mathcal{C}^{F}$ is defined, as in \eqref{eq:F_twisted_identity_functor}, to be the functor which acts as the identity on objects and by fermion parity on morphisms. By construction, we have a form of anti-linearity over $\sHilb$,
\begin{equation}\label{eq:higher_antilinearity}
    \theta_{\CH \otimes X} = (-1)^{F_\CH} \otimes \theta_X, \quad \CH \in \sHilb, \quad X \in \mathcal{C}.
\end{equation}

Given a bosonic structure on a Cauchy complete super W*-category $\CC$, we define the \textit{bosonic subcategory} $\CC_\mathrm{b}$ to be the full subcategory on the objects $X$ with $\theta_X = \mathrm{id}_X$; any morphism in $\mathcal{C}_\mathrm{b}$ is necessarily bosonic. Analogously to \eqref{eq:tensor_over_R}, we have a canonical equivalence of super W*-categories,
\begin{equation}
    \mathcal{C} = \sHilb \otimes_\Hilb \mathcal{C}_\mathrm{b},
\end{equation}
which takes the given bosonic structure $\theta$ on $\mathcal{C}$ to the canonical bosonic structure $(-1)^F$ on $\sVec$. A bosonic structure $\theta$ is higher Galois descent data from $\sHilb$ to $\Hilb$: a twisted action of the Galois higher group \cite{Johnson-Freyd:2015fua},
\begin{equation}
    \mathrm{Gal}(\sHilb/\Hilb) = B \mathbb{Z}_2^F,
\end{equation}
realizing $\CC$ as having arisen from the underlying bosonic W*-category $\CC_\mathrm{b}$. There is a very physical way to understand a bosonic structure. A bosonic structure tells you, for each simple object $\mu \in \mathcal{C}$, which of $\mu$ or $\mathbb{C}^{0|1} \otimes \mu$ is a boson, while otherwise you only knew that they must have opposite fermion parity.\footnote{Unlike the case of real structures on $\mathbb{C}$-vector spaces, there are examples of super W*-categories $\mathcal{C}$ which do not admit any bosonic structure. This happens whenever $\mathcal{C}$ contains a Majorana simple object $\mu$, meaning that $\mathbb{C}^{0 | 1} \otimes \mu \cong \mu$, or equivalently, that $\mathrm{End}(\mu)$ is the complex Clifford algebra $\mathrm{Cl}_1(\mathbb{C})$.}

In fact, we have seen a bosonic structure previously, in our review of the proof of the DR reconstruction theorem. Given a symmetric W*-tensor category $\mathcal{C}$, there are \textit{two} canonical bosonic structures on the super W*-category,
\begin{equation}\label{eq:fermionized_tensor_category}
    \mathcal{C}_\mathrm{f} \defined \sHilb \otimes_\Hilb \mathcal{C}.
\end{equation}
The first is the obvious one, which acts on $\sHilb$ as $(-1)^F$ and acts trivially on $\mathcal{C}$. The second is defined by the canonical twist $\theta$ induced by the unitary structure. With this second bosonic structure, the associated bosonic subcategory $\mathcal{C}_\mathrm{b} \subset \mathcal{C}_\mathrm{f}$ is nothing but the bosonized category $\mathcal{C}_\mathrm{b}$ with twisted symmetry $s_{X_\mathrm{b}, Y_\mathrm{b}}$ defined in \eqref{eq:bosonization_symmetry}. As discussed previously, this twisted symmetry is obtained by dressing the fermionic objects in $\CC$ with additional transparent fermions.

Returning to the baby universe category, if we do not include fermionic sources, $\CBU = \sHilb \otimes_\Hilb \CBUold$ would also have two canonical bosonic structures. The first is the obvious one which acts only on the coefficient super-Hilbert spaces, and acts trivially on $\CBUold$. The second is given by,
\begin{equation}\label{eq:bosonic_structure_CBU}
    \theta_{\mathcal{H} \otimes B} \defined (-1)^{\widehat{F}_\mathcal{H}} \otimes \theta_B, \quad \mathcal{H} \in \sHilb, \quad B \in \BordX,
\end{equation}
where now $\theta_B$ is the $2 \pi$ rotation $\mathcal{X}$-automorphism of $B$. Let us note that, even without fermionic sources, $\CC_{\mathrm{BU, b}} \mathrel{\not\simeq} \CBUold$, as $\CC_{\mathrm{BU, b}}$ is always bosonic as a symmetric W*-tensor category, while $\CBUold$ will be fermionic when the underlying QFT (depending, say, on a spin structure) requires fermions.

However, with fermionic sources, we only have one canonical bosonic structure, given by \eqref{eq:bosonic_structure_CBU}, as the trivial one is not compatible with the symmetric monoidal structure. Thus, the output of our construction, in general, is a symmetric super W*-tensor category $\CBU$, equipped with its canonical bosonic structure. The bosonic category $\CC_{\mathrm{BU, b}}$ can then be fed directly into the DR reconstruction theorem, and we learn that,
\begin{equation}
    \CC_{\mathrm{BU, b}} \simeq \mathrm{Rep}(G), \quad \CBU = \mathrm{sRep}(G),
\end{equation}
where now $\mathrm{sRep}(G)$ is just the super W*-category of arbitrary unitary representations of a group $G$, with no additional structure. The whole discussion of fiber functors versus super fiber functors, and groups versus super groups, has disappeared, because we have at our disposal a reservoir of transparent fermions. For example, the theory considered in Sections \ref{sec:ex_fTQM}, \ref{sec:reconstructing_a_fermion}, and Example \ref{ex:cannot_add_fermion_in_bordism}, given in $d = 1$ by a single fermionic state, has,
\begin{equation}
    \CBU = \sHilb \otimes_\Hilb \CBUold \simeq \sHilb \otimes_\Hilb \mathrm{sRep}^\mathrm{old}(U(1)_{\text{spin}^c}) \simeq \mathrm{sRep}(U(1)).
\end{equation}
The central element $(-1)^F$ is not really a canonical part of our framework, arising only because of the additional, unneeded bosonic structure on $\sHilb \otimes_\Hilb \CBUold$.

As a final comment, note that the canonical bosonic structure $B \mapsto \theta_B$ on $\CBU$ is the direct analog of the canonical real structure $\mathcal{T} : \ket{M} \mapsto \ket{\overline{M}}$ on states in the baby-universe Hilbert space $\HBU$, as discussed for instance in \cite{Harlow:2023hjb, Witten:2025ayw}. In \cite{Harlow:2023hjb}, it was further argued that this real structure, corresponding to CPT, should be gauged in quantum gravity. This simply means that physically-reasonable holographic partition functions $\zeta : \mathcal{A}_\mathrm{BU} \to \mathbb{C}$ must satisfy $\zeta(\overline{M}) = \overline{\zeta(M)}$, as we have assumed. Luckily, every $\alpha$-state is automatically real in this sense, being a $*$-algebra homomorphism. One could even go so far as to say that the CPT theorem in QFT \textit{is} the statement that CPT must be gauged in holographic quantum gravity.

Just as \cite{Harlow:2023hjb} argued that CPT should be gauged in quantum gravity, we argue that fermion parity must also be gauged, in the sense that physically-reasonable QFTs (i.e., fiber functors) $\CBU \to \sHilb$ should intertwine the bosonic structures. By the spin-statistics theorem, any fiber functor $\CBU \to \sHilb$ must intertwine the bosonic structures in this way. So, in some sense, we could say that the spin-statistics theorem \textit{is} the statement that fermion parity must be gauged in holographic quantum gravity.

As argued in \cite{Harlow:2023hjb}, given the canonical real structure on the baby universe Hilbert space $\HBU$, we see that the condition for partition-function factorization is, equivalently,
\begin{equation}
\HBU = \mathbb{C} \iff \mathcal{H}_{\mathrm{BU}, \mathbb{R}} = \mathbb{R}.
\end{equation}
The analogous statements, for the baby universe category, are:
\begin{equation}
    \CBU = \sHilb \iff \CC_{\mathrm{BU,b}} = \Hilb \iff \CC_{\mathrm{BU}, \mathbb{R}} = \sHilb_\mathbb{R} \iff \CC_{\mathrm{BU, b}, \mathbb{R}} = \Hilb_\mathbb{R},
\end{equation}
corresponding to the Galois fixed points under the four obvious subgroups of
\begin{equation}
    \mathrm{Gal}(\sHilb/\Hilb_\mathbb{R}) = \mathbb{Z}_2^{\mathcal{T}} \times B\mathbb{Z}_2^F,
\end{equation}
as discussed in \cite{Johnson-Freyd:2015fua}.

\subsection{Dropping multiplicativity and including $\alpha$-sectors}\label{sec:including_alpha_sectors}

Our approach in this paper has been to restrict to factorizing partition functions, but from the perspective of the GPI, it is more natural to drop multiplicativity (Axiom \ref{axiom:multiplicativity}) and consider both partition-function factorization and Hilbert-space factorization at once.

For a discrete set of $\alpha$-sectors, our approach generalizes easily: we simply project to each $\alpha$ sector to obtain a factorizing partition function $\zeta_\alpha$ to which our methods directly apply. By taking a direct sum over $\alpha$, we have,
\begin{equation}
    \CBU \defined \bigoplus_\alpha \mathcal{C}_\mathrm{BU}^\alpha,
\end{equation}
with $\BordX$ mapping in diagonally.\footnote{If some closed manifolds $M$ have $\alpha$-sector partition functions which are not uniformly bounded in $\alpha$, the inclusion map from $\BordX$ technically lands not in $\CBU$ but in the associated category of trace-class operators.} The resulting category is an atomic, rigidly generated, symmetric W*-multitensor category, and we have,
\begin{equation}
    \mathcal{A}_\mathrm{BU} \defined \CBU(\varnothing \to \varnothing) = \bigoplus_\alpha \mathbb{C},
\end{equation}
so the baby universe algebra is just a direct sum of projectors onto the $\alpha$-sectors.

Applying Theorem \ref{thm:main} in each $\alpha$-sector, we learn that,
\begin{equation}
    \CBU \simeq \bigoplus_\alpha \mathrm{sRep}(G_\alpha),
\end{equation}
for a family of compact Hausdorff super-groups $G_\alpha$, and obtain the formula,
\begin{equation}\label{eq:oplus_and_Galpha_invts}
    \mathcal{H}_{B_1 \sqcup \cdots \sqcup B_n} = \bigoplus_\alpha \left( \widetilde{\CH}_{B_1}^\alpha \otimes \cdots \otimes \widetilde{\CH}_{B_n}^\alpha \right)^{G_\alpha},
\end{equation}
for the multi-boundary Hilbert spaces obtained from the GPI. 

If the measure space of $\alpha$-sectors is diffuse, let us assume the moment problem mentioned in Footnote \ref{footnote:moment} is solved, and we do at least have a commutative baby universe von Neumann algebra $\mathcal{A}_\mathrm{BU}$. As the space $\mathrm{Spec}(\mathcal{A}_\mathrm{BU})$ of $\alpha$-sectors is diffuse, we cannot project to an $\alpha$-sector, and so we cannot directly apply our construction of $\CBU$. While we could still use $\zeta$ to build universal Hilbert spaces, we would not be able to prove the trace inequality (Proposition \ref{prop:trace_inequality}), as its proof relies on partition-function factorization. A natural approach would be to attempt to define universal self-dual $\mathcal{A}_\mathrm{BU}$-Hilbert modules for each manifold $B$, with an $\mathcal{A}_\mathrm{BU}$-valued inner product defined by pairing $M_1, M_2 : \varnothing \to B$ to the closed manifold $M_1^\dagger \circ M_2 \in \mathcal{A}_\mathrm{BU}$.

We leave the full exploration of this idea to future work. We expect that there is some version of the DR reconstruction theorem which applies to non-atomic, but still rigidly generated, symmetric W*-multitensor categories. If so, we would learn that,
\begin{equation}
    \CBU = \int^\oplus d\alpha\ \mathrm{sRep}(G_\alpha),
\end{equation}
and obtain the formula,
\begin{equation}\label{eq:alpha_integral_discussion}
    \mathcal{H}_{B_1 \sqcup \cdots \sqcup B_n} = \int^\oplus d\alpha\ \left( \widetilde{\CH}_{B_1}^\alpha \otimes \cdots \otimes \widetilde{\CH}_{B_n}^\alpha \right)^{G_\alpha},
\end{equation}
discussed in the Introduction.

Assuming the above works, then the higher ensemble averaging perspective of Section \ref{sec:higher_averaging} applies, verbatim, to the multitensor case with multiple $\alpha$-sectors. For example, if the space of $\alpha$-sectors is discrete, we have, non-canonically,
\begin{equation}
    \mathrm{Spec}(\CBU) \simeq \coprod_\alpha BG_\alpha,
\end{equation}
and thus we may rewrite \eqref{eq:oplus_and_Galpha_invts} as,
\begin{equation}
    \mathcal{H}_{B_1 \sqcup \cdots \sqcup B_n} = \int^\oplus_{\mathrm{Spec}(\CBU)} \widetilde{\CH}_{B_1} \otimes \cdots \otimes \widetilde{\CH}_{B_n} \cong \bigoplus_\alpha \left( \widetilde{\CH}_{B_1}^\alpha \otimes \cdots \otimes \widetilde{\CH}_{B_n}^\alpha \right)^{G_\alpha},
\end{equation}
where the first equality is completely canonical. This is the most direct analog of writing a non-factorizing partition function as a quenched average of the underlying $\alpha$-sector partition functions. We note that one natural measure on $\mathrm{Spec}(\CBU)$, which weights each $\alpha$-sector by the inverse of the volume of $G_\alpha$, has recently appeared in the context of ensemble holography \cite{Barbar:2025vvf,Dymarsky:2026asf}.

We can now say what an $\alpha$-sector is, once both partition function and Hilbert-space factorization are taken into account: an $\alpha$-sector is nothing but a fiber functor on $\CBU$. This matches the fact that, physically, projecting to an $\alpha$-sector is supposed to mean picking a holographic theory out of the ensemble, not just an isomorphism class of holographic theories.

We may even write a formula for the categorical wavefunction of the $\alpha$-sector, which we now do in the case of a factorizing partition function, for simplicity. Note that choosing tensorators,
\begin{equation}
\phi_{\mu \nu} : \mathbb{C}^{d_\mu} \otimes \mathbb{C}^{d_\nu} \to \bigoplus_{\rho} \mathcal{V}_{\mu \nu}^\rho \otimes \mathbb{C}^{d_\rho},
\end{equation}
defining a fiber functor is the same thing as equipping the categorical $\alpha$-state,
\begin{equation}\label{eq:categorical_alpha_state}
    \kket{\alpha} \defined \bigoplus_\mu \overline{\mathbb{C}^{d_\mu}} \kket{\mu},
\end{equation}
with the structure of a commutative algebra object in $\CBU$. Conversely, given a commutative algebra structure on $\kket{\alpha} \in \CBU$, we recover the fiber functor as $\bbra{\alpha}$, the functor defined by taking a categorical wavefunction overlap with $\kket{\alpha}$. We see that being an $\alpha$-state, while a condition in codimension-one, is data in codimension-two.

The correspondence between fiber functors and commutative algebra structures on \eqref{eq:categorical_alpha_state} is quite standard in the context of generalized symmetries. By construction, the object $\kket{\alpha} \in \CBU$ corresponds, under the identification $\CBU \simeq \mathrm{Rep}(G)$, to the regular representation of $G$, and passing to the fiber functor is known as ``gauging the $\mathrm{Rep}(G)$ symmetry.'' To see which $\mathrm{Rep}(G)$ symmetry this is, note that we may view the canonical functor $\BordX \to \CBU$ as a relative QFT, in the sense of \cite{Freed:2012bs}, whose bulk is given by flat $G$ gauge theory in any higher number of dimensions, which has a dual $\mathrm{Rep}(G)$ symmetry. Picking a fiber functor corresponds to picking a gapped Dirichlet boundary condition for this bulk, which picks a global form of $\BordX \to \CBU$ and turns it into an absolute QFT. This perspective fits well with the observations of \cite{Benini:2022hzx,Torres:2025jcb,Yu:2026gdf}. More generally, when there are multiple non-isomorphic $\alpha$-sectors, arising from a non-factorizing partition function, the bulk of this relative QFT will decompose into a direct sum over multiple vacua, matching the discussion in \cite{Sharpe:2023lfk,Perez-Lona:2025add} about the relationship between ensembles and decomposition.

\section{Discussion}\label{sec:conclusions_future_directions}

In this paper, we have studied the issue of Hilbert-space factorization in quantum gravity. We have worked within the axiomatic framework, described in Sections \ref{sec:axiomatic_approach_philosophy}-\ref{sec:spaces_of_grav_states}, in which we assume we have been given a factorizing, reflection-positive partition function $\zeta$ defined on some class $\mathcal{X}$ of source manifolds which serve as the formal boundaries of a putative GPI. By providing a proof of Theorem \ref{thm:main}, we have learned that, within this framework, any breakdown of Hilbert-space factorization is a red herring: there is always an underlying unitary QFT with partition function $\zeta$, and the Hilbert spaces of gravitational states are precisely the invariant subsector of this underlying QFT with respect to a global symmetry group $G$. Thus, we have learned that the necessary condition of Harlow \cite{Harlow:2015lma} for ER = EPR is, in fact, also sufficient.

\subsection{Baby universe field theory and averaged holography}

Given our results, there are now two perspectives one could take on a GPI which fails to satisfy ER = EPR, which parallel the perspectives called \textit{baby universe field theory} and \textit{averaged holography} in \cite{Harlow:2026hky} on a GPI which produces non-factorizing partition functions. To describe the corresponding perspective on non-factorizing Hilbert spaces, let us first review these perspectives.
\begin{itemize}
    \item \textbf{Baby universe field theory:} The baby universe Hilbert space $\HBU$ is genuinely nontrivial, and multi-boundary partition functions do not factorize. The different $\alpha$-states are actual distinct states in a theory of quantum gravity. Information which falls into a black hole ends up in a baby universe, and does not come out in the Hawking radiation. Quantum gravity can have global 0-form symmetries.

    \item \textbf{Averaged holography:} The baby universe Hilbert space $\HBU$ is an auxiliary tool, and we are secretly studying a classical ensemble, parametrized by the $\alpha$-sectors, of different boundary conditions for the GPI. Each member of this ensemble has a trivial baby universe Hilbert space and holographically factorizing partition function. In each member, black holes evaporate unitarily. Quantum gravity has no global 0-form symmetries.
\end{itemize}

The result of Coleman, Giddings, Strominger, Marolf, and Maxfield \cite{Coleman:1988cy,Giddings:1988cx,Giddings:1988wv,Marolf:2020xie} justifies the perspective of averaged holography by showing that any baby universe field theory is, secretly, nothing but the ensemble average of underlying holographic partition functions. In this paper, we have provided a justification for the analogous perspective, which we call \textit{higher averaged holography}, on a GPI with non-factorizing Hilbert spaces. Higher averaged holography stands in contrast to what might be called \textit{extended baby universe field theory}. These perspectives are given as follows.

\begin{itemize}

    \item \textbf{Extended baby universe field theory:} The baby universe category $\CBU$ is genuinely nontrivial in quantum gravity. The universal Hilbert spaces $\mathcal{H}_B$ are the actual multi-boundary Hilbert spaces, and do not factorize. The Bekenstein--Hawking formula is incorrect.\footnote{In the sense that the actual entropy of black hole microstates, derived from actual states in the one-sided Hilbert spaces, does not satisfy the Bekenstein--Hawking formula, and is instead smaller (as is necessarily the case, given Proposition \ref{prop:thermal_bound}).} Quantum gravity can have global 1-form symmetries and an incomplete spectrum of charged states. ER $\neq$ EPR.

    \item \textbf{Higher averaged holography:} The baby universe category $\CBU$ is an auxiliary tool, and we are secretly studying a higher ensemble average of underlying holographically dual QFTs. The missing one-sided black hole microstates are present in the underlying UV-complete theory, which has a trivial baby universe category. The Bekenstein--Hawking formula is correct. Quantum gravity has no global 1-form symmetries and a complete spectrum of charged states. ER = EPR.
    
\end{itemize}

While our expectation is that a UV-complete theory of quantum gravity is holographic, let us note that baby universe field theories (BUFTs) may still prove to be objects of mathematical interest. It would be interesting to establish a full set of axioms for functorial BUFT; our expectation is that they would involve an appropriately lax form of symmetric monoidal functors $\BordX \to \sHilb$. Further, we expect BUFTs, in general, to sit somewhere in between the universal construction and the underlying QFT, potentially having more states than the universal construction but fewer than the full holographic QFT. In fact, we expect BUFTs on $\BordX$ with partition function $\zeta$ to have a sort of Galois theory, being classified by the subgroup $H \subset G$ under which all of their states are neutral.\footnote{If such a Galois theory exists, then the universal construction plays the role of the ground field, while the underlying holographic QFT plays the role of the algebraic closure.}

\subsection{Comparison with past literature}

We now compare our results with previous studies of Hilbert-space factorization in the literature. First, our conclusion fits well with the previously-observed pattern \cite{Penington:2023dql,Kolchmeyer:2023gwa,Chua:2023ios,Boruch:2024kvv,Balasubramanian:2024yxk,Li:2024nft,Banerjee:2024fmh,Balasubramanian:2025zey,Balasubramanian:2025jeu} that gravitational Hilbert spaces seem to automatically factorize once enough degrees of freedom are included. Moreover, it fits especially well with the observations of \cite{Maxfield:2023mdj} where our exact conclusions were demonstrated in a toy model (reviewed in Section \ref{sec:Ex_TopoQM}). It also, of course, fits well with the previous work \cite{Colafranceschi:2023urj,Marolf:2024adj} which served as the technical backbone of this paper. This paper completes the analysis of \cite{Colafranceschi:2023urj,Marolf:2024adj} by studying the case with more than two boundaries, and explaining that the $\mu$-sectors are nothing but the irreducible representations of the symmetry group $G$.

Let us emphasize an important distinction regarding the proper meaning of Hilbert-space factorization. Holography requires, not just that multi-boundary Hilbert spaces happen to factorize as \textit{some} tensor product (indeed, this requirement is vacuous), but that the canonical map,
\begin{equation}\label{eq:disjoint_union_map_conclusions}
    \sqcup : \CH_{B_1} \otimes \CH_{B_2} \to \CH_{B_1 \sqcup B_2}, \quad \ket{M_1} \otimes \ket{M_2} \mapsto \ket{M_1 \sqcup M_2},
\end{equation}
is an isomorphism. This distinction is relevant, for example, in the case of JT gravity with matter (as studied in, e.g., \cite{Boruch:2024kvv}), where the two-boundary Hilbert space factorizes, but the map \eqref{eq:disjoint_union_map_conclusions} cannot possibly be an isomorphism. This is because, in the absence of end-of-the-world (EotW) branes, the one-boundary Hilbert space vanishes, due to the obstruction arising from bordism.

What seems to be happening is that, in JT gravity with matter, the \textit{only} remaining obstruction is the one arising from bordism, so that the object $\mathrm{pt}^+ \in \CBU$ is invertible under disjoint union (with inverse $\mathrm{pt}^-$). As with any invertible object, we must have $d_{\mathrm{pt}^+} = 1$, and so it only takes a single species of EotW brane to get the map \eqref{eq:disjoint_union_map_conclusions} to be an isomorphism. Thus, every two-sided state in JT gravity with matter can be viewed as arising from gluing this unique EotW brane to its conjugate, with any profile of matter excitations coming along for the ride. The invertibility of $\mathrm{pt}^+ \in \CBU$ under disjoint union makes this gluing invisible, which is why it played no role in the analysis of \cite{Boruch:2024kvv}.

Our analysis also correctly matches the known reason why the two-sided Hilbert space of pure JT gravity, \textit{without} matter, does not factorize in any meaningful sense. Putting aside issues with drawing from the ensemble of \cite{Saad:2019lba}, whatever the $\alpha$-sector theories may be, they must be some quantum mechanical systems. As we saw in Example \ref{sec:ex_QM}, quantum mechanical systems always have, at least, the symmetry group which rotates each energy level independently. Gravitational states in pure JT gravity only produce neutral states, or in other words, states with equal energy as measured on the two sides; this is merely the gravitational gauge constraint of JT gravity. This constraint is lifted in the presence of matter, which allows us to insert charged operators as additional sources (compare with \cite{Lin:2022rbf,Penington:2023dql,Kolchmeyer:2023gwa}). This breaks the corresponding symmetry to the diagonal $U(1)$, which is associated with the bordism group $\Omega_1^{SO} = \mathbb{Z}$ and can only be broken by EotW branes.

\subsection{Bulk interpretation of the symmetry group}
\label{sec:bulk_interpretation_of_symmetry_group}

The symmetry group $G$ appearing in Theorem \ref{thm:main} has, so far, only been discussed as a global symmetry of the holographically dual QFT associated to a bulk theory of quantum gravity. We now comment on the bulk interpretation. At first glance, $G$ should be interpreted as a bulk gauge group without any charged states, in line with the motivating example of Einstein--Maxwell theory \cite{Harlow:2015lma}. In some sense, this interpretation is correct, as every global symmetry of the holographically dual QFT is usually viewed as the dual of a bulk gauge symmetry under the holographic dictionary.

However, putting aside our specific discussion, the matching between a boundary global symmetry and a bulk gauge symmetry requires caveats. First of all, it does not correctly type check: a boundary global symmetry is a duality invariant notion, while a bulk gauge symmetry is not invariant under duality. The resolution is that the correct bulk dual of a boundary global symmetry is a boundary-localized global symmetry of the theory of quantum gravity, as defined with the boundary held fixed as a source. Such boundary-localized symmetries arise when the bulk has a deconfined gauge field with a Dirichlet boundary condition, but may arise in many other ways. In the context of AdS/CFT, having a boundary-localized global symmetry is the same as having a long-range gauge symmetry in the sense of \cite{Harlow:2018tng}. Notably, the logic of \cite{Harlow:2018tng} only rules out global symmetries in gravity which act in the bulk, and is perfectly compatible with boundary-localized global symmetries.

Even with this said, the bulk interpretation of our symmetry group $G$ requires an additional caveat, if we view our input partition function $\zeta$ as arising from projecting the output of some GPI onto an $\alpha$-sector. This is because there is no guarantee that the symmetry group $G$ must be visible as a gauge field in the GPI, or even as a long-range gauge symmetry of the bulk EFT. One obvious way this can happen is when the symmetry is associated with a nontrivial bordism class, in which case the associated ``gauge flux'' is simply the nontrivial topology of the bulk spacetime itself, as discussed in \cite{McNamara:2019rup,McNamara:2022lrw}. For example, in JT gravity, the ``gauge flux'' associated to $\Omega_1^{SO} = \mathbb{Z}$ is merely the bulk 2-manifold itself, viewed as the flux radiated by an EotW brane.

More subtly, the symmetry group $G$ might get washed out in the ensemble over $\alpha$-sectors, only reappearing once one projects to a fixed $\alpha$-sector. This happens in the topological toy model of Marolf and Maxfield \cite{Marolf:2020xie}, whose ensemble of $\alpha$-sectors corresponds to an ensemble of TQMs (as in Section \ref{sec:Ex_TopoQM}) with a Poisson-random ground state degeneracy $k$. The global symmetry in the $\alpha$-sector with ground state degeneracy $k$ is $U(k)$, but only the common diagonal $U(1)$ is visible in the bulk description as the cobordism symmetry associated to oriented manifolds. The rest of the presumptive bulk $U(k)$ gauge theory is washed out in the ensemble, with its only remnant being the nonzero mean-field Euler counterterm visible in the bulk action functional.

One upshot is that we need not get the charged states we may have expected to come for free when $G$ is non-abelian. If we knew the bulk had dynamical $G$ Yang--Mills fields, and had made sure to turn on background $G$ gauge fields in our class $\mathcal{X}$ of sources, then we would automatically get charged states from the adjoint action of $G$ on itself. Perturbatively, these charged states are built from bulk gluons, which transform in the adjoint representation of $G$ on its Lie algebra, and which provide all neutral states under the non-invertible 1-form centralizer symmetry of \cite{Rudelius:2020orz,Heidenreich:2021xpr}. Globally, this also includes the states of $G$ bundles on topologically nontrivial manifolds, which correspond to gravitational solitons and break the non-invertible 1-form centralizer symmetry to the invertible 1-form center symmetry \cite{McNamara:2021cuo}. When $G$ is not visible in the bulk EFT as a gauge field, or even as a symmetry of the boundary condition, we will not get these states, which is why \textit{all} gravitational states can be neutral under a non-abelian group.

We offer one further interpretation of the symmetry group $G$, from the perspective on holography in terms of quantum error correction (see, e.g., \cite{Harlow:2018fse}). From this perspective, the collection of semiclassical states in a holographic CFT are described by a \textit{code subspace}, in which the semiclassical bulk physics are encoded via a map from semiclassical states to states of the holographic CFT. From our perspective, the universal construction provides a canonical code subspace of geometric states in any QFT, though the encoding map will be far from isometric \cite{Akers:2022qdl} when the bulk is far from semiclassical.\footnote{Even to the point of failing to emerge, in the sense of \cite{Engelhardt:2026blp}.} Theorem \ref{thm:main}, in this language, states that the universal construction is a \textit{stabilizer code}, meaning, a code subspace specified as the invariant subspace under a group. More generally, Theorem \ref{thm:main} implies that the only families of code subspaces in a local unitary QFT which are compatible with cutting and gluing of manifolds are stabilizer codes for global symmetries.

\subsection{Reflection positivity in the gravitational path integral}

A major deficiency of our axiomatic framework, as with any axiomatic framework, is that it is only as strong as its axioms. We have already discussed the weaknesses of Axiom \ref{axiom:multiplicativity}, the multiplicativity of $\zeta$, at length, as it corresponds precisely to the original issue of partition-function factorization. We have also touched on the weaknesses of Axiom \ref{axiom:finiteness}, the finiteness/well-definition of $\zeta$, which, while severe, are an unavoidable aspect of trying to reason about the GPI. Axioms \ref{axiom:reality}, the reality of $\zeta$, and \ref{axiom:continuity}, the continuity of $\zeta$ are essentially kinematic, and we do not see any major hurdles arising from either, especially as the continuity of $\zeta$ plays a very minimal role in our construction and results.

It remains to discuss our most important assumption: Axiom \ref{axiom:reflection_positivity}, that $\zeta$ is reflection positive. As in Section \ref{sec:axioms_for_partition_function}, we emphasize just how strong of an assumption reflection positivity really is, particularly when imposed for all possible manifolds and possibly-disconnected spatial slices. Strong though it is, there are many reflection-positive partition functions, simply because there are many unitary QFTs. One might wonder, however, whether it would ever be practically possible to find a reflection-positive partition function without knowing the underlying QFT in the first place.

Our hope is that the GPI could help us find such reflection-positive partition functions. This hope applies even in the absence of partition-function factorization; as we have seen in our review of the proof of the Doplicher--Roberts reconstruction theorem, it can sometimes be easier to write down an ensemble, prove that it is reflection positive, and project to an $\alpha$-sector, than to write down an $\alpha$-sector quantity directly.

Nevertheless, the output of the GPI is, in general, not reflection-positive. This has been observed, for instance, in the GPI for three-dimensional pure gravity \cite{Maloney:2007ud,Keller:2014xba,Benjamin:2019stq,Benjamin:2020mfz,Maxfield:2020ale,DiUbaldo:2023hkc}, and was used recently to argue for the axionic weak gravity conjecture \cite{DiUbaldo:2026rly,Maldacena:2026jqd}. The results of this paper move the target for the GPI forward: while a generic GPI may not produce a reflection-positive partition function, our results imply that as soon as one does, then there exists a unitary holographic dual, which is specified uniquely up to unitary isomorphism. For instance, consider the ensemble of two-dimensional CFT data studied in \cite{Chandra:2022bqq,Wang:2025bcx}. Our results imply that the failure of a generic member of this ensemble to satisfy crossing must correspond directly to actual negativity of the corresponding GPI.\footnote{Working this out in detail would be of independent interest for the purpose of connecting our framework more concretely to the conjecture of Friedan and Shenker \cite{Friedan:1986ua} which it, in principle, resolves.}

\subsection{Towards fully-extended holography}

We close this paper by speculating on the full tower of Factorization Paradoxes described in the Introduction. The first rung on this tower, as studied by Coleman, Giddings, Strominger, Marolf, and Maxfield \cite{Coleman:1988cy,Giddings:1988wv,Giddings:1988cx,Marolf:2020xie}, is the problem of partition-function factorization arising from Euclidean wormholes, or black instantons. The story does not end there, and the next rung is Hilbert-space factorization, as studied in this paper. This second rung is threatened by ER bridges, or black holes.

But Hilbert-space factorization is also not the end of the story. Just as the baby universe Hilbert space $\HBU$ is only one among the whole family of universal Hilbert spaces $\CH_B$, the baby universe category should merely be one among a whole family of universally constructed categories $\CC_K$, one for each closed $(d-2)$-dimensional $\mathcal{X}$-manifold $K$. The definition of $\CC_K$ is not quite clear, but the rough picture is as follows. The objects of $\CC_K$ should be $(d-1)$-dimensional bordisms $B : \varnothing \to K$, with categorical wavefunction overlap,
\begin{equation}
    \bbrakket{B_1 | B_2} = \CH_{B_2 \cup_K \overline{B}_1}, \quad \kket{B_i} \in \CC_K.
\end{equation}
The Hilbert space on the right hand side is the universal Hilbert space on the \textit{connected} $(d-1)$-manifold obtained by gluing $B_2$ to $\overline{B}_1$ along their shared boundary $K$.

The morphisms in $\CC_K$ should be given, roughly, by bordisms with corners along $K$. More precisely, bordisms with corners of opening angle $2 \pi/p$ along $K$ should form the non-commutative $L^p$ spaces \cite{haagerup1979lp} of morphisms associated to $\CC_K$. Thus, opening angle $\pi$ ($p = 2$) would give the space of states in $\CH_{B_2 \cup_K \overline{B}_1}$, and opening angle $2 \pi$ ($p = 1$) would give trace-class morphisms, which can be glued into closed manifolds and fed into $\zeta$. Finally, and likely only formally, opening angle $0$ ($p = \infty$) would form the morphisms of $\CC_K$ itself, which would almost certainly be Type III.

We could then ask: do the categories $\CC_K$ factorize on disjoint unions,
\begin{equation}
    \CC_{K_1 \sqcup K_2} \stackrel{?}{=} \CC_{K_1} \otimes \CC_{K_2}.
\end{equation}
This issue, the Factorization Paradox for categories, is threatened not just by black instantons, and not just by black holes, but also by black strings. An obvious obstruction, related to completeness of the spectrum of charged strings, is the possibility of a global 2-form symmetry in the bulk, related to a global 1-form symmetry of the holographic dual.

So we should build a baby universe 2-category, whose objects are condensates \cite{Gaiotto:2019xmp} of closed $(d-2)$-dimensional $\mathcal{X}$-manifolds $K$, and prove a higher version of DR reconstruction, saying that the stack of fiber functors is $B^2G_0$ for some abelian group $G_0$, at least when partition functions and Hilbert spaces already factorize. If only partition functions factorize, this stack would be $BG$ for some 2-group $G$, and in complete generality we would just have some 2-stack of $\alpha$-sectors. We would learn that any breakdown of factorization of categories is, again, a red herring, corresponding to an ensemble average of an extended QFT over this 2-stack.

Yet again, this would not be the end of the story, and we should keep going... until we landed, finally, on a baby universe $d$-category, and obtained a higher ensemble average of fully-local, i.e., fully-extended, functorial, unitary, holographic QFTs. Even the notion of once-extended functorial QFT is far from clear, outside of the topological case. We hope that the framework described in this paper might help to find it.

As a final comment, note that the original obstruction to holography arises from Euclidean wormholes, which are simply handles of Morse index one in spacetime. In Euclidean signature, ER bridges, which obstruct Hilbert-space factorization, are handles of Morse index two, as in the Euclidean Schwarzschild solution. The full tower of Factorization Paradoxes, then, arises from the full Morse decomposition of the bulk spacetime, or in other words, the entire bulk topology. If our framework does generalize to the fully extended case, we would learn that the bulk topology is, itself, a red herring \cite{McNamara:2024kitp}, as envisioned in the paper \cite{Stephens:1993an} immediately preceding 't Hooft's introduction \cite{tHooft:1993dmi} of the holographic principle.

\section*{Acknowledgments}
We thank Xi Dong and Henry Maxfield for discussions that inspired this work during the 2024 KITP conference on Spacetime and String Theory. We thank Corey Jones and Dave Penneys for patiently explaining the theory of W*-categories, and Marius Junge for pointing out the relation of \cite{Colafranceschi:2023urj,Marolf:2024adj} to Tomita--Takesaki theory. We thank Dan Freed, Mike Hopkins, Cameron Krulewski, Tomer Schlank, Luuk Stehouwer, and Constantin Teleman for discussions on bordism categories, dagger structures, and reflection positivity, and we further thank Luuk Stehouwer for helping us sort out conventions for unitary structures. We thank David Ayala, Tom Banks, Kasia Budzik, Anatoly Dymarsky, Tom Faulkner, Elliott Gesteau, Daniel Harlow, Simeon Hellerman, Javier Magan, Don Marolf, Vinicius Nevoa, Sanjay Raman, David Reutter, Cumrun Vafa, and Gabriel Wong for helpful discussions. Finally, we especially thank Theo Johnson-Freyd for many discussions and clarifying observations regarding the entire construction.

JM is currently supported by the Simons Center for Geometry and Physics, and was previously supported by the DOE (HEP) Award DE-SC0011632 as a research fellow at the Walter Burke Institute for Theoretical Physics at the California Institute of Technilogy, where part of this work was completed. ZW is supported by the DOE award number DE-SC0015655. JM and ZW thank the Kavli Institute for Theoretical Physics for hospitality during the Spring 2024 program on Spacetime and String Theory where this work was initiated. JM thanks the Kavli Institute for Theoretical Physics for hospitality during the Spring 2025 program on Generalized Symmetries in Quantum Field Theory, the Simons Center for Geometry and Physics for hospitality during the 2025 Summer workshop, the Simons Foundation for hospitality during the November 2025 Annual Meeting of the Simons Collaboration on Global Categorical Symmetries, and the Aspen Center for Physics for hospitality during the February 2026 conference on Generalized Symmetries and Defects in QFT and Gravity, where various parts of this work were completed.

\appendix

\section{Daggers, duals, and conjugates}
\label{app:daggers_duals_conjugates}

In this appendix, we establish the basic properties of the unitary structure on $\CBUpre$, its trace, and the map $\mathcal{O} \mapsto \ket{\mathcal{O}}$ turning a morphism into a ket.\footnote{All of the same properties hold in $\BordX$ with appropriate inclusions of the subscript ``$\Bord$'' on traces, bras, and kets.} In this appendix, we will not fix any partition function, so the traces and inner products will all be valued in $\CBUpre(\varnothing \to \varnothing)$.

Recall that the ket associated to a morphism $\mathcal{O} : A \to B$ is defined by,
\begin{equation}
    \ket{\mathcal{O}} \defined (\mathcal{O} \circ \rho_{A}^{-1})^{\vee_{\overline{A}^\vee}},
\end{equation}
and the associated bra morphism is defined by taking the adjoint,
\begin{equation}
    \bra{\mathcal{O}} \defined \ket{\mathcal{O}}^\dagger. 
\end{equation}
We also have the twist, defined by,
\begin{equation}
    \theta_B \defined \rho_{\overline{B}^\vee} \circ \rho_B,
\end{equation}
which is balanced, so that,
\begin{equation}\label{eq:balancing_appendix}
    \theta_B^2 = \mathrm{id}_B,
\end{equation}
as $\mathcal{X}$-automorphisms. Finally, for a morphism $\mathcal{O} : B \to B$, we have the trace,
\begin{equation}
    \tr(\mathcal{O}) \defined \str(\theta_B \circ \mathcal{O}),
\end{equation}
where we recall that the super-trace is defined by naively gluing the incoming and outgoing boundaries, or more formally, using the canonical pivotal structure induced from the symmetric monoidal structure.

We will also need a few relations satisfied by the Hermitian structure. We recall that our conventions are that $(-)^\vee$ and $\overline{(-)}$ are involutive and commute on the nose.\footnote{We thank Luuk Stehouwer for help sorting out our conventions, which differ slightly from the conventions of \cite{Freed:2016rqq,stehouwer2024unitary,stehouwer2024spin,ferrer2024dagger} in service of better matching the standard conventions for super-Hilbert spaces used in the physics community. In the language of \cite[Appendix B]{stehouwer2024spin}, our conventions correspond to a third, undiscussed option, given by multiplying the Hermitian structures discussed there by $i^F$.} These are given by,
\begin{equation}\label{eq:bar_and_vee_of_rho}
    \overline{\rho}_B = \rho_{\overline{B}} \circ \theta_{\overline{B}}^{-1} = \rho_{B^\vee}^{-1}, \quad \rho_{B}^\vee = \rho_{B^\vee}^{-1} \circ \theta_{\overline{B}} = \rho_{\overline{B}}, \quad \overline{\rho}_B^\vee = \rho_{\overline{B}^\vee}^{-1} = \rho_B \circ \theta_B.
\end{equation}
These imply, for the twist,
\begin{equation}
    \overline{\theta}_B = \theta_{\overline{B}}, \quad \theta_B^\vee = \theta_{B^\vee}.
\end{equation}
We now check that the unitary structure satisfies the following essential conditions.

\begin{proposition}
    For any $B \in \CBUpre$, the Hermitian structure $\rho_B : B \to \overline{B}^\vee$ is unitary,
    \begin{equation}
        \rho_B^\dagger = \rho_B^{-1}.
    \end{equation}
    Moreover, for any morphism $\mathcal{O} : A \to B$, we have,
    \begin{equation}
        \mathcal{O}^{\dagger \dagger} = \mathcal{O},
    \end{equation}
    so that $(-)^\dagger$ does indeed define a $\dagger$-structure.
\end{proposition}

\begin{proof}
    For the unitarity of $\rho_B$, we compute,
    \begin{equation}
        \rho_B^\dagger = \rho_B^{-1} \circ \overline{\rho}_B^\vee \circ \rho_{\overline{B}^\vee} = \rho_B^{-1},
    \end{equation}
    by \eqref{eq:bar_and_vee_of_rho}. For the involutory nature of $(-)^\dagger$, we compute,
    \begin{align}
        \mathcal{O}^{\dagger \dagger} &= \rho_B^{-1} \circ \overline{\mathcal{O}^\dagger}^\vee \circ \rho_A \\
        &= \rho_B^{-1} \circ \overline{\rho_A^{-1} \circ \overline{\mathcal{O}}^\vee \circ \rho_B}^\vee \circ \rho_A \\
        &= \rho_B^{-1}  \circ \overline{\rho}_B^\vee \circ \mathcal{O} \circ (\overline{\rho}_A^\vee)^{-1} \circ \rho_A \\
        &= \rho_B^{-1}  \circ \rho_{\overline{B}^\vee}^{-1} \circ \mathcal{O} \circ \rho_{\overline{A}^\vee} \circ \rho_A \\
        &= \theta_{B}^{-1} \circ \CO \circ \theta_A = \CO,
    \end{align}
    using \eqref{eq:bar_and_vee_of_rho} and the naturality of $\theta_B$.\footnote{Based on this calculation, one might worry that $\CO^{\dagger \dagger} = - \CO$ for fermionic morphisms in the $\sVec$-enriched generalization discussed in Section \ref{sec:fermionic_sources}. However, it does not, as we pick up an additional sign from the fact that $\theta_B$ becomes a twisted $B\mathbb{Z}_2$ action. Thus, $\CO^{\dagger \dagger} = \CO$, even for fermionic morphisms.}
\end{proof}

We now establish the properties of the trace and the kets associated to morphisms.

\begin{proposition}\label{prop:trace_pairing_is_inner_product}
    Let $\mathcal{O}_1, \CO_2 : A \to B$ be two parallel morphisms in $\CBUpre$. Then we have,
    \begin{equation}
        \tr(\CO_1^\dagger \CO_2) = \braket{\CO_1 | \CO_2},
    \end{equation}
    so that the trace pairing agrees with the inner product.
\end{proposition}

\begin{proof}
    First, we evaluate the inner product more explicitly. Expanding out the definition of $(-)^\dagger$, we have,
    \begin{equation}
        \bra{\mathcal{O}_1} = \overline{\mathcal{O}_1 \circ \rho_A^{-1}}^{\vee_{\overline{B}}} \circ \rho_{B \sqcup \overline{A}} = \overline{\mathcal{O}}_1^{\vee_{\overline{B}}} \circ (\rho_B \sqcup \theta_{\overline{A}}),
    \end{equation}
    using the naturality of taking partial duals (as we continue to use below without comment). Thus, we have,
    \begin{align}
        \braket{\CO_1 | \CO_2} &= \overline{\mathcal{O}}_1^{\vee_{\overline{B}}} \circ (\rho_B \sqcup \theta_{\overline{A}}) \circ (\CO_2 \circ \rho_A^{-1})^{\vee_{\overline{A}^\vee}} \\
        &= \overline{\mathcal{O}}_1^{\vee_{\overline{B}}} \circ (\rho_B \sqcup \theta_{\overline{A}}) \circ (\mathrm{id}_B \sqcup (\rho_A^\vee)^{-1})\circ \CO_2^{\vee_{A}} \\
        &= \overline{\mathcal{O}}_1^{\vee_{\overline{B}}} \circ (\rho_B \sqcup \theta_{\overline{A}}) \circ (\mathrm{id}_B \sqcup \theta_{\overline{A}} \circ \rho_{A^\vee})\circ \CO_2^{\vee_{A}} \\
        &= \overline{\mathcal{O}}_1^{\vee_{\overline{B}}} \circ (\rho_B \sqcup \rho_{A^\vee})\circ \CO_2^{\vee_{A}}\label{eq:inner_prod_computed}
    \end{align}
    Now, we evaluate the trace pairing. We compute,
    \begin{align}
        \tr(\CO_1^\dagger \CO_2) &= \str(\theta_A \circ \rho_A^{-1} \circ \overline{\CO}_1^\vee \circ \rho_B \circ \CO_2) \\
        &= \str(\rho_{\overline{A}^\vee} \circ \overline{\CO}_1^\vee \circ \rho_B \circ \CO_2). \label{eq:intermezzo_trace_pairing}
    \end{align}
    
    We need the following fact: the super-trace of a composite is given by composing the partial duals along the outermost cut. This follows from the definition of the symmetric monoidal pivotal structure, but can be seen more geometrically from the fact that both are given by the same naive gluing. Thus, continuing from \eqref{eq:intermezzo_trace_pairing}, we have,
    \begin{align}
        \tr(\CO_1^\dagger \CO_2) 
        &= (\rho_{\overline{A}^\vee} \circ \overline{\CO}_1^\vee)^{\vee_A} \circ (\rho_B \circ \CO_2)^{\vee_A} \\
        &= \overline{\CO}_1^{\vee_{\overline{B}}} \circ (\rho_B \sqcup \rho_{A^\vee}) \circ \CO_2^{\vee_A}
    \end{align}
    Comparing with \eqref{eq:inner_prod_computed}, we obtain the desired result.
\end{proof}

\begin{proposition}\label{prop:reality_of_trace}
    Let $\CO : B \to B$ be a morphism in $\CBUpre$. Then we have,
    \begin{equation}
        \overline{\tr(\CO)} = \tr(\CO^\dagger).
    \end{equation}
    Thus, the trace is real.
\end{proposition}

\begin{proof}
    We compute, using the manifest reality and sphericality of the super-trace,
    \begin{equation}
        \overline{\tr(\CO)} = \overline{\str(\theta_B \circ \CO)} = \str(\theta_{\overline{B}^{\vee}} \circ \overline{\CO}^\vee) = \str( \theta_{B} \circ \CO^\dagger) = \tr(\CO^\dagger),
    \end{equation}
    having used cyclicity of the super-trace and naturality of the twist in the third step to insert copies of the Hermitian structure.
\end{proof}

\begin{corollary}\label{corr:Hermitian_and_dagger_is_antiunitary}
    The inner product (equivalently, the trace pairing) on parallel morphisms $\CO_1, \CO_2 : A \to B$ is Hermitian,
    \begin{equation}
        \overline{\braket{\CO_1 | \CO_2}} = \braket{\CO_2 | \CO_1},
    \end{equation}
    and the $\dagger$-structure is anti-unitary,
    \begin{equation}
        \braket{\CO_1^\dagger | \CO_2^\dagger} = \braket{\CO_2 | \CO_1}.
    \end{equation}
\end{corollary}

\begin{proof}
    Hermiticity follows from Proposition \ref{prop:trace_pairing_is_inner_product} and Proposition \ref{prop:reality_of_trace} by,
    \begin{equation}
        \overline{\braket{\CO_1 | \CO_2}} = \overline{\tr(\CO_1^\dagger \CO_2)} = \tr(\CO_2^\dagger \CO_1) = \braket{\CO_2 | \CO_1}.
    \end{equation}
    Anti-unitarity of the $\dagger$-structure follows from cyclicity of the trace,
    \begin{equation}
        \braket{\CO_1^\dagger | \CO_2^\dagger} = \tr(\CO_1 \CO_2^\dagger) = \tr(\CO_2^\dagger \CO_1) = \braket{\CO_2 | \CO_1}.
    \end{equation}
\end{proof}

\begin{proposition}\label{prop:dagger_and_bar_kets_same}
    Let $\mathcal{O} : A \to B$ be a morphism in $\CBUpre$. We have,
    \begin{equation}
        \ket{\CO^\dagger} = \ket{\overline{\CO}},
    \end{equation}
    suppressing the use of the symmetry. As a consequence, we learn that $\overline{(-)}$ is also anti-unitary.
\end{proposition}

\begin{proof}
    We compute,
    \begin{align}
        \ket{\CO^\dagger} &= (\CO^\dagger \circ \rho_B^{-1})^{\vee_{\overline{B}^\vee}} \\
        &= (\rho_A^{-1} \circ \overline{\CO}^\vee)^{\vee_{\overline{B}^\vee}} \\
        &= (\overline{\CO} \circ (\rho_A^\vee)^{-1})^{\vee_{A^\vee}} \\
        &= (\overline{\CO} \circ \rho_{\overline{A}}^{-1})^{\vee_{A^\vee}} = \ket{\overline{\CO}},
    \end{align}
    as claimed.
\end{proof}

\begin{proposition}\label{prop:dinatuality_of_kets}
    Let $\CO_1 : B_1 \to B_2$ and $\CO_2: B_2 \to B_3$ be composable morphisms in $\CBUpre$. We have,
    \begin{equation}
        \ket{\CO_2 \circ \CO_1} = \CO_2 \circ_{B_2} \ket{\CO_1} = \overline{\CO}_1^\dagger \circ_{\overline{B}_2} \ket{\CO_2},
    \end{equation}
    where we recall from Section \ref{sec:analytic_core} that the notation $\circ_B$ denotes a partial composition along the boundary component $B$.
\end{proposition}

\begin{proof}
    The first equality is immediate from the naturality of partial dualization,
    \begin{align}
        \ket{\CO_2 \circ \CO_1} &= (\CO_2 \circ \CO_1 \circ \rho_{B_1}^{-1})^{\vee_{\overline{B}_1^\vee}} \\ &= \CO_2 \circ_{B_2} (\CO_1 \circ \rho_{B_1}^{-1})^{\vee_{\overline{B}_1^\vee}} \\ &= \CO_2 \circ_{B_2} \ket{\CO_1}.
    \end{align}
    For the second equality, we compute,
    \begin{align}
        \ket{\CO_2 \circ \CO_1} &= (\CO_2 \circ \CO_1 \circ \rho_{B_1}^{-1})^{\vee_{\overline{B}_1^\vee}} \\
        &= (\CO_2 \circ \rho_{B_2}^{-1} \circ \rho_{B_2} \circ \CO_1 \circ \rho_{B_1}^{-1})^{\vee_{\overline{B}_1^\vee}} \\
        &= (\rho_{\overline{B}_1}^{-1} \circ \CO_1^\vee \circ \rho_{\overline{B}_2}) \circ_{\overline{B}_2} (\CO_2 \circ \rho_{B_2}^{-1})^{\vee_{\overline{B}_2^\vee}} \\
        &= \overline{\CO}_1^\dagger \circ_{\overline{B}_2} \ket{\CO_2},
    \end{align}
    as claimed.
\end{proof}

\begin{corollary}\label{cor:modular_conjugation}
    The modular conjugation map $J : \ket{\CO} \mapsto \ket{\CO^\dagger}$ is anti-unitary and satisfies,
    \begin{equation}
        J : \CO_2 \circ_{B_2} \ket{\CO_1} \mapsto \overline{\CO_2} \circ_{\overline{B}_2} \ket{\CO_1^\dagger},
    \end{equation}
    for composable bordisms $\CO_1 : B_1 \to B_2$ and $\CO_2: B_2 \to B_3$.
\end{corollary}

\begin{proof}
    Anti-unitarity was shown in Corollary \ref{corr:Hermitian_and_dagger_is_antiunitary}. For the claim regarding the actions, we compute, using Proposition \ref{prop:dinatuality_of_kets},
    \begin{equation}
        CO_2 \circ_{B_2} \ket{\CO_1} = \ket{\CO_2 \circ \CO_1} \mapsto \ket{\CO_1^\dagger \circ \CO_2^\dagger} = \overline{\CO_2} \circ_{\overline{B}_2} \ket{\CO_1^\dagger},
    \end{equation}
    as claimed.
\end{proof}

\section{Analytic properties of the universal construction}
\label{app:analytic_properties}

Fix a partition function $\zeta$ on a unitary bordism category $\BordX$ that satisfies Axioms \ref{axiom:finiteness}-\ref{axiom:reflection_positivity}. In this appendix, we describe various analytic properties of the universal construction based on $\mathcal{X}$ and $\zeta$.

\subsection{The trace inequality}
\label{app:trace_inequality}

The most important analytic property of the universal construction is the fundamental trace inequality of \cite{Colafranceschi:2023urj}. We reproduce this trace inequality as the following proposition, which goes slightly beyond \cite{Colafranceschi:2023urj} in that our proof manifestly applies to fermionic theories.

\begin{proposition}[Trace inequality of \cite{Colafranceschi:2023urj}] \label{prop:trace_inequality}
    Let $\mathcal{O}_1 : B_1 \to B_2$ and $\mathcal{O}_2 : B_2 \to B_3$ be any pair of composable morphisms in the pre-baby universe category $\CBUpre = \mathbb{C}[\BordX]$. Then we have the inequality,
    \begin{equation}\label{eq:trace_inequality_states}
        \braket{\mathcal{O}_2 \circ \mathcal{O}_1 | \mathcal{O}_2 \circ \mathcal{O}_1} \leq \braket{\mathcal{O}_2 | \mathcal{O}_2} \braket{\mathcal{O}_1 | \mathcal{O}_1},
    \end{equation}
    or, equivalently,
    \begin{equation}\label{eq:trace_inequality_traces}
        \tr(\mathcal{O}_1^\dagger \circ \mathcal{O}_2^\dagger \circ \mathcal{O}_2 \circ \mathcal{O}_1) \leq \tr(\mathcal{O}_2^\dagger \circ \mathcal{O}_2) \tr(\mathcal{O}_1^\dagger \circ \mathcal{O}_1).
    \end{equation}
\end{proposition}

\begin{proof}
    Our proof will be, essentially, the argument given in \cite{Colafranceschi:2023urj}, phrased algebraically in a way that could be applied to any category with the structures we have put on $\BordX$. We prove the trace inequality by folding the trace twice in order to obtain an inner product in a four-boundary pre-Hilbert space, where we may apply the Cauchy--Schwarz inequality. We illustrate the argument in Figure \ref{fig:trace_inequality}.

    \begin{figure}
    \centering
    \includegraphics[width=\linewidth]{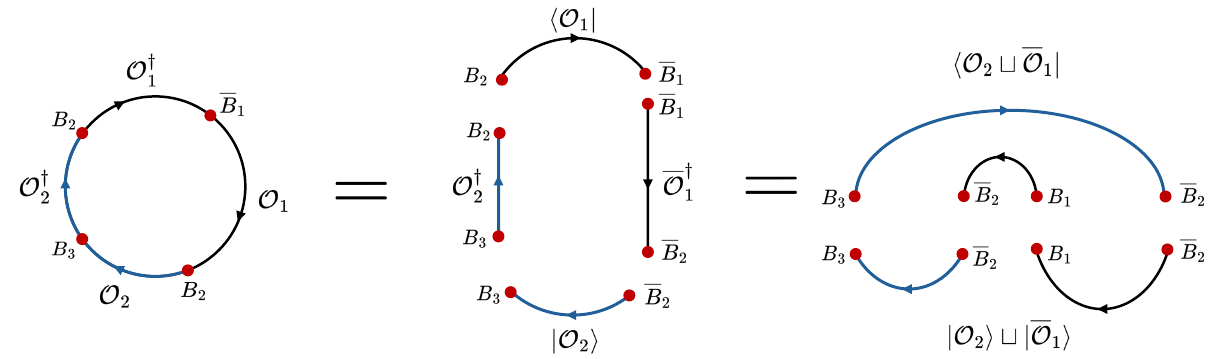}
    \caption{Diagrammatic representation of the proof of the trace inequality. Starting from $\braket{\mathcal{O}_2 \circ \mathcal{O}_1 | \mathcal{O}_2 \circ \mathcal{O}_1}$ (left), we fold it once to obtain $\langle \mathcal{O}_1|\circ (\mathcal{O}_2^\dagger \sqcup \overline{\mathcal{O}}_1^\dagger)\circ|\mathcal{O}_2\rangle $ (middle). Finally, we fold it a second time, to obtain $\braket{\mathcal{O}_2 \sqcup \overline{\mathcal{O}}_1 | s_{\overline{B}_2, \overline{B}_2} | \mathcal{O}_2 \sqcup \overline{\mathcal{O}}_1}$ (right). The insertion of the symmetry arises because the copies of $\overline{B}_2$ are connected differently on the top and bottom of the inner product.}
    \label{fig:trace_inequality}
    \end{figure}

    First, we compute,
    \begin{equation}
        \braket{\mathcal{O}_2 \circ \mathcal{O}_1 | \mathcal{O}_2 \circ \mathcal{O}_1} = \braket{\mathcal{O}_1 | \mathcal{O}_2^\dagger \circ \mathcal{O}_2 \circ \mathcal{O}_1},
    \end{equation}
    using that $\mathcal{O}_2^\dagger$ is an adjoint of $\mathcal{O}_2$. Then, using Proposition \ref{prop:dinatuality_of_kets}, we have,
    \begin{equation}\label{eq:folded_once}
        \braket{\mathcal{O}_1 | \mathcal{O}_2^\dagger \circ \mathcal{O}_2 \circ \mathcal{O}_1} = \braket{\mathcal{O}_1 | (\mathcal{O}_2^\dagger \sqcup \overline{\mathcal{O}}_1^\dagger) |  \mathcal{O}_2}.
    \end{equation}
    We have thus folded the trace once.
    
    To fold it again, we turn the right hand side of \eqref{eq:folded_once} into another trace,
    \begin{equation}
        \braket{\mathcal{O}_1 | (\mathcal{O}_2^\dagger \sqcup \overline{\mathcal{O}}_1^\dagger) |  \mathcal{O}_2} = \tr\big( \bra{\mathcal{O}_1} \circ (\mathcal{O}_2^\dagger \sqcup \overline{\mathcal{O}}_1^\dagger) \circ \ket{\mathcal{O}_2} \big),
    \end{equation}
    and use cyclicity of the trace together with Propositions \ref{prop:trace_pairing_is_inner_product} and \ref{prop:dagger_and_bar_kets_same} to write,
    \begin{align}
        \tr\big( \bra{\mathcal{O}_1} \circ (\mathcal{O}_2^\dagger \sqcup \overline{\mathcal{O}}_1^\dagger) \circ \ket{\mathcal{O}_2} \big) &= \tr\big( (\mathcal{O}_2 \sqcup \overline{\mathcal{O}}_1)^\dagger \circ (\ket{\mathcal{O}_2} \circ \bra{\mathcal{O}_1}) \big) \\
        &= \bra{\CO_2 \sqcup \overline{\CO}_1} (\ket{\CO_2} \sqcup \ket{\CO_1})
    \end{align}
    Finally, we note that,
    \begin{equation}
        \ket{\CO_2} \sqcup \ket{\CO_1} = s_{\overline{B}_2, \overline{B}_2} \ket{\mathcal{O}_2 \sqcup \overline{\mathcal{O}}_1},
    \end{equation}
    as the copies of $\overline{B}_2$ are swapped between the disjoint union of kets and the ket of a disjoint union (see Figure \ref{fig:trace_inequality}). Altogether, we have,
    \begin{equation}\label{eq:folded_twice}
        \braket{\mathcal{O}_2 \circ \mathcal{O}_1 | \mathcal{O}_2 \circ \mathcal{O}_1} = \braket{\mathcal{O}_2 \sqcup \overline{\mathcal{O}}_1 | s_{\overline{B}_2, \overline{B}_2} | \mathcal{O}_2 \sqcup \overline{\mathcal{O}}_1},
    \end{equation}
    successfully folding the trace twice.
    
    We now apply the Cauchy--Schwarz inequality and Proposition \ref{prop:dagger_and_bar_kets_same}, obtaining
    \begin{align}
        \braket{\mathcal{O}_2 \sqcup \overline{\mathcal{O}}_1 | s_{\overline{B}_2, \overline{B}_2} | \mathcal{O}_2 \sqcup \overline{\mathcal{O}}_1} &\leq \sqrt{\braket{\mathcal{O}_2 \sqcup \overline{\mathcal{O}}_1 | \mathcal{O}_2 \sqcup \overline{\mathcal{O}}_1} \braket{\mathcal{O}_2 \sqcup \overline{\mathcal{O}}_1 | s_{\overline{B}_2, \overline{B}_2}^2 | \mathcal{O}_2 \sqcup \overline{\mathcal{O}}_1}} \\
        &=\sqrt{\braket{\mathcal{O}_1 | \mathcal{O}_1}^2 \braket{\mathcal{O}_2 |\mathcal{O}_2}^2} = \braket{\mathcal{O}_1 | \mathcal{O}_1} \braket{\mathcal{O}_2 | \mathcal{O}_2}, \label{eq:end_of_Cauchy_Schwarz}
    \end{align}
    Combining \eqref{eq:folded_twice} and \eqref{eq:end_of_Cauchy_Schwarz} we have,
    \begin{equation}
        \braket{\mathcal{O}_2 \circ \mathcal{O}_1 | \mathcal{O}_2 \circ \mathcal{O}_1} \leq \braket{\mathcal{O}_1 | \mathcal{O}_1} \braket{\mathcal{O}_2 | \mathcal{O}_2},
    \end{equation}
    as claimed.
\end{proof}

\begin{corollary}\label{cor:composition_is_jointly_continuous}
    The composition map,
    \begin{equation}
        \circ : \mathcal{H}_{B_3 \sqcup \overline{B}_2}^\mathrm{pre} \otimes \mathcal{H}_{B_2 \sqcup \overline{B}_1}^\mathrm{pre} \to \mathcal{H}_{B_3 \sqcup \overline{B}_1}^\mathrm{pre}, \quad \ket{\mathcal{O}_2} \otimes \ket{\CO_1} \mapsto \ket{\CO_2 \circ \CO_1},
    \end{equation}
    is jointly continuous in the Hilbert space topology. Thus it extends to a jointly continuous bilinear map,
    \begin{equation}
        \circ : \mathcal{H}_{B_3 \sqcup \overline{B}_2} \otimes \mathcal{H}_{B_2 \sqcup \overline{B}_1} \to \mathcal{H}_{B_3 \sqcup \overline{B}_1},
    \end{equation}
    on the universal Hilbert spaces.
\end{corollary}

\begin{corollary}
\label{cor:operator_norm_bound}
    The operator norm of any morphism $\mathcal{O}$ in $\CBUpre$, acting by pre- or post-composition, is uniformly bounded,
    \begin{equation}
        \lvert \lvert \mathcal{O} \rvert \rvert \leq \sqrt{\braket{\mathcal{O} | \mathcal{O}}},
    \end{equation}
    by the Hilbert space norm.
\end{corollary}

\subsection{Continuity in sources}
\label{app:source_continuity}

We now establish the continuity of states and operators in the universal construction, as functions of sources.

\begin{proposition}
    \label{prop:continuity_of_state}
    Fix $B \in \BordX$, and consider the state $\ket{M}$ in the universal Hilbert space $\CH_B$ associated to an $\mathcal{X}$-bordism $M : \varnothing \to B$. Then $\ket{M}$ depends continuously on sources in the Hilbert space topology.
\end{proposition}

\begin{proof}
    Recall that a net $\ket{\psi_i}$ of states in a Hilbert space converges to a state $\ket{\psi}$ if and only if it converges weakly and the norms converge:
    \begin{equation}
        \ket{\psi_i} \to \ket{\psi} \iff \braket{\phi | \psi_i} \to \braket{\phi | \psi}, \forall \ket{\phi} \text{ and } \lvert\lvert \psi_i \rvert\rvert \to \lvert\lvert \psi \rvert\rvert.
    \end{equation}
    The implication $(\Longrightarrow)$ is immediate, and the converse $(\Longleftarrow)$ follows from,
    \begin{equation}
        \lvert\lvert \psi - \psi_i \rvert\rvert^2 = \lvert\lvert \psi \rvert\rvert^2 + \lvert\lvert \psi_i \rvert\rvert^2 - 2 \mathrm{Re} \braket{\psi | \psi_i} \to \lvert\lvert \psi \rvert\rvert^2 + \lvert\lvert \psi \rvert\rvert^2 - 2 \lvert\lvert \psi \rvert\rvert^2 = 0.
    \end{equation}
    Now, suppose we have a net $M_i \in \BordX(\varnothing \to B)$ of source manifolds converging to a fixed source manifold $M \in \BordX(\varnothing \to B)$. Both the inner products of $\ket{M_i}$ with any fixed state in $\mathcal{H}_B^\mathrm{pre}$ and the squared norms $\braket{M_i |M_i}$ are obtained by evaluating $\zeta$ on closed manifolds which depend continuously on $M_i$ (by continuity of the gluing of $\mathcal{X}$-structures). Thus, by continuity of $\zeta$, the norms of $\ket{M_i}$ converge to the norm of $\ket{M}$, and $\ket{M_i}$ converges weakly to $\ket{M}$ when tested against a dense subspace, hence everywhere by convergence of the norm. Thus, the states $\ket{M_i}$ converge to $\ket{M}$ in the Hilbert space topology, as claimed.
\end{proof}

\begin{proposition}\label{prop:continuity_of_morphisms}
    Fix $A, B \in \BordX$, and consider the morphisms in $\CBUpre$ or $\CBU$ associated to $\mathcal{X}$-bordisms $N : A \to B$. Then these morphisms depend continuously on sources in the norm topology.
\end{proposition}

\begin{proof}
    This follows from continuity of the states $\ket{N}$ in sources, as was just shown in Proposition \ref{prop:continuity_of_state}, combined with the bound of Corollary \ref{cor:operator_norm_bound} showing that the operator norm of $N$ is controlled by the Hilbert space norm of $\ket{N}$.
\end{proof}

\begin{proposition}
    \label{prop:seperable_hilbert_spaces}
    The universal Hilbert spaces $\mathcal{H}_B$ are separable.
\end{proposition}

\begin{proof}
     For fixed $B$, the set of diffeomorphism classes of compact $d$-manifolds with boundary $B$ is countable.\footnote{To see this, note that the set of $PL$-isomorphism classes of such manifolds is countable, as each compact $PL$-manifold is determined by a finite amount of combinatorial data. Moreover, due to smoothing theory \cite{hirsch2016smoothings}, the number of smoothings of a fixed compact $PL$-manifold is at most countable, and every smooth manifold admits a $PL$ structure.} Combining this fact with the assumed separability of the space of sources on any compact manifold, we may choose a countable dense set $\{M_i\} \subset \BordX(\varnothing \to B)$ of $\mathcal{X}$-bordisms. By Proposition \ref{prop:continuity_of_state}, the $\mathbb{Q}$-linear span of the states $\ket{M_i}$ is dense in $\mathcal{H}_B$, and thus $\mathcal{H}_B$ is a separable Hilbert space, as desired.
\end{proof}

\subsection{Spacetime symmetries}
\label{app:geometric_symmetries}

We now consider the action of spacetime symmetries on the universal Hilbert spaces. First, we construct a strongly continuous unitary representation of the group $\mathrm{Iso}^\mathcal{X}(B)$ of $\mathcal{X}$-isomorphisms acting on the Hilbert space $\mathcal{H}_B$. This unitary representation is required in order for the universal construction to be well-defined, as our manifolds are only defined up to $\mathcal{X}$-isomorphism in any case.

\begin{definition}\label{defn:unitaries_from_isometries}
    Let $u \in \mathrm{Iso}^\mathcal{X}(B)$. We define the operator $\widehat{u}$ on $\mathcal{H}_B^\mathrm{pre}$ on manifold states by the formula
    \begin{equation}
        \widehat{u} \ket{M} = \ket{u \circ M},
    \end{equation}
    where we recall that $u \circ M$ denotes the bordism obtained by twisting the outgoing boundary of $M$ by $u$.
\end{definition}

\begin{proposition}
    The operators $\widehat{u}$ are unitary, and so extend uniquely to unitary operators on $\mathcal{H}_B$. Moreover, the representation $u \mapsto \widehat{u}$ is continuous in the strong operator topology.
\end{proposition}

\begin{proof}
    Unitarity of $\widehat{u}$ follows immediately from the unitarity of $u$ under the $\dagger$-structure on $\BordX$. To check continuity in the strong operator topology, we use that a unitary group of operators is strongly continuous if and only if it is weakly continuous at the identity. For weak continuity, we have
    \begin{equation}
        \braket{M_1 | \widehat{u} | M_2} = \zeta( M_1^\dagger \circ u \circ M_2 ).
    \end{equation}
    But the $\mathcal{X}$-manifolds $M_1^\dagger \circ u \circ M_2$ converge to the $\mathcal{X}$-manifold $M_1^\dagger \circ M_2$ as $u \to \mathrm{id}_B$. Thus, by continuity of $\zeta$, the operators $\widehat{u}$ converge to the identity in the weak operator (hence strong operator) topology as $u \to 1$.
\end{proof}

Next, we construct the Hamiltonian, generating time translations. Note that the operators $\widehat{C}_B(\beta)$ of evolution along a cylinder of length $\beta > 0$ form a semigroup of self-adjoint operators for $\beta > 0$. We would like to conclude that
\begin{equation}
    \widehat{C}_B(\beta) = e^{- \beta \widehat{H}_B},
\end{equation}
for some self-adjoint, bounded-below Hamiltonian operator $\widehat{H}_B$. By the Hille--Yosida theorem, it suffices to prove the following.

\begin{proposition}[Corollary 4 of \cite{Colafranceschi:2023urj}]\label{prop:approx_identity}
    The operators $\widehat{C}_B(\beta)$ converge to the identity in the strong operator topology as $\beta \to 0$.
\end{proposition}

\begin{proof}
    By Proposition \ref{prop:continuity_of_morphisms}, the operators $\widehat{C}_B(\beta)$ are continuous in $\beta$ in the norm topology. For rational $\beta$, we have,
    \begin{equation}
        \widehat{C}_B(\beta) = \widehat{C}_B(1)^{\beta},
    \end{equation}
    by the semigroup property, and for all $\beta$ by norm continuity. By the spectral theorem, we may write $\widehat{C}_B(\beta)$ in terms of its spectral projectors $\widehat{P}_\lambda$ of $\widehat{C}_B(1)$ as,
    \begin{equation}
        \widehat{C}_B(\beta) = \int d\widehat{P}_\lambda\ \lambda^\beta.
    \end{equation}
    By dominated convergence, we have,
    \begin{equation}
        \lvert\lvert \widehat{C}_B(\beta) \ket{\psi} - \ket{\psi} \rvert\rvert^2 = \int \braket{\psi | d\widehat{P}_\lambda | \psi} \lvert  \lambda^\beta - 1\rvert^2 \to \braket{\psi | \widehat{P}_{\ker(\widehat{C}_B(\beta))} | \psi},
    \end{equation}
    where $\widehat{P}_{\ker(\widehat{C}_B(\beta))}$ is the projector onto the mutual kernel of $\widehat{C}_B(\beta)$. Thus, we are done if we can show that this kernel vanishes.

    By self-adjointness, it suffices to show instead that the mutual range of $\widehat{C}_B(\beta)$ is dense in $\CH_B$. This follows because every manifold state $\ket{M}$ can be written as,
    \begin{equation}
        \ket{M} = \widehat{C}_B(\beta) \ket{M_\mathrm{cut}},
    \end{equation}
    for sufficiently small $\beta > 0$, where $M_\mathrm{cut}$ is obtained from $M$ by removing the cylinder $C_B(\beta)$. But the linear span of the manifold states is dense in $\CH_B$ by construction.
\end{proof}

\begin{corollary}\label{cor:Hamiltonian}
    There exists a self-adjoint, bounded-below Hamiltonian operator $\widehat{H}_B$ on $\mathcal{H}_B$ such that we have $\widehat{C}_B(\beta) = e^{- \beta \widehat{H}_B}$.
\end{corollary}

We will need the following slight generalization of the above, where we allow the cylinders to act on one component of a multi-boundary universal Hilbert space. The proof is directly analogous to the proof of Proposition \ref{prop:approx_identity}.

\begin{proposition}\label{prop:approx_identities_in_two_sided_reps}
    Fix $A, B \in \BordX$. The operators $\CH_{B \sqcup \overline{A}} \to \CH_{B \sqcup \overline{A}}$ defined by partial composition with $C_B(\beta)$ along $B$ or $C_{\overline{A}}(\beta)$ along $\overline{A}$ converge to the identity in the strong operator topology as $\beta \to 0$.
\end{proposition}

\subsection{A thermal bound}
\label{app:thermal_bound}

We now prove a thermal bound, comparing the thermal partition functions as computed by $\zeta$ and as computed by the Hilbert space trace of $e^{- \beta \widehat{H}_B}$ over the universal one-boundary Hilbert space $\CH_B$.

\begin{proposition}\label{prop:thermal_bound}
    Let $B \in \BordX$ be a closed $(d-1)$-manifold, and consider the Hamiltonian $\widehat{H}_B$ on $\CH_B$ as defined in Corollary \ref{cor:Hamiltonian}. We have the bound
    \begin{equation}\label{eq:thermal_bound_appendix}
        \mathrm{tr}_{\mathcal{H}_B}\big(e^{- \beta \widehat{H}_B}\big) \leq \zeta\left(S^1_\beta \times B  \right),
    \end{equation}
    where the manifold $S^1_\beta \times B$ is equipped with an $\mathcal{X}$-structure via,
    \begin{equation}
        S^1_\beta \times B \defined \tr_\Bord(C_B(\beta)).
    \end{equation}
    In particular, $S^1_\beta$ is equipped with the anti-periodic spin structure when necessary. Moreover, we have equality in \eqref{eq:thermal_bound_appendix} if and only if the map
    \begin{equation}
        \sqcup : \mathcal{H}_B \otimes \mathcal{H}_{\overline{B}} \to \mathcal{H}_{B \sqcup \overline{B}},
    \end{equation}
    is an isomorphism.
\end{proposition}

\begin{proof}
    We have,
    \begin{equation}\label{eq:thermal_bound_intermezzo}
        \zeta\left( S^1_\beta \times B \right) = \braket{C_B(\beta/2) | C_B(\beta/2)},
    \end{equation}
    where the inner product is taken in $\mathcal{H}_{B \sqcup \overline{B}}$. Let $\widehat{P}_{\mathrm{fact}}$ denote the orthogonal projector onto the factorizing subspace $\mathcal{H}_B \otimes \mathcal{H}_{\overline{B}} \subset \mathcal{H}_{B \sqcup \overline{B}}$. Since $\widehat{P}_{\mathrm{fact}} \leq 1$, we have
    \begin{equation}
        \zeta\left( S^1_\beta \times B \right) \geq \braket{C_B(\beta/2) | \widehat{P}_{\mathrm{fact}} | C_B(\beta/2)},
    \end{equation}
    with equality if and only if $\ket{C_B(\beta/2)} \in \mathcal{H}_B \otimes \mathcal{H}_{\overline{B}}$.
    
    Now, note that the state,
    \begin{equation}
        \widehat{P}_{\mathrm{fact}} \ket{ C_B(\beta/2)} \in \mathcal{H}_B \otimes \mathcal{H}_{\overline{B}},
    \end{equation}
    corresponds to the Hilbert-Schmidt operator,
    \begin{equation}
        \widehat{C}_B(\beta/2) : \CH_{B} \to \CH_B.
    \end{equation}
    So, continuing from \eqref{eq:thermal_bound_intermezzo}, we have,
    \begin{align}
        \zeta\left( S^1_\beta \times B \right) &\geq \braket{C_B(\beta/2) | \widehat{P}_{\mathrm{fact}}^2 | C_B(\beta/2)} \\
        &= \tr_{\CH_B}\big(\widehat{C}_B(\beta/2)^2\big) \\
        &= \mathrm{tr}_{\mathcal{H}_B}\big(e^{- \beta \widehat{H}_B}\big),
    \end{align}
    as desired, again with equality if and only if $\ket{C_B(\beta/2)} \in \mathcal{H}_B \otimes \mathcal{H}_{\overline{B}}$.
    
    It remains to show that
    \begin{equation}
        \big(\ket{C_B(\beta_0/2)} \in \mathcal{H}_B \otimes \mathcal{H}_{\overline{B}}\big) \implies \big( \mathcal{H}_B \otimes \mathcal{H}_{\overline{B}} = \mathcal{H}_{B \sqcup \overline{B}} \big),
    \end{equation}
    for some fixed $\beta_0$. First of all, note that as soon as the cylinder state $\ket{C_B(\beta_0/2)}$ factorizes for some $\beta_0$, the cylinder states $\ket{C_B(\beta)}$ automatically factorize for all $\beta$. To see this, note that for $\beta > \beta_0$, we have,
    \begin{equation}
        \ket{C_B(\beta/2)} = C_B(\beta/2 - \beta_0/2) \circ_B \ket{C_B(\beta_0/2)} \in \mathcal{H}_B \otimes \mathcal{H}_{\overline{B}},
    \end{equation}
    since the partial composition of anything with a factorizing state must factorize. For $\beta < \beta_0$, note that we have,
    \begin{equation}
        C_B(\beta_0/2 - \beta/2) \circ_B (1 - \widehat{P}_{\mathrm{fact}}) \ket{C_B(\beta/2)} = (1 - \widehat{P}_{\mathrm{fact}}) \ket{C_B(\beta_0/2)} = 0,
    \end{equation}
    using that $\ket{C_B(\beta_0/2)}$ factorizes and that the self-adjoint operator of partial composition with $C_B(\beta_0/2 - \beta/2)$ preserves the factorizing subspace, hence commutes with $\widehat{P}_\mathrm{fact}$. But as we saw in the proof of Proposition \ref{prop:approx_identity} (or, strictly speaking, the analogous argument which would be made in proving Proposition \ref{prop:approx_identities_in_two_sided_reps}), the cylinder operators have no kernel, which implies that $\ket{C_B(\beta/2)}$ factorizes itself.

    Now, let $\ket{\psi} \in \mathcal{H}_{B \sqcup \overline{B}}$ be arbitrary. We compute, using Proposition \ref{prop:approx_identities_in_two_sided_reps},
    \begin{equation}
        \ket{\psi} = \lim_{\beta \to 0} C_B(\beta/2) \circ_B \ket{\psi} = \lim_{\beta \to 0} \overline{\psi}^\dagger \circ_{\overline{B}} \ket{C_B(\beta/2)} \in \mathcal{H}_B \otimes \mathcal{H}_{\overline{B}},
    \end{equation}
    as desired, having used Corollary \ref{cor:composition_is_jointly_continuous} to extend partial composition to the state $\ket{\psi}$. 
\end{proof}

\section{Universal property of the baby universe category}
\label{app:univ_property}

In this appendix, we provide a proof of the following universal property of the baby universe category $\CBU$. First, we need the notion of a $\mathcal{D}$-valued unitary QFT, for an arbitrary symmetric monoidal W*-tensor category $\mathcal{D}$.

\begin{definition}\label{defn:D_valued_QFT}
    Let $\BordX$ be a unitary bordism category, and let $\mathcal{D}$ be any rigidly generated symmetric monoidal W*-tensor category. A \textit{$\mathcal{D}$-valued unitary QFT} on $\BordX$ is a non-degenerate, unitary, symmetric monoidal functor,
    \begin{equation}
        \mathcal{Z} : \BordX \to \mathcal{D}.
    \end{equation}
    Non-degeneracy means that, for every $B \in \BordX$, we have
    \begin{equation}
        \mathcal{Z}(C_B(\beta)) \to \mathrm{id}_{\mathcal{Z}(B)}, \quad \beta \to 0,
    \end{equation}
    in the $\sigma$-strong topology.\footnote{By the fact that the cylinder operators $C_B(\beta)$ form a self-adjoint semigroup, we could also equivalently ask for convergence in the $\sigma$-weak topology.} By simplicity of the tensor unit in $\mathcal{D}$, $\mathcal{Z}(M)$ is just some complex number for each closed $d$-manifold $M$, which we refer to as the \textit{partition function of $Z$}.
\end{definition}

\begin{proposition}[Universal property of $\CBU$]\label{prop:univ_property_CBU}
    Let $\zeta$ be a partition function on a unitary bordism category $\BordX$, satisfying Axioms \ref{axiom:finiteness}-\ref{axiom:reflection_positivity}, and let $\CBU$ be its baby universe category. Let $\mathcal{D}$ be a rigidly generated symmetric W*-tensor category, and let $\mathcal{Z} : \BordX \to \mathcal{D}$ be a $\mathcal{D}$-valued unitary QFT with partition function $\zeta$. Then there exists a unique symmetric W*-tensor functor $\CBU \to \mathcal{D}$ making the following diagram commute.
    \begin{equation}\label{eq:univ_property_appendix}
    \begin{tikzcd}
        & \CBU \arrow[dr, dashed, "\exists !"] & \\
        \BordX \arrow[ur] \arrow[rr] & & \mathcal{D}
    \end{tikzcd}
    \end{equation}
\end{proposition}

\begin{proof}
    The result follows as a specific case of Lemma \ref{lem:multiplier_extension_for_CBU} below, provided we can show that the free extension of a $\mathcal{D}$-valued QFT with partition function $\zeta$ to $\CBUpre$ is bounded in operator norm. This is because symmetric monoidality is an algebraic constraint, which transfers automatically once we have a W*-functor making diagram \eqref{eq:univ_property_appendix} commute.

    Thus, we must show that a $\mathcal{D}$-valued QFT with partition function $\zeta$ is bounded in operator norm. In fact, we will show more: it is automatically isometric in operator norm. To see this, let $\tr_\mathcal{D}$ denote the canonical trace on $\mathcal{D}$ induced from the canonical positive trace on the symmetric unitary tensor category $\mathcal{D}^{\mathrm{f.d.}}$. Now, we have the following identity,
    \begin{equation}\label{eq:replica_moments}
        \lvert \lvert \mathcal{Z}(\mathcal{O}) \rvert \rvert = \lim_{n \to \infty} \sqrt[2n]{\tr_\mathcal{D} \mathcal{Z}\big((\mathcal{O}^\dagger \mathcal{O})^n\big)},
    \end{equation}
    expressing the operator norm as a limit of $L^p$ norms for $p = 2n$, where $n \in \mathbb{Z}_{> 0}$. This follows from continuity in the $L^p$ norms as $p \to \infty$ for any rigidly generated (hence atomic) symmetric W*-tensor category.
    
    Because unitary symmetric monoidal functors take unitary traces to unitary traces, we have
    \begin{equation}\label{eq:trace_of_Z_is_Z_of_trace}
        \tr_\mathcal{D} \mathcal{Z}\big((\mathcal{O}^\dagger \mathcal{O})^n\big) =\mathcal{Z}\big(\tr \big((\mathcal{O}^\dagger \circ \mathcal{O})^n \big)\big) =  \zeta\big(\tr \big((\mathcal{O}^\dagger \circ \mathcal{O})^n \big)\big),
    \end{equation}
    by our assumption that $\mathcal{Z}$ has partition function $\zeta$. But, by applying the same argument to $\CBUvN$, the operator norm $\lvert \lvert \mathcal{O} \rvert \rvert$ in $\CBUpre$ is also given by the $2n$-th root of \eqref{eq:trace_of_Z_is_Z_of_trace} as well. Thus, $\mathcal{Z} : \CBUpre \to \mathcal{D}$ is isometric in operator norm and non-degenerate by assumption, hence extends uniquely by Lemma \ref{lem:multiplier_extension_for_CBU}.
\end{proof}

\begin{lemma}\label{lem:multiplier_extension_for_CBU}
    Let $\zeta$ be a partition function on a unitary bordism category $\BordX$, satisfying Axioms \ref{axiom:finiteness}-\ref{axiom:reflection_positivity}, and let $\CBUpre$ be its pre-baby universe category. Let $\mathcal{Z} : \CBUpre \to \mathcal{D}$ be a $\dagger$-functor valued in some W*-category $\mathcal{D}$, which is norm-bounded and non-degenerate in the same sense as Definition \ref{defn:D_valued_QFT}. Then there exists a unique W*-functor $\CBUvN \to \mathcal{D}$ making the following diagram commute.\footnote{In other words, $\CBUvN$ is the multiplier category of $\CBUpre$, or equivalently, of its C*-completion (see \cite{bunke2020additive}).}
    \begin{equation}
    \begin{tikzcd}
        & \CBUvN \arrow[dr, dashed, "\exists !"] & \\
        \CBUpre \arrow[ur] \arrow[rr] & & \mathcal{D}
    \end{tikzcd}
    \end{equation}
    Moreover, if $\mathcal{D}$ is Cauchy complete, the same conclusion holds with $\CBUvN$ replaced with $\CBU$.
\end{lemma}

\begin{proof}
     To see this, recall the following standard fact \cite{TakesakiI}: any non-degenerate norm-bounded $*$-homomorphism,
     \begin{equation}
         \mathcal{K}(\CH) \to \mathcal{A},
     \end{equation}
     from the C*-algebra $\mathcal{K}(\CH)$ of compact operators on some Hilbert space to an arbitrary von Neumann algebra $\mathcal{A}$ extends uniquely to a normal $*$-algebra map,
     \begin{equation}
         \mathcal{B}(\CH) \to \mathcal{A}.
     \end{equation}
     This is because the representation theory of $\mathcal{K}(\CH)$ is very simple, and every representation is an amplification of $\CH$, i.e., a possibly infinite direct sum of copies of $\CH$ itself.

     Now, to extend our given $\dagger$-functor $\mathcal{Z} : \CBUpre \to \mathcal{D}$, we first note that $\CBUpre$ is dense in $\CBUHS$, which is itself dense in its norm closure, the compact ideal $\mathcal{K}(\CBUvN)$. By norm-boundedness, $\mathcal{Z}$ extends uniquely to a $\dagger$-functor $\mathcal{K}(\CBUvN) \to \mathcal{D}$. By Proposition \ref{prop:approx_identities_in_two_sided_reps}, the cylinder morphisms $C_B(\beta)$ form a family of approximate identities for $\mathcal{K}(\CBUvN)$, and so our assumption of non-degeneracy of $\mathcal{Z}$ implies that $\mathcal{K}(\CBUvN) \to \mathcal{D}$ is non-degenerate. Applying the extension result described above block diagonally to the atomic W*-category $\CBUvN$, we obtain our desired extension. Finally, when $\mathcal{D}$ is Cauchy complete, the extension to $\CBU$ follows from the universal property of W*-Cauchy completion.
\end{proof}

\newpage
\bibliographystyle{JHEP}
\bibliography{ref}

@article{Colafranceschi:2023urj,
    author = "Colafranceschi, Eugenia and Dong, Xi and Marolf, Donald and Wang, Zhencheng",
    title = "{Algebras and Hilbert spaces from gravitational path integrals. Understanding Ryu-Takayanagi/HRT as entropy without AdS/CFT}",
    eprint = "2310.02189",
    archivePrefix = "arXiv",
    primaryClass = "hep-th",
    doi = "10.1007/JHEP10(2024)063",
    journal = "JHEP",
    volume = "10",
    pages = "063",
    year = "2024"
}

@article{hirsch2016smoothings,
  title={Smoothings of Piecewise Linear Manifolds.(AM-80)},
  author={Hirsch, Morris W and Mazur, Barry},
  year={2016},
  publisher={Princeton university press}
}

@article{khovanov2020universal,
  title={Universal construction of topological theories in two dimensions},
  author={Khovanov, Mikhail},
  journal={arXiv preprint arXiv:2007.03361},
  year={2020}
}

@article{Penneys2018UnitaryDF,
  title={Unitary dual functors for unitary multitensor categories},
  author={David Penneys},
  journal={Higher Structures},
  year={2018},
  url={https://api.semanticscholar.org/CorpusID:119269075},
  eprint={1808.00323},
  archivePrefix={arXiv},
  primaryClass={math.QA}
}

@article{Freed:2012bs,
    author = "Freed, Daniel S. and Teleman, Constantin",
    title = "{Relative quantum field theory}",
    eprint = "1212.1692",
    archivePrefix = "arXiv",
    primaryClass = "hep-th",
    doi = "10.1007/s00220-013-1880-1",
    journal = "Commun. Math. Phys.",
    volume = "326",
    pages = "459--476",
    year = "2014"
}

@article{lashof1963poincare,
  title={Poincar{\'e} duality and cobordism},
  author={Lashof, Richard},
  journal={Transactions of the American Mathematical Society},
  pages={257--277},
  year={1963},
  publisher={JSTOR}
}

@article{grady2026higher,
  title={Higher categories of bordisms with geometric structures},
  author={Grady, Daniel and Pavlov, Dmitri},
  journal={arXiv preprint arXiv:2605.03453},
  year={2026}
}

@article{Marolf:2024jze,
    author = "Marolf, Donald",
    title = "{On the nature of ensembles from gravitational path integrals}",
    eprint = "2407.04625",
    archivePrefix = "arXiv",
    primaryClass = "hep-th",
    doi = "10.1088/1751-8121/ad9cd4",
    journal = "J. Phys. A",
    volume = "58",
    number = "3",
    pages = "035204",
    year = "2025"
}

@article{Chen:2025fwp,
    author = "Chen, Hong Zhe",
    title = "{Observers seeing gravitational Hilbert spaces: abstract sources for an abstract path integral}",
    eprint = "2505.15892",
    archivePrefix = "arXiv",
    primaryClass = "hep-th",
    doi = "10.1007/JHEP10(2025)139",
    journal = "JHEP",
    volume = "10",
    pages = "139",
    year = "2025"
}

@article{Harlow:2023hjb,
    author = "Harlow, Daniel and Numasawa, Tokiro",
    title = "{Gauging spacetime inversions in quantum gravity}",
    eprint = "2311.09978",
    archivePrefix = "arXiv",
    primaryClass = "hep-th",
    reportNumber = "MIT-CTP/5647",
    doi = "10.1007/JHEP01(2026)098",
    journal = "JHEP",
    volume = "01",
    pages = "098",
    year = "2026"
}

@article{stehouwer2024spin,
  title={The spin-statistics theorem for topological quantum field theories},
  author={Stehouwer, Luuk},
  journal={Communications in Mathematical Physics},
  volume={405},
  number={11},
  pages={253},
  year={2024},
  publisher={Springer}
}

@article{Freed:2016rqq,
    author = "Freed, Daniel S. and Hopkins, Michael J.",
    title = "{Reflection positivity and invertible topological phases}",
    eprint = "1604.06527",
    archivePrefix = "arXiv",
    primaryClass = "hep-th",
    doi = "10.2140/gt.2021.25.1165",
    journal = "Geom. Topol.",
    volume = "25",
    pages = "1165--1330",
    year = "2021"
}

@article{Rudelius:2020orz,
    author = "Rudelius, Tom and Shao, Shu-Heng",
    title = "{Topological Operators and Completeness of Spectrum in Discrete Gauge Theories}",
    eprint = "2006.10052",
    archivePrefix = "arXiv",
    primaryClass = "hep-th",
    doi = "10.1007/JHEP12(2020)172",
    journal = "JHEP",
    volume = "12",
    pages = "172",
    year = "2020"
}

@article{Heidenreich:2021xpr,
    author = "Heidenreich, Ben and McNamara, Jacob and Montero, Miguel and Reece, Matthew and Rudelius, Tom and Valenzuela, Irene",
    title = "{Non-invertible global symmetries and completeness of the spectrum}",
    eprint = "2104.07036",
    archivePrefix = "arXiv",
    primaryClass = "hep-th",
    reportNumber = "ACFI-T21-03",
    doi = "10.1007/JHEP09(2021)203",
    journal = "JHEP",
    volume = "09",
    pages = "203",
    year = "2021"
}

@article{Lin:2022rbf,
    author = "Lin, Henry W.",
    title = "{The bulk Hilbert space of double scaled SYK}",
    eprint = "2208.07032",
    archivePrefix = "arXiv",
    primaryClass = "hep-th",
    doi = "10.1007/JHEP11(2022)060",
    journal = "JHEP",
    volume = "11",
    pages = "060",
    year = "2022"
}

@article{DiUbaldo:2026rly,
    author = "Di Ubaldo, Gabriele and Iliesiu, Luca V. and Lin, Henry W. and Yan, Cynthia",
    title = "{Positivity of the gravitational path integral implies the axionic weak gravity conjecture}",
    eprint = "2605.05305",
    archivePrefix = "arXiv",
    primaryClass = "hep-th",
    reportNumber = "RIKEN-iTHEMS-Report-26",
    month = "5",
    year = "2026"
}

@article{Binder:2019zqc,
    author = "Binder, Damon J. and Rychkov, Slava",
    title = "{Deligne Categories in Lattice Models and Quantum Field Theory, or Making Sense of $O(N)$ Symmetry with Non-integer $N$}",
    eprint = "1911.07895",
    archivePrefix = "arXiv",
    primaryClass = "hep-th",
    reportNumber = "PUPT-2601",
    doi = "10.1007/JHEP04(2020)117",
    journal = "JHEP",
    volume = "04",
    pages = "117",
    year = "2020"
}

@article{Engelhardt:2026blp,
    author = "Engelhardt, Netta and Gesteau, Elliott",
    title = "{A Semiclassical Diagnostic for Spacetime Emergence}",
    eprint = "2605.06780",
    archivePrefix = "arXiv",
    primaryClass = "hep-th",
    reportNumber = "MIT-CTP/6035",
    month = "5",
    year = "2026"
}

@article{Akers:2022qdl,
    author = "Akers, Chris and Engelhardt, Netta and Harlow, Daniel and Penington, Geoff and Vardhan, Shreya",
    title = "{The black hole interior from non-isometric codes and complexity}",
    eprint = "2207.06536",
    archivePrefix = "arXiv",
    primaryClass = "hep-th",
    doi = "10.1007/JHEP06(2024)155",
    journal = "JHEP",
    volume = "06",
    pages = "155",
    year = "2024"
}

@article{Harlow:2018fse,
    author = "Harlow, Daniel",
    title = "{TASI Lectures on the Emergence of Bulk Physics in AdS/CFT}",
    eprint = "1802.01040",
    archivePrefix = "arXiv",
    primaryClass = "hep-th",
    doi = "10.22323/1.305.0002",
    journal = "PoS",
    volume = "TASI2017",
    pages = "002",
    year = "2018"
}

@article{Dymarsky:2026asf,
    author = "Dymarsky, Anatoly and Shapere, Alfred",
    title = "{Mass formula for topological boundary conditions from TQFT gravity}",
    eprint = "2602.00224",
    archivePrefix = "arXiv",
    primaryClass = "hep-th",
    month = "1",
    year = "2026"
}

@article{Chandra:2022bqq,
    author = "Chandra, Jeevan and Collier, Scott and Hartman, Thomas and Maloney, Alexander",
    title = "{Semiclassical 3D gravity as an average of large-c CFTs}",
    eprint = "2203.06511",
    archivePrefix = "arXiv",
    primaryClass = "hep-th",
    doi = "10.1007/JHEP12(2022)069",
    journal = "JHEP",
    volume = "12",
    pages = "069",
    year = "2022"
}

@article{Maldacena:2026jqd,
    author = "Maldacena, Juan and Maloney, Alexander and McPeak, Brian",
    title = "{Wormholes and the imaginary distance bound}",
    eprint = "2605.05336",
    archivePrefix = "arXiv",
    primaryClass = "hep-th",
    month = "5",
    year = "2026"
}

@article{McNamara:2022lrw,
    author = "McNamara, Jacob and Reece, Matthew",
    title = "{Reflections on Parity Breaking}",
    eprint = "2212.00039",
    archivePrefix = "arXiv",
    primaryClass = "hep-th",
    month = "11",
    year = "2022"
}

@article{Perez-Lona:2025add,
    author = "Perez-Lona, Alonso and Sharpe, Eric and Yu, Xingyang and Zhang, Hao",
    title = "{Total instanton restriction via multiverse interference: Noncompact gauge theories and ({\ensuremath{-}}1)-form symmetries}",
    eprint = "2508.00050",
    archivePrefix = "arXiv",
    primaryClass = "hep-th",
    doi = "10.1007/JHEP05(2026)214",
    journal = "JHEP",
    volume = "05",
    pages = "214",
    year = "2026"
}

@article{Sharpe:2023lfk,
    author = "Sharpe, Eric",
    title = "{Dilaton shifts, probability measures, and decomposition}",
    eprint = "2312.08438",
    archivePrefix = "arXiv",
    primaryClass = "hep-th",
    doi = "10.1088/1751-8121/ad8196",
    journal = "J. Phys. A",
    volume = "57",
    number = "44",
    pages = "445401",
    year = "2024"
}

@article{Krein1949Duality,
  author  = {Krein, M. G.},
  title   = {A principle of duality for bicompact groups and quadratic block algebras},
  journal = {Doklady Akademii Nauk SSSR. Novaya Seriya},
  volume  = {69},
  number  = {6},
  year    = {1949},
  pages   = {725--728},
  note    = {In Russian}
}

@article{tannaka1939dualitatssatz,
  title={{\"U}ber den Dualit{\"a}tssatz der nichtkommutativen topologischen Gruppen},
  author={Tannaka, Tadao},
  journal={Tohoku Mathematical Journal, First Series},
  volume={45},
  pages={1--12},
  year={1939},
  publisher={Mathematical Institute, Tohoku University}
}

@article{milne2013motives,
  title={Motives—Grothendieck’s dream},
  author={Milne, James S},
  journal={Open problems and surveys of contemporary mathematics},
  volume={6},
  pages={325--342},
  year={2013}
}

@article{Czech:2023zmq,
    author = "Czech, Bartlomiej and de Boer, Jan and Esp{\'\i}ndola, Ricardo and Najian, Bahman and van der Heijden, Jeremy and Zukowski, Claire",
    title = "{Changing states in holography: From modular Berry curvature to the bulk symplectic form}",
    eprint = "2305.16384",
    archivePrefix = "arXiv",
    primaryClass = "hep-th",
    doi = "10.1103/PhysRevD.108.066003",
    journal = "Phys. Rev. D",
    volume = "108",
    number = "6",
    pages = "066003",
    year = "2023"
}

@article{Belin:2018fxe,
    author = "Belin, Alexandre and Lewkowycz, Aitor and S{\'a}rosi, G{\'a}bor",
    title = "{The boundary dual of the bulk symplectic form}",
    eprint = "1806.10144",
    archivePrefix = "arXiv",
    primaryClass = "hep-th",
    doi = "10.1016/j.physletb.2018.10.071",
    journal = "Phys. Lett. B",
    volume = "789",
    pages = "71--75",
    year = "2019"
}

@article{Marolf:2017kvq,
    author = "Marolf, Donald and Parrikar, Onkar and Rabideau, Charles and Izadi Rad, Ali and Van Raamsdonk, Mark",
    title = "{From Euclidean Sources to Lorentzian Spacetimes in Holographic Conformal Field Theories}",
    eprint = "1709.10101",
    archivePrefix = "arXiv",
    primaryClass = "hep-th",
    doi = "10.1007/JHEP06(2018)077",
    journal = "JHEP",
    volume = "06",
    pages = "077",
    year = "2018"
}

@article{Bak:2017rpp,
    author = "Bak, Dongsu and Trivella, Andrea",
    title = "{Quantum Information Metric on $\mathbb{R} \times S^{d-1}$}",
    eprint = "1707.05366",
    archivePrefix = "arXiv",
    primaryClass = "hep-th",
    doi = "10.1007/JHEP09(2017)086",
    journal = "JHEP",
    volume = "09",
    pages = "086",
    year = "2017"
}

@article{Lashkari:2015hha,
    author = "Lashkari, Nima and Van Raamsdonk, Mark",
    title = "{Canonical Energy is Quantum Fisher Information}",
    eprint = "1508.00897",
    archivePrefix = "arXiv",
    primaryClass = "hep-th",
    doi = "10.1007/JHEP04(2016)153",
    journal = "JHEP",
    volume = "04",
    pages = "153",
    year = "2016"
}

@article{Miyaji:2015woj,
    author = "Miyaji, Masamichi and Numasawa, Tokiro and Shiba, Noburo and Takayanagi, Tadashi and Watanabe, Kento",
    title = "{Distance between Quantum States and Gauge-Gravity Duality}",
    eprint = "1507.07555",
    archivePrefix = "arXiv",
    primaryClass = "hep-th",
    reportNumber = "YITP-15-62, IPMU15-0119, YITP-15-62, IPMU15-0119",
    doi = "10.1103/PhysRevLett.115.261602",
    journal = "Phys. Rev. Lett.",
    volume = "115",
    number = "26",
    pages = "261602",
    year = "2015"
}

@article{McNamara:2021cuo,
    author = "McNamara, Jacob",
    title = "{Gravitational Solitons and Completeness}",
    eprint = "2108.02228",
    archivePrefix = "arXiv",
    primaryClass = "hep-th",
    month = "8",
    year = "2021"
}

@article{stolz2012traces,
  title={Traces in monoidal categories},
  author={Stolz, Stephan and Teichner, Peter},
  journal={Transactions of the American Mathematical Society},
  volume={364},
  number={8},
  pages={4425--4464},
  year={2012}
}

@article{Witten:2025ayw,
    author = "Witten, Edward",
    title = "{Bras and kets in Euclidean path integrals}",
    eprint = "2503.12771",
    archivePrefix = "arXiv",
    primaryClass = "hep-th",
    doi = "10.4310/bpam.260113013520",
    journal = "Beijing J. Pure Appl. Math.",
    volume = "3",
    number = "1",
    pages = "1--34",
    year = "2026"
}

@article{bunke2020additive,
  title={Additive C*-categories and K-theory},
  author={Bunke, Ulrich and Engel, Alexander},
  journal={arXiv preprint arXiv:2010.14830},
  year={2020}
}

@article{Harlow:2026hky,
    author = "Harlow, Daniel",
    title = "{Observers, $\alpha$-parameters, and the Hartle-Hawking state}",
    eprint = "2602.03835",
    archivePrefix = "arXiv",
    primaryClass = "hep-th",
    reportNumber = "MIT-CTP/5900",
    month = "2",
    year = "2026"
}

@article{Atiyah1988TQFT,
  author  = {Atiyah, Michael F.},
  title   = {Topological quantum field theory},
  journal = {Publications Mathématiques de l'IHÉS},
  volume  = {68},
  pages   = {175--186},
  year    = {1988},
  doi     = {10.1007/BF02698547}
}

@incollection{segal1988definition,
  title={The definition of conformal field theory},
  author={Segal, Graeme B},
  booktitle={Differential geometrical methods in theoretical physics},
  pages={165--171},
  year={1988},
  publisher={Springer}
}

@article{ayala2017local,
  title={Local structures on stratified spaces},
  author={Ayala, David and Francis, John and Tanaka, Hiro Lee},
  journal={Advances in Mathematics},
  volume={307},
  pages={903--1028},
  year={2017},
  publisher={Elsevier}
}

@article{OsterwalderSchrader1973,
  author  = {Osterwalder, Konrad and Schrader, Robert},
  title   = {Axioms for Euclidean Green's Functions},
  journal = {Communications in Mathematical Physics},
  volume  = {31},
  number  = {2},
  pages   = {83--112},
  year    = {1973},
  doi     = {10.1007/BF01645738}
}

@article{OsterwalderSchrader1975,
  author  = {Osterwalder, Konrad and Schrader, Robert},
  title   = {Axioms for Euclidean Green's Functions II},
  journal = {Communications in Mathematical Physics},
  volume  = {42},
  number  = {3},
  pages   = {281--305},
  year    = {1975},
  doi     = {10.1007/BF01608978}
}

@book{StreaterWightman2000,
  author    = {Streater, Raymond F. and Wightman, Arthur S.},
  title     = {PCT, Spin and Statistics, and All That},
  publisher = {Princeton University Press},
  series    = {Princeton Landmarks in Mathematics and Physics},
  year      = {2000}
}

@article{quillen1997module,
  title={Module theory over nonunital rings},
  author={Quillen, Daniel},
  journal={Notes},
  volume={20},
  pages={1445--1459},
  year={1997}
}

@misc{johnsonfreyd2026rigidfirm,
  author       = {Johnson-Freyd, Theo},
  title        = {Rigid Firm Categories},
  howpublished = {Report for talk at \emph{Higher Structures from Symmetries in Quantum Field Theory}, Mathematisches Forschungsinstitut Oberwolfach. Available at \url{https://categorified.net/JohnsonFreyd-MFO2612.pdf}.},
  month        = mar,
  year         = {2026},
  note         = {Talk delivered March 18, 2026; joint work with Jake McNamara and David Reutter},
  url          = {https://categorified.net/JohnsonFreyd-MFO2612.pdf}
}

@article{Wightman1956,
  author  = {Wightman, Arthur S.},
  title   = {Quantum Field Theory in Terms of Vacuum Expectation Values},
  journal = {Physical Review},
  volume  = {101},
  number  = {2},
  pages   = {860--866},
  year    = {1956},
  doi     = {10.1103/PhysRev.101.860}
}

@article{Banerjee:2022pmw,
    author = "Banerjee, Anindya and Moore, Gregory W.",
    title = "{Comments on summing over bordisms in TQFT}",
    eprint = "2201.00903",
    archivePrefix = "arXiv",
    primaryClass = "hep-th",
    doi = "10.1007/JHEP09(2022)171",
    journal = "JHEP",
    volume = "09",
    pages = "171",
    year = "2022"
}

@article{dewitt1967quantum,
  title={Quantum theory of gravity. I. The canonical theory},
  author={DeWitt, Bryce S},
  journal={Physical Review},
  volume={160},
  number={5},
  pages={1113},
  year={1967},
  publisher={APS}
}

@article{Kuchar:1993ne,
    author = "Kuchar, Karel V.",
    title = "{Canonical quantum gravity}",
    eprint = "gr-qc/9304012",
    archivePrefix = "arXiv",
    reportNumber = "UU-REL-92-12-10",
    month = "4",
    year = "1993"
}

@article{Isham:1992ms,
    author = "Isham, C. J.",
    editor = "Ibort, L. A. and Rodriguez, M. A.",
    title = "{Canonical quantum gravity and the problem of time}",
    eprint = "gr-qc/9210011",
    archivePrefix = "arXiv",
    reportNumber = "IMPERIAL-TP-91-92-25",
    journal = "NATO Sci. Ser. C",
    volume = "409",
    pages = "157--287",
    year = "1993"
}

@article{ayala2015factorization,
  title={Factorization homology of topological manifolds},
  author={Ayala, David and Francis, John},
  journal={Journal of Topology},
  volume={8},
  number={4},
  pages={1045--1084},
  year={2015},
  publisher={Oxford University Press}
}

@incollection{deligne1982tannakian,
  title={Tannakian categories},
  author={Deligne, Pierre and Milne, James S},
  booktitle={Hodge cycles, motives, and Shimura varieties},
  pages={101--228},
  year={1982},
  publisher={Springer}
}

@article{deligne1990categories,
  title={Cat{\'e}gories tannakiennes, the grothendieck festschrift, vol. ii, 111--195},
  author={Deligne, Pierre},
  journal={Progr. Math},
  volume={87},
  year={1990}
}

@article{deligne2002categories,
  title={Cat{\'e}gories tensorielles},
  author={Deligne, Pierre},
  journal={Mosc. Math. J},
  volume={2},
  number={2},
  pages={227--248},
  year={2002}
}

@article{ferrer2024dagger,
  title={Dagger $ n $-categories},
  author={Ferrer, Giovanni and Hungar, Brett and Johnson-Freyd, Theo and Krulewski, Cameron and M{\"u}ller, Lukas and Penneys, David and Reutter, David and Scheimbauer, Claudia and Stehouwer, Luuk and Vuppulury, Chetan and others},
  eprint = "2403.01651",
  year={2024}
}

@article{ghez1985w,
  title={W$^*$-categories},
  author={Ghez, Paul and Lima, Ricardo and Roberts, John},
  journal={Pacific Journal of Mathematics},
  volume={120},
  number={1},
  pages={79--109},
  year={1985},
  publisher={Mathematical Sciences Publishers}
}

@article{Hamada:2021yxy,
    author = "Hamada, Yuta and Montero, Miguel and Vafa, Cumrun and Valenzuela, Irene",
    title = "{Finiteness and the swampland}",
    eprint = "2111.00015",
    archivePrefix = "arXiv",
    primaryClass = "hep-th",
    doi = "10.1088/1751-8121/ac6404",
    journal = "J. Phys. A",
    volume = "55",
    number = "22",
    pages = "224005",
    year = "2022"
}

@article{henriques2024completewcategories,
      title={Complete W*-categories}, 
      author={André Henriques and Nivedita and David Penneys},
      year={2024},
      eprint={2411.01678},
      archivePrefix={arXiv},
      primaryClass={math.OA},
      url={https://arxiv.org/abs/2411.01678}, 
}

@article{Marolf:2024adj,
    author = "Marolf, Donald and Zhang, Daiming",
    title = "{When left and right disagree: entropy and von Neumann algebras in quantum gravity with general AlAdS boundary conditions}",
    eprint = "2402.09691",
    archivePrefix = "arXiv",
    primaryClass = "hep-th",
    doi = "10.1007/JHEP08(2024)010",
    journal = "JHEP",
    volume = "08",
    pages = "010",
    year = "2024"
}

@article{Marolf:2020xie,
    author = "Marolf, Donald and Maxfield, Henry",
    title = "{Transcending the ensemble: baby universes, spacetime wormholes, and the order and disorder of black hole information}",
    eprint = "2002.08950",
    archivePrefix = "arXiv",
    primaryClass = "hep-th",
    doi = "10.1007/JHEP08(2020)044",
    journal = "JHEP",
    volume = "08",
    pages = "044",
    year = "2020"
}

@article{Coleman:1988cy,
    author = "Coleman, Sidney R.",
    title = "{Black holes as red herrings: Topological fluctuations and the loss of quantum coherence}",
    reportNumber = "HUTP-88/A008",
    doi = "10.1016/0550-3213(88)90110-1",
    journal = "Nucl. Phys. B",
    volume = "307",
    pages = "867--882",
    year = "1988"
}

@article{Giddings:1988cx,
    author = "Giddings, Steven B. and Strominger, Andrew",
    title = "{Loss of incoherence and determination of coupling constants in quantum gravity}",
    reportNumber = "HUTP-88/A006",
    doi = "10.1016/0550-3213(88)90109-5",
    journal = "Nucl. Phys. B",
    volume = "307",
    pages = "854--866",
    year = "1988"
}

@article{Giddings:1988wv,
    author = "Giddings, Steven B. and Strominger, Andrew",
    title = "{Baby Universes, Third Quantization and the Cosmological Constant}",
    reportNumber = "HUTP-88/A036",
    doi = "10.1016/0550-3213(89)90353-2",
    journal = "Nucl. Phys. B",
    volume = "321",
    pages = "481--508",
    year = "1989"
}

@article{Maxfield:2023mdj,
    author = "Maxfield, Henry",
    title = "{Counting states in a model of replica wormholes}",
    eprint = "2311.05703",
    archivePrefix = "arXiv",
    primaryClass = "hep-th",
    month = "11",
    year = "2023"
}

@article{Casini:2019kex,
    author = "Casini, Horacio and Huerta, Marina and Mag\'an, Javier M. and Pontello, Diego",
    title = "{Entanglement entropy and superselection sectors. Part I. Global symmetries}",
    eprint = "1905.10487",
    archivePrefix = "arXiv",
    primaryClass = "hep-th",
    doi = "10.1007/JHEP02(2020)014",
    journal = "JHEP",
    volume = "02",
    pages = "014",
    year = "2020"
}

@article{Casini:2020rgj,
    author = "Casini, Horacio and Huerta, Marina and Magan, Javier M. and Pontello, Diego",
    title = "{Entropic order parameters for the phases of QFT}",
    eprint = "2008.11748",
    archivePrefix = "arXiv",
    primaryClass = "hep-th",
    doi = "10.1007/JHEP04(2021)277",
    journal = "JHEP",
    volume = "04",
    pages = "277",
    year = "2021"
}

@article{Casini:2021zgr,
    author = "Casini, Horacio and Magan, Javier M.",
    title = "{On completeness and generalized symmetries in quantum field theory}",
    eprint = "2110.11358",
    archivePrefix = "arXiv",
    primaryClass = "hep-th",
    doi = "10.1142/S0217732321300251",
    journal = "Mod. Phys. Lett. A",
    volume = "36",
    number = "36",
    pages = "2130025",
    year = "2021"
}

@article{Benedetti:2022zbb,
    author = "Benedetti, Valentin and Casini, Horacio and Magan, Javier M.",
    title = "{Generalized symmetries and Noether{\textquoteright}s theorem in QFT}",
    eprint = "2205.03412",
    archivePrefix = "arXiv",
    primaryClass = "hep-th",
    doi = "10.1007/JHEP08(2022)304",
    journal = "JHEP",
    volume = "08",
    pages = "304",
    year = "2022"
}

@article{Benedetti:2024dku,
    author = "Benedetti, Valentin and Casini, Horacio and Kawahigashi, Yasuyuki and Longo, Roberto and Magan, Javier M.",
    title = "{Modular invariance as completeness}",
    eprint = "2408.04011",
    archivePrefix = "arXiv",
    primaryClass = "hep-th",
    doi = "10.1103/PhysRevD.110.125004",
    journal = "Phys. Rev. D",
    volume = "110",
    number = "12",
    pages = "125004",
    year = "2024"
}

@article{Casini:2025lfn,
    author = "Casini, Horacio and Magan, Javier M.",
    title = "{A generalization of the DHR theorem for higher form symmetries}",
    eprint = "2511.21810",
    archivePrefix = "arXiv",
    primaryClass = "hep-th",
    month = "11",
    year = "2025"
}

@article{Shao:2025mfj,
    author = "Shao, Shu-Heng and Sorce, Jonathan and Srivastava, Manu",
    title = "{Additivity, Haag duality, and non-invertible symmetries}",
    eprint = "2503.20863",
    archivePrefix = "arXiv",
    primaryClass = "hep-th",
    reportNumber = "MIT-CTP/5853, YITP-SB-2025-06",
    doi = "10.1007/JHEP08(2025)009",
    journal = "JHEP",
    volume = "08",
    pages = "009",
    year = "2025"
}

@article{Harlow:2025cqc,
    author = "Harlow, Daniel and Shao, Shu-Heng and Sorce, Jonathan and Srivastava, Manu",
    title = "{Disjoint additivity and local quantum physics}",
    eprint = "2509.03589",
    archivePrefix = "arXiv",
    primaryClass = "hep-th",
    reportNumber = "MIT-CTP/5901",
    doi = "10.1007/JHEP06(2026)025",
    journal = "JHEP",
    volume = "06",
    pages = "025",
    year = "2026"
}

@article{McNamara:2020uza,
    author = "McNamara, Jacob and Vafa, Cumrun",
    title = "{Baby Universes, Holography, and the Swampland}",
    eprint = "2004.06738",
    archivePrefix = "arXiv",
    primaryClass = "hep-th",
    month = "4",
    year = "2020"
}

@article{McNamara:2019rup,
    author = "McNamara, Jacob and Vafa, Cumrun",
    title = "{Cobordism Classes and the Swampland}",
    eprint = "1909.10355",
    archivePrefix = "arXiv",
    primaryClass = "hep-th",
    month = "9",
    year = "2019"
}

@article{Harlow:2015lma,
    author = "Harlow, Daniel",
    title = "{Wormholes, Emergent Gauge Fields, and the Weak Gravity Conjecture}",
    eprint = "1510.07911",
    archivePrefix = "arXiv",
    primaryClass = "hep-th",
    doi = "10.1007/JHEP01(2016)122",
    journal = "JHEP",
    volume = "01",
    pages = "122",
    year = "2016"
}

@article{Harlow:2018tqv,
    author = "Harlow, Daniel and Jafferis, Daniel",
    title = "{The Factorization Problem in Jackiw-Teitelboim Gravity}",
    eprint = "1804.01081",
    archivePrefix = "arXiv",
    primaryClass = "hep-th",
    doi = "10.1007/JHEP02(2020)177",
    journal = "JHEP",
    volume = "02",
    pages = "177",
    year = "2020"
}

@article{DoplicherRoberts89,
    author = "Doplicher, S. and Roberts, J.E.",
    title = "{A new duality theory for compact groups}",
    doi = "10.1007/BF01388849",
    journal = "Inventiones mathematicae",
    volume = "98",
    pages = "157–218",
    year = "1989"
}

@article{Saad:2019lba,
    author = "Saad, Phil and Shenker, Stephen H. and Stanford, Douglas",
    title = "{JT gravity as a matrix integral}",
    eprint = "1903.11115",
    archivePrefix = "arXiv",
    primaryClass = "hep-th",
    month = "3",
    year = "2019"
}

@article{Kolchmeyer:2023gwa,
    author = "Kolchmeyer, David K.",
    title = "{von Neumann algebras in JT gravity}",
    eprint = "2303.04701",
    archivePrefix = "arXiv",
    primaryClass = "hep-th",
    doi = "10.1007/JHEP06(2023)067",
    journal = "JHEP",
    volume = "06",
    pages = "067",
    year = "2023"
}

@article{Chua:2023ios,
    author = "Chua, Wan Zhen and Jiang, Yikun",
    title = "{Hartle-Hawking state and its factorization in 3d gravity}",
    eprint = "2309.05126",
    archivePrefix = "arXiv",
    primaryClass = "hep-th",
    doi = "10.1007/JHEP03(2024)135",
    journal = "JHEP",
    volume = "03",
    pages = "135",
    year = "2024"
}

@article{Boruch:2024kvv,
    author = "Boruch, Jan and Iliesiu, Luca V. and Lin, Guanda and Yan, Cynthia",
    title = "{How the Hilbert space of two-sided black holes factorises}",
    eprint = "2406.04396",
    archivePrefix = "arXiv",
    primaryClass = "hep-th",
    doi = "10.1007/JHEP06(2025)092",
    journal = "JHEP",
    volume = "06",
    pages = "092",
    year = "2025"
}

@article{Banerjee:2024fmh,
    author = "Banerjee, Souvik and Erdmenger, Johanna and Karl, Jonathan",
    title = "{Nonlocality induces isometry and factorisation in holography}",
    eprint = "2411.09616",
    archivePrefix = "arXiv",
    primaryClass = "hep-th",
    doi = "10.1103/tm45-3lz8",
    journal = "Phys. Rev. D",
    volume = "112",
    number = "2",
    pages = "L021902",
    year = "2025"
}

@article{Li:2024nft,
    author = "Li, Pan",
    title = "{Notes on the factorisation of the Hilbert space for two-sided black holes in higher dimensions}",
    eprint = "2410.23886",
    archivePrefix = "arXiv",
    primaryClass = "hep-th",
    doi = "10.1007/JHEP02(2025)060",
    journal = "JHEP",
    volume = "02",
    pages = "060",
    year = "2025"
}

@article{Balasubramanian:2024yxk,
    author = "Balasubramanian, Vijay and Craps, Ben and Hernandez, Juan and Khramtsov, Mikhail and Knysh, Maria",
    title = "{Factorization of the Hilbert space of eternal black holes in general relativity}",
    eprint = "2410.00091",
    archivePrefix = "arXiv",
    primaryClass = "hep-th",
    doi = "10.1007/JHEP01(2025)046",
    journal = "JHEP",
    volume = "01",
    pages = "046",
    year = "2025"
}

@article{Gesteau:2020wrk,
    author = "Gesteau, Elliott and Kang, Monica Jinwoo",
    title = "{Holographic baby universes: an observable story}",
    eprint = "2006.14620",
    archivePrefix = "arXiv",
    primaryClass = "hep-th",
    reportNumber = "CALT-TH-2020-029",
    month = "6",
    year = "2020"
}

@article{VanRaamsdonk:2010pw,
    author = "Van Raamsdonk, Mark",
    title = "{Building up spacetime with quantum entanglement}",
    eprint = "1005.3035",
    archivePrefix = "arXiv",
    primaryClass = "hep-th",
    doi = "10.1142/S0218271810018529",
    journal = "Gen. Rel. Grav.",
    volume = "42",
    pages = "2323--2329",
    year = "2010"
}

@article{Balasubramanian:2025zey,
    author = "Balasubramanian, Vijay and Yildirim, Tom",
    title = "{The nonperturbative Hilbert space of quantum gravity with one boundary}",
    eprint = "2506.04319",
    archivePrefix = "arXiv",
    primaryClass = "hep-th",
    doi = "10.1007/JHEP03(2026)040",
    journal = "JHEP",
    volume = "03",
    pages = "040",
    year = "2026"
}

@article{Maldacena:2013xja,
    author = "Maldacena, Juan and Susskind, Leonard",
    title = "{Cool horizons for entangled black holes}",
    eprint = "1306.0533",
    archivePrefix = "arXiv",
    primaryClass = "hep-th",
    doi = "10.1002/prop.201300020",
    journal = "Fortsch. Phys.",
    volume = "61",
    pages = "781--811",
    year = "2013"
}

@article{Gibbons:1978ac,
    author = "Gibbons, G. W. and Hawking, S. W. and Perry, M. J.",
    title = "{Path Integrals and the Indefiniteness of the Gravitational Action}",
    reportNumber = "PRINT-78-0375 (CAMBRIDGE)",
    doi = "10.1016/0550-3213(78)90161-X",
    journal = "Nucl. Phys. B",
    volume = "138",
    pages = "141--150",
    year = "1978"
}

@article{Gibbons:1976ue,
    author = "Gibbons, G. W. and Hawking, S. W.",
    title = "{Action Integrals and Partition Functions in Quantum Gravity}",
    reportNumber = "PRINT-76-0995 (CAMBRIDGE)",
    doi = "10.1103/PhysRevD.15.2752",
    journal = "Phys. Rev. D",
    volume = "15",
    pages = "2752--2756",
    year = "1977"
}

@article{Page:1993wv,
    author = "Page, Don N.",
    title = "{Information in black hole radiation}",
    eprint = "hep-th/9306083",
    archivePrefix = "arXiv",
    reportNumber = "ALBERTA-THY-24-93",
    doi = "10.1103/PhysRevLett.71.3743",
    journal = "Phys. Rev. Lett.",
    volume = "71",
    pages = "3743--3746",
    year = "1993"
}

@article{Penington:2019npb,
    author = "Penington, Geoffrey",
    title = "{Entanglement Wedge Reconstruction and the Information Paradox}",
    eprint = "1905.08255",
    archivePrefix = "arXiv",
    primaryClass = "hep-th",
    doi = "10.1007/JHEP09(2020)002",
    journal = "JHEP",
    volume = "09",
    pages = "002",
    year = "2020"
}

@article{Gaberdiel:2009rd,
    author = "Gaberdiel, Matthias R. and Volpato, Roberto",
    title = "{Higher genus partition functions of meromorphic conformal field theories}",
    eprint = "0903.4107",
    archivePrefix = "arXiv",
    primaryClass = "hep-th",
    doi = "10.1088/1126-6708/2009/06/048",
    journal = "JHEP",
    volume = "06",
    pages = "048",
    year = "2009"
}

@article{Gaberdiel:2010jf,
    author = "Gaberdiel, Matthias R. and Keller, Christoph A. and Volpato, Roberto",
    title = "{Genus Two Partition Functions of Chiral Conformal Field Theories}",
    eprint = "1002.3371",
    archivePrefix = "arXiv",
    primaryClass = "hep-th",
    doi = "10.4310/CNTP.2010.v4.n2.a2",
    journal = "Commun. Num. Theor. Phys.",
    volume = "4",
    pages = "295--364",
    year = "2010"
}

@article{muger2007abstract,
  title={Abstract duality theory for symmetric tensor *-categories},
  author={M{\"u}ger, Michael},
  journal={Philosophy of Physics},
  pages={865--922},
  year={2007},
  publisher={North-Holland}
}

@article{Codogni:2019sub,
    author = "Codogni, Giulio",
    title = {{Vertex algebras and Teichm{\"u}ller modular forms}},
    eprint = "1901.03079",
    archivePrefix = "arXiv",
    primaryClass = "math.AG",
    month = "1",
    year = "2019"
}

@article{Carpi:2026bfh,
    author = "Carpi, Sebastiano and Codogni, Giulio",
    title = {{Vertex operator algebras, partition functions and Teichm{\"u}ller modular forms}},
    eprint = "2605.26972",
    archivePrefix = "arXiv",
    primaryClass = "math.QA",
    month = "5",
    year = "2026"
}

@article{Balasubramanian:2025jeu,
    author = "Balasubramanian, Vijay and Yildirim, Tom",
    title = "{A Nonperturbative Toolkit for Quantum Gravity}",
    eprint = "2504.16986",
    archivePrefix = "arXiv",
    primaryClass = "hep-th",
    month = "4",
    year = "2025"
}

@article{Almheiri:2019psf,
    author = "Almheiri, Ahmed and Engelhardt, Netta and Marolf, Donald and Maxfield, Henry",
    title = "{The entropy of bulk quantum fields and the entanglement wedge of an evaporating black hole}",
    eprint = "1905.08762",
    archivePrefix = "arXiv",
    primaryClass = "hep-th",
    doi = "10.1007/JHEP12(2019)063",
    journal = "JHEP",
    volume = "12",
    pages = "063",
    year = "2019"
}

@article{Penington:2019kki,
    author = "Penington, Geoff and Shenker, Stephen H. and Stanford, Douglas and Yang, Zhenbin",
    title = "{Replica wormholes and the black hole interior}",
    eprint = "1911.11977",
    archivePrefix = "arXiv",
    primaryClass = "hep-th",
    doi = "10.1007/JHEP03(2022)205",
    journal = "JHEP",
    volume = "03",
    pages = "205",
    year = "2022"
}

@article{Almheiri:2019qdq,
    author = "Almheiri, Ahmed and Hartman, Thomas and Maldacena, Juan and Shaghoulian, Edgar and Tajdini, Amirhossein",
    title = "{Replica Wormholes and the Entropy of Hawking Radiation}",
    eprint = "1911.12333",
    archivePrefix = "arXiv",
    primaryClass = "hep-th",
    doi = "10.1007/JHEP05(2020)013",
    journal = "JHEP",
    volume = "05",
    pages = "013",
    year = "2020"
}

@article{Balasubramanian:2020jhl,
    author = "Balasubramanian, Vijay and Kar, Arjun and Ross, Simon F. and Ugajin, Tomonori",
    title = "{Spin structures and baby universes}",
    eprint = "2007.04333",
    archivePrefix = "arXiv",
    primaryClass = "hep-th",
    doi = "10.1007/JHEP09(2020)192",
    journal = "JHEP",
    volume = "09",
    pages = "192",
    year = "2020"
}

@article{Maldacena:2001kr,
    author = "Maldacena, Juan Martin",
    title = "{Eternal black holes in anti-de Sitter}",
    eprint = "hep-th/0106112",
    archivePrefix = "arXiv",
    reportNumber = "NSF-ITP-01-59",
    doi = "10.1088/1126-6708/2003/04/021",
    journal = "JHEP",
    volume = "04",
    pages = "021",
    year = "2003"
}

@article{Polchinski:2003bq,
    author = "Polchinski, Joseph",
    editor = "Baer, H. and Belyaev, A.",
    title = "{Monopoles, duality, and string theory}",
    eprint = "hep-th/0304042",
    archivePrefix = "arXiv",
    doi = "10.1142/S0217751X0401866X",
    journal = "Int. J. Mod. Phys. A",
    volume = "19S1",
    pages = "145--156",
    year = "2004"
}

@article{Banks:2010zn,
    author = "Banks, Tom and Seiberg, Nathan",
    title = "{Symmetries and Strings in Field Theory and Gravity}",
    eprint = "1011.5120",
    archivePrefix = "arXiv",
    primaryClass = "hep-th",
    doi = "10.1103/PhysRevD.83.084019",
    journal = "Phys. Rev. D",
    volume = "83",
    pages = "084019",
    year = "2011"
}

@article{stolz2004elliptic,
  title={What is an elliptic object?},
  author={Stolz, Stephan and Teichner, Peter},
  journal={London Mathematical Society Lecture Note Series},
  volume={308},
  pages={247},
  year={2004},
  publisher={Cambridge University Press}
}

@misc{ayala2009geometriccobordismcategories,
      title={Geometric Cobordism Categories}, 
      author={David Ayala},
      year={2009},
      eprint={0811.2280},
      archivePrefix={arXiv},
      primaryClass={math.AT},
      url={https://arxiv.org/abs/0811.2280}, 
}

@article{lurie2008classification,
  title={On the classification of topological field theories},
  author={Lurie, Jacob},
  journal={Current developments in mathematics},
  volume={2008},
  number={1},
  pages={129--280},
  year={2008},
  publisher={International Press of Boston}
}

@article{calaque2019note,
  title={A note on the $(\infty, n)$--category of cobordisms},
  author={Calaque, Damien and Scheimbauer, Claudia},
  journal={Algebraic \& Geometric Topology},
  volume={19},
  number={2},
  pages={533--655},
  year={2019},
  publisher={Mathematical Sciences Publishers}
}

@article{stolz2011supersymmetric,
  title={Supersymmetric field theories and generalized cohomology},
  author={Stolz, Stephan and Teichner, Peter},
  journal={arXiv preprint arXiv:1108.0189},
  year={2011}
}

@book{freed2019lectures,
  title={Lectures on field theory and topology},
  author={Freed, Daniel S},
  volume={133},
  year={2019},
  publisher={American Mathematical Soc.}
}

@article{Penington:2023dql,
    author = "Penington, Geoff and Witten, Edward",
    title = "{Algebras and States in JT Gravity}",
    eprint = "2301.07257",
    archivePrefix = "arXiv",
    primaryClass = "hep-th",
    month = "1",
    year = "2023"
}

@article{Ryu:2006bv,
    author = "Ryu, Shinsei and Takayanagi, Tadashi",
    title = "{Holographic derivation of entanglement entropy from AdS/CFT}",
    eprint = "hep-th/0603001",
    archivePrefix = "arXiv",
    reportNumber = "NSF-KITP-06-11, NSF-KITP-06-11",
    doi = "10.1103/PhysRevLett.96.181602",
    journal = "Phys. Rev. Lett.",
    volume = "96",
    pages = "181602",
    year = "2006"
}

@article{Kitaev:2000nmw,
    author = "Kitaev, Alexei",
    title = "{Unpaired Majorana fermions in quantum wires}",
    eprint = "cond-mat/0010440",
    archivePrefix = "arXiv",
    doi = "10.1070/1063-7869/44/10S/S29",
    journal = "Phys. Usp.",
    volume = "44",
    number = "10S",
    pages = "131--136",
    year = "2001"
}

@inproceedings{haagerup1979lp,
  title={Lp-spaces associated with an arbitrary von Neumann algebra},
  author={Haagerup, Uffe and others},
  booktitle={Algebres d’op{\'e}rateurs et leurs applications en physique math{\'e}matique (Proc. Colloq., Marseille, 1977)},
  volume={274},
  pages={175--184},
  year={1979}
}

@article{Friedan:2023vxx,
    author = "Friedan, Daniel",
    title = "{Global structure of euclidean quantum gravity}",
    eprint = "2306.00019",
    archivePrefix = "arXiv",
    primaryClass = "hep-th",
    month = "5",
    year = "2023"
}

@article{Gaiotto:2019xmp,
    author = "Gaiotto, Davide and Johnson-Freyd, Theo",
    title = "{Condensations in higher categories}",
    eprint = "1905.09566",
    archivePrefix = "arXiv",
    primaryClass = "math.CT",
    month = "5",
    year = "2019"
}

@article{Stephens:1993an,
    author = "Stephens, Christopher R. and 't Hooft, Gerard and Whiting, Bernard F.",
    title = "{Black hole evaporation without information loss}",
    eprint = "gr-qc/9310006",
    archivePrefix = "arXiv",
    reportNumber = "THU-93-20, UF-RAP-93-11",
    doi = "10.1088/0264-9381/11/3/014",
    journal = "Class. Quant. Grav.",
    volume = "11",
    pages = "621--648",
    year = "1994"
}

@article{Johnson-Freyd:2015fua,
    author = "Johnson-Freyd, Theo",
    title = "{Spin, statistics, orientations, unitarity}",
    eprint = "1507.06297",
    archivePrefix = "arXiv",
    primaryClass = "math-ph",
    doi = "10.2140/agt.2017.17.917",
    journal = "Algebr. Geom. Topol.",
    volume = "17",
    number = "2",
    pages = "917--956",
    year = "2017"
}

@article{Cheung:2024ypq,
    author = "Cheung, Clifford and Derda, Maria and Kim, Joon-Hwi and Nevoa, Vinicius and Rothstein, Ira and Shah, Nabha",
    title = "{Generalized symmetry in dynamical gravity}",
    eprint = "2403.01837",
    archivePrefix = "arXiv",
    primaryClass = "hep-th",
    reportNumber = "CALT-TH 2024-009",
    doi = "10.1007/JHEP10(2024)007",
    journal = "JHEP",
    volume = "10",
    pages = "007",
    year = "2024"
}

@article{deligne2007categorie,
  title={La cat{\'e}gorie des repr{\'e}sentations du groupe sym{\'e}trique St, lorsque t n’est pas un entier naturel},
  author={Deligne, Pierre},
  journal={Algebraic groups and homogeneous spaces},
  volume={19},
  pages={209--273},
  year={2007},
  publisher={Tata Institute of Fundamental Research Mumbai}
}

@article{Heckman:2024obe,
    author = "Heckman, Jonathan J. and McNamara, Jacob and Montero, Miguel and Sharon, Adar and Vafa, Cumrun and Valenzuela, Irene",
    title = "{Fate of stringy noninvertible symmetries}",
    eprint = "2402.00118",
    archivePrefix = "arXiv",
    primaryClass = "hep-th",
    reportNumber = "CERN-TH-2024-019",
    doi = "10.1103/PhysRevD.110.106001",
    journal = "Phys. Rev. D",
    volume = "110",
    number = "10",
    pages = "106001",
    year = "2024"
}

@article{Hubeny:2007xt,
    author = "Hubeny, Veronika E. and Rangamani, Mukund and Takayanagi, Tadashi",
    title = "{A Covariant holographic entanglement entropy proposal}",
    eprint = "0705.0016",
    archivePrefix = "arXiv",
    primaryClass = "hep-th",
    reportNumber = "DCPT-07-13, KUNS-2069",
    doi = "10.1088/1126-6708/2007/07/062",
    journal = "JHEP",
    volume = "07",
    pages = "062",
    year = "2007"
}

@article{Faulkner:2013ana,
    author = "Faulkner, Thomas and Lewkowycz, Aitor and Maldacena, Juan",
    title = "{Quantum corrections to holographic entanglement entropy}",
    eprint = "1307.2892",
    archivePrefix = "arXiv",
    primaryClass = "hep-th",
    doi = "10.1007/JHEP11(2013)074",
    journal = "JHEP",
    volume = "11",
    pages = "074",
    year = "2013"
}

@article{Engelhardt:2014gca,
    author = "Engelhardt, Netta and Wall, Aron C.",
    title = "{Quantum Extremal Surfaces: Holographic Entanglement Entropy beyond the Classical Regime}",
    eprint = "1408.3203",
    archivePrefix = "arXiv",
    primaryClass = "hep-th",
    doi = "10.1007/JHEP01(2015)073",
    journal = "JHEP",
    volume = "01",
    pages = "073",
    year = "2015"
}

@article{Harlow:2020bee,
    author = "Harlow, Daniel and Shaghoulian, Edgar",
    title = "{Global symmetry, Euclidean gravity, and the black hole information problem}",
    eprint = "2010.10539",
    archivePrefix = "arXiv",
    primaryClass = "hep-th",
    doi = "10.1007/JHEP04(2021)175",
    journal = "JHEP",
    volume = "04",
    pages = "175",
    year = "2021"
}

@article{tHooft:1993dmi,
    author = "'t Hooft, Gerard",
    title = "{Dimensional reduction in quantum gravity}",
    eprint = "gr-qc/9310026",
    archivePrefix = "arXiv",
    reportNumber = "THU-93-26",
    journal = "Conf. Proc. C",
    volume = "930308",
    pages = "284--296",
    year = "1993"
}

@article{Susskind:1994vu,
    author = "Susskind, Leonard",
    title = "{The World as a hologram}",
    eprint = "hep-th/9409089",
    archivePrefix = "arXiv",
    reportNumber = "SU-ITP-94-33",
    doi = "10.1063/1.531249",
    journal = "J. Math. Phys.",
    volume = "36",
    pages = "6377--6396",
    year = "1995"
}

@article{Bousso:2002ju,
    author = "Bousso, Raphael",
    title = "{The Holographic principle}",
    eprint = "hep-th/0203101",
    archivePrefix = "arXiv",
    reportNumber = "NSF-ITP-02-17, NSF-ITP-02-17",
    doi = "10.1103/RevModPhys.74.825",
    journal = "Rev. Mod. Phys.",
    volume = "74",
    pages = "825--874",
    year = "2002"
}

@article{Hawking:1987mz,
author = "Hawking, S. W.",
title = "{Quantum coherence down the wormhole}",
doi = "10.1016/0370-2693(87)90028-1",
journal = "Phys. Lett. B",
volume = "195",
pages = "337--343",
year = "1987"
}

@article{Hawking:1988ae,
author = "Hawking, S. W.",
title = "{Wormholes in space-time}",
doi = "10.1103/PhysRevD.37.904",
journal = "Phys. Rev. D",
volume = "37",
pages = "904--910",
year = "1988"
}

@article{Gubser:1998bc,
    author = "Gubser, S. S. and Klebanov, Igor R. and Polyakov, Alexander M.",
    title = "{Gauge theory correlators from noncritical string theory}",
    eprint = "hep-th/9802109",
    archivePrefix = "arXiv",
    reportNumber = "PUPT-1767",
    doi = "10.1016/S0370-2693(98)00377-3",
    journal = "Phys. Lett. B",
    volume = "428",
    pages = "105--114",
    year = "1998"
}

@article{Witten:1998qj,
    author = "Witten, Edward",
    title = "{Anti de Sitter space and holography}",
    eprint = "hep-th/9802150",
    archivePrefix = "arXiv",
    reportNumber = "IASSNS-HEP-98-15",
    doi = "10.4310/ATMP.1998.v2.n2.a2",
    journal = "Adv. Theor. Math. Phys.",
    volume = "2",
    pages = "253--291",
    year = "1998"
}

@article{Einstein:1935rr,
    author = "Einstein, Albert and Podolsky, Boris and Rosen, Nathan",
    title = "{Can quantum mechanical description of physical reality be considered complete?}",
    doi = "10.1103/PhysRev.47.777",
    journal = "Phys. Rev.",
    volume = "47",
    pages = "777--780",
    year = "1935"
}

@article{Blommaert:2021fob,
    author = "Blommaert, Andreas and Iliesiu, Luca V. and Kruthoff, Jorrit",
    title = "{Gravity factorized}",
    eprint = "2111.07863",
    archivePrefix = "arXiv",
    primaryClass = "hep-th",
    doi = "10.1007/JHEP09(2022)080",
    journal = "JHEP",
    volume = "09",
    pages = "080",
    year = "2022"
}

@article{Saad:2021uzi,
    author = "Saad, Phil and Shenker, Stephen H. and Yao, Shunyu",
    title = "{Comments on wormholes and factorization}",
    eprint = "2107.13130",
    archivePrefix = "arXiv",
    primaryClass = "hep-th",
    doi = "10.1007/JHEP10(2024)076",
    journal = "JHEP",
    volume = "10",
    pages = "076",
    year = "2024"
}

@article{Saad:2021rcu,
    author = "Saad, Phil and Shenker, Stephen H. and Stanford, Douglas and Yao, Shunyu",
    title = "{Wormholes without averaging}",
    eprint = "2103.16754",
    archivePrefix = "arXiv",
    primaryClass = "hep-th",
    doi = "10.1007/JHEP09(2024)133",
    journal = "JHEP",
    volume = "09",
    pages = "133",
    year = "2024"
}

@article{Gesteau:2024gzf,
    author = "Gesteau, Elliott and Marcolli, Matilde and McNamara, Jacob",
    title = "{Wormhole Renormalization: The gravitational path integral, holography, and a gauge group for topology change}",
    eprint = "2407.20324",
    archivePrefix = "arXiv",
    primaryClass = "hep-th",
    month = "7",
    year = "2024"
}

@article{Zhao:2026mpl,
    author = "Zhao, Ying",
    title = "{''It from Bit'': The Hartle-Hawking state and quantum mechanics for de Sitter observers}",
    eprint = "2602.05939",
    archivePrefix = "arXiv",
    primaryClass = "hep-th",
    reportNumber = "MIT-CTP/6000",
    month = "2",
    year = "2026"
}

@article{Almheiri:2019hni,
    author = "Almheiri, Ahmed and Mahajan, Raghu and Maldacena, Juan and Zhao, Ying",
    title = "{The Page curve of Hawking radiation from semiclassical geometry}",
    eprint = "1908.10996",
    archivePrefix = "arXiv",
    primaryClass = "hep-th",
    doi = "10.1007/JHEP03(2020)149",
    journal = "JHEP",
    volume = "03",
    pages = "149",
    year = "2020"
}

@article{Dong:2024tjx,
    author = "Dong, Xi and Kolanowski, Maciej and Liu, Xiaoyi and Marolf, Donald and Wang, Zhencheng",
    title = "{Null states and time evolution in a toy model of black hole dynamics}",
    eprint = "2405.04571",
    archivePrefix = "arXiv",
    primaryClass = "hep-th",
    doi = "10.1007/JHEP08(2024)199",
    journal = "JHEP",
    volume = "08",
    pages = "199",
    year = "2024"
}

@misc{McNamara:2024kitp,
  author       = {McNamara, Jacob},
  title        = {{Cobordism, ER = EPR, and the Sum Over Topologies}},
  howpublished         = {Talk at Spacetime and String Theory conference, Kavli Institute for Theoretical Physics.
                  Available at \url{https://online.kitp.ucsb.edu/online/strings-c24/mcnamara/}},
  year         = {2024}
  }

@article{Friedan:1986ua,
    author = "Friedan, Daniel and Shenker, Stephen H.",
    title = "{The Analytic Geometry of Two-Dimensional Conformal Field Theory}",
    reportNumber = "EFI-86-18A-CHICAGO",
    doi = "10.1016/0550-3213(87)90418-4",
    journal = "Nucl. Phys. B",
    volume = "281",
    pages = "509--545",
    year = "1987"
}

@article{Blanchet:1995TQFT,
    author = "Blanchet, Christian and Habegger, Nathan and Masbaum, Gregor and Vogel, Pierre",
    title = "{Topological quantum field theories derived from the Kauffman bracket}",
    journal = "Topology",
    volume = "34",
    number = "4",
    pages = "883--927",
    year = "1995",
    doi = "10.1016/0040-9383(94)00051-4"
}

@phdthesis{McNamara:2022xkg,
    author = "McNamara, Jacob Miles",
    title = "{The Kinematics of Quantum Gravity}",
    school = "Harvard U. (main)",
    year = "2022"
}

@article{Gelfand:1943normed,
    author = "Gel'fand, I. M. and Naimark, M. A.",
    title = "{On the imbedding of normed rings into the ring of operators in Hilbert space}",
    journal = "Mat. Sb.",
    volume = "12",
    number = "2",
    pages = "197--217",
    year = "1943"
}

@article{Segal:1947operator,
    author = "Segal, I. E.",
    title = "{Irreducible representations of operator algebras}",
    journal = "Bull. Amer. Math. Soc.",
    volume = "53",
    number = "2",
    pages = "73--88",
    year = "1947",
    doi = "10.1090/S0002-9904-1947-08742-5"
}

@article{Jafferis:2017tiu,
    author = "Jafferis, Daniel Louis",
    title = "{Bulk reconstruction and the Hartle-Hawking wavefunction}",
    eprint = "1703.01519",
    archivePrefix = "arXiv",
    primaryClass = "hep-th",
    month = "3",
    year = "2017"
}

@book{TakesakiI,
  author    = {Takesaki, Masamichi},
  title     = {Theory of Operator Algebras I},
  publisher = {Springer},
  address   = {New York},
  year      = {1979},
  doi       = {10.1007/978-1-4612-6188-9}
}

@article{Longo:1997yc,
    author = "Longo, Roberto and Roberts, John E.",
    title = "{A theory of dimension}",
    eprint = "funct-an/9604008",
    archivePrefix = "arXiv",
    journal = "K Theory",
    volume = "11",
    number = "2",
    pages = "103--159",
    year = "1997",
    doi = "10.1023/A:1007714415067"
}

@article{Torres:2025jcb,
    author = "Torres, Ethan and Yu, Xingyang",
    title = "{Symmetry topological field theory entanglement and holographic nonfactorization}",
    eprint = "2510.06319",
    archivePrefix = "arXiv",
    primaryClass = "hep-th",
    reportNumber = "CERN-TH-2025-194",
    doi = "10.1103/9wgn-ntnn",
    journal = "Phys. Rev. D",
    volume = "113",
    number = "10",
    pages = "106028",
    year = "2026"
}

@article{Witten:1999xp,
    author = "Witten, Edward and Yau, Shing-Tung",
    editor = "D'Hoker, Erik and Phong, Duong and Yau, Shing-Tung",
    title = "{Connectedness of the boundary in the AdS / CFT correspondence}",
    eprint = "hep-th/9910245",
    archivePrefix = "arXiv",
    doi = "10.4310/ATMP.1999.v3.n6.a1",
    journal = "Adv. Theor. Math. Phys.",
    volume = "3",
    pages = "1635--1655",
    year = "1999"
}

@book{stehouwer2024unitary,
  title={Unitary fermionic topological field theory},
  author={Stehouwer, Luuk},
  year={2024},
  publisher={Rheinische Friedrich-Wilhelms-Universitaet Bonn (Germany)}
}

@article{Wen:2019ylt,
    author = "Wen, Xueda and Wen, Xiao-Gang",
    title = "{Distinguish modular categories and 2+1D topological orders beyond modular data: Mapping class group of higher genus manifold}",
    eprint = "1908.10381",
    archivePrefix = "arXiv",
    primaryClass = "cond-mat.str-el",
    month = "8",
    year = "2019"
}

@article{Kaidi:2024cbx,
    author = "Kaidi, Justin and Tachikawa, Yuji and Yonekura, Kazuya",
    title = "{On non-supersymmetric heterotic branes}",
    eprint = "2411.04344",
    archivePrefix = "arXiv",
    primaryClass = "hep-th",
    reportNumber = "TU-1248, KYUSHU-HET-294",
    doi = "10.1007/JHEP03(2025)211",
    journal = "JHEP",
    volume = "03",
    pages = "211",
    year = "2025"
}

@article{Abdalla:2025gzn,
    author = "Abdalla, Ahmed I. and Antonini, Stefano and Iliesiu, Luca V. and Levine, Adam",
    title = "{The gravitational path integral from an observer{\textquoteright}s point of view}",
    eprint = "2501.02632",
    archivePrefix = "arXiv",
    primaryClass = "hep-th",
    doi = "10.1007/JHEP05(2025)059",
    journal = "JHEP",
    volume = "05",
    pages = "059",
    year = "2025"
}

@article{Harlow:2025pvj,
    author = "Harlow, Daniel and Usatyuk, Mykhaylo and Zhao, Ying",
    title = "{Quantum mechanics and observers for gravity in a closed universe}",
    eprint = "2501.02359",
    archivePrefix = "arXiv",
    primaryClass = "hep-th",
    reportNumber = "MIT-CTP/5824",
    doi = "10.1007/JHEP02(2026)108",
    journal = "JHEP",
    volume = "02",
    pages = "108",
    year = "2026"
}

@article{Usatyuk:2024isz,
    author = "Usatyuk, Mykhaylo and Zhao, Ying",
    title = "{Closed universes, factorization, and ensemble averaging}",
    eprint = "2403.13047",
    archivePrefix = "arXiv",
    primaryClass = "hep-th",
    doi = "10.1007/JHEP02(2025)052",
    journal = "JHEP",
    volume = "02",
    pages = "052",
    year = "2025"
}

@article{Usatyuk:2024mzs,
    author = "Usatyuk, Mykhaylo and Wang, Zi-Yue and Zhao, Ying",
    title = "{Closed universes in two dimensional gravity}",
    eprint = "2402.00098",
    archivePrefix = "arXiv",
    primaryClass = "hep-th",
    doi = "10.21468/SciPostPhys.17.2.051",
    journal = "SciPost Phys.",
    volume = "17",
    number = "2",
    pages = "051",
    year = "2024"
}

@article{Giddings:1987cg,
    author = "Giddings, Steven B. and Strominger, Andrew",
    title = "{Axion Induced Topology Change in Quantum Gravity and String Theory}",
    reportNumber = "HUTP-87-A067",
    doi = "10.1016/0550-3213(88)90446-4",
    journal = "Nucl. Phys. B",
    volume = "306",
    pages = "890--907",
    year = "1988"
}

@article{freedman2005universal,
  title={Universal manifold pairings and positivity},
  author={Freedman, Michael H and Kitaev, Alexei and Nayak, Chetan and Slingerland, Johannes K and Walker, Kevin and Wang, Zhenghan},
  journal={Geometry \& Topology},
  volume={9},
  number={4},
  pages={2303--2317},
  year={2005},
  publisher={Mathematical Sciences Publishers}
}

@article{Calegari:2008cw,
    author = "Calegari, Danny and Freedman, Michael and Walker, Kevin",
    title = "{Positivity of the universal pairing in 3 dimensions}",
    eprint = "0802.3208",
    archivePrefix = "arXiv",
    primaryClass = "math.GT",
    month = "2",
    year = "2008"
}

@article{Brennan:2024fgj,
    author = "Brennan, T. Daniel and Sun, Zhengdi",
    title = "{A SymTFT for continuous symmetries}",
    eprint = "2401.06128",
    archivePrefix = "arXiv",
    primaryClass = "hep-th",
    doi = "10.1007/JHEP12(2024)100",
    journal = "JHEP",
    volume = "12",
    pages = "100",
    year = "2024"
}

@article{Bonetti:2024cjk,
    author = "Bonetti, Federico and Del Zotto, Michele and Minasian, Ruben",
    title = "{SymTFTs for continuous non-Abelian symmetries}",
    eprint = "2402.12347",
    archivePrefix = "arXiv",
    primaryClass = "hep-th",
    doi = "10.1016/j.physletb.2025.140010",
    journal = "Phys. Lett. B",
    volume = "871",
    pages = "140010",
    year = "2025"
}

@article{Eberhardt:2021jvj,
    author = "Eberhardt, Lorenz",
    title = "{Summing over Geometries in String Theory}",
    eprint = "2102.12355",
    archivePrefix = "arXiv",
    primaryClass = "hep-th",
    doi = "10.1007/JHEP05(2021)233",
    journal = "JHEP",
    volume = "05",
    pages = "233",
    year = "2021"
}

@article{Cordova:2022rer,
    author = "Cordova, Clay and Ohmori, Kantaro and Rudelius, Tom",
    title = "{Generalized symmetry breaking scales and weak gravity conjectures}",
    eprint = "2202.05866",
    archivePrefix = "arXiv",
    primaryClass = "hep-th",
    doi = "10.1007/JHEP11(2022)154",
    journal = "JHEP",
    volume = "11",
    pages = "154",
    year = "2022"
}

@article{Harlow:2018tng,
    author = "Harlow, Daniel and Ooguri, Hirosi",
    title = "{Symmetries in quantum field theory and quantum gravity}",
    eprint = "1810.05338",
    archivePrefix = "arXiv",
    primaryClass = "hep-th",
    doi = "10.1007/s00220-021-04040-y",
    journal = "Commun. Math. Phys.",
    volume = "383",
    number = "3",
    pages = "1669--1804",
    year = "2021"
}

@article{Maldacena:2004rf,
    author = "Maldacena, Juan Martin and Maoz, Liat",
    title = "{Wormholes in AdS}",
    eprint = "hep-th/0401024",
    archivePrefix = "arXiv",
    reportNumber = "ITFA-2003-57",
    doi = "10.1088/1126-6708/2004/02/053",
    journal = "JHEP",
    volume = "02",
    pages = "053",
    year = "2004"
}

@misc{M-JF-R_wip,
    author = "Johnson-Freyd, Theo and McNamara, Jacob and Reutter, David",
    howpublished = "Work in progress"
}

@article{Hartle:1983ai,
    author = "Hartle, J. B. and Hawking, S. W.",
    editor = "Fang, Li-Zhi and Ruffini, R.",
    title = "{Wave Function of the Universe}",
    reportNumber = "PRINT-83-0937 (CAMBRIDGE)",
    doi = "10.1103/PhysRevD.28.2960",
    journal = "Phys. Rev. D",
    volume = "28",
    pages = "2960--2975",
    year = "1983"
}

@article{Halliwell:1989dy,
    author = "Halliwell, Jonathan J. and Hartle, James B.",
    title = "{Integration Contours for the No Boundary Wave Function of the Universe}",
    reportNumber = "NSF-ITP-89-147",
    doi = "10.1103/PhysRevD.41.1815",
    journal = "Phys. Rev. D",
    volume = "41",
    pages = "1815",
    year = "1990"
}

@article{Witten:2021nzp,
    author = "Witten, Edward",
    title = "{A Note On Complex Spacetime Metrics}",
    eprint = "2111.06514",
    archivePrefix = "arXiv",
    primaryClass = "hep-th",
    month = "11",
    year = "2021"
}

@article{Kontsevich:2021dmb,
    author = "Kontsevich, Maxim and Segal, Graeme",
    title = "{Wick Rotation and the Positivity of Energy in Quantum Field Theory}",
    eprint = "2105.10161",
    archivePrefix = "arXiv",
    primaryClass = "hep-th",
    doi = "10.1093/qmath/haab027",
    journal = "Quart. J. Math. Oxford Ser.",
    volume = "72",
    number = "1-2",
    pages = "673--699",
    year = "2021"
}

@article{Marolf:2022ybi,
    author = "Marolf, Donald",
    title = "{Gravitational thermodynamics without the conformal factor problem: partition functions and Euclidean saddles from Lorentzian path integrals}",
    eprint = "2203.07421",
    archivePrefix = "arXiv",
    primaryClass = "hep-th",
    doi = "10.1007/JHEP07(2022)108",
    journal = "JHEP",
    volume = "07",
    pages = "108",
    year = "2022"
}

@article{Maloney:2007ud,
    author = "Maloney, Alexander and Witten, Edward",
    title = "{Quantum Gravity Partition Functions in Three Dimensions}",
    eprint = "0712.0155",
    archivePrefix = "arXiv",
    primaryClass = "hep-th",
    doi = "10.1007/JHEP02(2010)029",
    journal = "JHEP",
    volume = "02",
    pages = "029",
    year = "2010"
}

@article{Benjamin:2019stq,
    author = "Benjamin, Nathan and Ooguri, Hirosi and Shao, Shu-Heng and Wang, Yifan",
    title = "{Light-cone modular bootstrap and pure gravity}",
    eprint = "1906.04184",
    archivePrefix = "arXiv",
    primaryClass = "hep-th",
    reportNumber = "CALT-TH 2019-020, IPMU19-0086, PUPT-2586",
    doi = "10.1103/PhysRevD.100.066029",
    journal = "Phys. Rev. D",
    volume = "100",
    number = "6",
    pages = "066029",
    year = "2019"
}

@book{EGNO,
  author    = {Etingof, Pavel and Gelaki, Shlomo and Nikshych, Dmitri and Ostrik, Victor},
  title     = {Tensor Categories},
  series    = {Mathematical Surveys and Monographs},
  volume    = {205},
  publisher = {American Mathematical Society},
  address   = {Providence, RI},
  year      = {2015},
  isbn      = {978-1-4704-2024-6}
}

@book{TakesakiII,
  author    = {Takesaki, Masamichi},
  title     = {Theory of Operator Algebras II},
  series    = {Encyclopaedia of Mathematical Sciences},
  volume    = {125},
  publisher = {Springer},
  address   = {Berlin},
  year      = {2003},
  doi       = {10.1007/978-3-662-10451-4}
}

@article{Arkani-Hamed:2007cpn,
    author = "Arkani-Hamed, Nima and Orgera, Jacopo and Polchinski, Joseph",
    title = "{Euclidean wormholes in string theory}",
    eprint = "0705.2768",
    archivePrefix = "arXiv",
    primaryClass = "hep-th",
    doi = "10.1088/1126-6708/2007/12/018",
    journal = "JHEP",
    volume = "12",
    pages = "018",
    year = "2007"
}

@article{DiUbaldo:2023hkc,
    author = "Di Ubaldo, Gabriele and Perlmutter, Eric",
    title = "{AdS3 Pure Gravity and Stringy Unitarity}",
    eprint = "2308.01787",
    archivePrefix = "arXiv",
    primaryClass = "hep-th",
    doi = "10.1103/PhysRevLett.132.041602",
    journal = "Phys. Rev. Lett.",
    volume = "132",
    number = "4",
    pages = "041602",
    year = "2024"
}

@article{Keller:2014xba,
    author = "Keller, Christoph A. and Maloney, Alexander",
    title = "{Poincare Series, 3D Gravity and CFT Spectroscopy}",
    eprint = "1407.6008",
    archivePrefix = "arXiv",
    primaryClass = "hep-th",
    reportNumber = "RUNHETC-2014-13",
    doi = "10.1007/JHEP02(2015)080",
    journal = "JHEP",
    volume = "02",
    pages = "080",
    year = "2015"
}

@article{Benjamin:2020mfz,
    author = "Benjamin, Nathan and Collier, Scott and Maloney, Alexander",
    title = "{Pure Gravity and Conical Defects}",
    eprint = "2004.14428",
    archivePrefix = "arXiv",
    primaryClass = "hep-th",
    doi = "10.1007/JHEP09(2020)034",
    journal = "JHEP",
    volume = "09",
    pages = "034",
    year = "2020"
}

@article{Maxfield:2020ale,
    author = "Maxfield, Henry and Turiaci, Gustavo J.",
    title = "{The path integral of 3D gravity near extremality; or, JT gravity with defects as a matrix integral}",
    eprint = "2006.11317",
    archivePrefix = "arXiv",
    primaryClass = "hep-th",
    doi = "10.1007/JHEP01(2021)118",
    journal = "JHEP",
    volume = "01",
    pages = "118",
    year = "2021"
}

@article{Doplicher:1969fvo,
    author = "Doplicher, Sergio and Haag, Rudolf and Roberts, John E.",
    title = "{Fields, observables and gauge transformations I}",
    journal = "Commun. Math. Phys.",
    volume = "13",
    pages = "1--23",
    year = "1969",
    doi = "10.1007/BF01645267"
}

@article{Doplicher:1969fvg,
    author = "Doplicher, Sergio and Haag, Rudolf and Roberts, John E.",
    title = "{Fields, observables and gauge transformations II}",
    journal = "Commun. Math. Phys.",
    volume = "15",
    pages = "173--200",
    year = "1969",
    doi = "10.1007/BF01645674"
}

@article{Doplicher:1971wk,
    author = "Doplicher, Sergio and Haag, Rudolf and Roberts, John E.",
    title = "{Local observables and particle statistics. 1}",
    doi = "10.1007/BF01877742",
    journal = "Commun. Math. Phys.",
    volume = "23",
    pages = "199--230",
    year = "1971"
}

@article{Doplicher:1973at,
    author = "Doplicher, Sergio and Haag, Rudolf and Roberts, John E.",
    title = "{Local observables and particle statistics. 2}",
    doi = "10.1007/BF01646454",
    journal = "Commun. Math. Phys.",
    volume = "35",
    pages = "49--85",
    year = "1974"
}

@article{Maldacena:1997re,
    author = "Maldacena, Juan Martin",
    title = "{The Large $N$ limit of superconformal field theories and supergravity}",
    eprint = "hep-th/9711200",
    archivePrefix = "arXiv",
    reportNumber = "HUTP-97-A097, HUTP-98-A097",
    doi = "10.4310/ATMP.1998.v2.n2.a1",
    journal = "Adv. Theor. Math. Phys.",
    volume = "2",
    pages = "231--252",
    year = "1998"
}

@article{DiFrancesco:1987ez,
    author = "Di Francesco, P. and Saleur, H. and Zuber, J. B.",
    title = "{Critical Ising Correlation Functions in the Plane and on the Torus}",
    reportNumber = "SACLAY-SPH-T-87-097",
    doi = "10.1016/0550-3213(87)90202-1",
    journal = "Nucl. Phys. B",
    volume = "290",
    pages = "527",
    year = "1987"
}

@article{Itzykson:1986pj,
    author = "Itzykson, C. and Zuber, J. B.",
    title = "{Two-Dimensional Conformal Invariant Theories on a Torus}",
    reportNumber = "SACLAY-PHT-85-019",
    doi = "10.1016/0550-3213(86)90576-6",
    journal = "Nucl. Phys. B",
    volume = "275",
    pages = "580--616",
    year = "1986"
}

@article{Bagger:1988yc,
    author = "Bagger, Jonathan and Nemeschansky, Dennis and Zuber, Jean-Bernard",
    title = "{Minimal Model Correlation Functions on the Torus}",
    reportNumber = "HUTP-88/A051, USC-88/009",
    doi = "10.1016/0370-2693(89)91122-2",
    journal = "Phys. Lett. B",
    volume = "216",
    pages = "320--324",
    year = "1989"
}

@article{Ambrose1945,
  author = {Ambrose, Warren},
  title = {Structure theorems for a special class of {Banach} algebras},
  journal = {Transactions of the American Mathematical Society},
  volume = {57},
  year = {1945},
  pages = {364--386},
  doi = {10.1090/S0002-9947-1945-0013235-8}
}

@article{Barbar:2025vvf,
    author = "Barbar, Ahmed",
    title = "{Automorphism-weighted ensembles from TQFT gravity}",
    eprint = "2511.04311",
    archivePrefix = "arXiv",
    primaryClass = "hep-th",
    month = "11",
    year = "2025"
}

@article{Colin-Ellerin:2020mva,
    author = "Colin-Ellerin, Sean and Dong, Xi and Marolf, Donald and Rangamani, Mukund and Wang, Zhencheng",
    title = "{Real-time gravitational replicas: Formalism and a variational principle}",
    eprint = "2012.00828",
    archivePrefix = "arXiv",
    primaryClass = "hep-th",
    doi = "10.1007/JHEP05(2021)117",
    journal = "JHEP",
    volume = "05",
    pages = "117",
    year = "2021"
}

@article{Held:2026huj,
    author = "Held, Jesse and Kaplan, Molly and Marolf, Donald and Wang, Zhencheng",
    title = "{Axion Wormholes and the AdS/CFT Factorization Problem}",
    eprint = "2601.02507",
    archivePrefix = "arXiv",
    primaryClass = "hep-th",
    month = "1",
    year = "2026"
}

@article{Chen:2025leq,
    author = "Chen, Hong Zhe",
    title = "{Thermodynamic stability from Lorentzian path integrals and codimension-two singularities}",
    eprint = "2501.08409",
    archivePrefix = "arXiv",
    primaryClass = "hep-th",
    doi = "10.1007/JHEP05(2025)180",
    journal = "JHEP",
    volume = "05",
    pages = "180",
    year = "2025"
}

@article{Louko:1995jw,
    author = "Louko, Jorma and Sorkin, Rafael D.",
    title = "{Complex actions in two-dimensional topology change}",
    eprint = "gr-qc/9511023",
    archivePrefix = "arXiv",
    reportNumber = "SU-GP-95-5-1, WISC-MILW-95-TH-16, MDDP-PP-96-40",
    doi = "10.1088/0264-9381/14/1/018",
    journal = "Class. Quant. Grav.",
    volume = "14",
    pages = "179--204",
    year = "1997"
}

@article{Held:2026bbo,
    author = "Held, Jesse and Kaplan, Molly and Marolf, Donald and Wang, Zhencheng",
    title = "{Lorentzian Path Integrals and Jackiw-Teitelboim wormholes with imaginary scalars}",
    eprint = "2601.09932",
    archivePrefix = "arXiv",
    primaryClass = "hep-th",
    month = "1",
    year = "2026"
}

@article{Colafranceschi:2023txs,
    author = "Colafranceschi, Eugenia and Marolf, Donald and Wang, Zhencheng",
    title = "{A trace inequality for Euclidean gravitational path integrals (and a new positive action conjecture)}",
    eprint = "2309.02497",
    archivePrefix = "arXiv",
    primaryClass = "hep-th",
    doi = "10.1007/JHEP04(2024)140",
    journal = "JHEP",
    volume = "04",
    pages = "140",
    year = "2024"
}

@article{Kolanowski:2026gii,
    author = "Kolanowski, Maciej and Marolf, Donald",
    title = "{How to tame your (black hole) saddles: Lessons from the Lorentzian Gravitational Path Integral}",
    eprint = "2603.24681",
    archivePrefix = "arXiv",
    primaryClass = "hep-th",
    month = "3",
    year = "2026"
}

@article{Benini:2022hzx,
    author = "Benini, Francesco and Copetti, Christian and Di Pietro, Lorenzo",
    title = "{Factorization and global symmetries in holography}",
    eprint = "2203.09537",
    archivePrefix = "arXiv",
    primaryClass = "hep-th",
    reportNumber = "SISSA 05/2022/FISI, SISSA 05/2022/FISI",
    doi = "10.21468/SciPostPhys.14.2.019",
    journal = "SciPost Phys.",
    volume = "14",
    number = "2",
    pages = "019",
    year = "2023"
}

@article{Yu:2026gdf,
    author = "Yu, Xingyang",
    title = "{From Baby Universes to Narain Moduli: Topological Boundary Averaging in SymTFTs}",
    eprint = "2605.06653",
    archivePrefix = "arXiv",
    primaryClass = "hep-th",
    month = "5",
    year = "2026"
}

@article{Held:2025mai,
    author = "Held, Jesse and Maxfield, Henry",
    title = "{Gravitational Hilbert spaces: invariant and co-invariant states, inner products, gauge-fixing, and BRST}",
    eprint = "2509.05412",
    archivePrefix = "arXiv",
    primaryClass = "hep-th",
    month = "9",
    year = "2025"
}

@article{Wang:2025bcx,
    author = "Wang, Diandian and Wang, Zhencheng and Wei, Zixia",
    title = "{Wormholes with ends of the world}",
    eprint = "2504.12278",
    archivePrefix = "arXiv",
    primaryClass = "hep-th",
    doi = "10.1007/JHEP09(2025)166",
    journal = "JHEP",
    volume = "09",
    pages = "166",
    year = "2025"
}

\end{document}